\documentclass[american,aps,prx,reprint,superscriptaddress,nofootinbib,longbibliography,floatfix]{revtex4-2}

\usepackage{amsmath,amssymb}

\usepackage{amsthm}

\usepackage{thmtools}
\declaretheorem[name=Theorem]{theorem}

\usepackage{amsfonts}

\usepackage{color}
\usepackage{graphicx}
\usepackage[american]{babel}
\usepackage[utf8]{inputenc}
\usepackage{times}
\usepackage{braket} 				%Dirac-Notation

\usepackage{subcaption}
\usepackage{verbatim}
\usepackage{quantikz}
\usepackage{enumitem}
\usepackage{float}
\usepackage{bbold}
\usepackage{mathrsfs}  
\usepackage{titlesec}
\newcommand{\ketbra}[2]{|#1\rangle\!\langle#2|}
\newcommand{\shortminus}{\mathbin{\scalebox{0.7}[1.0]{$-$}}}

\definecolor{mygrey}{gray}{0.35}
\definecolor{myblue}{rgb}{0.2,0.2,0.8}
\definecolor{myzard}{cmyk}{0,0,0.05,0}
\definecolor{mywhite}{rgb}{1,1,1}
\definecolor{myred}{rgb}{0.9,0.1,0.}
\definecolor{goldenyellow}{rgb}{1.0, 0.87, 0.0}
\definecolor{cornellred}{rgb}{0.7, 0.11, 0.11}
\definecolor{textgreen}{RGB}{25,160,34}
\definecolor{dartmouthgreen}{rgb}{0.05, 0.5, 0.06}
\newcommand{\mR}{\mathcal{R}}
\DeclareMathOperator*{\argmax}{arg\,max}

\DeclareMathOperator{\Prob}{P}
\DeclareMathOperator{\mOO}{mO}
\usepackage[colorlinks=true,citecolor=myblue,linkcolor=myblue,urlcolor=myblue]{hyperref}

\newtheoremstyle{customStyle1}  % name of the style to be used
{0pt}       % measure of space to leave above the theorem. E.g.: 3pt
{0pt}       % measure of space to leave below the theorem. E.g.: 3pt
{\normalfont}   % name of font to use in the body of the theorem
{\parindent}        % measure of space to indent
{\em}  % name of head font
{. --}   	 % punctuation between head and body
{.5em}       % space after theorem head
{\thmname{#1}\thmnumber{ #2}\thmnote{ (#3)}}  % Manually specify head

\newtheorem{proposition}[theorem]{Proposition}

\newtheorem{defin}[theorem]{Definition}

\newtheorem{lem}[theorem]{Lemma}
\newtheorem{lemma}[theorem]{Lemma}
\newtheorem{cor}[theorem]{Corollary}

\usepackage{cancel}
\usepackage{soul}

\newcommand{\id}{{\mathbb{1}}}

\DeclareMathOperator{\CPTP}{CPTP}

\DeclareMathOperator{\idChan}{id}
\DeclareMathOperator{\idChannel}{id}

\DeclareMathOperator{\F}{F}

\newcommand{\cldots}{\mathord{.\kern-0.08em.\kern-0.08em.}}

\let\O\relax 
\DeclareMathOperator{\O}{O}
\DeclareMathOperator{\Pure}{Pure}

\newcommand{\norm}[1]{\left\lVert#1\right\rVert}
\DeclareMathOperator{\Comb}{Comb}

\DeclareMathOperator{\Herm}{Herm}
\DeclareMathOperator{\trace}{Tr}
\newcommand{\Tr}[1]{\trace\left[#1\right]}
\newcommand{\partTr}[2]{\trace_{#1}\left[#2\right]}
\newcommand{\I}{\text{I}}

\DeclareMathOperator{\Dist}{Dist}
\DeclareMathOperator{\Cost}{Cost}
\DeclareMathOperator{\T}{T}
\DeclareMathOperator{\supp}{supp}
\DeclarePairedDelimiter\floor{\lfloor}{\rfloor}
\DeclarePairedDelimiter{\ceil}{\lceil}{\rceil}

\DeclareMathOperator{\cMIO}{cMIO}
\DeclareMathOperator{\cDI}{cDI}
\DeclareMathOperator{\cDIO}{cDIO}

\DeclareMathOperator{\mMIO}{mMIO}
\DeclareMathOperator{\mDI}{mDI}

\DeclareMathOperator{\mFO}{mO}
\DeclareMathOperator{\cO}{cO}
\DeclareMathOperator{\nO}{nO}
\DeclareMathOperator{\mO_1}{mO_1}
\DeclareMathOperator{\cfO}{cfO}
\DeclareMathOperator{\nMIO}{nMIO}
\DeclareMathOperator{\nDI}{nDI}

\newcommand{\Choi}{Choi-Jamiołkowski}

\DeclareMathOperator{\ChoiF}{Choi}
\DeclareMathOperator{\MIO}{MIO}
\DeclareMathOperator{\DIO}{DIO}
\DeclareMathOperator{\DI}{DI}
\newcommand{\In}{\text{in}}
\newcommand{\Out}{\text{out}}
\newcommand{\all}{\text{all}}

\def\mE{\mathcal{E}}

\def\mF{\mathcal{F}}
\def\mN{\mathcal{N}}
\def\mM{\mathcal{M}}

\def\mL{\mathcal{L}}
\def\mP{\mathcal{P}}
\def\mS{\mathcal{S}}
\def\mD{\mathcal{D}}
\def\mT{\mathcal{T}}
\def\mV{\mathcal{V}}
\def\mG{\mathcal{G}}
\def\mQ{\mathcal{Q}}

\def\mU{\mathcal{U}}

\def\mK{\mathcal{K}}

\DeclareMathOperator{\D}{D}

\DeclareMathOperator{\aff}{aff}

\newcommand{\TITEL}{The Structure of Higher-Order Quantum Resources with an Application to Coherence}

\begin{document}

	\author{Felix Ahnefeld}
	\email{felix.ahnefeld@uni-ulm.de}
	\affiliation{Institute of Theoretical Physics, Ulm University, Albert-Einstein-Allee 11, D-89081 Ulm, Germany}
	\author{Thomas Theurer}
	\email{tth@math.ku.dk}
	\affiliation{Department of Mathematical Sciences, University of Copenhagen, Universitetsparken 5, 2100, Denmark}
	\author{Martin B. Plenio}
	\email{martin.plenio@uni-ulm.de}
	\affiliation{Institute of Theoretical Physics, Ulm University, Albert-Einstein-Allee 11, D-89081 Ulm, Germany}
    \affiliation{Center for Integrated Quantum Science and Technology (IQST), D-89081 Ulm, Germany }

    \title{ \TITEL}
    
	\date{\today}
    \begin{abstract}
        We develop a general framework for higher-order quantum resources that treats the resourcefulness of quantum states, channels, and supermaps (also called processes) in a unified manner. By generalizing well-established concepts from resource theories of states and channels, we introduce and study supermap resource monotones and supermap distillation and dilution. As free supermaps, we consider free networks (supermaps that can be implemented with free channels), compatible supermaps (sets of supermaps closed under arbitrary composition), and maximally free superchannels. Specializing to approximate channel-to-channel conversions, we show that under mild assumptions, the conversion error with respect to the diamond distance and the worst-case infidelity coincide. We then apply the general framework to resource theories of dynamical coherence, where we establish the existence of largest sets of compatible supermaps. By characterizing these sets with semidefinite constraints, we show that they are strictly larger than the corresponding free networks but do not contain all maximally free superchannels.  We finish by investigating the operational differences of these sets in one-shot channel distillation and dilution, providing solutions for both tasks in the parallel setting and for dilution in the adaptive setting.
    \end{abstract}
	\maketitle

\section{INTRODUCTION}
Entanglement, coherence, asymmetry, athermality, and non-stabilizerness~\cite{Bennett1996,Horodecki2009,Aberg2006,Baumgratz2014,Gour2008,Brandao2013,Bravyi2005,Veitch2014} are examples of quantum resources that provide advantages in communication~\cite{Bennett1992,Bennett1993,LipkaBartosik2020,Korzekwa2022}, computation~\cite{Bruss2011,Howard2014,Shi2017,Ahnefeld2022,Naseri2022,Ahnefeld2025},  sensing~\cite{Giovannetti2006,Pezze2009,Tan2021,Lecamwasam2024,Ahnefeld2025,Chen2026}, and other applications. Quantum resource theories provide a rigorous framework for describing, quantifying, and manipulating such resources. Historically, state-based, or \textit{static}, resource theories~\cite{Chitambar2019, Gour2025}, which focus on the resources contained in quantum states, were developed first. However, the primary goal of quantum technologies is to perform tasks, e.g., to transmit information or to solve a computational problem. Ultimately, this is all achieved by quantum operations, which motivated the development of channel-based, or \textit{dynamical}, resource theories, in which quantum channels themselves are treated as resources that can be quantified and manipulated~\cite{Theurer2019,Liu2019arxiv,Liu2020,Liu2019,Gour2019,Regula2021b,Regula2021c,Kim2021,Gour2020,Saxena2020,Gour2020b,Gour2021b,Luo2025,Berk2021}. This approach includes and generalizes static resource theories in the sense that not all dynamical resource theories can be reduced to static resource theories~\cite{Theurer2019}.

Formulating such theories requires specifying how the available channel-resources may be processed. Initially, this was approached by defining free superchannels, i.e., maps from quantum channels to quantum channels. An operational choice is to call the superchannels that can be implemented by a free pre- and postprocessing free~\cite{Theurer2019, Liu2019arxiv, Liu2020}, while the most permissive choice is the so-called maximally free superchannels~\cite{Regula2021b, Regula2021c, Kim2021}, which guarantee the minimal requirement that free channels are mapped to free channels. Superchannels that also guarantee this when applied to subsystems are called completely resource non-generating~\cite{Gour2020,Gour2020b,Saxena2020,Gour2021b}, and all three choices allow building resource theories well-suited to describe, for example, channel-to-channel conversions. However, multiple channel-resources can be used not only in parallel but also sequentially, with intermediate processing and quantum memory, as formalized by so-called quantum supermaps or processes~\cite{Chiribella2008,Chiribella_2008,Chiribella2009, Taranto2025}. When considering multi-copy settings, it is thus insufficient to only define free superchannels. Instead, one must define free supermaps -- and this is exactly what we do in this work. As we will show, this not only allows us to consider multi-copy channel settings, but also to build higher-order quantum resource theories which treat general quantum objects, including states, channels, and supermaps in a unified manner.

\begin{figure*}[ht]
    \centering
    \includegraphics[width=1\linewidth]{ 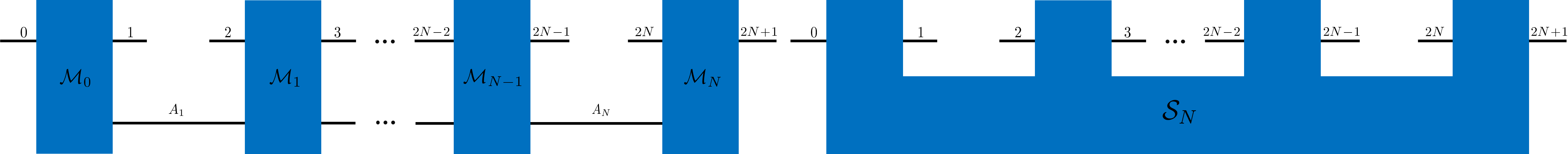}
    \caption{A quantum network comprised of channels $\mathcal{M}_0,\ldots, \mathcal{M}_N$ as depicted on the left side is a specific implementation of a supermap $\mathcal{S}_N$ shown on the right.}
    \label{fig:NetworkvsComb}
\end{figure*}

As a core concept, we introduce \textit{compatible} sets of free supermaps, i.e., sets of supermaps that are closed under arbitrary compositions. This encapsulates the golden rule of resource theories that resources cannot be created for free. If a set of free channels is fixed and satisfies a few natural assumptions, an operationally well-motivated set of compatible free supermaps is again the set of supermaps implementable by free channels, which we call free networks. However, we stress again that our framework allows us to handle general resource theories that are not reducible to channels. Generalizing notions from dynamic and static resource theories, we then define resource monotones and one-shot conversions of supermaps, including distillation and dilution.

After introducing the general framework, we return to approximate channel-to-channel conversion under maximally free superchannels and extend the error bounds of Ref.~\cite{Regula2021b} from worst-case fidelity to diamond distance. For isometric target channels satisfying suitable resource-theoretic conditions, these bounds become tight and establish an exact equality between the optimal diamond distance and the optimal worst-case infidelity. Moreover, they establish bounds on the conversion error under the more restrictive sets of free networks and compatible supermaps. 

As an application, we study general resource theories derived from dynamical coherence theories. For creation-incoherent~\cite{Aberg2006,Diaz2018}, detection-incoherent~\cite{Theurer2019}, and dephasing-covariant channels~\cite{Chitambar2016b,Marvian2016,Chitambar2017,Meznaric2013}, we prove the existence of largest sets of compatible supermaps. By characterizing these sets with semidefinite constraints, we show that they are strictly larger than the corresponding free networks but do not contain all maximally free superchannels. We conclude by investigating the operational differences between these sets in one-shot channel distillation and dilution, providing solutions for both tasks in the parallel setting and for dilution in the adaptive setting.

\section{NOTATION AND PRELIMINARIES}\label{sec:Notation}
In this work, we restrict ourselves to finite-dimensional quantum systems. If we deal with few systems, we label them with capital Roman letters such as $A,B$, and when we need to enumerate many systems for convenience, we label them by $j$ instead of $A_j$. Throughout this work, we use both sub- and superscripts to denote the systems an operator acts on. If we require matrix exponents, we explicitly state this. Quantum states will be denoted by small Greek letters, and the set of all quantum states is denoted by $\D$. Quantum channels, also called quantum operations, i.e., linear, completely positive, and trace-preserving maps that transform quantum states, will be denoted by calligraphic Latin letters, with the exception of the identity channel, which is denoted by $\idChannel$. If required, we explicitly denote the systems on which a quantum channel acts as $\smash{\mathcal{M}^{A\to B}}$ to emphasize that $\mathcal{M}$ is a channel from system $A$ to $B$, i.e., from the set of linear operators $\mathcal{L}(\mathcal{H}_A)$ to $\mathcal{L}(\mathcal{H}_B)$.  The set of all quantum channels from $A$ to $B$ is denoted as $\CPTP(A\to B)$. Quantum supermaps, i.e., linear maps that act on (up to) $N$ quantum channels, are denoted by calligraphic letters too, e.g., $\mathcal{S}_N$. Throughout this work, we extensively use the {\Choi} representation of quantum channels and supermaps. The (unnormalized) Choi state~\cite{Choi1975} of a linear map (and particularly a quantum channel) $\mathcal{N}^{A\to B}$  is defined as
\begin{align}\label{eq:ChoiState}
    J_\mathcal{N}^{AB}&:= d_A\left(\idChannel^A \otimes\, \mathcal{N}^{ \Tilde{A} \to B} \right) \ketbra{\phi}{\phi}_{A,\tilde{A}} \nonumber\\
    &= \sum_{n,m} \ketbra{n}{m}_A \otimes \mathcal{N}^{ \Tilde{A}\to B}\left(\ketbra{n}{m}_{\Tilde{A}} \right),
\end{align}
where $\tilde{A}\cong A$ is a copy of the input system $A$ of dimension $d_A$, and $\phi$ denotes a maximally entangled state. The action of such a linear map in the Choi representation can be expressed as
\begin{align}
    \mathcal{N}^{A\to B}(X_A):=\partTr{A}{(X_A^T\otimes \id_B)J_\mathcal{N}^{AB}},
\end{align}
where $T$ denotes the transpose in the same basis as in Eq.~\eqref{eq:ChoiState}. This provides a one-to-one correspondence between linear maps and linear operators known as the {\Choi} isomorphism~\cite{Jamiolkowski1972, Choi1975}. Moreover, $\mathcal{N}$ is completely positive iff $J_\mathcal{N} \geq 0$, and $\mathcal{N}^{A\to B}$ is trace-preserving iff $\partTr{B}{J_\mathcal{N}^{AB}}=\id_A$.
Whenever we speak of a Choi ``state'' of a channel, we refer to an unnormalized state as in Eq.~\eqref{eq:ChoiState}. If clear from the context, we use the Choi state of a channel $J_\mN$ and the channel $\mN$ itself interchangeably. We also omit system labels for operators, channels, and superchannels if there is no risk of confusion.

\subsection{Quantum supermaps}\label{sec:Supermaps}
Next, primarily following Ref.~\cite{Chiribella2009}, we briefly review how the Choi representation of quantum channels can be extended to quantum supermaps. Throughout this work, we restrict ourselves to quantum processes with a \textit{definite causal order}; for more details on quantum processes with indefinite causal order, see, e.g., Ref.~\cite{Taranto2025}. 

The left-hand side of Fig.~\ref{fig:NetworkvsComb} depicts a quantum network composed of quantum channels $\mathcal{M}_0,\ldots, \mathcal{M}_N$. For $0\leq j\leq N$, the systems $2j$ are accessible inputs to the network, the systems $2j+1$ are accessible outputs, and the systems $A_{j+1}$ are inaccessible auxiliary systems. Importantly, different networks can lead to the same effective dynamics on the accessible input and output systems: On the auxiliary system $A_j$ between two successive channels $\mathcal{M}_j$ and $\mathcal{M}_{j+1}$, one can, for example, insert any unitary followed by its inverse and absorb the unitaries into new channels $\mathcal{M}_j^\prime$ and $\mathcal{M}_{j+1}^\prime$ without changing the accessible dynamics. 
All networks leading to the same effective, accessible dynamics are described by a single supermap $\mathcal{S}_N$. To depict it, one uses the serrated object on the right-hand side of Fig.~\ref{fig:NetworkvsComb}, where the inaccessible auxiliary systems are hidden. A supermap $\mS_N$ of order $N$ (with definite causal order) acting on (up to ) $N$ quantum channels $\mathcal{N}_1,\ldots, \mathcal{N}_N$ can thus always be implemented by a quantum network consisting of a set of channels $\mathcal{M}_0,\ldots, \mathcal{M}_N$ in the sense that~\cite{Chiribella2009}
\begin{align}
    &\mathcal{S}_N[\mathcal{N}_1,\ldots,\mathcal{N}_N]\\
    &:= \mathcal{M}_{N}(\idChannel \otimes \mathcal{N}_N)\mathcal{M}_{N\!-\!1}  \ldots 
\mathcal{M}_1 (\idChannel \otimes \mathcal{N}_1)\mathcal{M}_0, \nonumber
\end{align}
where we suppressed all system indices. To the supermap $S_N$ implemented by this network, one can now assign an analogue of the Choi state, which is called a quantum comb: For each of the channels $\mathcal{M}_j$, let $J_{\mathcal{M}_j}$ denote its Choi state as defined in Eq.~\eqref{eq:ChoiState}.  Using the system labels from Fig.~\ref{fig:NetworkvsComb}, recall that this is an unnormalized state on the composite system $2j, 2j+1, A_j, A_{j+1}$. The comb $J_{\mathcal{S}_N}$ associated with the network is then defined as
\begin{align}\label{eq:NetworkLinkProduct}
J_{\mathcal{S}_N}^{0,\ldots,2N+1}
&:=\trace_{A_1,\ldots,A_N}
\Bigl[
\bigl(\id \otimes (J_{\mathcal{M}_0}
)^{T_{A_1}}\bigr)
\nonumber\\
&\qquad
\bigl(\id \otimes (J_{\mathcal{M}_1})^{T_{A_2}}\bigr)
\cdots \bigl(\id \otimes J_{\mathcal{M}_N}\bigr)\Bigr],
\end{align}
where $T_X$ denotes the partial transpose on system $X$. 
Importantly, two networks implement the same supermap iff their combs are the same~\cite{Chiribella2008,Chiribella_2008,Chiribella2009}, which justifies writing $ J_{\mathcal{S}_N}$ in Eq.~\eqref{eq:NetworkLinkProduct}. An operator $J_{\mS_N}$ that is a comb, as in Eq.~\eqref{eq:NetworkLinkProduct}, satisfies certain constraints that arise from being implemented by a sequence of quantum channels; in particular, for a $J_{\mathcal{S}_N} \in \mathcal{L} ( \otimes_{j=0}^{2N+1} \mathcal{H}_j )$, there exists a sequence of positive semidefinite operators $J_{\mathcal{S}_N}^{(j)} \in \mathcal{L} ( \otimes_{k=0}^{2j+1} \mathcal{H}_k) $ with $0\leq j \leq N$  and $J_{\mathcal{S}_N}=J_{\mathcal{S}_N}^{(N)}$ satisfying
\begin{subequations}\label{eq:TPcondition}
    \begin{align}
    & \partTr{2j+1}{J_{\mathcal{S}_N}^{(j)}}= J^{(j-1)}_{\mathcal{S}_N} \otimes \id_{2j} \quad \forall\, 1\leq j \leq N,  \\
    & \partTr{1}{J_{\mathcal{S}_N}^{(0)}}=\id_0.
\end{align}
\end{subequations}
In essence, this formalizes that discarding the output system of $\mM_N$ implies that its accessible input system is irrelevant - one could thus directly trace out the (accessible and inaccessible) input of $\mM_N$ and consider a new network with one fewer channel - and then do this recursively. As shown in Refs.~\cite{Chiribella2008,Chiribella2009}, these conditions are not only \textit{necessary} for $J_{\mS_N}$ to be the comb representation of a deterministic supermap $\mathcal{S}_N$ with definite causal order, but also \textit{sufficient}. Therefore, Eqs.~\eqref{eq:TPcondition}, together with the positive semidefinite constraint, characterize all deterministic supermaps with definite causal order. These conditions are semidefinite constraints and describe the action of a supermap via a single object rather than a product of (Choi representations of) channels. As intuitive from the network representation, from here on, 
we will use the convention that for a given supermap $\mS_N$, the system labels denote a specific quantum system at a specific (time) step. The causal order is thus fixed by the system labels, and each label can only appear once.
We denote the set of all quantum combs of order $N$ (with arbitrary finite-dimensional input and output systems) as
\begin{align}\label{eq:CombSet}
\Comb_N\!:=\! \Bigl\{&\!J_{\mS_N}\!:\! J_{\mS_N}\!=\!J_{\mS_N}^{(N)}\geq 0, \partTr{1}{J_{\mS_N}^{(0)}}\!=\!\id_0,\\ 
&\partTr{2j+1}{J_{\mS_N}^{(j)}}\!=\!J_{\mS_N}^{(j-1)}\!\otimes\!\id_{2j}\ \forall j: 1\leq j\leq N\! \Bigr\},\nonumber
\end{align} 
where we use the convention that for every $1\le j\le N$, not \textit{both} systems $2j$ and $2j-1$ are trivial, i.e., $\Comb_N$ contains combs with precisely $N$ ``input slots". The corresponding sets of supermaps are defined as
\begin{align}
    \mathfrak{S}_N := \left\{ \mS_N: J_{\mS_N} \in \Comb_N \right\},
\end{align}
and $\mathfrak{S}:=\bigcup_{N\in \mathbb{N}_0} \mathfrak{S}_N$, where $\mathfrak{S}_0 :=\CPTP$ is the set of quantum channels and $\mathfrak{S}_1$ the set of superchannels. With this choice of convention, we have that the sets of supermaps of distinct orders are disjoint, i.e.,  $\mathfrak{S}_N \cap \mathfrak{S}_M =\emptyset $ if $N\neq M$.

\begin{figure}[ht]
    \centering

    \begin{subfigure}{\columnwidth}
        \centering
        \scalebox{0.6}{\includegraphics[width=\linewidth]{ 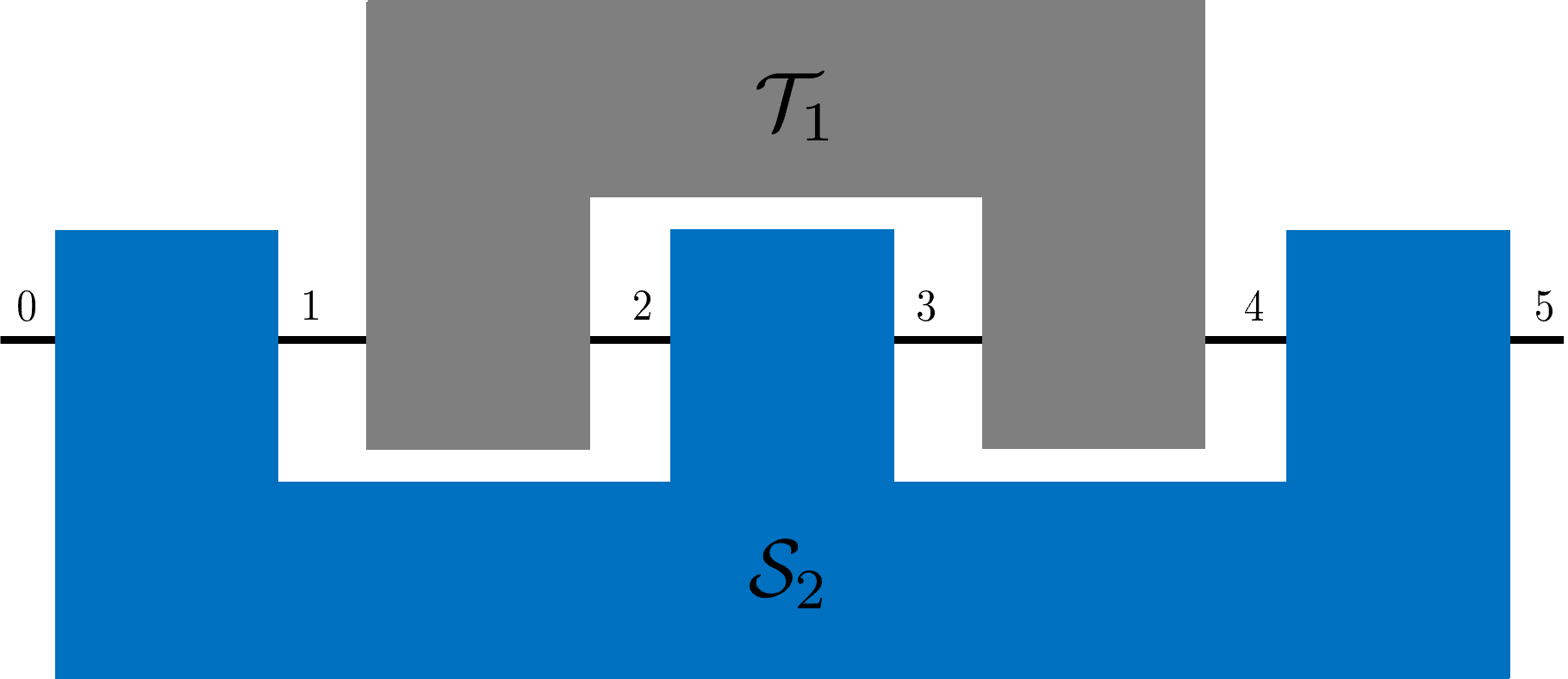}}
        \caption{Composing the supermaps $\mS_2$ and $\mT_1$ results in a channel.}
        \label{fig:first}
    \end{subfigure}
    
    \vspace{0.5em}

        \begin{subfigure}{\columnwidth}
        \centering
        \includegraphics[width=\linewidth]{ 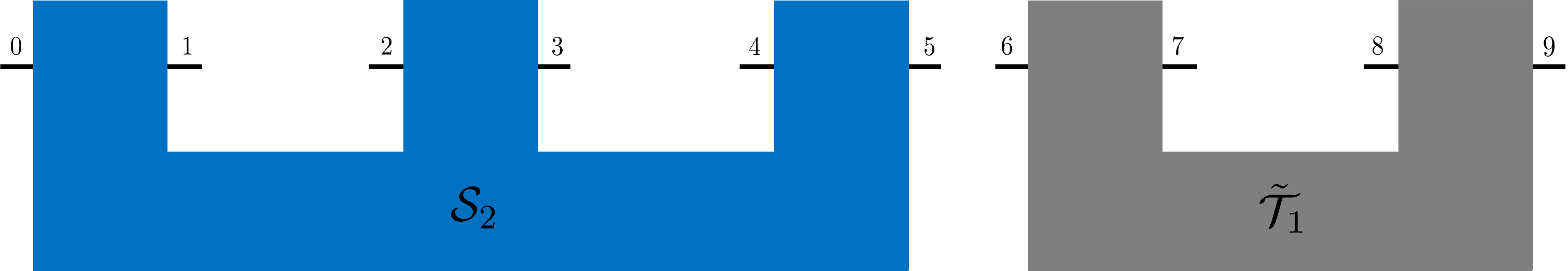}
        \caption{Composing the supermaps $\mS_2$ and $\tilde{\mT}_1$ results in a supermap.}
        \label{fig:second}
    \end{subfigure}

        \vspace{0.5em}

    \begin{subfigure}{\columnwidth}
        \centering
        \scalebox{0.72}{\includegraphics[width=\linewidth]{ 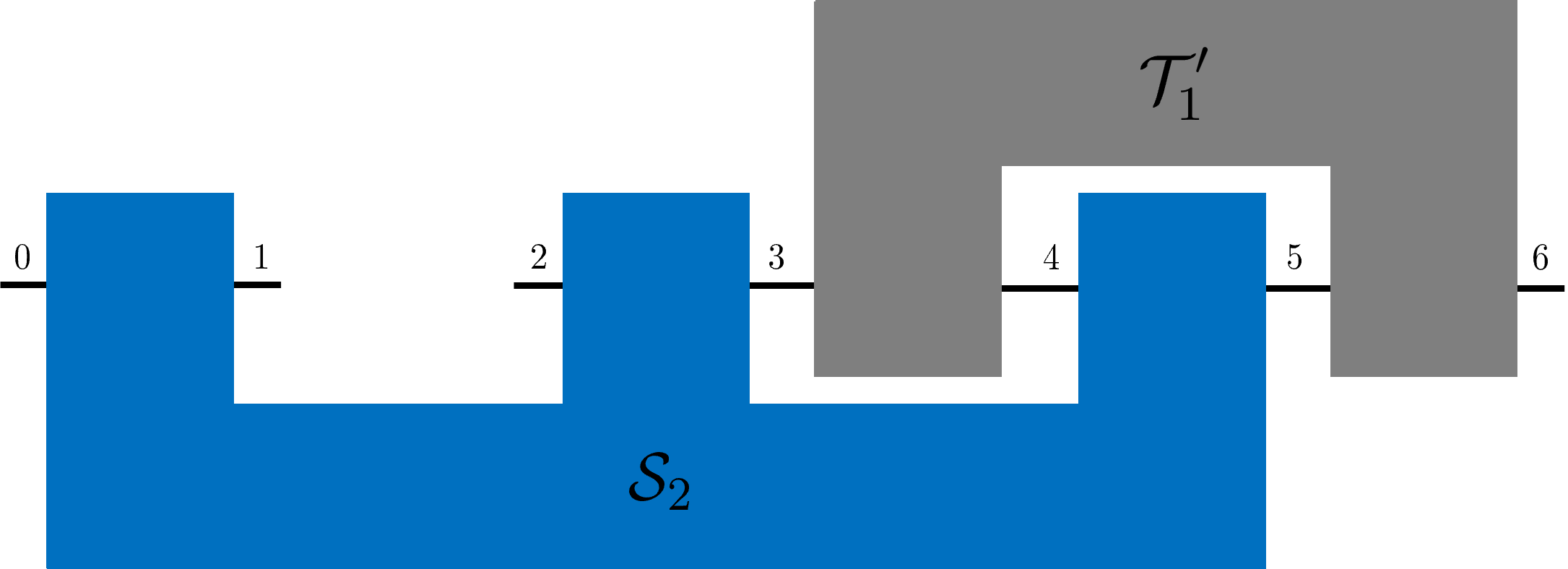}}
        \caption{Composing the supermaps $\mS_2$ and $\mT_1^\prime$ results in a superchannel.}
        \label{fig:third}
    \end{subfigure}

    \caption{Composing quantum supermaps. A quantum supermap $\mS_2$ is composed with $\mT_1$, $\tilde{\mT}_1$, and $\mT_1'$ by connecting all input and output systems that appear in both supermaps. It is always assumed that this is consistent; see the main text for details.}
    \label{fig:Composition}
\end{figure}

Next, consider two supermaps $\mS_N, \mT_M \in \mathfrak{S}$, let $N_{\In}$ and $N_{\Out}$ denote the set of all input and output systems of $\mS_N$, respectively, and let $N_{\all}=N_{\In} \cup N_{\Out}$. Analogously, define the sets $M_{\In}, M_{\Out}$ and $M_{\all}$ for the supermap $\mT_M$. If  
\begin{enumerate}[label=\roman*)]
    \item $N_{\In}\cap M_{\In}=\emptyset$ and $N_{\Out}\cap M_{\Out}=\emptyset$, and 
    \item the two supermaps have a consistent causal order, i.e., if $A,B \in N_{\all}\cap M_{\all}$, then either $A$ is causally prior to $B$ or  $B$ is causally prior to $A$ in \textit{both} supermaps,
\end{enumerate}
we can compose $\mS_N$ and $\mT_M$ in a well-defined manner by interlacing them as illustrated in Fig.~\ref{fig:Composition}.  i.e., we ``connect" all systems with the same labels, while systems with distinct labels remain accessible. Here, condition i) is due to our convention that each system label denotes a quantum system at a specific (time) step, and ii) due to our assumption of an underlying definite causal order. The comb representing the resulting supermap $\mR_K$ is then obtained via the so-called link product~\cite{Chiribella2009} defined as
\begin{align}\label{eq:defLinkProduct}
    J_{\mR_K}\!&= \!\partTr{ M_{\all} \cap N_{\all}}{\big(\mkern-1mu\id_{ N_{\all} \mkern-1mu\backslash\mkern-1mu  M_{\all} }\!\otimes \!J_{\mT_M}^{T_{ M_{\all}\cap N_{\all}}}\mkern-1mu\big)\mkern-1mu\big(\mkern-1muJ_{\mS_N}\!\otimes \!\id_{ M_{\all} \backslash  N_{\all}}\mkern-1mu\big)} \nonumber \\
    &=: J_{\mS_N}*J_{\mT_M},
\end{align}
where we recall that $T_X$ denotes the partial transpose on system $X$. We emphasize that if i) and ii) are satisfied, which we will implicitly assume from now on in any link product, $\mR_K$ is again a supermap with a consistent and definite causal order in which any system label appears at most once. We also denote compositions of supermaps simply as $\mR_K=\mS_N*\mT_M$, by which we mean that the comb representation of $\mR_K$ is given by the link product in Eq.~\eqref{eq:defLinkProduct}. As seen in Fig.~\ref{fig:second}, the order $K$ of the supermap $\mR_K$ is bounded between $0\leq K\leq N+M+1$.

\subsection{Distances for quantum channels and supermaps}
To assess errors in transformations of quantum states, channels, and supermaps, we require distance measures that quantify differences between such objects. For quantum states $\rho$ and $\sigma$, the  trace distance $\frac{1}{2}\norm{\rho-\sigma}_1$ induced by the trace norm $\norm{X}_1= \text{Tr}[ \smash{\sqrt{X^\dagger X}}]$  is such a quantity. A closely related quantity is the (squared) Uhlmann fidelity $F(\rho,\sigma)= || \smash{\sqrt{\rho}} \smash{\sqrt{\sigma}}||_1^2$. In the following, we will consider different well-known extensions of these quantities to the level of channels: the diamond distance induced by the diamond norm $\norm{\cdot}_\diamond$~\cite{Kitaev1997,Aharonov1998,Kitaev2002,Watrous2009,Watrous2013,Watrous2018}, the entanglement-assisted worst-case channel fidelity $F_{\min}$~\cite{Belavkin2005,Gilchrist2005}, and the (normalized) Choi state fidelity $F_{\ChoiF}$~\cite{Raginsky2001,Nielsen2002,Gilchrist2005}. For quantum channels $\mN^{A\to B}$ and $\mM^{A\to B}$, they are defined as
\begin{align}
    &\tfrac{1}{2}\norm{\mN\!-\!\mM}_{\diamond} \!:=\! \tfrac{1}{2}\sup_{\rho_{AR}}\! \norm{\left(\idChan_R \otimes \mN\!-\!\idChan_R \otimes \mM\right)(\rho_{AR})}_1 \!, \label{eq:def_DiamondDistance}\\
    &\smash[b]{F_{\min}\mkern-1mu(\mkern-2mu\mN,\mkern-2.5mu\mM\mkern-2mu)} \!:= \mkern-5mu\inf_{\rho_{AR}}\! F\Big(\mkern-4mu \idChan_R \otimes \mN(\mkern-1mu\rho_{AR}\mkern-1mu),\mkern-2mu\idChan_R \otimes\mM(\mkern-1mu\rho_{AR}\mkern-1mu) \mkern-4mu\Big), \label{eq:WorstCaseFidelity}\\
    &F_{\ChoiF}(\mN,\mM) \!:=\! \tfrac{1}{d_{A}^2}F\left(J_{\mN}^{AB},J_{\mM}^{AB}\right) \label{eq:ChoiStateFidelity},
\end{align}
where the infimum and supremum over $\rho_{AR}$ are to be understood as an optimization over bipartite states with an auxiliary system $R$ of arbitrary dimension. For finite-dimensional systems, it is sufficient to consider auxiliary systems $R\cong A$, see, e.g., Refs.~\cite{Belavkin2005, Puzzuoli2017}. The diamond distance and the entanglement-assisted worst-case fidelity quantify complementary operational notions of distinguishability and similarity for quantum channels. The diamond distance determines the optimal bias in single-shot discrimination of two channels when arbitrary entangled inputs and joint measurements are allowed. Conversely, the worst-case fidelity evaluates the smallest output fidelity achievable over all (possibly entangled) input states, thereby quantifying the degree to which two channels are similar. This notion can be made quantitative via a channel-based analogue of the Fuchs-van de Graaf inequality~\cite{Fuchs1999}. In particular, the worst-case fidelity and the diamond distance bound one another in direct analogy with fidelity and trace distance for quantum states, i.e., 
\begin{equation}\label{eq:FuchsVanDeGraaf_channels}
    1\!-\!\sqrt{\smash[b]{F_{\min}(\mkern-2mu\mN,\mkern-2.5mu\mM\mkern-2mu)}} \!\leq \! \tfrac{1}{2} \norm{\mN\!-\!\mM}_\diamond \!\leq\! \sqrt{1\!-\!\smash[b]{F_{\min}(\mkern-2mu\mN,\mkern-2.5mu\mM\mkern-1mu)}}.
\end{equation}
Thus, a small diamond distance implies high worst-case fidelity, while low worst-case fidelity implies a large operational distinguishability between the channels. As we show in Appendix~\ref{sec:Appen_GenDiaDist}, the lower bound in Eq.~\eqref{eq:FuchsVanDeGraaf_channels} can be sharpened if $\mN$ maps pure states to pure states in a complete sense, i.e., if $\idChan \otimes\,\mN(\psi)$ is pure for every pure $\psi$. This is equivalent to $\mN$ being an isometry in the sense that $\mN(\rho)=V\rho V^\dagger$ for some isometry $V$, see Ref.~\cite[Thm.~3.1]{Davies1976} and Ref.~\cite[Lem.~3.2]{Belzig2025}. Note that this also includes $\mN$ being a pure-state preparation channel, that is, an isometry from a one-dimensional system to a fixed pure state. For an isometric channel $\mV$, the sharpened Fuchs-van de Graaf inequalities are
\begin{align}\label{eq:FuchsVanDeGraaf_channels_sharp}
    1\!-\!F_{\min}(\mV,\mM)  \!\leq \! \tfrac{1}{2} \norm{\mV\!-\!\mM}_\diamond \! \leq\! \sqrt{1\!-\! F_{\min}(\mV,\mM)}.
\end{align}
Moreover, for any channel $\mM$ and any isometric channel $\mV$, 
\begin{align}\label{eq:InequalityChain} 
     1\!-\!F_{\min}(\mM,\mV) \!\geq \!1\!-\!F_{\ChoiF}(\mM,\mV),
\end{align}
which follows from choosing a maximally entangled state $\phi$ in Eq.~\eqref{eq:WorstCaseFidelity}, and if one of the two channels $\mM,\mN$ is an isometry, then $F_{\ChoiF}(\mN,\mM)=\tfrac{1}{d_A^2}\Tr{J_{\mN}^{AB}J_{\mM}^{AB}}$. The Choi state fidelity is closely related to the average-case fidelity over uniform pure states, see Refs.~\cite{Horodecki1999,Nielsen2002}

Throughout this work, we will consider conversion distances between quantum channels quantified by the diamond distance, worst-case fidelity, and Choi state fidelity. 
The generalized diamond distance between two quantum supermaps $\mS_N,\mT_N \in \mathfrak{S}$ defined as
\begin{align}\label{eq:diamond_generalized_main}
    d_{\diamond,N} (\mS_N,\mT_N) \!:= \!\min\big\{ &\lambda\!\geq\! 0: \! J_{\mS_N}\!-\!J_{\mT_N}\!\leq\! \lambda J_{\mM_N}, \nonumber \\
    &J_{\mM_N} \!\in\! \Comb_N \big\}
\end{align}
is an extension of the diamond distance to quantum supermaps~\cite{Chiribella2009,Gutoski2012}. The generalized diamond distance is contractive under compositions with quantum supermaps, see Lem.~\ref{def:GenDiamondDist} in Appendix~\ref{sec:Appen_GenDiaDist}, where we also discuss its interpretation and show how it can be evaluated via a semidefinite program (SDP)~\cite{Vandenberghe1996,Boyd2004}.

\section{CONSTRUCTING HIGHER-ORDER RESOURCE THEORIES}\label{sec:FreeSupermaps} 

Starting from a set of free channels $\O$ (also called free operations), we show in this section how to construct consistent higher-order resource theories. We start with free superchannels and then extend our definitions to free supermaps of arbitrary order. As defined in the last section, let $\mathfrak{S}_N$ denote the set of all deterministic quantum supermaps of order $N$ with arbitrary finite-dimensional input and output systems. Moreover, recall that a superchannel $\mS^{(A\to B)\to (C\to D)} \in \mathfrak{S}_1$ can be used to transform a channel $\mN^{ A\to B}$ into another channel $\mM^{C\to D}=S[\mN]=S*\mN$ by connecting systems $A$ and $B$ as shown on the right-hand side of Fig.~\ref{fig:freeSuperchannelsSets}. 

To define free supermaps, we start by fixing a set of free channels\footnote{Going back even one step further, the zero step would be to fix the systems we are considering. In this work, for the sake of a clean presentation, we consider all finite-dimensional systems. One can, however, extend the framework to include infinite-dimensional systems, and it is straightforward to consider only subsets of the finite-dimensional systems. In this case, a natural choice is to demand that if systems $A$ and $B$ are considered, then, also $A\otimes B$ is considered.
}.
Recall that $\CPTP(A \to B)$ denotes the set of all quantum channels from quantum systems $A$ to $B$, and let $\CPTP:=\cup_{A,B}\CPTP(A \to B)$ denote the set of all quantum channels with arbitrary finite-dimensional input and output systems. Let $\O(A\to B)\subseteq \CPTP(A\to B)$ denote a non-empty set of free channels from $A$ to $B$, and let $\O:=\cup_{A,B} O(A\to B)$~\footnote{The assumption that for all $A,B$, $\O(A\to B)$ is non-empty is again for the sake of a cleaner presentation.}. From here on, we use the shorthand notation $\mM^{A\to B} \in \O$ to denote that $\mM^{A\to B}$ is contained in $\O(A\to B)$. Standard assumptions on the set of free channels $\O$ are
\begin{enumerate}[label=\alph*), series=list1]
    \item The identity channels $\idChan$ are free, i.e., for all $A$, $\idChan^A \in \O(A\to A)$. \label{assump:identity}
    \item The set $\O$ is closed under sequential compositions, i.e., $\forall A,B,C$ and $ \mN^{A\to B},\mM^{B\to C} \in \O$, it holds that  $\,\mM^{B\to C}\circ \mN^{A\to B}  \in \O(A\to C)$. \label{assump:Composition}
\end{enumerate} 
Assumption~\ref{assump:identity} expresses that leaving a system unchanged requires no resources. Assumption~\ref{assump:Composition} guarantees that sequentially composing free operations cannot generate resources, thereby capturing the intuition that doing something free twice in a row is still free.

In static resource theories, for a given set of free states, the choice of free channels is generally not unique; see, e.g., Refs.~\cite{Chitambar2019, Gour2025}. The same applies to the choice of free superchannels when considering a fixed set of free channels. In this work, we consider three classes of free superchannels, as well as their generalizations to free supermaps, namely \textit{free networks}, \textit{completely free superchannels}, and \textit{maximally free superchannels}, which we define now. The different defining features are illustrated in Fig.~\ref{fig:freeSuperchannelsSets}. A minimal requirement~\cite{Gour2020,Regula2021b, Regula2021c} for a superchannel $\mS^{(A\to B)\to (C\to D)} \in \mathfrak{S}_1$ to be considered free is that  $\mS^{(A\to B)\to (C\to D)}[\mM^{A\to B}] \in \O(C\to D)$ for all $\mM^{A\to B}\in \O$. This means that if one applies a free superchannel to a free channel, as shown on the right-hand side of Fig.~\ref{fig:freeSuperchannelsSets}, the resulting channel is also free, and it is thus impossible to generate resources at no cost in this manner. The set of all superchannels satisfying this property is called the set of maximally free superchannels, denoted by
\begin{align}\label{eq:maximally_free_Superchannels}
    \mFO_1((A\mkern-2mu\to \mkern-2mu B)\mkern-2mu\to &(C\mkern-2mu\to\mkern-2mu D)):= \Big\{ \mS^{(A\to B)\to (C\to D)} \in \mathfrak{S}_1: \nonumber \\
    &\mS*\mM \in  \O \ \forall \mM\in \O(A\to B)\Big\}, 
\end{align}
and we again write $\mFO_1$ if we consider the union over arbitrary finite-dimensional input and output systems. 

In static resource theories, whenever one intends to talk about the resource content of composite systems, a typical convention is to demand that the set of free channels is completely free~\cite{Chitambar2019, Gour2025}, i.e., 
\begin{enumerate}[label=\alph*), resume=list1]
    \item If $\mM^{A\to B}\in \O(A\to B)$, then $\idChan^R \otimes \,\mM^{A\to B} \in \O(AR\to B R)$, \label{assumpt:complfree}
\end{enumerate}
which ensures that if one applies a free operation to a subsystem of a free state, the resulting state remains free. A further common assumption is that
\begin{enumerate}[label=\alph*), resume=list1]
    \item For all $A$ and $B$, the set $\O(A\to B)$ is topologically closed and convex. \label{assump:ClosedConvex}
\end{enumerate}
This ensures that (classical) mixing does not increase the resources at hand and that operations which can be approximated arbitrarily well by free operations are also coined as free. This assumption is also convenient from a technical perspective.

The notion of complete resource non-generation can be extended to channel transformations via superchannels by demanding that a superchannel $\mS^{(A\to B)\to (C\to D)} \in \mathfrak{S}$ is considered free only if it converts any free channel $\mM\in O(AR_1\to BR_2)$ to a free channel $\mE:=\mS*\mM \in \O(CR_1\to DR_2)$, see the middle of Fig.~\ref{fig:freeSuperchannelsSets}. The set of all superchannels satisfying this property is often called ``completely free" in the literature~\cite{Gour2020, Saxena2020,Gour2021b} and is formally defined as
\begin{align}\label{eq:completely_free_Superchannels}
    \cfO_1&((A\mkern-2mu\to\mkern-2mu B)\to  (C\mkern-2mu\to\mkern-2mu D))\!:=\Bigr
    \{ \mS^{(A\to B)\to (C\to D)}\!\in\!\mathfrak{S}_1: \nonumber \\
    &\quad \mS\!*\!\mM\in \!\O \,\forall\mM\in\! \O(AR_1\mkern-2mu\to\mkern-2muBR_2), \forall R_1, \!R_2\!\Big\}.
\end{align}
Lastly, as discussed for general networks in the previous section, any superchannel can be implemented using a pre- and post-processing channel together with an auxiliary system. An operational approach is thus to define free superchannels as the superchannels that can be implemented via free pre- and post-processing, see the left-hand side of Fig.~\ref{fig:freeSuperchannelsSets}, i.e.,
\begin{align} \label{eq:free_network_Superchannels}
    &\nO_1((A\mkern-2mu\to\mkern-2mu B)\mkern-2mu\to\mkern-2mu(C\mkern-2mu\to\mkern-2mu D))\!:= \big\{ \mS^{(A\mkern-2mu\to\mkern-2mu B)\to(C\mkern-2mu\to\mkern-2mu D)}\!\in\!\mathfrak{S}_1: \quad \, \nonumber \\
    &\,\,\mS \mkern-5mu=\mkern-4mu \mM_2\!*\!\mM_1, \mM_1\!\in\mkern-4mu \O(C\mkern-6mu\to\mkern-6mu AR),\mM_2 \!\in\mkern-4mu \O(BR\mkern-6mu\to \mkern-6mu D)\!\big\}.
\end{align}
If the set of free channels satisfies assumptions~\ref{assump:identity} to \ref{assumpt:complfree}, then 
\begin{align}\label{eq:SuperChannelInclusions}
    \nO_1 \subseteq \cfO_1 \subseteq \mFO_1.
\end{align}
Here, $\cfO_1 \subseteq \mFO_1$ follows immediately from choosing trivial systems $R_1,R_2$, and $\nO_1 \subseteq \cfO_1$ holds since assumptions \ref{assump:identity}, \ref{assump:Composition}, and \ref{assumpt:complfree} ensure that the concatenation of free channels leads to free channels. As we will show later, these sets are, in general, different (Prop.~\ref{prop:nonfreenetworks}). 
\begin{figure*}[ht]
    \centering
    \scalebox{0.8}{\includegraphics[width=1\linewidth]{ 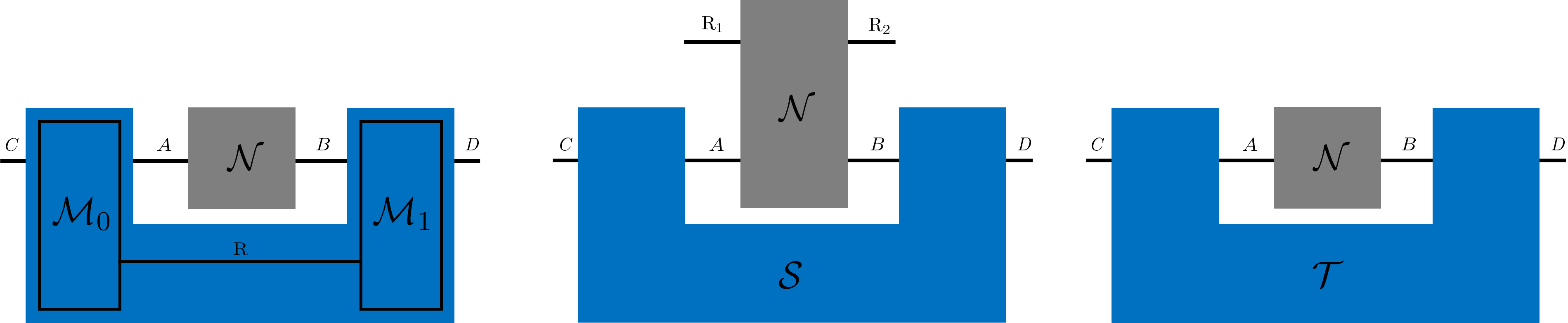}}
    \caption{A free network is composed of free pre-and postprocessing channels $\mM_0, \mM_1 \in \O$ as shown on the left. A completely free superchannel $\mS^{(A\to B) \to (C\to D)}$  transforms any free channel $\mN$ to a free channel in a complete sense, i.e., the resulting channel $\mS*\mN$ shown in the middle must be free for arbitrary systems $R_1, R_2$ and all free channels $\mN$. A maximally free superchannel $\mT^{(A\to B) \to (C\to D)}$ transforms any free channel $\mN^{A\to B}$ to another free channel $\mT*\mN$ as shown on the right. }
    \label{fig:freeSuperchannelsSets}
\end{figure*}

Next, we extend these definitions to arbitrary (deterministic) quantum supermaps. A straightforward generalization of $\nO_1$ are the supermaps $\mS_N\in \mathfrak{S}_N$ that can be implemented by a network consisting only of free operations, see Fig.~\ref{fig:NetworkvsComb}.

\begin{defin}\label{def:freeNetwork} 
The set of free networks with respect to a set of free channels $\O$ is given by
\begin{align}
    &\nO_N \!:=\!\big\{ \mS_N\!\in\!\mathfrak{S}_N\!:\!\mathcal{S}_N[\mathcal{N}_1,\!\cldots,\mathcal{N}_N]\!=\!\mathcal{M}_{N}(\idChannel \otimes \mathcal{N}_N)\mathcal{M}_{N\!-\!1} \nonumber \\
    &\quad (\idChannel \otimes \mathcal{N}_{N-1})\mathcal{M}_{N-2}\cldots \mathcal{M}_1 (\idChannel \otimes \mathcal{N}_1)\mathcal{M}_0,\mM_i\mkern-4mu\in \mkern-4mu\O \forall i\big\}. 
\end{align}
The set of all free networks will be denoted by $\nO := \cup_{N \in \mathbb{N}_0} \nO_N$, where $ \nO_0:= \O$.
\end{defin}

 As argued, e.g., in Ref.~\cite{Theurer2019}, considering free networks as free supermaps is convincing from an operational perspective; for instance, if one thinks about a quantum optical experiment where one combines beam splitters and phase shifters. However, characterizing free networks is generally difficult (see, e.g., Sec.~\ref{sec:charac_compatitbleSupermaps} and Refs.~\cite{Theurer2019,Ahnefeld2025}), and as in the case of superchannels (recall Eq.~\eqref{eq:SuperChannelInclusions}), it is thus convenient to consider relaxations. 

If $\O$ in Def.~\ref{def:freeNetwork} satisfies assumptions~\ref{assump:identity}, \ref{assump:Composition}, and~\ref{assumpt:complfree}, any network resulting from the concatenation of two free networks is free too, i.e., the free networks inherit the ``no free lunch" property of the free channels. This motivates the following definition, which captures the spirit of the notion of ``completely free".
 \begin{defin}\label{def:cO}
We call a set $\mathfrak{C}\subseteq\mathfrak{S}$ \textit{compatible} if it is closed under all compositions, i.e., iff
\begin{align}
    \mS* \mT\in \mathfrak{C} \ \forall \mS, \mT \in \mathfrak{C}.
\end{align}
If $\mathfrak{C}$ is compatible and $\mathfrak{C}_0:=\mathfrak{C}\cap \mathfrak{S}_0=\O$, we call $\mathfrak{C}$ compatible with respect to $\O$.
\end{defin}
Operationally, demanding that free supermaps are compatible guarantees that it is impossible to generate resources by composing free supermaps. Moreover, any set of supermaps that is compatible with respect to $\O$ contains only completely free superchannels (see Eq.~\eqref{eq:completely_free_Superchannels} and the middle of Fig.~\ref{fig:freeSuperchannelsSets}). However, since composing a superchannel and a channel does not necessarily yield another channel (recall Fig.~\ref{fig:Composition}), compatibility is the stronger requirement. 
 
For a fixed set of free channels $\O$, a priori, it is not guaranteed that there exists \textit{any} set of supermaps that is compatible with $\O$  (e.g., if the free channels themselves are not closed under composition). However, as noted above, if the free channels satisfy assumptions \ref{assump:identity} to \ref{assumpt:complfree}, the free networks are. On the contrary, it is possible that multiple sets of supermaps $\mathfrak{C}, \mathfrak{C}^\prime$ that are compatible with $\O$ exist, which could also be incomparable in the sense that neither $\mathfrak{C}\subseteq \mathfrak{C}^\prime$ nor $\mathfrak{C}^\prime\subseteq \mathfrak{C}$. If a largest set of supermaps that is compatible with $\O$ exists (in the sense that it contains all other sets compatible with $\O$), we call it ``the compatible supermaps" with respect to $\O$ and denote it as $\cO$, with  $\cO_N:=\cO\cap\mathfrak{S}_N$ (and in particular $\cO_0=\O$). In the following, we will see examples of maximal sets of compatible supermaps (Thm.~\ref{thm:CompatibleSupermaps}) that contain more than just free networks (Prop.~\ref{prop:nonfreenetworks}). It is also important to note that one can define compatible sets of free supermaps that are non-empty but do not contain any channels -- start, for example, by choosing a set of free superchannels with the property that their output and input systems are disjoint and consider their compositions as free. Such an approach might, for example, be relevant in an open system setting, see Ref.~\cite{Berk2021}.

It is possible to imagine different generalizations of the set of maximally free superchannels to the level of arbitrary supermaps which might be relevant in different scenarios: One example is to demand that if one applies a free supermap $\mS_N$ to $N$ free channels as on the right-hand side of Fig.~\ref{fig:freeSuperchannelsSets}, i.e., such that the first input and the last output of $\mS_N$ are the only accessible systems of the resulting supermap, then the resulting channel is free. This is of relevance in the case of adaptive multi-copy channel distillation~\cite{Regula2021c,Fang2022}. However, other potential generalizations are possible, e.g., to inductively demand that $S_{N+1}*S_N$ is free if $S_N$ was and only the first input and the last output of $\mS_{N+1}$ remain accessible (see Fig.~\ref{fig:Composition}(a)). We stress, however, that only the notion of compatibility guarantees ``no free lunch" under arbitrary compositions. 

\section{ONE-SHOT CONVERSION OF SUPERMAPS}\label{sec:GeneralConversion}
Two fundamental tasks in quantum resource manipulation are resource distillation and dilution~\cite{Chitambar2019,Gour2025}. They have been investigated extensively for quantum states~\cite{Winter2016,Zhao2018,Regula2018,Horodecki2013,Liu2019,Takagi2020} and, more recently, in specific settings for quantum channels~\cite{Gour2019,Liu2019arxiv,Liu2020,Kim2021,Regula2021b, Regula2021c,Saxena2020,Gour2020}. Importantly, these tasks can be formulated at the general level of higher-order quantum transformations~\cite{Takagi2024}, thereby unifying resource theories of states, channels, and supermaps within a common framework. To this end, fix a set of free supermaps and a designated family $\T$ of target supermaps that serve as reference resources.

To quantify the resourcefulness of the target supermaps, we will use resource monotones. Resource monotones are functionals that are \textit{monotonic under free supermaps}, which reflects the intuition that resources cannot be generated at no cost. To formalize this notion, consider the example of channels first: A functional $M$ from quantum channels to the real numbers is called a channel resource monotone under the set of free superchannels $F\subseteq\mathfrak{S}_1$ iff for all $\mN\in \CPTP$, $ M(\mS[\mN]) \leq M(\mN)$ whenever the application of $\mS\in F$ to $\mN$ is free (which implicitly implies that it is well-defined). Further desirable properties may include an operational interpretation, faithfulness (that is, $M(\mN)=0$ iff $\mN$ is free), and convexity, while for practical applications, computational tractability is often equally important~\cite{Regula2018,Chitambar2019,Gour2025}. It is straightforward to extend the notion of monotones to supermaps by demanding that transforming supermaps by means of free supermaps cannot increase the resources at hand. 
\begin{defin}
    A functional $M$ from $\mathfrak{S}$ to the real numbers is called a resource monotone under a set of free supermaps $F\subset\mathfrak{S}$ iff for all $ \mT_M \in \mathfrak{S}$
    \begin{align}
        M(\mS_N*\mT_M) \leq M(\mT_M)
    \end{align}
    whenever the transformation of $\mT_M$ by $\mS_N\in F $ is free.
\end{defin} 
To construct such monotones, it is convenient to start from a distance between quantum supermaps that is contractive under composition with quantum supermaps. One such distance is the generalized diamond distance (see Eq.~\eqref{eq:diamond_generalized_main}), which we discuss in more detail in Appendix~\ref{sec:Appen_GenDiaDist}.
\begin{restatable}{proposition}{SupermapMonotone}
\label{prop:SupermapMonotone}
The functional $M_{\mathfrak{C}}$ defined for $\mS_N \in \mathfrak{S}_N$ as
\begin{align}\label{eq:SupermapMontone}
    M_{\mathfrak{C}}(\mS_N)&\!:=\!\inf_{\mR_N \in \mathfrak{C}_N} d_{\diamond,N}(\mS_N,\mR_N).
\end{align}
is a resource monotone under a compatible set of free supermaps $\mathfrak{C}$. If the sets $\mathfrak{C}_N$ can be characterized by semidefinite constraints, this monotone can be computed via a semidefinite program.
\end{restatable}
In the above Proposition, the optimization over $\mR_N$ is understood to range over free supermaps with the same input and output systems as $\mS_N$ such that the diamond distance is well-defined. From now on, this will always be implicitly assumed.
For $N=0$, the monotone in Prop.~\ref{prop:SupermapMonotone} is precisely the minimal diamond distance of a channel to the set of free channels, and by associating a state with its preparation channel, we obtain the minimal trace distance to the set of free states.

With this at hand, we return to distillation and dilution of higher-order quantum resources. Distillation asks for the most resourceful target in $\T$ that can be approximately obtained from a given supermap $\mS_N$ via free transformations, whereas dilution asks for the least resourceful target in $\T$ from which a given supermap $\mS_N$
can be approximately simulated with the help of free transformations. To quantify how well such an approximate transformation can be achieved, one fixes a suitable metric for quantum supermaps, for example the generalized diamond distance in Eq.~\eqref{eq:diamond_generalized_main}. For an allowed error $\epsilon$, the corresponding one-shot distillable resource of a supermap $\mS_n$ under a set of free supermaps $\F$ is then formally defined as
\begin{align}\label{eq:Dist}
    \Dist_{\T, \F}^{\epsilon,\diamond}(\mS_N)=\sup\big\{ M(\mT)& : d_{\diamond,K}(\mS_N*\mR_M,\mT_K) \leq \epsilon, \nonumber \\
    &\mT_K \in \T, \mR_M\in \F \big\},
\end{align}
where $M$ is a monotone under $F$. Analogously, the resource cost is defined as 
\begin{align}\label{eq:Dil}
    \Cost_{\T,\F}^{\epsilon,\diamond}(\mS_N)= \inf\Big\{ M(\mT)& : d_{\diamond,N}(\mR_M*\mT,\mS_N) \leq \epsilon, \nonumber \\
    &\mT \in \T, \mR_M\in \F \Big\}.
\end{align}

\subsection{One-shot manipulation of quantum channels under maximally free superchannels} \label{sec:channel_conversion_general}
In this section, we specialize this general framework to the manipulation of quantum channels by superchannels. Accordingly, both the input $\mE$ and the target $\mT$ are quantum channels (and $\T$ contains only channels), while the allowed transformations are superchannels. The manipulation of quantum channels with free superchannels has been investigated in a number of dynamical resource theories, see, e.g., Refs.~\cite{Kim2021, Liu2020,Saxena2020,Regula2021b}. In particular, Ref.~\cite{Regula2021b} considers the maximally free superchannels in Eq.~\eqref{eq:maximally_free_Superchannels} as free superchannels and characterizes approximate one-shot dynamical resource conversions with respect to the worst-case fidelity (as defined in Eq.~\eqref{eq:WorstCaseFidelity}) for a broad class of general resource theories. In the remainder of this section, we show that the underlying arguments of Ref.~\cite{Regula2021b} also apply for approximation with respect to the diamond distance. This will allow us to show that under suitable conditions, the distillable resources defined with respect to the diamond distance and the worst-case fidelity coincide. 

To establish these extensions of the arguments of Ref.~\cite{Regula2021b}, we require several key definitions, resource quantifiers, and structural results introduced and employed in Ref.~\cite{Regula2021b}. Among them are the standard robustness $R_{s,\O}$~\cite{Vidal1999,Yuan2021,Takagi2021} and the generalized robustness $R_{\max,\O}$~\cite{Steiner2003,Harrow2003,Datta2009,Takagi2019} of a quantum channel 
\begin{align}
    R_{s,\mkern-1mu\O}(\mkern-1mu\mN\mkern-1mu)\mkern-4mu &=\mkern-4mu\inf_{\lambda\geq 0, \mK,\mM\in \O} \{ \mkern-1mu 1\!+\!\lambda \!:\! \mN\!+\!\lambda \mM \mkern-4mu =\mkern-4mu (1\mkern-4mu+\mkern-4mu\lambda) \mK \mkern-1mu\} \label{eq:def_R_s},\\
    R_{\max,\mkern-1mu\O}(\mkern-1mu\mN\mkern-1mu)\mkern-4mu&=\mkern-20mu\smash{\inf_{\substack{\lambda\geq 0 \\ \mM\in \CPTP,\mK\in \O}}} \{ \mkern-1mu 1\mkern-4mu+\mkern-4mu\lambda \!:\! \mN\!+\!\lambda \mM \mkern-4mu =\mkern-4mu (1\mkern-4mu+\mkern-4mu\lambda) \mK \mkern-1mu\}.\label{eq:def_R_max}
\end{align}
\vspace{0.1cm}

\noindent  The distinction between standard and generalized robustness becomes particularly relevant when considering full-dimensional resource theories, i.e., those for which $\CPTP(A\!\to\! B)\subseteq\operatorname{span}(\O(A\!\to\! B))$, and reduced-dimensional resource theories, i.e., those for which $\CPTP(A\!\to \!B) \nsubseteq \operatorname{span}(\O(A\!\to\! B))$~\cite{Regula2020}. 

Moreover, we employ two channel monotones based on the hypothesis testing relative entropy  $R_H^\epsilon(\rho||\sigma)= \max\{\Tr{P\sigma}^{-1}: 0\leq P\leq \id, \Tr{P\rho}\geq 1-\epsilon\}$ (see Ref.~\cite{Buscemi2010,Yuan2019,Gour2021,Wang2021}), namely
\begin{align}
    R_{\!H\mkern-2mu,\mkern-2mu\O\mkern-2mu}^\epsilon(\mkern-1mu\mE\mkern-1mu) \mkern-4mu&=\mkern-8mu \inf_{\mM\in \O\mkern-2mu}\mkern-6mu \max_\psi  R_H^\epsilon(\mkern-1mu\idChan \otimes\mE(\mkern-1mu\psi\mkern-1mu)|\mkern-2mu|\mkern-1mu\idChan\otimes \mM(\mkern-1mu\psi\mkern-1mu)), \label{eq:R_H,O} \\
    R_{\!H\mkern-2mu,\aff(\mkern-2mu\O\mkern-2mu)}^\epsilon(\mkern-1mu\mE\mkern-1mu) \mkern-4mu&=\mkern-8mu \inf_{\mM\mkern-2mu\in \aff(\mkern-2mu\O\mkern-2mu)}\mkern-6mu \max_\psi  R_H^\epsilon(\mkern-1mu\idChan \otimes\mE(\mkern-1mu\psi\mkern-1mu)|\mkern-2mu|\mkern-1mu\idChan\otimes \mM(\mkern-1mu\psi\mkern-1mu)), \label{eq:R_H,O_aff}
\end{align}
where the maximization is over pure states $\psi^{AR}$ with an auxiliary system $R\cong A$ fixed by the channel $\mE^{A\to B}$ and $\aff$ denotes the affine hull. In the case of $\epsilon=0$, the hypothesis testing relative entropy reduces to the min-relative entropy~\cite{Datta2009}, which motivates the definition 
\begin{align}
    R_{\min,\O}(\mE):=\mkern-8mu \inf_{\mM\in \O\mkern-2mu}\mkern-6mu \max_\psi  R_H^{\epsilon=0}(\mkern-1mu\idChan \otimes\mE(\mkern-1mu\psi\mkern-1mu)|\mkern-2mu|\mkern-1mu\idChan\otimes \mM(\mkern-1mu\psi\mkern-1mu)).
\end{align}
In the affine case, we \textit{define}
\begin{align}
    R_{\min,\aff(\O)}(\mN)\!:=\! \sup_{\psi}\big\{&\tfrac{1}{\lambda}\! :\! \Tr{\Pi_{\mN,\psi} (\idChan \otimes \mM)(\psi)}\!=\!\lambda  \nonumber \\
    &\quad \forall \mM\in \O\big\}
\end{align}
where $\Pi_{\mN,\psi}:= \Pi_{\supp\left(\idChan \otimes \mN (\psi) \right)}$ is the projector onto the support of $\idChan \otimes \,\mN(\psi)$, and where we adopt the convention that $\sup \emptyset =0$. We emphasize that, unlike in the non-affine case, there potentially is a gap between $R_{\min,\aff(\O)}(\mN)$ and $R_{\!H\mkern-2mu,\aff(\mkern-2mu\O\mkern-2mu)}^{\epsilon=0}(\mkern-1mu\mE\mkern-1mu) $, see Ref.~\cite[SM Remark below Thm.~1]{Regula2021b}. Moreover, for $m>0$, we will also use the following monotones introduced in Ref.~\cite{Regula2021b},
\begin{subequations}\label{eq:G_O}
\begin{align}
G_{\O}(\mE,m)=\max  &\big\{ \Tr{J_\mE W}: 0\leq W \leq \rho \otimes \id, \rho \in \D, \nonumber \\
&\,\, \Tr{W J_\mM}\leq \frac{1}{m}\, \forall \mM \in \O\big\},\\
    G_{\aff(\O)}(\mE,m)=\sup &\big\{ \Tr{J_\mE W}: 0\leq W \leq \rho \otimes \id, \rho \in \D, \nonumber \\
    &\,\, \Tr{W J_\mM}=\frac{1}{m} \,\forall \mM \in \O\big\},
\end{align}
\end{subequations}
and we will use the convention that $G_{\aff(\O)}(\mE,0)=1$. Using these monotones, for a target channel $\mN$ that is an isometry, Ref.~\cite{Regula2021b} establishes bounds on the optimal conversion distance with respect to the worst-case fidelity under the set of maximally free superchannels. In particular, for any full-dimensional resource theory satisfying assumption \ref{assump:ClosedConvex},  Ref.~\cite[Thm.~6]{Regula2021b} shows that $\max_{\mS_1\in \mO_1} F_{\min}(\mS_1[\mE],\mN)$ is upper and lower bounded by $G_{\O}(\mE,R_{\min ,\O}(\mN))$ and $G_{\O}(\mE,R_{s,\O}(\mN))$ respectively, and similarly for reduced-dimensional resource theories. 

We now turn to the first part of our extension and show that, under these assumptions, the same bounds hold if the conversion distance is measured with the diamond distance instead of the worst-case fidelity.

\begin{restatable}{theorem}{DiamondBoundBoundsG}
\label{thm:DiamondBound_Bounds_G}
    Let $\mE$ and $\mN$ be quantum channels with $\mN$ an isometry and let $\O$ denote a set of free channels satisfying assumption~\ref{assump:ClosedConvex}. If $\O$ is full-dimensional, then
    \begin{align}
       1- G_{ \O}(\mE,R_{s ,\O}(\mN)) &\geq  \min_{\mS \in \mO_1} \tfrac{1}{2}\norm{\mS[\mE]-\mN}_\diamond \nonumber \\
       &\geq 1- G_{\O}(\mE,R_{\min ,\O}(\mN)),
       \end{align}
       and if $\O$ is reduced-dimensional, then
       \begin{align}
       1\mkern-6mu \shortminus\mkern-4mu G_{\aff\mkern-2mu(\mkern-2mu\O\mkern-2mu)}\mkern-1mu(\mkern-2mu\mE\mkern-1mu,\mkern-1mu\mkern-2mu R_{\max ,\mkern-1mu\O\mkern-1mu}(\mkern-2mu\mN\mkern-1mu)\mkern-1mu) \mkern-5mu&\geq\!  \min_{\mS \in\mO_1 } \tfrac{1}{2}\norm{\mS[\mE]-\mN}_\diamond \quad\quad \quad \quad\quad \,\,\,  \, \, \, \quad \nonumber \\
       &\geq\mkern-6mu 1\mkern-6mu \shortminus \mkern-4mu G_{\aff(\mkern-2mu\O\mkern-2mu)}\mkern-2mu(\mkern-2mu\mE\mkern-2mu,\mkern-1mu\mkern-2mu R_{\min ,\mkern-1mu\aff(\mkern-2mu\O\mkern-2mu)}(\mkern-2mu\mN\mkern-1mu)\mkern-1mu).
    \end{align}
    If additionally $R_{\min,\O}(\mN)\!=\!R_{s,\O}(\mN)$ or $R_{\min,\aff(\O)}(\mN)\!=\!R_{\max,\O}(\mN)$, respectively, then the diamond conversion distance and the conversion distance with respect to the worst-case fidelity coincide, i.e., 
    \begin{align}\label{eq:diamond_fidelity_equal_convDist}
        \min_{\mS \in \mO_1} \tfrac{1}{2}\norm{\mS[\mE]-\mN}_\diamond &=1\!-\!\max_{\mS \in \mO_1} F_{\min}(\mS[\mE],\mN) \nonumber \\
        &= 1\!-\! G_{ \O}(\mE,R_{s ,\O}(\mN)),
    \end{align}
    in the full-dimensional case 
    and analogously in the reduced-dimensional case.
    \end{restatable}
    
We emphasize that Eq.~\eqref{eq:diamond_fidelity_equal_convDist} goes beyond the sharpened Fuchs-van de Graaf inequality in Eq.~\eqref{eq:FuchsVanDeGraaf_channels_sharp}, from which one cannot deduce equality of the two conversion distances except when an exact conversion is possible. Moreover, we note that neither Thm.~\ref{thm:DiamondBound_Bounds_G} nor Thm.~6 in Ref.~\cite{Regula2021b} is stronger, since via the sharpened  Fuchs-van de Graaf inequality, each implies parts of the other. As a straightforward extension of Thm.~\ref{thm:DiamondBound_Bounds_G}, and as already mentioned in Ref.~\cite{Regula2021b}, if one considers the Choi state fidelity as defined in Eq.~\eqref{eq:ChoiStateFidelity}, and if additionally for 
\begin{align}\label{eq:R_min_tilde}
    \tilde{R}_{\min , \O}(\mN)= \min_{\mM \in \O} R_{\min}\left( \idChan \otimes \, \mN (\phi) || \idChan\otimes \,\mM(\phi)\right),
\end{align}
where $\phi$ denotes a maximally entangled state, it holds that $\tilde{R}_{\min , \O}(\mN)=R_{s,\O}(\mN)$, then 
\begin{align}
    \min_{\mS \in \mO_1} \tfrac{1}{2}\norm{\mS[\mE]-\mN}_\diamond &=  1-\max_{\mS \in \mO_1} F_{\min}(\mS[\mE],\mN)  \nonumber \\
    &= 1-\max_{\mS \in \mO_1} F_{\ChoiF}(\mS[\mE],\mN) \nonumber \\
    &=1- G_{\O}(\mE,R_{s ,\O}(\mN)).
\end{align}

As we will see later in Sec.~\ref{sec:CoherenceManipulation}, even if we consider sets of free superchannels other than the maximally free superchannels, similar results can be established in the resource theories of dynamical coherence. 

Next, we relate the diamond conversion distance to inequalities between \textit{smoothed} resource monotones, thereby establishing the analogue of Ref.~\cite[Thm.~1 and 3]{Regula2021b}: Let $R_{\O}$ denote any resource monotone, and let $R_{\O}^{\epsilon, \cdot}$ denote the monotone smoothed with respect to the diamond distance and the worst-case fidelity, respectively, i.e.,
\begin{subequations}
    \begin{align}
        R_{\O}^{\epsilon,\diamond}(\mE)&:= \inf_{ \mE^\prime \in \CPTP}\big\{R_{\O}(\mE^\prime)\!:\! \tfrac{1}{2}\norm{\mE-\mE^\prime}_\diamond \leq \epsilon \big\}, \label{eq:DiamondSmoothing}\\
        R_{\O}^{\epsilon,F}(\mE)&:= \inf_{ \mE^\prime \in \CPTP}\big\{R_{\O}(\mE^\prime)\!:\! F_{\min}(\mE,\mE^\prime) \!\geq\! 1\!- \!\epsilon \big\} \label{eq:FidelitySmoothing}.
    \end{align}
\end{subequations}
We now have the following diamond distance analogue of Ref.~\cite[Thm.~1]{Regula2021b}.

\begin{restatable}{theorem}{ThmOneequivalent}
\label{thm:Thm1_equivalent}
Let $\mE$ and $\mN$ be quantum channels and let $\O$  be a set of free channels satisfying assumption~\ref{assump:ClosedConvex}. Furthermore, let $R_{\O}$ denote any monotone under $\mO_1$. If there exists a superchannel $\mS_1 \in \mO_1$ such that $\frac{1}{2}\norm{\mS_1[\mE]-\mN}_\diamond \leq \epsilon$, then $R_{\O}^{\delta,\diamond}(\mE) \geq R_{\O}^{\epsilon+\delta,\diamond}(\mN)$ for any $ 0\leq \delta \leq 1$. Conversely, for any choices of $\epsilon,\delta\geq 0$ with $\epsilon+\delta\leq 1$, if $R_{H,\O}^{\delta} (\mE)\geq R_{s,\O}^{\epsilon,\diamond}(\mN)$ or $R_{H,\aff(\O)}^{\delta} (\mE)\geq R_{\max , \O}^{\epsilon,\diamond}(\mN)$, then there exists a superchannel $\mS_1\in \mO_1$ such that $\frac{1}{2}\norm{\mS_1[\mE]-\mN}_\diamond\leq \epsilon+\delta $.
\end{restatable}

We discuss in Appendix~\ref{sec:Appen_GeneralConversion} how this Theorem is related to Ref.~\cite[Thm.~1]{Regula2021b}. If $\mN$ is an isometry, by specializing to the monotones based on the hypothesis-testing relative entropy, as a simple Corollary of Thm.~3 in Ref.~\cite{Regula2021b}, we can strengthen the direct part of the above Theorem.
\begin{restatable}{cor}{ThmThreeequivalent}
\label{thm:Thm3_equivalent}
    Let $\mE$ and $\mN$ be quantum channels with $\mN$ an isometry, and let $\O$  be a set of free channels satisfying assumption~\ref{assump:ClosedConvex}. If there exists a free superchannel $\mS_1\in \mO_1$ such that $\frac{1}{2}\norm{\mS_1[\mE]-\mN}_\diamond \leq \epsilon$, then $R_{H,\O}^{\epsilon}(\mE) \geq R_{\min ,\O}(\mN)$ and $R_{H,\aff(\O)}^{\epsilon}(\mE) \geq R_{\min ,\aff(\O)}(\mN)$.
\end{restatable}

The above one-shot conversion bounds allow us to address distillation and dilution with maximally free superchannels $\mO_1$. For the remainder of this section, we choose the resource monotone to be $M(\mT)=R_{\min, \O}(\mT)$ (see Eq.~\eqref{eq:min_entropy}). While this choice might seem arbitrary at first glance, it is meaningful in that it recovers conventional definitions of the resource cost and distillable resource in various resource theories, for more information see, e.g., Ref.~\cite{Regula2021b}. With this choice, 
\begin{align}
    \Dist_{\T, \mO_1}^{\epsilon,\diamond}(\mE)=&\sup\Big\{R_{\min ,\O}(\mT): \frac{1}{2}\norm{\mS[\mE]-\mT}_\diamond \leq \epsilon, \nonumber \\
    &\quad\qquad\mT \in \T, \mS\in \mO_1 \Big\},  \\
    \Cost_{\T,\mO_1}^{\epsilon,\diamond}(\mE)=& \inf\Big\{ R_{\min ,\O}(\mT) : \frac{1}{2}\norm{\mE-\mS[\mT]}_\diamond \leq \epsilon, \nonumber \\
    &\quad\qquad\mT \in \T, \mS\in  \mO_1\},
\end{align}
and analogously, define 
\begin{align}
    \Dist_{\T, \mO_1}^{\epsilon,F}(\mE):=&\sup\Big\{R_{\min ,\O}(\mT): F_{\min}(\mS[\mE],\mT) \geq 1\!-\!\epsilon,  \nonumber\\
    &\quad\qquad\mT \in \T, \mS\in \mO_1 \Big\},  \\
    \Cost_{\T,\mO_1}^{\epsilon,F}(\mE):=& \inf\Big\{ R_{\min ,\O}(\mT) : F_{\min}(\mE,\mS[\mT]) \geq 1\!-\!\epsilon, \nonumber \\
    &\quad\qquad\mT \in \T, \mS\in  \mO_1\}.
\end{align}
Furthermore, let  
\begin{subequations}
\begin{align}
    \floor{x}_{\T} :=& \sup\{R_{\min, \O}(\mT):  R_{\min , \O}(\mT)\leq x, \mT\in \T\},  \\
    \ceil{x}_{\T} :=& \inf\{R_{\min , \O}(\mT):  R_{\min, \O}(\mT)\geq x, \mT\in \T \}.
\end{align}
\end{subequations}
With this notation at hand, the following is a direct Corollary of Ref.~\cite[Cor.~2 and~4]{Regula2021b}, Thm.~\ref{thm:Thm1_equivalent}, and Cor.~\ref{thm:Thm3_equivalent}.

\begin{cor}
\label{cor:Cost_Dist_mO}
Let $\mE$ be a quantum channel, $\star \in \{F,\diamond \}$, and $\O$ a set of free channels satisfying assumption~\ref{assump:ClosedConvex}. If $R_{\min,\O}(\mT)=R_{s,\O}(\mT) \, \forall \mT \in \T$, then
\begin{align}
    \Cost_{\T, \mO_1}^{\epsilon,\star}(\mE) =\ceil{R_{s, \O}^{\epsilon,\star}(\mE)}_T,
\end{align}
and if additionally $\mT$ is an isometry for every $\mT\in\T$, then
\begin{align}
    \Dist_{\T, \mO_1}^{\epsilon,\star}(\mE)=\floor{R_{H,\O}^\epsilon(\mE)}_T.
\end{align}
If $R_{\min,\aff(\O)}(\mT)=R_{\max,\O}(\mT) \, \forall \mT \in \T$, then
\begin{align}
    \Cost_{\T, \mO_1}^{\epsilon,\star}(\mE)= \ceil{R_{\max , \O}^{\epsilon,\star}(\mE)}_T,
\end{align}
and if additionally $\mT$ is an isometry for every $\mT\in\T$, then
\begin{align}
    \Dist_{\T, \mO_1}^{\epsilon,\star}(\mE)=\floor{R_{H,\aff(\O)}^\epsilon(\mE)}_T.
\end{align}
\end{cor}

Before we discuss the content of the above Corollary, we want to comment on its applicability. As detailed in Refs.~\cite{Liu2019, Regula2020, Regula2021b}, the condition $R_{\min,\O}(\mT)=R_{s,\O}(\mT)$ holds in many full-dimensional resource theories for common sets of target channels, and the condition that $R_{\min,\aff(\O)}(\mT)=R_{\max,\O}(\mT)$ holds in many reduced-dimensional ones. In either case, the distillable resources with respect to the diamond distance and worst-case fidelity coincide. In contrast, the costs can differ, since the relevant robustness is smoothed over a different metric. This concludes our extension of the fidelity-based framework of Ref.~\cite{Regula2021b} to the diamond distance. We emphasize again that in this subsection, we considered the set of maximally free superchannels. This is crucial for the proofs of Thm.~\ref{thm:DiamondBound_Bounds_G} and Thm.~\ref{thm:Thm1_equivalent}, and likewise for the proofs of Ref.~\cite[Thm.~1, Thm.~3, and Thm.~6]{Regula2021b}. At their core are ``measure-and-act" superchannels of the form
\begin{align}\label{eq:Superchannel_maxFree}
    \mS[\mE]\mkern-4mu=\mkern-4mu\Tr{ \idChan \otimes\,\mE(\psi) P} \mN \mkern-4mu+\mkern-4mu\Tr{\idChan \otimes\,\mE(\psi)(\id\mkern-4mu -\mkern-4mu P)} \mQ,
\end{align}
where $\mE$ is the input channel, $\mN$ the target channel, and the state $\psi$, the POVM element $P$, and the channel $Q$ are solutions of the optimization problems related to the monotones appearing in these Theorems (see the Appendix for details). This ensures that $\mS \in \mO_1$. However, without further assumptions, it neither guarantees that $\mS$ is completely free in the sense of Eq.~\eqref{eq:completely_free_Superchannels}, nor that it admits an implementation by a free network as in Eq.~\eqref{eq:free_network_Superchannels}. In fact, we will see in Prop.~\ref{cor:DIConversionDistance} in the next section that in coherence theory, neither is true. Maximally free superchannels are thus, in general, not closed under arbitrary composition. 

To conclude this section, we note that one-shot distillation and dilution of dynamical resources with respect to the diamond distance were previously considered in Ref.~\cite {Yuan2020}. However, Ref.~\cite{Yuan2020} contains no proofs and considers resource manipulation with free networks, whereas  Ref.~\cite{Regula2021b} and the present section consider the maximal set of free superchannels.

\section{APPLICATION TO DYNAMICAL COHERENCE}\label{sec:Coherence}
Resource theories of coherence~\cite{Streltsov2017} quantify the value of superposition with respect to a fixed orthonormal basis $\{\ket{i}\}_i$ referred to as the incoherent basis. A quantum state $\sigma$ is called incoherent (or free) iff it is diagonal in the incoherent basis, i.e., iff $\Delta(\sigma)=\sigma$, where $\Delta(\rho)=\sum_i \ketbra{i}{i} \rho \ketbra{i}{i}$ denotes total dephasing in the incoherent basis. We denote the set of incoherent states by $\I$. States not contained in $\I$ are called coherent and are thus considered resourceful. The literature considers several classes of free channels~\cite {Streltsov2017,Aberg2006,Baumgratz2014,Winter2016,Yadin2016,Chitambar2016,Chitambar2016b,Chitambar2017,Marvian2016} that emerge, for example, from practical or conceptual considerations. In the following, we focus on the largest sets of free channels that cannot create, detect, or both create and detect coherence, respectively~\cite{Aberg2006,Diaz2018,Theurer2019, Meznaric2013,Chitambar2016,Chitambar2016b,Marvian2016,Chitambar2017,Liu2017}.

\begin{defin}
\label{def:FreeOperation}
    A quantum channel $\mathcal{N}^{A\to B}$ is called
    \begin{enumerate}[label=\alph*)]
        \item Creation-incoherent or maximally incoherent iff $\mathcal{N}^{A\to B} \circ\Delta_A =\Delta_B\circ\mathcal{N}^{A\to B}\circ \Delta_A$,
        \item Detection-incoherent iff $\Delta_B\circ\mathcal{N}^{A\to B}  =\Delta_B \circ\mathcal{N}^{A\to B}\circ \Delta_A$,
        \item Creation-detection-incoherent or dephasing-covariant iff $\Delta_B\circ\mathcal{N}^{A\to B}  =\mathcal{N}^{A\to B}\circ\Delta_A$.
    \end{enumerate}
    The set of all creation-incoherent, detection-incoherent, and dephasing-covariant channels (of arbitrary input and output systems) will be labeled as $\MIO$, $\DI$, and $\DIO$ respectively.
\end{defin}

For these resource theories of dynamical coherence,  assumptions~\ref{assump:identity} to \ref{assump:ClosedConvex} are satisfied. Using the Choi representation of quantum channels, it becomes clear that a channel $\mN$ is contained in $\MIO$ iff $(\Delta_A\otimes \idChan_B) J_\mathcal{N}^{AB}=(\Delta_A \otimes \Delta_B) J_\mathcal{N}^{AB}$. Throughout this work, we will make extensive use of total dephasing maps acting on subsystems only, and will simply denote the channel action $(\Delta_A\otimes \idChan_B)(X_{AB})$ as $\Delta_A X_{AB}$ from now on. Analogously, a channel is in $\DI$ iff $\Delta_B J_{\mN}^{AB}=\Delta_{AB} J_{\mN}^{AB}$ and in $\DIO$ iff $\Delta_A J_{\mN}^{AB}=\Delta_{B} J_{\mN}^{AB}$. 

The largest compatible sets $\cMIO, \cDI$, and $\cDIO$, as defined in the discussion following Def.~\ref{def:cO}, exist. Using the system labels from Fig.~\ref{fig:NetworkvsComb}, they admit the following characterization. 

\begin{restatable}{theorem}{CompatibleSupermaps}
\label{thm:CompatibleSupermaps}
Let $\O \in\{\MIO, \DI, \DIO \} $.  Then, $\cO$ exists and
\begin{enumerate}[label=\alph*)]
    \item A supermap satisfies  $\mS_N \in \cMIO$ iff $\forall j: 0\leq j\leq N$ it holds that~\cite{Ahnefeld2025}
        \begin{align}\label{eq:MIOCompFromMaxSet}
        \Delta_{0,2,\cldots, 2j} J_{\mathcal{S}_N}=\Delta_{0,2,\cldots, 2j} \Delta_{1,3,\cldots, 2j+1} J_{\mathcal{S}_N}.
    \end{align}
    \item A supermap satisfies  $\mS_N \in \cDI$ iff $\forall j: 0\leq j\leq N$ it holds that
        \begin{align}\label{eq:DICompFromMaxSet}
        \Delta_{2N+1,\cldots,2j+1} J_{\mathcal{S}_N}=\Delta_{2N+1,\cldots,2j+1} \Delta_{2N,\cldots, 2j} J_{\mathcal{S}_N}.
    \end{align}
    \item A supermap satisfies  $\mS_N \in \cDIO$ iff $\forall j: 0\leq j\leq N$ it holds that
    \begin{subequations}\label{eq:DIOCompFromMaxSet}
        \begin{align}
        &\Delta_{0,2,\cldots, 2j} J_{\mathcal{S}_N}=\Delta_{0,2,\cldots, 2j} \Delta_{1,3,\cldots, 2j+1} J_{\mathcal{S}_N},  \\
        &\Delta_{2N+1,\cldots,2j+1} J_{\mathcal{S}_N}=\Delta_{2N+1,\cldots,2j+1} \Delta_{2N,\cldots, 2j} J_{\mathcal{S}_N}.
    \end{align}
    \end{subequations}
\end{enumerate}
\end{restatable}

As we show in Appendix~\ref{sec:Appen_Coherence}, for $\O\in \{\MIO,\DI,\DIO \}$ we have $\cfO_1=\cO_1$, i.e., the superchannels in $\cO$ are exactly the completely free superchannels $\cfO_1$ introduced in  Eq.~\eqref{eq:completely_free_Superchannels}. 

In Sec.~\ref{sec:FreeSupermaps}, we discussed two other classes of free superchannels, namely free networks and maximally free superchannels. To demonstrate that these sets are generally different, we provide an explicit counterexample using MIO and DI as free operations.

\begin{restatable}{proposition}{nonfreenetworks}
\label{prop:nonfreenetworks}
Let $\O\in \{\MIO,\DI\}$. Then it holds that $\nO_1 \subsetneq \cO_1 \subsetneq \mO_1$.
\end{restatable}
We mentioned earlier that the set of compatible supermaps provides a useful relaxation to the set of free networks. For coherence, a concrete example is  Ref.~\cite{Ahnefeld2025}, where this relaxation was used to prove a quantitative link between coherence and the performance of optimal phase estimation protocols when the available coherence is limited. 

The constraints in Thm.~\ref{thm:CompatibleSupermaps} and Eq.~\eqref{eq:CombSet} are semidefinite constraints, allowing us to employ the powerful tools of convex optimization and SDPs. For $\cMIO,\cDI$ or $\cDIO$ as free supermaps, this implies that the monotones in Prop.~\ref{prop:SupermapMonotone} can be efficiently computed\footnote{Here, ``efficiently" means that an approximate solution within a prescribed accuracy $\epsilon$ can be computed in polynomial time in the input size and $\log(1/\epsilon)$; see Ref.~\cite{Vandenberghe1996, Boyd2004}.}. Additionally, Thm.~\ref{thm:CompatibleSupermaps} allows us to express conversion distances, e.g., with respect to the generalized diamond distance, as an SDP. 

\subsection{Manipulation of dynamical coherence}\label{sec:CoherenceManipulation}
In Sec.~\ref{sec:GeneralConversion}, we considered channel distillation and dilution with the maximal set of free superchannels for general resource theories and sets of target channels $\T$. Here, we specialize this to dynamical coherence but consider the more complicated compatible supermaps and free networks instead. For dynamical coherence, a natural choice of a set of target channels is $\T:= \{\mF_d: d\in \mathbb{N}\}$, where $\mF_d$ denotes the $d$-dimensional quantum Fourier transform. We will justify this shortly with the help of the following Theorem, in which we consider the conversion distance from multiple copies of $\mF_d$ to arbitrary channels under both adaptive and parallel protocols.

\begin{restatable}{theorem}{DynamicalCoherenceCost}
\label{thm:Cost}
Let $\mE$ be any quantum channel, let $\mF_d$ denote the d-dimensional Fourier transform, and let $N\geq 1, d\geq 2$. Then, for  $\F \in \{\nMIO_1,\cMIO_1,\mMIO_1\}$,
   \begin{align}\label{eq:Cost_SDP_MIO}
    \inf_{\mS \in \F} \tfrac{1}{2}\norm{\mS[\mF_d^{\,\otimes N}]\!-\!\mE}_\diamond =&\min \big\{\tfrac{1}{2}\norm{\mM-\mE}_\diamond: \mM \in \CPTP, \nonumber \\
    & \quad R_{\max,\MIO}(\mM)\leq d^N\big\} , \\    
    \sup_{\mS \in \F} F_{\min}(\mS[\mF_d^{\otimes N}],\mE) =&\max\! \big\{ F_{\min}(\mM,\mE): \mM \in \CPTP, \nonumber \\
    &\quad   R_{\max,\MIO}(\mM)\leq d^N\big\},
\end{align}
and for $\F \in \{\cDI_1,\mDI_1\}$, 
   \begin{align}\label{eq:Cost_SDP_DI}
    \min_{\mS \in \F} \tfrac{1}{2}\norm{\mS[\mF_d^{\,\otimes N}]\!-\!\mE}_\diamond =&\min \big\{\tfrac{1}{2}\norm{\mM-\mE}_\diamond: \mM \in \CPTP, \nonumber \\
    & \quad  R_{\max,\DI}(\mM)\leq d^N\big\},
     \\    
    \sup_{\mS \in \F} F_{\min}(\mS[\mF_d^{\otimes N}],\mE) =&\max \big\{ F_{\min}(\mM,\mE): \mM \in \CPTP, \nonumber \\
    &\quad   R_{\max,\DI}(\mM)\leq d^N\big\}.
\end{align}
However, there exists a unitary channel $\mU$ for which
    \begin{align}\label{eq:different_network_cost_DI}
       \inf_{\mS \in \nDI_1} \tfrac{1}{2}\norm{\mS[\mF_2]-\mU}_\diamond > \min_{\mS \in \cDI_1} \tfrac{1}{2}\norm{\mS[\mF_2]-\mU}_\diamond.  
    \end{align}
Furthermore, recalling that $\mS_N[\mF_d,\cldots, \mF_d]$ denotes the sequential application of the supermap $\mS_N$ to $N$ copies of $\mF_d$, for $\O \in\{\MIO,\DI\}$, 
    \begin{align}\label{eq:Cost_sequential}
    \!\min_{\mS_N \mkern-2mu\in \cO_N} \mkern-7mu\tfrac{1}{2}\mkern-3mu\norm{\mS_N[\mF_d,\mkern-2mu\cldots, \mkern-2mu \mF_d ]\mkern-4mu \shortminus\mkern-4mu\mE}_\diamond \mkern-5mu =\mkern-5mu \min_{\mS \mkern-2mu\in \cO_1} \mkern-5mu\tfrac{1}{2}\mkern-3mu\norm{\mS[\mF_d^{\otimes N}]\mkern-4mu \shortminus\mkern-4mu\mE}_\diamond\!.
    \end{align}
\end{restatable}
As promised, this justifies choosing the set of target channels as $\T= \{\mF_d: d\in \mathbb{N}\}$: We show in Appendix~\ref{sec:Appen_Coherence} (see Lem.~\ref{cor:R_Fd}) that for $\O\in \{\MIO,\DI\}$, it holds that $R_{\max,\O}(\mF_d)=d$, $R_{\max, \MIO}(\mE^{A\to B}) \leq d_B$, and $R_{\max, \DI}(\mE^{A\to B}) \leq d_A$. The above Theorem then implies that Fourier transforms can be mapped into all other channels of equal or lower dimensions, therefore serving as ``golden units".

Moreover, the first part of the Theorem implies that even though the sets of free supermaps $\F\in\{\nMIO,\cMIO,\mMIO\}$ are distinct (cf. Prop.~\ref{prop:nonfreenetworks}), they all lead to the same resource cost 
\begin{align}
    \Cost_{\T,\F}^{\epsilon,\star}(\mE) = \ceil{R_{\max, \MIO }^{\epsilon, \star}(\mE)},
\end{align}
where $\star \in\{\diamond,F\}$, and analogously for $\F \in\{\cDI,\mDI \}$
\begin{align}
    \Cost_{\T,\F}^{\epsilon,\star}(\mE) = \ceil{R_{\max, \DI }^{\epsilon, \star}(\mE)}.
\end{align}
In all these cases, the resource cost is determined by the generalized robustness of the target channel $\mE$ smoothed over different metrics depending on how the conversion error is defined. Notice the absence of $\nDI$ here; Eq.~\eqref{eq:different_network_cost_DI} in the Theorem implies that even in the parallel case, the cost under $\nDI$ differs from the one under $\cDI$ and $\mDI$. It is currently unclear whether this is just a coincidence or reflects a deeper reason. One fundamental difference between MIO and DI is that DI  is a genuine dynamical resource theory in the sense that all states are free: The resource lies in a channel’s ability to convert input coherence into detectable population differences, a capability that cannot be reduced to a property of states (see Ref.~\cite[SM Sec.~VIII]{Theurer2019}). In contrast, MIO is defined by the preservation of incoherent states, and a channel's resourcefulness is therefore closely tied to the underlying state resource theory. Computing the cost under $\nDI$ is an open problem, in part because it is unclear how to characterize the set of free networks and because a relaxation to $\cDI$ changes the cost.

To conclude the discussion of the Theorem, we note that Eq.~\eqref{eq:Cost_sequential} implies that for the resource cost under compatible supermaps, adaptive and parallel protocols perform equally well.

We emphasize that the resource cost of channels has been studied before: In the creation-incoherent setting, Ref.~\cite{Diaz2018} related the implementation cost of channels to the smoothed max-robustness. In particular, given a target channel $\mE$ and a maximally coherent resource state $\psi_d$ of dimension $d$, a MIO operation consuming this coherence is used to simulate $\mE$ up to an error $\epsilon$ in diamond distance. Notice that by associating a resource state with its preparation channel $\mP_{\psi_d}$, the Fourier transform $\mF_d$ and $\mP_{\psi_d}$ can be interconverted using MIO networks: Starting from $\mF_d$, simply prepare $\psi_d=F_d\ket{0}$. Conversely, starting from $\mP_{\psi_d}$, use it to simulate the target channel $\mF_d$, which can be achieved exactly (according to Ref.~\cite{Diaz2018} or Thm.~\ref{thm:Cost}). Moreover, Ref.~\cite{Liu2020} considered adaptive protocols with free supermaps $\nMIO$ and obtained an asymptotic characterization in terms of a regularized smoothed max-robustness. 
We extend these results to the one-shot regime by showing that adaptive $\cMIO$ (and thus, by the first part of Thm.~\ref{thm:Cost}, $\nMIO$) protocols that use $N$ copies of a Fourier transform offer no advantage over a parallel protocol. 
To the best of our knowledge, results similar to those in Thm.~\ref{thm:Cost} were not previously known in the detection-incoherent setting. 

We now turn to the distillation of Fourier transforms, distinguishing again between parallel and adaptive distillation protocols on $N$ copies of a channel $\mE$. We begin with parallel distillation, where we require the monotones
\begin{subequations}\label{def:M_cO}
    \begin{align}
        M_{\cMIO,d}(\mE) \! &:=\!\max\big\{\!\Tr{X \mE(\ketbra{j}{j})} \!:\ 0\leq \!X\!\leq \!d\id, \nonumber \\ 
        &\qquad \qquad \quad  \Delta X=\id, 0\le j \le d_A\!-\!1\big\}, \label{eq:C_D} \\
        M_{\cDI,d}(\mE) \!: &=\! \max \Bigg\{ \sum_{i=0}^{d_B-1}\braket{i|\mE(\rho_i)| i} \!:\! \rho_i,\sigma \in \D, \rho_i\! \leq\! d \sigma, \nonumber \\
        &\qquad \qquad  \Delta \rho_i \!=\!\Delta \sigma \, \forall i: 0 \!\leq \!i \!\leq\! d_B\!-\!1\Bigg\},\label{eq:D_d}
    \end{align}
\end{subequations}
defined on all quantum channels $\mE^{A\to B}$.
As we show in the following Theorem, these monotones determine the optimal parallel conversion distances under compatible superchannels to a Fourier transform $\mF_d$ of fixed dimension $d$. To do so, we cannot use the general tools from Ref.~\cite{Regula2021b} and Sec.~\ref{sec:GeneralConversion} for the reason explained around Eq.~\eqref{eq:Superchannel_maxFree}. Instead, we rely on Thm.~\ref{thm:CompatibleSupermaps} and a twirling map on the level of Choi states (see Sec.~\ref{sec:Twirling} in the Appendix for more details) as the main technical ingredients to arrive at the following Theorem.

\begin{restatable}{theorem}{CoherenceOptimalConversion}
\label{thm:CoherenceOptimalConversion_main}
Let $\mE$ be any quantum channel, and let $\mF_d$ denote the d-dimensional Fourier transform. Suppose $d\geq 2$ and let $ \O \in \{\MIO,\DI \}$. Then, 
    \begin{align}
        \min_{\mS \in \cO_1} \tfrac{1}{2}\norm{\mS[\mE]-\mF_d}_\diamond &=  1-\max_{\mS \in \cO_1}F_{\min}(\mS[\mE],\mF_d) \nonumber \\
        &= 1-\max_{\mS \in \cO_1}F_{\ChoiF}(\mS[\mE],\mF_d) \nonumber \\
        &=1-\tfrac{1}{d} M_{\cO,d}(\mE).
    \end{align}
    Additionally, for $\O=\MIO$ this optimum can be achieved using a free ``measure-and-act" network.
\end{restatable}
The conversion distances to Fourier transforms are thus determined by $M_{\cO,d}(\mE)$, which makes it relevant to understand the properties of these monotones. For every channel $\mE$, it holds that $M_{\cO,d}$ is bounded between $\tfrac{1}{d}\leq \tfrac{1}{d} M_{\cO,d}(\mE) \leq 1$ and in Appendix~\ref{sec:Appen_Coherence} (Cor.~\ref{cor:Monotones_properties}),  we show that $ M_{\cO,d}(\mE)/d$ is non-increasing in $d$, i.e., for all $d_1\leq d_2$ we have $ M_{\cO,d_1}(\mE)/d_1 \geq M_{\cO,d_2}(\mE)/d_2$. This reflects the intuition that at fixed resources $\mE$, distilling a Fourier transform of higher dimension becomes increasingly hard. 
In Appendix~\ref{sec:Appen_Coherence} (Prop.~\ref{cor:DIO_Conv}), we show that a statement analogous to Thm.~\ref{thm:CoherenceOptimalConversion_main} holds in the creation-detection-incoherent setting, i.e., for $\O=\DIO$. 

A direct consequence of Thm.~\ref{thm:CoherenceOptimalConversion_main} is that for $\O\in \{\MIO,\DI\}$ and $\star \in\{\diamond, F\}$,
\begin{align}\label{eq:dist_M_O}
        \Dist_{\T, \cO}^{\epsilon,\star}(\mE)=\sup\big\{ d\in \mathbb{N} : M_{\cO,d}(\mE) \geq d(1-\epsilon)\big\}.
\end{align}
For $\O=\MIO$, Thm.~\ref{thm:CoherenceOptimalConversion_main} further implies that the same distillable resources are attained under $\nMIO$ supermaps, whereas it remains an open question whether the same is true in the detection-incoherent setting. We emphasize that although we are not considering the set of maximally free superchannels as in Sec.~\ref{sec:GeneralConversion}, the optimal conversion distances with respect to the diamond distance, worst-case fidelity, and Choi fidelity coincide. Therefore, identical conversion distances (and thus distillable resources) with respect to these different metrics are not necessarily a property arising from considering the set of maximally free superchannels. 
This is an interesting observation in itself, and leads us to pose the question if it is possible to identify general conditions under which the conversion distances for different sets of free supermaps coincide.

We now return to optimal resource conversion under $\mMIO_1$ and $\mDI_1$ and compare it to Thm.~\ref{thm:CoherenceOptimalConversion_main}. To this end, let $\Prob_n:=\{(p_0,\ldots,p_{n-1}), p_i\geq 0\, \forall i,\sum_i p_i=1\}$ denote sets of probability distributions and consider the monotones
\begin{subequations}\label{eq:M_mO,d}
    \begin{align}
        M_{\mMIO,d}(\mE)\! &:=\! \max\Bigg\{\sum_{i=0}^{d_A-1} \Tr{X_i\mE(\ketbra{i}{i})}: p,q\in \Prob_{d_A}, \nonumber \\
        & \quad \quad 0\leq X_i\leq d q_i \id,\ \Delta X_i=p_i \id \Bigg \}, \label{eq:M_mMIO,d} \\
        M_{\mDI,d}(\mE) \! &:=\! \max \Bigg\{\!\sum_{i=0}^{d_B-1} \!\braket{i|\mE(\rho_i)| i} \!:\! \rho_i,\sigma \!\in \!\D, \rho_i\! \leq\! d \sigma, \nonumber \\
        & \quad\Delta \rho_i \!=\!\Delta \rho_j \, \forall i,j: 0\!\leq\! i,j \!\leq \!d_B\!-\!1\Bigg\}\label{eq:M_mDI,d}
    \end{align}
\end{subequations}
on channels $\mE^{A\to B}$. The following Proposition characterizes the optimal conversion distance to Fourier transforms under $\mMIO_1$ and $\mDI_1$ in terms of these monotones.

\begin{restatable}{proposition}{ConversionMaximalSet}
\label{prop:ConversionMaximalSet}
Let $\mE$ be any quantum channel, $\mF_d$ the $d$-dimensional Fourier transform, $d\geq 2$, and $ \O \in \{\MIO,\DI \}$, then
\begin{align}
   \min_{\mS \in \mO_1}\tfrac{1}{2}\norm{\mS[\mE]-\mF_d}_\diamond&= 1-\max_{\mS \in \mO_1} F_{\min}(\mS[\mE],\mF_d)  \nonumber \\
   &=1-\frac{1}{d}M_{\mOO,d}(\mE).
\end{align}
\end{restatable}
Again, note that here the worst-case fidelity and diamond distance based conversion distances yield the same conversion error. However, the conversion error is different from the one under compatible superchannels derived in Thm.~\ref{thm:CoherenceOptimalConversion_main}.
\begin{restatable}{proposition}{DIConversionDistance}
\label{cor:DIConversionDistance}
    Let $\O \in \{\MIO,\DI\}$.  There exist quantum channels $\mE$ such that
    \begin{align}
        M_{\mOO,2}(\mE) > M_{\cO,2}(\mE).
    \end{align}
    In particular, this implies that the optimal conversion distances under $\cO_1$ and $\mO_1$ can differ.
\end{restatable}

Next, we turn to the question of whether adaptive protocols provide an advantage for coherence distillation. If the set of free channels contains identity channels (cf. assumption~\ref{assump:identity}) and SWAP channels (which is the case for dynamical coherence), adaptive protocols must perform at least as well as parallel protocols. However, as stated in Eq.~\eqref{eq:Cost_sequential} in Thm.~\ref{thm:Cost}, for the dilution of dynamical coherence, adaptive protocols provide no advantage. In contrast, we show in the following Proposition that adaptive $\cMIO$ protocols can outperform the parallel ones when it comes to one-shot distillation.

\begin{restatable}{proposition}{SequentialAdvantage}
\label{prop:SequentialAdvantage}
There exists a unitary channel $\mU$ such that
\begin{align*}
    \smash{\min_{\mS_2 \in \cMIO_2}} \tfrac{1}{2}\norm{\mS_2[\mU,\mU]\!-\!\mF_2}_\diamond \!<\! \smash{\min_{\mS_1 \in \cMIO_1}} \tfrac{1}{2}\norm{\mS_1[\mU^{\otimes 2}]\!-\!\mF_2}_\diamond\!.
\end{align*}
In particular, this implies that in the finite-copy setting, adaptive distillation can outperform parallel distillation even for unitary channels. 
\end{restatable}

Here, it is worth noting that Ref.~\cite[Cor.~1 and the following discussion]{Liu2020} showed that even in the asymptotic setting, adaptive distillation of unitary channels can outperform parallel distillation when MIO-networks are considered free. Solving the problem of adaptive finite-copy channel distillation under cMIO and cDI is currently an open question, but is in principle approachable via our SDPs. This is related to the question of whether optimal adaptive distillation protocols can be reduced to ``quasi-parallel" ones, as in Ref.~\cite{Ahnefeld2025} for phase-estimation. 

\section{DISCUSSION AND OUTLOOK}\label{sec:discussion_outlook}
In this work, we develop a general framework for higher-order quantum resources that treats the resourcefulness of quantum states, channels, and higher-order processes in a unified manner. By generalizing the spirit of completely resource non-generating superchannels to arbitrary supermaps, we introduce the notion of \textit{compatibility}, which guarantees that \textit{arbitrary} compositions of free supermaps remain free.
Additionally, we consider free networks, which decompose into free quantum channels, and maximally free superchannels.

In the case of approximate channel-to-channel conversion under maximally free superchannels, we extend the results of Ref.~\cite{Regula2021b} from the worst-case fidelity to the diamond distance and show that under suitable conditions, the approximate one-shot distillable resources with respect to the worst-case fidelity and the diamond distance coincide.

Specializing our framework to dynamical coherence, we prove the existence of largest sets of supermaps compatible with creation-incoherent, detection-incoherent, and dephasing-covariant channels and characterize them in terms of semidefinite constraints. With this, we prove that the three sets of free supermaps we consider are, in general, different. We continue by investigating the operational differences of these sets in one-shot channel distillation and dilution, for which we provide solutions in both the parallel and the adaptive settings: in the creation-incoherent setting, free networks and the maximal set of compatible supermaps lead, for example, to the same one-shot channel-cost, while in the detection-incoherent setting, there exist channels for which this is not the case. Moreover, we show that in the finite-copy setting, adaptive distillation protocols outperform parallel ones on certain channels.

Our findings suggest several directions for future research. Within resource theories of dynamical coherence, our characterizations of the compatible supermaps allow us to formulate the conversion distance between arbitrary supermaps as the solution of a semidefinite program. This opens the possibility to investigate the coherence cost and distillable coherence of arbitrary supermaps beyond channels. A similar approach is of course possible in other resource theories and could be used to determine, e.g., bounds on the non-stabilizerness needed to implement a quantum strategy that are automatically optimized over all gate decompositions.

When considering other resource theories, a natural first step is to see whether our coherence-specific results extend to resource theories that admit resource-destroying channels. A resource-destroying channel sends every state to a free state while leaving all free states invariant~\cite{Liu2017}, thereby providing a simple characterization of channels that cannot create, detect, or both create and detect resources analogous to Def.~\ref{def:FreeOperation}. It is then easy to see that variations of the conditions in Thm.~\ref{thm:CompatibleSupermaps} guarantee compatibility, and it remains to check whether they are necessary.

An example of such a resource theory is athermality under Gibbs-preserving operations, where the Gibbs-replacement channels are resource-destroying. In this setting, compatible supermaps might be a useful tool for investigating thermodynamic restrictions in relevant tasks, compare, e.g., Ref.~\cite{Chen2026}, which investigated parameter estimation under energy-constrained supermaps. As shown recently in Ref.~\cite{Yoeli2026}, it is possible to construct resource theories in which different systems are equipped with different resource-destroying maps, thereby unifying the resource theories of coherence, athermality, and nonuniformity.  Starting from the destruction-covariant free channels of Ref.~\cite{Yoeli2026}, one could seek compatible higher-order extensions that accommodate different resource restrictions within a single set of supermaps. This would provide a setting for studying conversions between coherence-related and thermodynamic channel resources, and it would be interesting to see which aspects of the unified state-conversion theory persist for channels and higher-order processes.

\section*{Acknowledgments}
We thank Koenraad Audenaert, Lea Lautenbacher, Rafael Wagner, Tore Friis, and F.A.'s cat Lina for discussions. T. T. acknowledges support from a Postdoc Scholarship on Quantum Algorithms or Quantum Software from the Danish e-infrastructure Consortium (DeiC). For numerics we used Refs.~\cite{cvx,Lofberg2004,Sturm1999}. We acknowledge OpenAI's GPT-5.6 Sol (and in part GPT-6 Astra) for text polishing and for simplifying the proof of Prop.~\ref{cor:DIConversionDistance}. All references were added manually. The Appendix also contains Refs.~\cite{Bhatia2000, Hayashi_2002, Winter2025,Hirche2023,Schwinger1960}.

\bibliography{bibliography.bib}

\appendix

\onecolumngrid
\makeatletter
\let\set@footnotewidth\@empty
\onecolumn@grid@setup
\makeatother

\section{GENERALIZED DIAMOND DISTANCE AND SUPERMAP MONOTONES}\label{sec:Appen_GenDiaDist}
In this appendix, we review and prove facts that we will use in other parts of this paper. We prove the sharpened lower bound in the Fuchs-van de Graaf inequalities in Eq.~\eqref{eq:FuchsVanDeGraaf_channels_sharp}, review an operational characterization of the generalized diamond distance, and show how to evaluate it via an SDP. We then discuss how this SDP is connected to other variational characterizations of the (generalized) diamond distance and prove Prop.~\ref{prop:SupermapMonotone}. 

We begin with the sharpened lower bound: For any state $\rho$ and pure state $\psi$, the variational characterization of the trace distance~\cite{Watrous2018} implies
\begin{align}\label{eq:boundTraceDist}
    \tfrac{1}{2}\norm{\rho-\psi}_1 &= \max_{0\leq P\leq \id}\Tr{P(\rho-\psi)} \geq \Tr{(\id-\psi)(\rho-\psi)}=1-\Tr{\psi\rho}=1-F(\psi,\rho).
\end{align}
For an arbitrary channel $\mM$ and an isometric channel $\mV=V\cdot V^\dagger$, this yields
\begin{align}\label{eq:FuchsVanDeGraaf_channels_sharp_appen}
     \tfrac{1}{2}\norm{\mM-\mV}_\diamond&=\sup_\rho \tfrac{1}{2}\norm{\idChan \otimes \,\mM(\rho)-\idChan \otimes \,\mV(\rho)}_1 \overset{\text{(i)}}{=} \sup_{\psi} \tfrac{1}{2}\norm{\idChan \otimes \,\mM(\psi)-\idChan \otimes \,\mV(\psi)}_1 \nonumber \\
 &\overset{\text{(ii)}}{\geq} 1- \inf_{\psi}F(\idChan \otimes\, \mM(\psi),\idChan \otimes\, \mV(\psi)) \overset{\text{(iii)}}{=} 1- F_{\min}(\mM,\mV),
\end{align}
where we used that (i) the trace norm is convex, and thus attains is maximum on pure states, (ii) the fact that $\idChan\otimes\,\mV(\psi)$ is pure together with Eq.~\eqref{eq:boundTraceDist}, and (iii) that the worst-case fidelity attains its minimum on pure states~\cite{Regula2021c}.

To review an operational characterization of the generalized diamond distance, we start from the diamond distance for quantum channels. The diamond distance in Eq.~\eqref{eq:def_DiamondDistance} admits an operational interpretation in terms of channel discrimination: A discrimination strategy consists of preparing an entangled input state, applying the unknown channel to one subsystem, and performing a joint measurement on the output and the reference system. The diamond distance equals the largest total-variation~\cite{Watrous2018} distance between the resulting outcome distributions, optimized over all such strategies~\cite{Watrous2018}. To extend this interpretation to quantum supermaps, the discrimination strategy must also account for their sequential structure. As illustrated in Fig.~\ref{fig:tester}, one prepares an initial entangled state, interacts with the unknown supermap through a sequence of (isometric) channels while retaining quantum memory, and finally performs a measurement, see also Ref.~\cite{Hirche2023}.

\begin{figure}[ht]
    \centering
    \scalebox{0.7}{\includegraphics[width=1\linewidth]{ 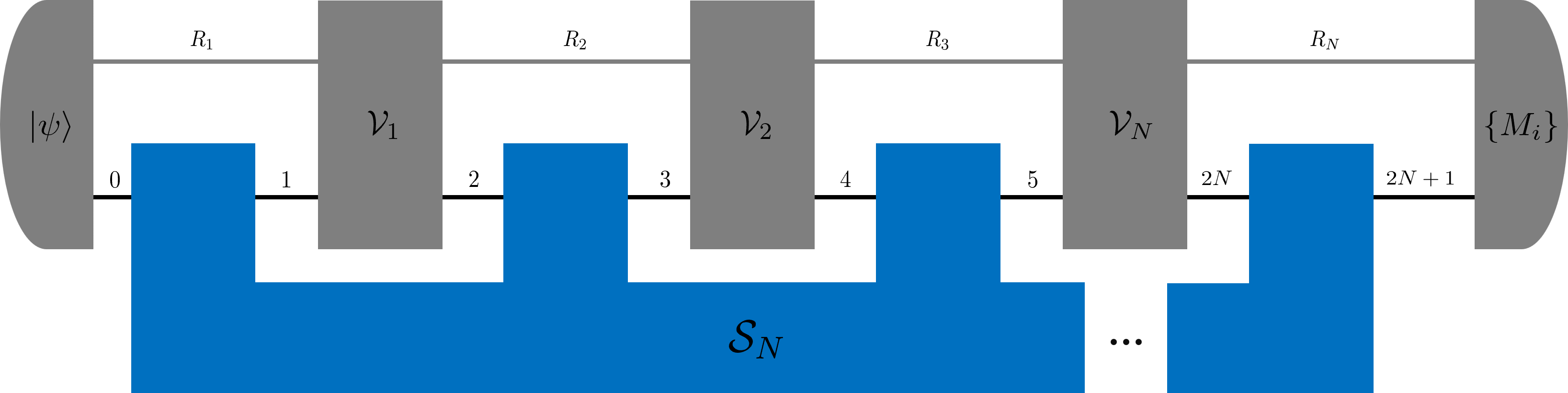}}
    \caption{Discrimination strategy between quantum supermaps consisting of a sequence of state preparation, isometries, and a final measurement. Without loss of generality, the initial state can be purified, and all inserted channels can be dilated to isometries.}
    \label{fig:tester}
\end{figure}

For two causally ordered supermaps $\mS_N,\mR_N$ with the same input and output systems, one therefore defines the generalized diamond distance as $d_{\diamond,N}(\mS_N,\mR_N):=\sup\tfrac12\sum_i\left|p(i|\mS_N)-p(i|\mR_N)\right|,$ 
where the supremum is over all such discrimination strategies. The comb formalism~\cite{Chiribella2009} reviewed in Sec.~\ref{sec:Supermaps} provides an exact characterization of those discrimination strategies in terms of so-called quantum testers~\cite{Chiribella2009}: A quantum tester is a collection of positive semidefinite operators $\{J_{\mP_i}\}$ such that $J_\mT:=\sum_i J_{\mP_i}$  is in the set (see Ref.~\cite[Def.~11 and Lem.~8]{Chiribella2009})
\begin{align}\label{eq:tester_conditions}
    \mathfrak{T}_N:=\Bigr\{& J_\mT\!=\! J_{\mT}^{(N)}\! \otimes\! \id^{2N+1}: J_{\mT}^{(j)} \geq 0,  \forall j: 1\leq j\leq N: \partTr{2j}{\!J_{\mT}^{(j)}}\!=\!J_{\mT}^{(j-1)} \!\otimes\! \id^{2j-1}, \Tr{J_{\mT}^{(0)}}\!=\!1  \Bigr\}.
\end{align}
Any discrimination strategy as in Fig.~\ref{fig:tester} corresponds to a tester $\{J_{\mP_i}\}$ in the sense that the probability of obtaining outcome $i$ in the final POVM is given by $p(i|\mS_N)=\Tr{J_{\mP_i}^T J_{\mS_N}}$, where the tester is independent of $\mS_N$.

Conversely, Ref.~\cite[Thm.~11]{Chiribella2009} establishes that every quantum tester can be realized by a discrimination strategy as in Fig.~\ref{fig:tester}. As such, the generalized diamond distance can equivalently be written as
\begin{align}\label{eq:generalizedDiamond_tester}
    d_{\diamond, N}(\mS_N,\mR_N) \!&= \!\max_{ \substack {J_{\mP_i}\geq 0\\ \sum_i J_{\mP_i} \in \mathfrak{T}_N}} \! \frac{1}{2} \sum_i \left| \Tr{J_{\mP_i}^T (J_{\mS_N}\!-\!J_{\mR_N})}  \right| = \max_{\!J_{\mT_N} \in \mathfrak{T}_N}  \frac{1}{2} \norm{\sqrt{J_{\mT_N}^T}(J_{\mS_N}\!-\!J_{\mR_N})\sqrt{J_{\mT_N}^T}}_1,
\end{align}
where the last equality follows from the tester realization in Ref.~\cite[Thm.~12]{Chiribella2009}, see Ref.~\cite[Chapter V]{Chiribella2009} for more details. 

The following Lemma connects this characterization of the diamond distance to the one in Eq.~\eqref{eq:diamond_generalized_main} in the main text.  Notably, similar characterizations of the generalized diamond distance appeared in Ref.~\cite{Gutoski2012} under the name ``r-norm". For the sake of self-containedness, we state the Lemma and its proof nevertheless.
\begin{restatable}{lemma}{LemmaDiamondDistance} \label{def:GenDiamondDist}
    Let $\mS_N,\mR_N\in \mathfrak{S}_N$ with the same input and output systems. Then
    \begin{align}\label{eq:Gen_Diamond_Dist}
        d_{\diamond, N}(\mS_N,\mR_N)=&\max_{\!J_{\mT_N} \in \mathfrak{T}_N}  \frac{1}{2} \norm{\sqrt{J_{\mT_N}^T}(J_{\mS_N}\!-\!J_{\mR_N})\sqrt{J_{\mT_N}^T}}_1=\min\left\{  \lambda\geq 0  : J_{\mS_N} - J_{\mR_N} \leq  \lambda J_{\mM_N},  \; J_{\mM_N} \in \Comb_N \right\},
    \end{align}
    which is the solution of the SDP 
    \begin{subequations}\label{eq:SDP_gen_diamond_dist}
    \begin{alignat}{2} 
        &\,\text{minimize}\quad &&  \frac{\lambda}{2} \\
        &\,\text{subject to}\quad &&  Y^{(j)} \in \Herm \,\forall j:0\leq j\leq N ,\lambda\in \mathbb{R}\\
        & && Y^{(N)} \geq J_{\mS_N}-J_{\mR_N} \geq -Y^{(N)} \\
        & &&  \partTr{2j+1}{Y^{(j)}}\leq \id^{2j} \otimes Y^{(j-1)} \, \forall j:1\leq j\leq N \label{eq:SDP_gen_diamond_dist_d}\\
        & &&  \partTr{1}{Y^{(0)}}\leq \lambda \id^0,\label{eq:SDP_gen_diamond_dist_e}
    \end{alignat}
    \end{subequations}
    where, without loss of generality, the inequalities in Eqs.~\eqref{eq:SDP_gen_diamond_dist_d} and~\eqref{eq:SDP_gen_diamond_dist_e} can be replaced by equalities.
    Moreover, the generalized diamond distance is contractive under compositions with supermaps~\cite{Chiribella2009}, i.e., for any $\mT_M\in \mathfrak{S}_M$ such that $\mT_M*\mS_N \in \mathfrak{S}_K$, it holds that $d_{\diamond, K}(\mT_M*\mS_N,\mT_M*\mR_N) \leq d_{\diamond, N}(\mS_N,\mR_N)$.
\end{restatable}

\begin{proof}
    First, we rewrite the generalized diamond distance from Eq.~\eqref{eq:generalizedDiamond_tester}. Note that from the variational characterization of the trace distance~\cite{Watrous2018}, for every Hermitian $X$, we get
\begin{align}
     \norm{\sqrt{J_{\mT_N}^T} X\sqrt{J_{\mT_N}^T}}_1 = \max_{-\id \leq H\leq \id} \Tr{H\sqrt{J_{\mT_N}^T} X\sqrt{J_{\mT_N}^T}} =  \max\left\{ \Tr{DX}: -J_{\mT_N}^T \leq D\leq J_{\mT_N}^T \right\},
\end{align} 
where the last equality is due to the following argument: For a feasible $H$, choose $D:=\sqrt{J_{\mT_N}^T} H\sqrt{J_{\mT_N}^T}$, and thus, e.g.,
\begin{align}
    D\le \sqrt{J_{\mT_N}^T} \sqrt{J_{\mT_N}^T}=J_{\mT_N}^T.
\end{align}
The right-hand side is therefore an upper bound. 

Conversely, note that $-J_{\mT_N}^T\leq D\leq J_{\mT_N}^T$ implies $\supp(D) \subseteq \supp(J_{\mT_N}^T)$, and thus defining $H:=\left(J_{\mT_N}^T\right)^{-1/2} D\left(J_{\mT_N}^T\right)^{-1/2} $ (where the inverse is to be understood as the generalized Moore-Penrose inverse on the support of $J_{\mT_N}^T$) yields $-\id\leq H\leq \id$. 

Next, let $X=J_{\mS_N}-J_{\mR_N}$, and notice that the generalized diamond distance can thus be written as
\begin{align}\label{eq:diamond_dist_generalized_intermediate}
      d_{\diamond,N}(\mkern-1mu\mS_N,\!\mR_N\mkern-1mu) \!=\! \tfrac{1}{2}\max\left\{\mkern-1mu \Tr{DX}\!:\! \shortminus J_{\mT_N}^T\!\leq\! D\!\leq\! J_{\mT_N}^T, J_{\mT_N} \!\in\! \mathfrak{T}_N   \mkern-1mu\right\}\!=\!\tfrac{1}{2}\max\left\{\mkern-1mu\Tr{(P_0\!\shortminus\!P_1)X}\!:\! P_0,P_1\!\geq\! 0, P_0\!+\!P_1\!\in\! \mathfrak{T}_N   \mkern-1mu\right\}.
\end{align}
The last equality follows from relaxing the problem first by taking
\begin{align}
    P_0=\tfrac{1}{2}\left(J_{\mT_N}^T+D\right) \quad \text{and}\quad  P_1=\tfrac{1}{2}\left(J_{\mT_N}^T-D\right)
\end{align}
for a feasible $D$.
Then, clearly $P_0,P_1 \geq 0$, and $P_0+P_1 \in \mathfrak{T}_N$, since this set is invariant under transposition. Conversely, the equality is achieved by taking $J_{\mT_N}=(P_0+P_1)^T$ and $D=P_0-P_1$. 

We now switch to the dual problem of the right-hand side of Eq.~\eqref{eq:diamond_dist_generalized_intermediate} where, for convenience, we consider the optimization problem corresponding to $2d_{\diamond,N}$: We introduce positive semidefinite auxiliary variables $J_{\mT_N}$ and $J_{\mT}^{(j)}$, $0\leq j\leq N$, subject to $J_{\mT_N}=P_0+P_1$ and the tester normalization constraints in Eq.~\eqref{eq:tester_conditions}. Furthermore, we introduce Lagrange multipliers $Z\in \Herm, Y^{(j)} \in \Herm \, \forall j: 0\leq j\leq N$ and $\lambda \in \mathbb{R}$. The Lagrangian is then
\begin{align}
    L(P_0,P_1,J_{\mT_N}, \{J_{\mT}^{(j)} \}; Z, \{Y^{(j)} \},\lambda) =& \Tr{(P_0-P_1)X} +\Tr{Z(J_{\mT_N}- P_0-P_1)}+ \Tr{Y^{(N)}\left(J_{\mT}^{(N)} \otimes \id^{2N+1}-J_{\mT_N} \right)} \nonumber\\
    &+\sum_{j=1}^{N} \Tr{Y^{(j-1)}\left(J_{\mT}^{(j-1)} \otimes \id^{2j-1} -\partTr{2j}{J_{\mT}^{(j)}}\right)} +\lambda\left(1-\Tr{J_{\mT}^{(0)}}\right) \nonumber \\
    =& \lambda +\Tr{\left(X-Z \right)P_0}+\Tr{\left(-X-Z \right)P_1}+\Tr{\left(Z-Y^{(N)} \right)J_{\mT_N}}  \\
    &+\sum_{j=1}^{N} \Tr{\left(\partTr{2j+1}{Y^{(j)}}\!-\!\id^{2j} \otimes Y^{(j-1)}\right) J_{\mT}^{(j)}} \!+\!\Tr{\left(\partTr{1}{Y^{(0)}}-\lambda \id^0\right)J_{\mT}^{(0)}}.\nonumber
\end{align}
Taking the supremum over the primal variables yields the dual characterization
\begin{subequations}
\begin{alignat}{2} 
    d_{\diamond,N}(\mS_N,\mR_N)= &\,\text{minimize}\quad &&  \frac{\lambda}{2} \\
    &\,\text{subject to}\quad && Z\in \Herm, Y^{(j)} \in \Herm \,\forall j:0\leq j\leq N ,\lambda\in \mathbb{R}\\
    & && Z \geq X, Z\geq -X, Y^{(N)} \geq Z \\
    & &&  \partTr{2j+1}{Y^{(j)}}\leq \id^{2j} \otimes Y^{(j-1)} \, \forall j:1\leq j\leq N\\
    & &&  \partTr{1}{Y^{(0)}}\leq \lambda \id^0,
\end{alignat}
\end{subequations}
where strong duality holds by Slater's conditions~\cite{Boyd2004} applied to the choice  $P_0=P_1=J_{\mT_N}/2$ for a full-rank tester $J_{\mT_N}$. We can now eliminate the variable $Z$ by noting that $Z \geq X, Z\geq -X, Y^{(N)} \geq Z$ implies $Y^{(N)} \geq X$ and $Y^{(N)} \geq -X$. By choosing $Z=Y^{(N)}$, these conditions are also sufficient. Moreover, this also implies that $ Y^{(N)}\geq 0$, and thus,  $ Y^{(j)}\geq 0$ for all $j$. Then, 
\begin{align}
    d_{\diamond,N}(\mS_N,\mR_N)\!=\!\min\left\{\frac{\lambda}{2}: Y^{(N)} \geq X\geq -Y^{(N)}, Y^{(j)}\geq 0, \partTr{2j+1}{Y^{(j)}}\leq \id^{2j} \otimes Y^{(j-1)}, \partTr{1}{Y^{(0)}}\leq \lambda \id^0 \right\}.
\end{align}
Next, we note that without loss of generality, we can assume that the constraints $\partTr{2j+1}{Y^{(j)}}\leq \id^{2j} \otimes Y^{(j-1)}, \partTr{1}{Y^{(0)}}\leq \lambda \id^0$ hold with equality, see also Ref.~\cite[Appendix~A]{Gutoski2012}: Assume that we have a set of feasible variables for which this is not the case, let $\tau^1$ be some quantum state, and define $\tilde{Y}^{(0)}:=Y^{(0)}+\left(\lambda \id^0-\partTr{1}{Y^{(0)}}\right) \otimes \tau^1$. Feasibility implies that $\tilde{Y}^{(0)} \geq Y^{(0)}$ and by construction, $\partTr{1}{\tilde{Y}^{(0)}}=\lambda \id^0$. We can now proceed recursively: Assume that $\tilde{Y}^{(j-1)}$ with $\tilde{Y}^{(j-1)} \geq Y^{(j-1)}  $  has been defined and choose
\begin{align}
    \tilde{Y}^{(j)}:= Y^{(j)}+\left(\id^{2j}\otimes \tilde{Y}^{(j-1)}-\partTr{2j+1}{Y^{(j)}} \right) \otimes  \tau^{2j+1}
\end{align}
for some quantum state $\tau^{2j+1}$. Then, $\partTr{2j+1}{\tilde{Y}^{(j)}} = \id^{2j}\otimes \tilde{Y}^{(j-1)}$, and 
\begin{align}
    \id^{2j}\otimes \tilde{Y}^{(j-1)} \geq \id^{2j}\otimes Y^{(j-1)} \geq \partTr{2j+1}{Y^{(j)}}
\end{align}
implies $\tilde{Y}^{(j)} \geq Y^{(j)}$ and thus $\tilde{Y}^{(N)} \geq X\geq -\tilde{Y}^{(N)}$. We thus constructed a new set of feasible variables satisfying the intended inequality constraints with equality. Comparing with Eq.~\eqref{eq:CombSet}, the constraints on $\{\tilde{Y}^{(j)}\}$ correspond to the constraints on a comb rescaled by $\lambda$. Finally, we thus arrive at
\begin{align}
    d_{\diamond,N}(\mS_N,\mR_N)
     = & \min\left\{ \frac{\lambda}{2}  :  -\lambda J_{\mM_N}  \leq  X  \leq  \lambda J_{\mM_N},  \; J_{\mM_N} \in \Comb_N \right\}\nonumber \\
     = & \min\left\{  \lambda\geq 0  : J_{\mS_N} - J_{\mR_N} \leq  \lambda J_{\mM_N},  \; J_{\mM_N} \in \Comb_N \right\},
\end{align}
where the last line follows from a standard argument: Let the pair $(\lambda, J_{\mM_N})$ be feasible in the first line. Hence, 
\begin{align}
    J_{\mM_N'}:=J_{\mM_N}+X/\lambda \ge J_{\mM_N}-J_{\mM_N} =0,
\end{align}
where the division by $\lambda$ is justified because if $\lambda=0$, then automatically $0=X=J_{\mS_N}-J_{\mR_N}$.
Also since $X=J_{\mS_N}-J_{\mR_N}$ is the difference of two normalized combs, adding $X/\lambda$ preserves the partial trace conditions a comb needs to satisfy (see Eq.~\eqref{eq:CombSet}). Hence, $J_{\mM_N'}$ is a comb too and 
\begin{align}
    X\leq \lambda J_{\mM_N}  = \lambda J_{\mM_N'} -X 
\end{align}
implies $X\leq (\lambda/2)J_{\mM_N'}$. Choosing $\lambda'=\lambda/2$, we thus constructed a feasible point $(\lambda', J_{\mM_N'})$ for the second line with the same objective value.
Conversely, starting from the second line, if $X\leq\lambda J_{\mM_N}$, then $J_{\mM_N'}:=J_{\mM_N}-X/\lambda\in\Comb_N$ and $ J_{\widetilde\mM_N}:=(J_{\mM_N}+J_{\mM_N'})/2\in\Comb_N$, which satisfies 
\begin{align}
    X\le \lambda J_{\mM_N}=\frac{\lambda}{2}\left( J_{\mM_N}+ J_{\mM_N}\right)= \frac{\lambda}{2}\left( \frac{X}{\lambda} + J_{\mM_N'}+ J_{\mM_N}\right)=\frac{X}{2} +\lambda J_{\widetilde\mM_N}.
\end{align}
Together with the analogous lower bound, this implies that $-2\lambda  J_{\widetilde\mM_N}\leq X
\leq2\lambda J_{\widetilde\mM_N}$. Any feasible point in the second line thus implies a feasible point in the first line with the same objective value. In combination, this concludes the argument.

Lastly, we need to show that the generalized diamond distance is contractive under compositions with arbitrary supermaps. This essentially follows directly from its construction involving an optimization over all discrimination strategies, see also Ref.~\cite{Chiribella2009}. More formally, let $\mT_M \in \mathfrak{S}_M$ denote an arbitrary supermap such that the composition $\mT_M*\mS_N \in \mathfrak{S}_K$ for some $K$. Let $(\lambda^\star,J_{\mM_N^\star})$ denote an optimal feasible point of
\begin{align}
    d_{\diamond,N}(\mS_N,\mR_N) &=\min\left\{ \lambda \geq 0: J_{\mS_N}-J_{\mR_N}\leq \lambda J_{\mM_N}, J_{\mM_N} \in \Comb_N \right\}=\lambda^\star.
\end{align}
Then, using the fact that $J_{\mT_M}*\left(J_{\mS_N}-J_{\mR_N}\right) \leq \lambda^\star J_{\mT_M}*J_{\mM_N^\star} $ (which is a direct consequence of the definition of the link product in Eq.~\eqref{eq:defLinkProduct}),  together with $J_{\mT_M}*J_{\mM_N^\star} \in \Comb_K$ (and likewise for $\mR_N$), yields a feasible point in
\begin{align}
    d_{\diamond,K}(\mT_M*\mS_N,\mT_M*\mR_N) &=\min\left\{ \lambda \geq 0: J_{\mT_M}*(J_{\mS_N}-J_{\mR_N})\leq \lambda J_{\mM_K}, J_{\mM_K} \in \Comb_K\right\}  \leq d_{\diamond,N}(\mS_N,\mR_N),
\end{align}
and thus, $d_{\diamond,N}$ is contractive under supermaps.
\end{proof}

If we consider channels, i.e., in our convention, if $N=0$, the generalized diamond distance reduces to the regular diamond distance between two quantum channels, see Ref.~\cite{Chiribella2009}. The diamond distance between two channls $\mN^{A\to B},\mM^{A\to B}$ can be calculated using the standard SDP~\cite{Watrous2009,Watrous2013} 
\begin{align}\label{eq:SPD_DiamondDistance}
    \tfrac{1}{2}\norm{\mN-\mM}_\diamond =\min \left\{ \norm{\partTr{B}{Z^{AB}}}_\infty: Z^{AB}\geq 0, Z^{AB} \geq J_\mN^{AB}-J_\mM^{AB}\right\},
\end{align}
where $\norm{\cdot}_\infty$ denotes the operator norm. This is clear from the same argument used in the previous Lemma: Specialized to channels, it follows from Eq.~\eqref{eq:Gen_Diamond_Dist} that
\begin{align}\label{eq:Gen_Diamond_Dist_channels}
        \tfrac{1}{2}\norm{\mN-\mM}_\diamond=d_{\diamond, 0}(\mN,\mM)=\min\left\{ \lambda\!\geq\! 0: \! J_{\mN}\!-\!J_{\mM}\!\leq\! \lambda J_{\mK}, J_{\mK} \!\in\! \Comb_0 \right\}.
    \end{align}
Take a feasible point $(\lambda, J_{\mK}^{AB})$ on the right-hand side of Eq.~\eqref{eq:Gen_Diamond_Dist_channels}. Define $Z^{AB}:= \lambda J_\mK^{AB}$, which is clearly feasible in Eq.~\eqref{eq:SPD_DiamondDistance}, and satisfies $\norm{\partTr{B}{Z^{AB}}}_\infty=\norm{\lambda \id^A}_\infty=\lambda$. We thus constructed a feasible point for the optimization problem in Eq.~\eqref{eq:SPD_DiamondDistance} with the same objective value. Conversely, take any feasible $Z^{AB}$ in Eq.~\eqref{eq:SPD_DiamondDistance} and define $\lambda:= \norm{\partTr{B}{Z^{AB}}}_\infty$. Thus, $\partTr{B}{Z^{AB}} \leq \lambda \id^A$, and without loss of generality, we assume $\lambda >0$ (otherwise there is nothing to show). Now, choosing any $\tau\in \D$, define $J_{\mK}^{AB} := \tfrac{1}{\lambda}\left(Z^{AB}+\left(\lambda \id^A-\partTr{B}{Z^{AB}}\right)\otimes \tau^B\right)$. Then, $J_{\mK}^{AB}\geq 0$, $\partTr{B}{J_{\mK}^{AB}}=\id^A$, and $\lambda J_{\mK}^{AB} \geq J_{\mN}^{AB}-J_{\mM}^{AB}$. Hence, $(\lambda,J_{\mK}^{AB})$ is feasible in Eq.~\eqref{eq:Gen_Diamond_Dist_channels} and yields the same objective value, establishing the equality of Eq.~\eqref{eq:SPD_DiamondDistance} and Eq.~\eqref{eq:Gen_Diamond_Dist_channels}.

Lastly, Prop.~\ref{prop:SupermapMonotone} from the main text, which we repeat for readability, follows immediately from the latter Lemma. 
\SupermapMonotone*
\begin{proof}
Let $\mT_M\in \mathfrak{C}_M$ be any free supermap for which $\mT_M*\mS_N \in \mathfrak{S}_K$ for some $K$. From the fact that the set $\mathfrak{C}$ is compatible and thus closed under \textit{arbitrary} compositions, it follows that
\begin{align}
    M_{\mathfrak{C}}(\mT_M*\mS_N)&= \inf_{\mR_K \in \mathfrak{C}_K} d_{\diamond,K}(\mT_M*\mS_N,\mR_K) \leq \inf_{\mR_N \in \mathfrak{C}_N} d_{\diamond,K}(\mT_M*\mS_N,\mT_M*\mR_N ) \leq \inf_{\mR_N \in \mathfrak{C}_N} d_{\diamond,N}(\mS_N,\mR_N) \nonumber \\
    &= M_{\mathfrak{C}}(\mS_N),
\end{align}
where we also used that the generalized diamond distance is contractive (see Lem.~\ref{def:GenDiamondDist}).

Note that if $\mathfrak{C}_N$ is characterized by semidefinite constraints, then Eq.~\eqref{eq:SDP_gen_diamond_dist} implies that this monotone can be computed via the SDP
\begin{subequations}
    \begin{alignat}{2} 
        &\,\text{minimize}\, \quad &&  \frac{\lambda}{2} \\
        &\,\text{subject to}\quad &&  Y^{(j)} \in \Herm \,\forall j:0\leq j\leq N ,\lambda\in \mathbb{R}\\
        & && Y^{(N)} \geq J_{\mS_N}-J_{\mR_N} \geq -Y^{(N)} \\
        & &&  \partTr{2j+1}{Y^{(j)}}\leq \id^{2j} \otimes Y^{(j-1)} \quad \forall j:1\leq j\leq N \\
        & &&  \partTr{1}{Y^{(0)}}\leq \lambda \id^0 \\
        & && \mR_N \in \mathfrak{C}_N,
    \end{alignat}
    \end{subequations}
    where again some of the inequalities can be replaced by equalities.
\end{proof}
Recall that by our assumption~\ref{assump:ClosedConvex}, the sets of free channels from $A$ to $B$, $\O(A\to B)$, are closed, and thus 
\begin{align}\label{eq:cO_closed}
    \overline{\cO}_0=\overline{\O}=\bigcup_{A,B} \overline{\O(A\to B)}=\bigcup_{A,B} \O(A\to B) =\cO_0.
\end{align}
Now assume that $\cO$ exists and take $\mS_N,\mT_M \in \overline{\cO} $. By the definition of the closure, there exist sequences of supermaps $\mS_N^{(n)},\mT_M^{(n)} \in \cO$ converging to $\mS_N,\mT_N$, respectively. Since the link product is a bilinear map between fixed finite-dimensional matrix spaces, it is jointly continuous, and thus $\mS_N^{(n)}*\mT_M^{(n)}$ converges to $\mS_N*\mT_M$ (whenever the link-product is well-defined). As such, $\mS_N*\mT_M\in \overline{\cO}$. Therefore, $\overline{\cO}$ is itself a compatible family with the fixed free channels $\O$. Since $\cO$ is the largest such family, it must contain $\overline{\cO}$, and thus $\overline{\cO}=\cO$. This is interesting in itself because it shows that if a maximal set of compatible supermaps exists, then it is closed, which is desirable for the same reasons that lead to assumption~\ref{assump:ClosedConvex}. One of the convenient technical consequences is that if $\mathfrak{C}=\cO$, then the infimum in Prop.~\ref{prop:SupermapMonotone} is achieved: For fixed input and output systems and causal order, the set of deterministic supermaps is compact in the Choi representation. The corresponding set of compatible supermaps is closed by the preceding argument and since $\cO \subset \mathfrak{S}$ it is therefore compact. Since it is non-empty (it contains free networks if assumptions~\ref{assump:identity} to~\ref{assumpt:complfree} hold), every continuous real-valued objective on this set attains its minimum and maximum.

\section{ONE-SHOT MANIPULATION OF QUANTUM CHANNELS UNDER MAXIMALLY FREE SUPERCHANNELS}\label{sec:Appen_GeneralConversion}

In this section, we consider free channels $\O$ for which assumption~\ref{assump:ClosedConvex} holds and prove Thm.~\ref{thm:DiamondBound_Bounds_G}, Thm.~\ref{thm:Thm1_equivalent}, and Cor.~\ref{thm:Thm3_equivalent}. To this end, recall that
\begin{align}
    R_H^\epsilon(\rho||\sigma)= \max\{\Tr{P\sigma}^{-1}: 0\leq P\leq \id, \Tr{P\rho}\geq 1-\epsilon\},
\end{align} 
and that $\aff(\cdot)$ denotes the affine hull. We first rewrite Eqs.~\eqref{eq:R_H,O}, i.e., the monotones
\begin{subequations}
\begin{align}\label{eq:R_H,O_appen}
    R_{H,\O}^\epsilon(\mE^{A\to B}) &= \min_{\mM\in \O(A\to B)} \max_{\psi^{AR}\in\Pure(AR)} R_H^\epsilon(\idChan^{R} \otimes\, \mE^{A\to B}(\psi)||(\idChan^R\otimes \,\mM^{A\to B}(\psi^{AR})),\\
    R_{H,\aff(\O)}^\epsilon(\mE^{A\to B}) &= \inf_{\mM\in \aff(\O(A\to B))} \max_{\psi^{AR}\in\Pure(AR)} R_H^\epsilon(\idChan^{R} \otimes\, \mE^{A\to B}(\psi)||(\idChan^R\otimes \,\mM^{A\to B}(\psi^{AR})),
\end{align}
\end{subequations}
where $R\cong A$ and $\Pure(AR)$ denotes the set of all pure states on $AR$. As shown in Ref.~\cite[Lem.~5]{Regula2021b}, for any topologically closed and convex set $\O(A\to B)$, 
\begin{subequations}\label{eq:R_H_epsilon}
\begin{align}
    R_{H,\O}^{\epsilon}(\mN)^{-1}= \min \left\{ \lambda: \Tr{P (\idChan \otimes \,\mM)(\psi)}\leq \lambda \, \forall \mM\in \O, 0\leq P\leq \id, \psi\in \Pure, \Tr{P (\idChan \otimes \, \mN)(\psi)}\geq 1 -\epsilon\right\}, \label{eq:R_H_smooth}\\
    R_{H,\aff(\O)}^{\epsilon}(\mN)^{-1}= \min \left\{ \lambda: \Tr{P (\idChan \otimes\, \mM)(\psi)}= \lambda \, \forall \mM\in \O, 0\leq P\leq \id, \psi\in \Pure, \Tr{P (\idChan \otimes\, \mN)(\psi)}\geq 1 -\epsilon\right\} \label{eq:R_H_aff_smooth},
\end{align}
\end{subequations}
where we omitted the system indices since they are fixed by the channel $\mN^{A\to B}$ (and the auxiliary system is still $R\cong A$). If clear from the context, we continue to omit the system labels from now on. Furthermore, for an isometric channel $\mV$, note that according to Ref.~\cite[Eqs.~(A14) and~(A15)]{Regula2021b}, it holds that 
\begin{subequations}\label{eq:min_entropy_combined}
\begin{align}
    R_{\min,\O}(\mV)^{-1}&= \min \left\{ \lambda: \Tr{(\idChan \otimes\, \mV)(\psi)(\idChan \otimes\, \mM)(\psi)}\leq \lambda \, \forall \mM\in \O, \psi\in \Pure\right\}, \label{eq:min_entropy} \\
    R_{\min,\aff(\O)}(\mV)^{-1}&= \inf \left\{ \lambda: \Tr{(\idChan \otimes\, \mV)(\psi)(\idChan \otimes\, \mM)(\psi)}= \lambda \, \forall \mM\in \O, \psi\in \Pure\right\}\label{eq:min_entropy_aff},
\end{align}
\end{subequations}
where we use the convention that $\inf \emptyset=\infty.$ As noted in Ref.~\cite[SM Remark beneath Thm.~1]{Regula2021b}, it holds that $R_{\min,\O}(\mN)=R_{H,\O}^{\epsilon=0}(\mN)$, whereas, in general, $R_{H,\aff(\O)}^{\epsilon=0}(\mN) \neq R_{\min,\aff(\O)}(\mN)$. For $m>0$, the monotones from Eq.~\eqref{eq:G_O} can be written as~\cite[SM Sec.~A]{Regula2021b}
\begin{subequations}\label{eq:G_O_alternative}
\begin{align}
    G_{O}(\mE,m) &= \max \left\{\Tr{W \left(\idChan\otimes \, \mE \right)(\psi)} : 0\leq W\leq \id, \psi\in \Pure, \Tr{W\left(\idChan \otimes \,\mM\right)(\psi)} \leq \tfrac{1}{m} \, \forall \mM \in\O\right\}, \nonumber\\
    &= \max \left\{ \Tr{J_\mE W}: 0\leq W \leq \rho \otimes \id, \rho\in \D, \Tr{W J_\mM}\leq \tfrac{1}{m}\, \forall \mM \in \O\right\}, \label{eq:eq:G_O_alternativeA}\\
    G_{\aff(O)}(\mE,m) &= \sup \left\{\Tr{W \left(\idChan\otimes\, \mE \right)(\psi)} : 0\leq W\leq \id,\psi\in \Pure, \Tr{W\left(\idChan \otimes \,\mM\right)(\psi)} = \tfrac{1}{m} \, \forall \mM \in\O\right\} \nonumber \\
    &= \sup\left\{ \Tr{J_\mE W}: 0\leq W \leq \rho \otimes \id, \rho\in \D, \Tr{W J_\mM}=\tfrac{1}{m} \,\forall \mM \in \O\right\},
\end{align}
\end{subequations}
where the systems are again fixed by $\mE$, and we recall that we adopt the convention that $G_{\aff(O)}(\mE,0):=1$. For fixed quantum systems $A,B$, define the dual cones~\cite{Regula2021b}
\begin{subequations}\label{eq:dualCones_appen}
\begin{align}
    \O(A\to B)^*:=\left\{X^{AB}\in \Herm: \Tr{X^{AB}J_{\mM}^{AB}} \geq 0 \, \forall \mM^{A\to B}\in\O(A\to B)\right\}, \label{eq:dualCones_appen_non_affine} \\
    \aff(\O(A\to B))^*: =\left\{X^{AB}\in \Herm: \Tr{X^{AB}J_{\mM}^{AB}} = 0 \, \forall \mM^{A\to B} \in\O(A\to B)\right\},\label{eq:dualCones_appen_aff}
\end{align}
\end{subequations}
and we will again write $\O^*, \aff(\O)^*$ if the systems are clear from the context. For $m>0$, this allows us to rewrite Eqs.~\eqref{eq:G_O_alternative} as
\begin{subequations}\label{eq:G_O_cone}
\begin{align}
    G_{\O}(\mE,m)&=\max \left\{ \Tr{J_\mE W}: 0\leq W \leq \rho \otimes \id, \rho \in \D, \tfrac{\id}{d_A}-mW \in \O^*\right\}, \\
    G_{\aff(\O)}(\mE,m)&=\sup \left\{ \Tr{J_\mE W}: 0\leq W \leq \rho \otimes \id, \rho \in \D, \tfrac{\id}{d_A}-mW \in \aff(\O)^*\right\}.\label{eq:G_aff_O}
\end{align}
\end{subequations}
To see the equality between Eqs.~\eqref{eq:G_O_cone} and Eqs.~\eqref{eq:G_O_alternative}, let $\mE^{A\to B}$ and let $m>0$. For every $\mM^{A\to B}$ we have that $\Tr{J_{\mM}^{AB}}=d_A$ since $\mM$ is trace-preserving, and as such, 
\begin{align}
    \Tr{\left(\tfrac{\id^{AB}}{d_A}-mW^{AB}\right)J_{\mM}^{AB}}=1-m\Tr{W^{AB}J_{\mM}^{AB}}.
\end{align}
Thus, for $m>0$, the constraint $\Tr{W^{AB}J_{\mM}^{AB}}\leq 1/m$ for every $\mM\in\O(A\to B)$ is equivalent to $\id^{AB}/d_A-mW^{AB}\in\O(A\to B)^*$, while $\Tr{W^{AB}J_{\mM}^{AB}}=1/m$ for every $\mM\in\O(A\to B)$ is equivalent to $\id^{AB}/d_A-mW^{AB}\in\aff(\O(A\to B))^*$. \\

As mentioned in the main text, Ref.~\cite[Thm.~6]{Regula2021b} provides upper and lower bounds on the conversion distance with respect to the worst-case fidelity: If $\mE$ and $\mN$ are quantum channels, with $\mN$ an isometry, and if $\O$ is a set of free channel satisfying  assumption~\ref{assump:ClosedConvex}, then
    \begin{align}\label{eq:thm:Fidelity_bound_G}
        G_{\O}\left(\mE, R_{\min ,\O}(\mN)\right)\!&\geq \!\max_{\mS \in \mO_1} \!F_{\min}(\mS[\mE],\mN) \geq G_{\O}\left(\mE, R_{s,\O}(\mN)\right)
    \end{align}
    if $\O$ is full-dimensional and 
    \begin{align}
        G_{\aff(\O)}\!\left(\mE,\! R_{\min ,\aff(\O)}(\mN)\right)\!&\geq\mkern-4mu\max_{\mS \in \mO_1} \!F_{\min}(\mS[\mE],\mN) \geq \mkern-4muG_{\aff(\O)}\!\left(\mE,\! R_{\max ,\O}(\mN)\right)
    \end{align}
    if $\O$ is reduced-dimensional.

Here, we want to note that Ref.~\cite[Thm.~6]{Regula2021b} requires that the target channel $\mN^{A\to B}$ maps pure states to pure states in a complete sense, i.e., $\idChan^R \otimes \,\mN^{A\to B}(\psi^{AR})\in \Pure(BR) \, \forall R, \forall  \psi^{AR}\in\Pure(AR)$. As mentioned in the main text, according to Ref.~\cite[Thm.~3.1]{Davies1976} and Ref.~\cite[Lem.~3.2]{Belzig2025}, this is equivalent to $\mN^{A\to B}$ being an isometry.

To justify the use of a maximum over $\mO_1$ here, notice that the set of superchannels is compact in the Choi representation. For each fixed free channel $\mM\in\O$, the map $J_{\mS}\mapsto J_{\mS[\mM]}$ is linear between finite-dimensional spaces and therefore continuous. Since the set of free channels $\O$ is closed, the condition $\mS[\mM]\in\O$ defines a closed subset of superchannels. Intersecting these subsets over all free input channels shows that $\mO_1$ is closed and hence compact. It is also nonempty, since it contains superchannels that output a fixed free channel (and we assumed that $\O(A\to B)$ is also non-empty). The worst-case fidelity (and diamond distance) are continuous on this set, so their respective minima and maxima are attained.

We are now ready to prove Thm.~\ref{thm:DiamondBound_Bounds_G}, i.e., the analog of Ref.~\cite[Thm. 6]{Regula2021b} for the diamond distance.

\DiamondBoundBoundsG*
\begin{proof}
Let $\mE^{A\to B}$ denote an arbitrary input channel, and let $\mN^{C\to D}$ denote an isometric target channel. We start with the full-dimensional case. The lower bound on the diamond distance follows directly from
Ref.~\cite[Thm. 6]{Regula2021b} combined with the sharpened Fuchs-van de Graaf inequality in Eq.~\eqref{eq:FuchsVanDeGraaf_channels_sharp}, as
\begin{align}
     \min_{ \mS \in \mO_1} \tfrac{1}{2}\norm{\mS[\mE^{A\to B}]-\mN^{C\to D}}_\diamond& \overset{\eqref{eq:FuchsVanDeGraaf_channels_sharp}}{\geq}   1-  \max_{ \mS \in \mO_1}F_{\min}(\mS[\mE],\mN) \overset{\eqref{eq:thm:Fidelity_bound_G}}{\geq}  1-G_{\O}(\mE,R_{\min,\O}(\mN)). 
\end{align}
For the upper bound, take a feasible point $(\psi^{AR}, W^{BR})$ of    
    \begin{align}\label{eq:G_O_Rs}
        G_{\O}(\mE,\mkern-2mu R_{s,\O}(\mN)) \!\overset{\eqref{eq:eq:G_O_alternativeA}}{=}\! \max \left\{\Tr{W \!\left(\idChan\otimes \mE \right)(\psi)} \!:\! 0\!\leq\! W\!\leq\! \id, \psi\in \Pure, \Tr{W\!\left(\idChan \otimes \mM\right)(\psi)} \leq R_{s,\O}(\mN)^{-1} \, \forall \mM \!\in\!\O\right\}.
    \end{align}
    Let $(\mQ^\star,\mK^\star)$ denote an optimal feasible point in the optimization in Eq.~\eqref{eq:def_R_s} for $R_{s,\O}$ , i.e., $\mQ^\star,\mK^\star \in \O$ such that $\mN +(R_{s,\O}(\mN)-1) \mQ^\star = R_{s,\O}(\mN) \mK^\star$. Such a feasible point exists since $\O(A\to B)$ is assumed to be closed and non-empty, and the standard robustness is finite for $\O$ full-dimensional.

    Using Eq.~\eqref{eq:SPD_DiamondDistance}, see also Refs.~\cite{Watrous2009,Watrous2013}, recall that the optimal diamond conversion distance is given by the optimization problem
    \begin{align}\label{eq:diamond_optimization}
        \min_{\mS \in \mO_1} \tfrac{1}{2}\norm{\mS[\mE]-\mN}_\diamond = \min \left\{\norm{\partTr{D}{Z^{CD}}}_\infty: Z^{CD}\geq 0,Z^{CD} \geq J_{\mN}^{CD}- J_{\mS[\mE]}^{CD}, \mS \in \mO_1 \right\}.
    \end{align}
    From this, we can construct a superchannel following  Ref.~\cite{Regula2021b}:  Choose
    \begin{subequations}\label{eq:OptSuperChannel}
    \begin{align} 
        \mS[\mL] &:= \Tr{W\left( \idChan \otimes \,\mL\right)(\psi)} \,\mN + \Tr{(\id-W)\left( \idChan \otimes \,\mL\right)(\psi)} \, \mQ^\star=: p_\mL \mN + \bar{p}_\mL \mQ^\star, \\
     Z&:=(1-p_\mE)J_\mN,
    \end{align}
    \end{subequations}
    where $\mL$ denotes an arbitrary placeholder channel, and    $p_\mL:=\Tr{W\left( \idChan \otimes \,\mL\right)(\psi)}$ and $\bar{p}_\mL:=\Tr{(\id-W)\left( \idChan \otimes \,\mL\right)(\psi)}$. Notice that for any channel $\mL$, we have that $p_\mL+\bar{p}_\mL=1$. That $\mS$ is a valid superchannel follows from its decomposition into a pre- and post-processing channel $\mM_0$ and $\mM_1$, which was shown in Ref.~\cite{Regula2021b}: Take 
    \begin{subequations}\label{eq:Superchannel_decomposition}
    \begin{align}
        \mM_0^{C\to AR C^\prime}(X^{C})&\!:=\! \psi^{AR}\otimes X^{C^\prime}, \\ \mM_1^{BRC^\prime \to D}(Y^{BRC^\prime})&\!:=\! \mN^{C^\prime\!\to D}\left(\partTr{RB}{(W^{RB} \!\otimes\!\id^{C^\prime}\!)Y^{BRC^\prime}}\right)\!+\!(\mQ^\star)^{C^\prime\!\to D}\left( \partTr{RB}{((\id\!-\!W)^{BR} \!\otimes\!\id^{C^\prime})Y^{BRC^\prime}} \right)
    \end{align}
    \end{subequations}
    where $C^\prime\cong C$. Due to the feasibility of $(\psi,W)$ and $\mN,\mQ^\star \in \CPTP$, $\mM_0$ simply appends the state $\psi$ and $\mM_1$ is a ``measure-and-act" channel, and as such, $\mS \in \mathfrak{S}_1$. Next, we show that $\mS\in \mO_1$. To this end, let $\mM \in \O(A\to B)$ be an arbitrary but fixed free channel, and set $q:=\Tr{W (\idChan \otimes\, \mM)(\psi)} $. Feasibility of $(\psi,W)$ implies $0\leq q\leq R_{s,\O}(\mN)^{-1}=:s^{-1}$. Using $\mN+(s-1)\mQ^\star=s\mK^\star$, we obtain $\mS[\mM]=q \mN +(1-q)\mQ^\star=qs \mK^\star+(1-qs) \mQ^\star$. Combining $0\leq qs\leq 1$ with $\mK^\star,\mQ^\star \in\O$, and $\O$ being convex gives $\mS[\mM] \in\O$, and thus $\mS\in \mO_1$. Lastly, note that $Z = (1-p_\mE) J_\mN \geq 0$. Moreover, recall that from $0\leq W \leq \id$, it follows that $0\leq p_\mE=\Tr{W\left( \idChan \otimes\, \mE\right)(\psi)}\leq 1$ and thus $(1-p_\mE) J_{\mQ^\star}\ge0$. Together with $J_{\mS[\mE]}= p_\mE J_\mN + (1-p_\mE) J_{\mQ^\star} $, we conclude that $Z=(1-p_\mE) J_\mN \geq J_\mN-J_{\mS[\mE]}$, and thus $(\mS,Z)$ is a feasible choice in Eq.~\eqref{eq:diamond_optimization}.
    Thus
\begin{align}\label{eq:diamond_upperbound}
    \min_{\tilde{\mS}_1 \in \mO_1} \tfrac{1}{2}\norm{\tilde{\mS}_1[\mE]-\mN}_\diamond \leq \norm{(1-p_\mE) \partTr{D}{J_{\mN}}   }_\infty = 1-p_\mE =1-\Tr{W\left( \idChan \otimes\, \mE\right)(\psi)}.
\end{align}
Since this holds for every feasible pair $(\psi,W)$, we can take the infimum over all such pairs and obtain
\begin{align}
    \min_{\mS \in \mO_1} \tfrac{1}{2}\norm{\mS[\mE]-\mN}_\diamond \leq 1-G_{\O}(\mE,R_{s,\O}(\mN)). 
\end{align}
If additionally $R_{\min,\O}(\mN)=R_{s,\O}(\mN)$ holds, the upper and lower bounds coincide, and the diamond and the worst-case fidelity conversion distance are equal, i.e., 
\begin{align}\label{eq:DistanceFallTogether}
    \min_{\mS \in \mO_1} \tfrac{1}{2}\norm{\mS[\mE]-\mN}_\diamond =  1-\max_{\mS \in \mO_1} F_{\min}(\mS[\mE],\mN)  = 1- G_{\O}(\mE,R_{s ,\O}(\mN)).
\end{align}

The bounds in the reduced-dimensional case follow analogously, except for the proof that the superchannel $\mS$ (constructed from the relevant robustness SDP) is contained in $\mO_1$ and the possibility that the relevant robustness may diverge. Let $r:=R_{\max,\O}(\mN)$. First, assume that $r<\infty$. Then, there exist $\mQ^\star\in\CPTP$ and $\mK^\star\in\O$ such that $\mN+(r-1)\mQ^\star=r\mK^\star$. Let $\mM\in\O(A\to B)$ be arbitrary and again let $q:=\Tr{W(\idChan\otimes\mM)(\psi)}=r^{-1}$, where $(\psi,W)$ is now a feasible point for the optimization problem determining $G_{\aff(\O)}(\mE,r)$, define a superchannel $\mS$ as in Eq.~\eqref{eq:OptSuperChannel}. As such, $\mS[\mM]=\frac{1}{r}\mN +(1-\frac{1}{r})\mQ^\star = \frac{\mN+(r-1)\mQ^\star}{r}=\mK^\star\in\O$ (since $r\ge 1$ by construction). Since $\mM\in\O(A\to B)$ was arbitrary, $\mS\in \mO_1$. 
Next, assume that $r=\infty$. Take any feasible point $(W,\rho)$ in
\begin{align}\label{eq:G_affO_infty}
    G_{\aff(O)}(\mE,\infty)= \sup \big\{ \Tr{J_\mE W}: 0\leq W \leq \rho \otimes \id, \rho \in \D,  \Tr{W J_\mM}=0 \,\forall \mM \in \O\big\},
\end{align}
which exists since $W=0$ and $\rho=\tfrac{\id}{d}$ is always feasible. Define 
\begin{align}
    p_\mL:= \Tr{J_\mL W}, \quad \text{and} \quad \bar{p}_\mL:= \Tr{J_\mL\left( \rho\otimes \id-W\right)}. 
\end{align}
Feasibility gives $0\leq p_\mL:= \Tr{J_\mL W}\leq 1$ for any channel $\mL$, and crucially $p_\mM= \Tr{J_\mM W}=0$ for every $\mM\in \O$. Now, choose any $\mK \in \O$, and define
\begin{align}
    \mS[\mL]:= p_\mL \mN + \bar{p}_\mL \mK.
\end{align}
Then, clearly $\mS[\mM] \in \O \, \forall \mM\in \O$, and thus, $\mS\in \mO_1.$ Moreover, we have
\begin{align}
    \min_{\tilde{\mS} \in \mO_1} \frac{1}{2}\norm{\tilde{\mS}[\mE]-\mN}_\diamond \leq (1-p_\mE) \frac{1}{2} \norm{\mK-\mN}_\diamond \leq 1-p_{\mE} = 1-\Tr{J_\mE W}.
\end{align}
Since the feasible point $(W,\rho)$ was arbitrary, we can take the infimum over all feasible points $(W,\rho)$ on the right hand-side to obtain $\min_{\tilde{\mS} \in \mO_1} \frac{1}{2}\norm{\tilde{\mS}[\mE]-\mN}_\diamond\leq 1-G_{\aff(O)}(\mE,\infty).$

\end{proof}

Next, we move on to the proofs of Thm.~\ref{thm:Thm1_equivalent} and Cor.~\ref{thm:Thm3_equivalent}. Recall that for any monotone $R_{\O}$, $R_{\O}^{\epsilon, \cdot}$ denotes its smoothed versions, i.e., 
\begin{subequations}\label{eq:Smoothed_Monotone_appen}
    \begin{align}
        R_{\O}^{\epsilon,\diamond}(\mE)&:= \inf\left\{R_{\O}(\mE^\prime): \tfrac{1}{2}\norm{\mE-\mE^\prime}_\diamond \leq \epsilon , \mE^\prime \in \CPTP\right\}, \label{eq:Smoothed_Monotone_appen_diamond}\\
        R_{\O}^{\epsilon,F}(\mE)&:= \inf\left\{R_{\O}(\mE^\prime): F_{\min}(\mE,\mE^\prime) \geq 1- \epsilon , \mE^\prime \in \CPTP\right\}.
    \end{align}
\end{subequations}

The analogue of Ref.~\cite[Thm.~1]{Regula2021b} for the diamond distance is the following.

\ThmOneequivalent*
\begin{proof}
We first note that $R_{\O}$ being a monotone under $\mO_1$ implies that $R_{\O}^{\mu,\diamond}$ is too:
Let $\mE^\prime \in \CPTP$ satisfy $\tfrac{1}{2}\norm{\mE-\mE^\prime}_\diamond \leq \mu$. Let $\mS_1\in\mO_1$, then the data-processing inequality implies $\tfrac{1}{2}\norm{\mS_1[\mE]-\mS_1[\mE^\prime]}_\diamond\leq \tfrac{1}{2}\norm{\mE-\mE^\prime}_\diamond \leq \mu$. Therefore, $\mS_1[\mE^\prime]$  is feasible in the optimization problem defining  $R_{\O}^{\mu,\diamond}(\mS_1[\mE])$. Thus, 
\begin{align}
    R_{\O}^{\mu,\diamond}(\mS_1[\mE]) \leq R_{\O}(\mS_1[\mE^\prime]) \leq R_{\O}(\mE^\prime),
\end{align}
where we used that $R_{\O}$ is a monotone under $\mO_1$. Taking the infimum over all such $\mE^\prime \in \CPTP$ yields 
\begin{align}\label{eq:R_O_smoothed_monotone}
    R_{\O}^{\mu,\diamond}(\mS_1[\mE]) \leq  \inf\{R_{\O}(\mE^\prime):\tfrac{1}{2}\norm{\mE-\mE^\prime}_\diamond \leq \mu,\ \mE^\prime \in \CPTP \}=  R_{\O}^{\mu,\diamond}(\mE).
\end{align}
Now, for a fixed input and target channel $\mE$ and $\mN$, respectively, suppose there exists a superchannel $\mS_1 \in \mO_1$ such that $\frac{1}{2}\norm{\mS_1[\mE]-\mN}_\diamond\leq \epsilon$. Let $\mN^\prime$ be a channel such that $\tfrac{1}{2} \norm{\mS_1[\mE]-\mN^\prime}_\diamond \leq \delta$. The triangle inequality yields
\begin{align}
    \tfrac{1}{2}\norm{\mN-\mN^\prime}_\diamond = \tfrac{1}{2}\norm{\mS_1[\mE]-\mN^\prime-(\mS_1[\mE]-\mN)}_\diamond \leq\tfrac{1}{2} \norm{\mS_1[\mE]-\mN}_\diamond+\tfrac{1}{2}\norm{\mS_1[\mE]-\mN^\prime}_\diamond\leq \epsilon +\delta.
\end{align}
Hence, $\mN^\prime $ is feasible in $R_{\O}^{\epsilon+\delta, \diamond}(\mN)$, and as such, $R_{\O}^{\epsilon+\delta, \diamond}(\mN) \leq R_{\O}(\mN^\prime) $. Taking the infimum over all such  $\mN^\prime $ yields
\begin{align}
    R_{\O}^{\epsilon+\delta, \diamond}(\mN) \leq \inf\{ R_{\O}(\mN^\prime): \tfrac{1}{2} \norm{\mS_1[\mE]-\mN^\prime}_\diamond \leq \delta,\ \mN^\prime \in \CPTP \} = R_{\O}^{\delta, \diamond}(\mS_1[\mE]) \overset{\eqref{eq:R_O_smoothed_monotone}}{\leq } R_{\O}^{\delta, \diamond}(\mE).
\end{align}

Conversely, suppose that one of the following holds:
    \begin{align}\label{eq:def_h_r}
        h:=R_{H,\O}^{\delta} (\mE)\geq R_{s,\O}^{\epsilon,\diamond}(\mN)=: r \qquad \text{or} \qquad h:=R_{H,\aff(\O)}^{\delta} (\mE)\geq R_{\max ,\O}^{\epsilon,\diamond}(\mN)=: r.
    \end{align}
Let $(\psi^\star, P^\star)$ denote an optimal feasible point in the respective optimization problem in Eq.~\eqref{eq:R_H_epsilon} and define the linear functionals
    \begin{align}
        p_\mL:= \Tr{P^\star (\idChan \otimes\,\mL)(\psi^\star)} \quad \text{and} \quad \bar{p}_\mL:= \Tr{(\id-P^\star) (\idChan \otimes\,\mL)(\psi^\star)}.
    \end{align}
    For the fixed input channel $\mE$, feasibility in  Eq.~\eqref{eq:R_H_epsilon} gives $p_{\mE}\geq1-\delta$. For every channel $\mL$, the pair $(p_{\mL},\bar p_{\mL})$ is a probability distribution since $\{P^\star, \id-P^\star\}$ is a POVM (and $\mL\in \CPTP$, $\psi^\star$ a state). Moreover, from feasibility in Eq.~\eqref{eq:R_H_epsilon}, we have
    \begin{align}
        p_\mM \leq \frac{1}{h} \qquad \text{or} \qquad p_\mM = \frac{1}{h} , \qquad \forall \mM\in \O
    \end{align}
    in the full- and reduced-dimensional case, respectively. 

Next, we distinguish between the cases  $h<\infty$ and $h=\infty$\footnote{This case can appear when for example $\O(A\to B)=\{\mM^{A\to B} \in\CPTP: \mM^{A\to B}(\ketbra{0}{0}_A)=\ketbra{0}{0}_B\}$: Let $\mE(\rho)=X\rho X$ be the qubit bit-flip channel and consider Eq.~\eqref{eq:min_entropy_combined}, which implies $R_{\min,\O}(\mE)^{-1}\leq R_{\min,\aff(\O)}(\mE)^{-1}$. In Eq.~\eqref{eq:min_entropy_aff}, choose the state $\psi=\ketbra{0}{0}_R\otimes \ketbra{0}{0}_A$. Then, we have that $\Tr{(\idChan \otimes\, \mE)(\psi) (\idChan \otimes\, \mM)(\psi)}= 0 \, \forall \mM\in \O$. As such $0\geq R_{\min,\aff(\O)}(\mE)^{-1} \geq R_{\min,\O}(\mE)^{-1} \geq 0$, where the last inequality follows from $R_{\min,\aff(\O)}(\mE)$ being non-negative.}. In the latter case, we directly obtain $p_{\mM}=0$ for every $\mM\in\O$, while in the full-dimensional case, the same conclusion follows from $0\leq p_{\mM}\leq h^{-1}=0$. Thus, $p_{\mM}=0$ and $ \bar{p}_{\mM}=1
    \,\forall\,\mM\in\O$. Choose any fixed free channel $\mK_0\in\O$ (which exists since $\O$ is non-empty) and define
\begin{align}
    \mS_{1}^\prime[\mL]:=p_{\mL}\mN+\bar p_{\mL}\mK_0.
\end{align}
This is a valid superchannel (see, e.g., the discussion surrounding Eq.~\eqref{eq:Superchannel_decomposition}). For every $\mM\in\O$, we have $\mS_{1}^\prime[\mM]=\mK_0\in\O$, and thus $\mS_{1}^\prime \in \mO_1$. Moreover, 
\begin{align}
    \tfrac{1}{2}\norm{\mS_{1}^\prime[\mE]-\mN}_\diamond=\tfrac{1-p_\mE}{2} \norm{\mK_0-\mN}_\diamond \leq 1-p_{\mE}\leq \delta\leq \delta+\epsilon.
\end{align}
    
Next, consider $h<\infty$. Then, $r<\infty$ and the infimum in Eq.~\eqref{eq:Smoothed_Monotone_appen_diamond} is thus attained\footnote{For completeness, let $r:=R_{\star,\O}^{\epsilon,\diamond}(\mN)<\infty$, where $\star\in\{s,\max\}$. There exists a sequence of feasible tuples $(t_n,\mN_n^\prime,\mQ_n,\mK_n)$ such that $t_n\to r$, $\frac12\norm{\mN-\mN_n^\prime}_\diamond\leq\epsilon$, and $\mN_n^\prime+(t_n-1)\mQ_n=t_n\mK_n$. Here, $\mQ_n,\mK_n\in\O(A\to B)$ for $\star=s$, whereas $\mQ_n\in\CPTP(A\to B)$ and $\mK_n\in\O(A\to B)$ for $\star=\max$. Since both $\CPTP(A\to B)$ and $\O(A\to B)$ are compact, a subsequence converges to $(\mN^\prime,\mQ,\mK)$. Taking the limit preserves both the diamond-norm constraint and the robustness decomposition, yielding $R_{\star,\O}(\mN^\prime)\leq r$; the reverse inequality follows because $\mN^\prime$ is feasible in the smoothing optimization. Hence $R_{\star,\O}(\mN^\prime)=r$. }\label{footnote:inf_attain}. As such, there exist channels $\mN^\prime,\mQ\in\CPTP$ and $\mK\in\O$ satisfying 
    \begin{align}\label{eq:exist_channels}
         \tfrac{1}{2} \norm{\mN^\prime-\mN}_\diamond\leq \epsilon, \quad \text{and} \quad \mN^\prime +\left(r-1\right) \mQ=  r \mK,
    \end{align}
    where additionally, $\mQ \in\O$ in the full-dimensional case. For the moment, assume in addition that $1<h$ and notice that from Eq.~\eqref{eq:def_R_s} and Eq.~\eqref{eq:def_h_r}, we have $h\geq r\geq 1$, and can thus define the channel
    \begin{align}\label{eq:Q_tilde_affine}
        \tilde{\mQ}:=\frac{(r-1)\mQ+(h-r)\mK}{h-1}.
    \end{align}
    In the full-dimensional case, we also have that $\tilde{\mQ} \in \O$ since $\O$ is convex. Additionally, it holds that
    \begin{align}
        \mN^\prime+(h-1) \tilde{\mQ}=\mN^\prime +(r-1) \mQ+(h-r)\mK= h\mK.
    \end{align}
    Define the superchannel 
    \begin{align}
        \mS_1[\mL]:= p_\mL \mN^\prime +\bar{p}_\mL \tilde{\mQ},
    \end{align}
    which is indeed a valid superchannel following from the discussion surrounding Eq.~\eqref{eq:Superchannel_decomposition}. Let $\mM\in \O$ be arbitrary. Then, the latter expression implies that $\mS_1$ acts as
    \begin{align}
        \mS_1[\mM]= hp_\mM \mK +(1-hp_\mM) \tilde{\mQ}.
    \end{align}
    In the full-dimensional case, from $0\leq h p_\mM \leq 1$ together with $\mK,\tilde{\mQ} \in \O$, it follows that $\mS_1[\mM] \in \O$. Moreover, in the reduced-dimensional case,     $p_\mM= h^{-1}$ immediately implies that $\mS_1[\mM]= \mK \in \O$, and as such, we have $\mS_1 \in \mO_1$ in both cases. This case is concluded by noting that 
    \begin{align}
        \tfrac{1}{2}\norm{\mS_1[\mE]-\mN}_\diamond &\leq \tfrac{p_\mE}{2}\norm{\mN^\prime-\mN}_\diamond+\tfrac{1-p_\mE}{2}\norm{\tilde{\mQ}-\mN}_\diamond \leq \epsilon p_\mE + (1-p_\mE) = 1-p_\mE(1-\epsilon) \nonumber \\
        &\leq 1-(1-\delta)(1-\epsilon) \leq \epsilon+\delta. \label{eq:DiamondBound_epsilon_delta}
    \end{align}
    It remains to treat the case $h=1$. As noted above, this implies $r=h=1$ and thus, by Eq.~\eqref{eq:exist_channels}, that $\mN^\prime=\mK \in \O$. Next, define the superchannel
      \begin{align}\label{eq:MS_TILDe}
        \tilde{\mS_1}[\mL]= p_\mL \mN^\prime +\bar{p}_\mL \mK= (p_\mL+\bar{p}_\mL) \mK ,
    \end{align}
    where we recall that $p_\mL+\bar{p}_\mL=\Tr{(\idChan \otimes\,\mL)(\psi^\star)} =1$ for every trace-preserving channel $\mL$. Thus, for any $\mM\in \O$, we have $\tilde{\mS_1}[\mM]=\mK\in\O$. Finally,
        \begin{align}
        \tfrac{1}{2}||\tilde{\mS_1}[\mE]-\mN||_\diamond = \tfrac{1}{2}||\mN^\prime-\mN||_\diamond \leq \epsilon \leq \epsilon+\delta    
    \end{align}
    concludes the proof. 
\end{proof}

As promised in the main text, we now discuss the relation between Thm.~\ref{thm:Thm1_equivalent} and Ref.~\cite[Thm.~1]{Regula2021b}. 
The smoothings of a monotone $R_{\O}$ with respect to the worst-case fidelity and the diamond distance are related via 
\begin{align}\label{eq:SmoothedComparison}
    R_{\O}^{2\epsilon,F}(\mE)
    &\leq R_{\O}^{\epsilon,\diamond}(\mE)
    \leq R_{\O}^{\epsilon^{2},F}(\mE),
\end{align}
which follows directly from the (general) Fuchs-van de Graaf inequality in Eq.~\eqref{eq:FuchsVanDeGraaf_channels}: If $\frac{1}{2}\norm{\mE-\mE^\prime}_\diamond \leq \epsilon$, then $ F_{\min}(\mE,\mE^\prime) \geq (1-\epsilon)^2 =1-2\epsilon+\epsilon^2\geq 1-2\epsilon$, and thus, $R_{\O}^{\epsilon,\diamond}(\mE)\geq R_{\O}^{2\epsilon,F}(\mE)$. Conversely, if  $F_{\min}(\mE,\mE^\prime) \geq 1- \epsilon^2$, then $\frac{1}{2}\norm{\mE-\mE^\prime}_\diamond\leq \epsilon$, and thus $R_{\O}^{\epsilon,\diamond}(\mE) \leq R_{\O}^{\epsilon^2,F}(\mE)$. This allows us to show that Thm.~\ref{thm:Thm1_equivalent} implies the direct part of Ref.~\cite[Thm.~1]{Regula2021b}: If there exists a superchannel $\mS \in \mO_1$ such that $1-F_{\min}(\mS[\mE],\mN)\leq \epsilon$, it follows from Fuchs-van de Graaf that $\tfrac{1}{2}\norm{\mS[\mE]-\mN}_\diamond \leq \sqrt{\epsilon}$ and thus
\begin{align}
    R_{\O}^{\delta,F}\mkern-4mu(\mkern-1mu\mE\mkern-1mu) \mkern-4mu \overset{\eqref{eq:SmoothedComparison}}{\geq}\mkern-4muR_{\O}^{\sqrt\delta,\diamond}(\mkern-1mu\mE\mkern-1mu)\mkern-4mu\overset{\text{Thm.}~\ref{thm:Thm1_equivalent}}{\geq }\mkern-4mu R_{\O}^{\sqrt{\epsilon}\mkern-1mu+\mkern-1mu\sqrt\delta,\diamond}(\mkern-1mu\mN\mkern-1mu) \mkern-4mu\overset{\eqref{eq:SmoothedComparison}}{\geq}\mkern-4mu R_{\O}^{2(\sqrt{\epsilon}\mkern-1mu+\mkern-1mu\sqrt\delta),F}\mkern-4mu(\mkern-1mu\mN\mkern-1mu),
\end{align}
for any $ 0\leq \epsilon,\delta \leq 1$. The converse statements are incomparable at fixed error parameters, since they involve robustness measures smoothed with respect to different metrics.

\ThmThreeequivalent*
\begin{proof}
    Let $\mS_1\in \mO_1$ be such that $ \tfrac{1}{2} \norm{\mS_1[\mE]-\mN}_\diamond \leq \epsilon$. Since by assumption, $\mN$ is an isometric channel, the sharpened Fuchs-van de Graaf inequality in Eq.~\eqref{eq:FuchsVanDeGraaf_channels_sharp} implies $\epsilon \geq \tfrac{1}{2} \norm{\mS_1[\mE]-\mN}_\diamond \geq 1-F_{\min}(\mS_1[\mE],\mN)$, and thus $F_{\min}(\mS_1[\mE],\mN)\geq 1-\epsilon$. Next, invoke Ref.~\cite[Thm.~3]{Regula2021b}, from which it follows that for $X\in \{\O,\aff(\O)\}$, we have $ R_{H,X}^{\epsilon}(\mE) \geq R_{\min ,X}(\mN)$. 
\end{proof}

\section{DYNAMICAL COHERENCE}\label{sec:Appen_Coherence}

\subsection{Characterization of compatible supermaps}\label{sec:charac_compatitbleSupermaps}
In this section, we show the existence of $\cMIO, \cDI$, and $\cDIO$, as defined in the discussion after Def.~\ref{def:cO} and characterize these sets in terms of semidefinite constraints. To this end, in the following Lemma, we construct families of free channels, which in the Theorem thereafter will then allow us to deduce necessary conditions for compatibility. These results are a generalization of the MIO case treated in  Ref.~\cite{Ahnefeld2025}.

\begin{lem}\label{lem:DIOTester}
Let $1,2$ and $B_{2j},B_{1},\ldots, B_{2j-3}$ denote quantum systems, and let $B$ denote the joint quantum system composed of $B_{2j},B_{1},\ldots, B_{2j-3}$\footnote{Within the convention of labeling systems at a specific (time) step, we technically assume $j\geq 2$ here. With some ``abuse of notation", we will later also consider the case where $B_1,\ldots, B_{2j-3}$ are chosen to be trivial systems. Then, the system $B$ is composed solely of $B_2$, and $\ket{\psi^\pm}=\tfrac{1}{\sqrt{2}} ( \ket{0}_{B_2} \pm \ket{1}_{B_{2}} )$. The reasons for this will become apparent in the proof of the following Thm.~\ref{thm:CompatibleSupermaps}. 
\label{footnote:j_2_system}}. Let $\ket{\psi^\pm}$ be normalized states of system $B$ defined as
\begin{align}\label{eq:PsiPm}
    \ket{\psi^\pm}_B &:= \frac{1}{\sqrt{2}} \left( \ket{0}_{B_{2j}} \otimes \ket{ k_{B_{1}}\ldots k_{B_{2j-3}}}_{B_{1},\ldots, B_{2j-3}} \pm \ket{1}_{B_{2j}} \otimes \ket{l_{B_{1}}\ldots l_{B_{2j-3}}}_{B_{1},\ldots, B_{2j-3}} \right) \nonumber \\
    &= \frac{1}{\sqrt{2}}  \left(\ket{0}_{B_{2j}} \otimes \ket{ \Vec{k}}_{B_{1},\ldots, B_{2j-3}} \pm \ket{1}_{B_{2j}} \otimes \ket{\Vec{l}}_{B_{1},\ldots, B_{2j-3}}\right),
\end{align}
where all $k_{B_i}$ and $l_{B_i}$ are arbitrary but fixed indices. Let $d_X$ denote the dimension of system $X$ and for $k,l\in \{0,\ldots, d_1-1\}$ arbitrary but fixed, let
\begin{align}
    &D=\left( \id -\ketbra{k}{k}-\ketbra{l}{l}\right)_1 \otimes \frac{\id_2}{d_2} \otimes \frac{\id_B}{d_B}, \\
    & Q=  \left( \id_{1,B}- \frac{1}{2}\left(\ketbra{k}{k}+\ketbra{l}{l}\right)_{1}\otimes ( \ketbra{\psi^+}{\psi^+}_B+\ketbra{\psi^-}{\psi^-}_B) \right)  \otimes \frac{\id_2}{d_2}.
\end{align}
If $k\ne l$, for all $n,m \in \{0,\ldots, d_2-1\}$, the matrices
     \begin{align}
         &J_\mathcal{N}\mkern-4mu=\mkern-4mu D \mkern-4mu+\mkern-4mu \frac{1}{2}\!\left(\ketbra{kn}{kn}\!+\!\ketbra{lm}{lm}\right)_{1,2}\!\otimes\! \left(\ketbra{\psi^+}{\psi^+}\mkern-4mu+\mkern-4mu\ketbra{\psi^-}{\psi^-}\right)_B \mkern-4mu+\mkern-4mu\frac{1}{2}\!\left(\ketbra{kn}{lm} \mkern-4mu+\mkern-4mu\ketbra{lm}{kn}\right)_{1,2}\!\otimes\! \left(\ketbra{\psi^+}{\psi^+}\mkern-4mu-\mkern-4mu\ketbra{\psi^-}{\psi^-}\right)_B,\\
         &J_\mathcal{M}\mkern-4mu=\mkern-4mu D\mkern-4mu+\mkern-4mu\frac{1}{2}\!\left(\ketbra{kn}{kn}\mkern-4mu+\mkern-4mu\ketbra{lm}{lm}\right)_{1,2}\!\otimes\! \left(\ketbra{\psi^+}{\psi^+}\mkern-4mu+\mkern-4mu\ketbra{\psi^-}{\psi^-}\right)_B\mkern-4mu+\mkern-4mu\frac{1}{2}\!\left(\ketbra{kn}{lm} \mkern-4mu-\mkern-4mu\ketbra{lm}{kn}\right)_{1,2}\!\otimes\!\left(\ketbra{\psi^+}{\psi^-}\mkern-4mu-\mkern-4mu\ketbra{\psi^-}{\psi^+}\right)_{\mkern-2mu B}
    \end{align}
define the Choi states of channels $\mathcal{N}^{1\to 2,B} \in \DIO\subset\MIO$ and $\mathcal{M}^{1\to 2,B} \in \DIO \subset \MIO$. Moreover, for all  $n,m \in \{0,\ldots, d_2-1\}$ with $n\neq m$, the matrices
     \begin{align}
         &J_\mathcal{K}\mkern-4mu=\mkern-4mu Q\mkern-4mu+\mkern-4mu\frac{1}{2}\!\left(\ketbra{kn}{kn}\mkern-4mu+\mkern-4mu\ketbra{lm}{lm}\right)_{1,2}\!\otimes\!\left(\ketbra{\psi^+}{\psi^+}\mkern-4mu+\mkern-4mu\ketbra{\psi^-}{\psi^-}\right)_B\!+\!\frac{1}{2}\left(\ketbra{kn}{lm}\mkern-4mu+\mkern-4mu\ketbra{lm}{kn}\right)_{1,2}\!\otimes\!\left(\ketbra{\psi^+}{\psi^+}\mkern-4mu-\mkern-4mu\ketbra{\psi^-}{\psi^-}\right)_B,\\
         &J_\mathcal{L}\mkern-4mu=\mkern-4mu Q \mkern-4mu+\mkern-4mu\frac{1}{2}\!\left(\ketbra{kn}{kn}\mkern-4mu+\mkern-4mu\ketbra{lm}{lm} \right)_{1,2}\! \otimes\!\left(\ketbra{\psi^+}{\psi^+}\mkern-4mu+\mkern-4mu\ketbra{\psi^-}{\psi^-}\right)_B\!+\!\frac{1}{2}\!\left(\ketbra{kn}{lm}\mkern-4mu-\mkern-4mu\ketbra{lm}{kn} \right)_{1,2}\!\otimes\!\left(\ketbra{\psi^+}{\psi^-}\mkern-4mu-\mkern-4mu\ketbra{\psi^-}{\psi^+}\right)_B
    \end{align}
define Choi states of channels $\mathcal{K}^{1,B\to 2} \in \DIO\subset\DI$ and $\mathcal{L}^{1,B\to 2} \in \DIO \subset \DI$.

\end{lem}
\begin{proof}
We start by showing that the matrices $J_\mathcal{N},J_\mathcal{M},J_\mathcal{K},J_\mathcal{L}$ correspond to quantum channels, and then show that they are contained in $\DIO.$ Firstly, let us introduce the composite system $\tilde{B}=(B_1,\ldots, B_{2j-3})$, and note that from
\begin{align}
    &\ketbra{\psi^\pm}{\psi^{\pm}}_B= \frac{1}{2} \left(\ketbra{0}{0}_{B_{2j}} \otimes \ketbra{\Vec{k}}{\Vec{k}}_{\tilde{B}} +\ketbra{1}{1}_{B_{2j}} \otimes \ketbra{\Vec{l}}{\Vec{l}}_{\tilde{B}} \pm \ketbra{0}{1}_{B_{2j}} \otimes \ketbra{\Vec{k}}{\Vec{l}}_{\tilde{B}} \pm \ketbra{1}{0}_{B_{2j}} \otimes \ketbra{\Vec{l}}{\Vec{k}}_{\tilde{B}} \right), \\
    &\ketbra{\psi^+}{\psi^{-}}_B= \frac{1}{2} \left(\ketbra{0}{0}_{B_{2j}} \otimes \ketbra{\Vec{k}}{\Vec{k}}_{\tilde{B}} -\ketbra{1}{1}_{B_{2j}} \otimes \ketbra{\Vec{l}}{\Vec{l}}_{\tilde{B}} - \ketbra{0}{1}_{B_{2j}} \otimes \ketbra{\Vec{k}}{\Vec{l}}_{\tilde{B}} + \ketbra{1}{0}_{B_{2j}} \otimes \ketbra{\Vec{l}}{\Vec{k}}_{\tilde{B}} \right),
\end{align}
it follows that 
\begin{subequations}\label{eq:Psi_pm_identites}
\begin{align}
    &\ketbra{\psi^+}{\psi^{+}}_B + \ketbra{\psi^-}{\psi^{-}}_B= \ketbra{0}{0}_{B_{2j}} \otimes \ketbra{\Vec{k}}{\Vec{k}}_{\tilde{B}} +\ketbra{1}{1}_{B_{2j}} \otimes \ketbra{\Vec{l}}{\Vec{l}}_{\tilde{B}} =\Delta_B \left(\ketbra{\psi^+}{\psi^{+}}_B + \ketbra{\psi^-}{\psi^{-}}_B \right), \label{eq:Psi_pm_diagonal}\\
    &\ketbra{\psi^+}{\psi^{+}}_B - \ketbra{\psi^-}{\psi^{-}}_B=  \ketbra{0}{1}_{B_{2j}} \otimes \ketbra{\Vec{k}}{\Vec{l}}_{\tilde{B}} + \ketbra{1}{0}_{B_{2j}} \otimes \ketbra{\Vec{l}}{\Vec{k}}_{\tilde{B}} =(\idChannel_B-\Delta_B)\left(\ketbra{\psi^+}{\psi^{+}}_B - \ketbra{\psi^-}{\psi^{-}}_B\right), \label{eq:Psi_pm2_diagonal} \\
    &\ketbra{\psi^+}{\psi^{-}}_B - \ketbra{\psi^-}{\psi^{+}}_B=  \ketbra{1}{0}_{B_{2j}} \otimes \ketbra{\Vec{l}}{\Vec{k}}_{\tilde{B}}-\ketbra{0}{1}_{B_{2j}} \otimes \ketbra{\Vec{k}}{\Vec{l}}_{\tilde{B}}= (\idChannel_B-\Delta_B)\left( \ketbra{\psi^+}{\psi^{-}}_B - \ketbra{\psi^-}{\psi^{+}}_B\right), \label{eq:Psi_pm3_diagonal}
\end{align}
\end{subequations}
or in other words, the sum of $\psi^\pm$ is a diagonal matrix, whereas the difference of $\psi^\pm$ and and the mixed term $\ketbra{\psi^+}{\psi^{-}} - \ketbra{\psi^-}{\psi^{+}}$ are off-diagonal matrices. Next, for some arbitrary $0\leq k,l \leq d_1-1$ and $0\leq n,m \leq d_2-1$ with either $k\neq l$ or $n\neq m$, define the states
\begin{subequations}
\begin{align}
    & \ket{\phi^{\pm}}_{1,2}:=\frac{1}{\sqrt{2}} \left( \ket{kn}_{1,2} \pm \ket{lm}_{1,2}\right), \label{eq:PhiKLNM} \\
    &\ket{\lambda}_{1,2,B}:= \frac{1}{\sqrt{2}} \left(\ket{kn}_{1,2} \otimes \ket{\psi^+}_{B} + \ket{lm}_{1,2} \otimes \ket{\psi^-}_{B} \right),\\
    &\ket{\xi}_{1,2,B}:=\frac{1}{\sqrt{2}} \left( \ket{kn}_{1,2} \otimes \ket{\psi^-}_{B} -\ket{lm}_{1,2} \otimes \ket{\psi^+}_{B} \right).
\end{align}
\end{subequations}
From expanding the states,
\begin{subequations}
\begin{align}
    &2\ketbra{\phi^{\pm}}{\phi^{\pm}}_{1,2}= \ketbra{kn}{kn}_{1,2}+ \ketbra{lm}{lm}_{1,2} \pm \ketbra{kn}{lm}_{1,2} \pm \ketbra{lm}{kn}_{1,2}, \\
    &2 \ketbra{\lambda}{\lambda}_{1,2,B}\mkern-4mu=\mkern-4mu \ketbra{kn}{kn}_{1,2} \!\otimes\! \ketbra{\psi^+}{\psi^+}_{ \mkern-2mu B} \mkern-4mu+\mkern-4mu \ketbra{lm}{lm}_{1,2} \!\otimes\!\ketbra{\psi^-}{\psi^-}_{ \mkern-2mu B} +\left(\ketbra{kn}{lm}_{1,2} \otimes \ketbra{\psi^+}{\psi^-}_{ \mkern-2mu B} \mkern-4mu + h.c. \right), \\
    &2 \ketbra{\xi}{\xi}_{1,2,B}\mkern-4mu=\mkern-4mu \ketbra{kn}{kn}_{1,2} \!\otimes\! \ketbra{\psi^-}{\psi^-}_{ \mkern-2mu B} \mkern-4mu+\mkern-4mu \ketbra{lm}{lm}_{1,2} \!\otimes\!\ketbra{\psi^+}{\psi^+}_{ \mkern-2mu B} - \left(\ketbra{kn}{lm}_{1,2} \otimes \ketbra{\psi^-}{\psi^+}_{ \mkern-2mu B} \mkern-4mu + h.c. \right),
\end{align}
\end{subequations}
it follows by a straightforward calculation that 
    \begin{align}\label{eq:N_M_decomposition}
        &J_\mathcal{N}=  D +\ketbra{\phi^+}{\phi^+}_{1,2} \otimes \ketbra{\psi^+}{\psi^+}_B +\ketbra{\phi^-}{\phi^-}_{1,2} \otimes \ketbra{\psi^-}{\psi^-}_B ,\\
        &   J_\mathcal{M}=D+ \ketbra{\lambda}{\lambda}_{1,2,B}+\ketbra{\xi}{\xi}_{1,2,B},
    \end{align}
    and 
    \begin{align}\label{eq:K_L_decomposition}
        &J_\mathcal{K}=  Q +\ketbra{\phi^+}{\phi^+}_{1,2} \otimes \ketbra{\psi^+}{\psi^+}_B +\ketbra{\phi^-}{\phi^-}_{1,2} \otimes \ketbra{\psi^-}{\psi^-}_B,\\
        &   J_\mathcal{L}=Q+ \ketbra{\lambda}{\lambda}_{1,2,B}+\ketbra{\xi}{\xi}_{1,2,B}.
    \end{align}
Furthermore, we have that
\begin{align}
    D=\left( \id -\ketbra{k}{k}-\ketbra{l}{l}\right)_1 \otimes \frac{\id_2}{d_2} \otimes \frac{\id_B}{d_B} \overset{k\neq l}{\geq} 0,
\end{align}
which together with Eq.~\eqref{eq:N_M_decomposition} implies that  $J_\mathcal{N},J_\mathcal{M}\geq 0$ for $k\neq l$ and all $n,m$. The channels are trace-preserving since
    \begin{align}
        \partTr{2,B}{J_\mathcal{N}}&=\left( \id -\ketbra{k}{k}-\ketbra{l}{l}\right)_1+  \ketbra{k}{k}_1 +\ketbra{l}{l}_1  =\id_1,
    \end{align}
and complete analogously we get $\partTr{2,B}{J_\mathcal{M}}=\id_1$. Thus, $\mathcal{N}^{1\to 2,B}$ and $\mathcal{M}^{1\to 2,B}$ are valid channels for all valid choices of $n,m$ and $k\neq l$. Additionally, from
\begin{align}
    Q&= \left( \id_{1,B}- \frac{1}{2}\left(\ketbra{k}{k}+\ketbra{l}{l}\right)_{1}\otimes ( \ketbra{\psi^+}{\psi^+}_B+\ketbra{\psi^-}{\psi^-}_B) \right)  \otimes \frac{\id_2}{d_2} \geq \left( \id_{1,B}- \id_1\otimes \left( \ketbra{\psi^+}{\psi^+}_B+\ketbra{\psi^-}{\psi^-}_B\right) \right)  \otimes \frac{\id_2}{d_2} \nonumber \\
    &\overset{\eqref{eq:Psi_pm_diagonal}}{=} \id_1 \otimes \left( \id_{B}- \left( \ketbra{0}{0} \otimes \ketbra{\Vec{k}}{\Vec{k}} +\ketbra{1}{1} \otimes \ketbra{\Vec{l}}{\Vec{l}}\right)_B \right)  \otimes \frac{\id_2}{d_2}\geq 0,
\end{align}
together with Eq.~\eqref{eq:K_L_decomposition} it follows that  $J_\mathcal{K},J_\mathcal{L}\geq 0$ for all choices of $k,l,n,m$. Furthermore, for any $k,l$ and $n\neq m$ it follows that
\begin{align}
    \partTr{2}{J_\mathcal{K}}&=\partTr{2}{Q}+\frac{1}{2} \left(\ketbra{k}{k}\!+\!\ketbra{l}{l}\right)_{1} \otimes \left(\ketbra{\psi^+}{\psi^+}\!+\!\ketbra{\psi^-} 
    {\psi^-}\right)_B \nonumber \\
    &= \left( \id_{1,B}-  \frac{1}{2}\left(\ketbra{k}{k}\!+\!\ketbra{l}{l}\right)_{1}\otimes ( \ketbra{\psi^+}{\psi^+}_B\!+\!\ketbra{\psi^-}{\psi^-}_B) \right) +\frac{1}{2} \left(\ketbra{k}{k}\!+\!\ketbra{l}{l}\right)_{1} \otimes \left(\ketbra{\psi^+}{\psi^+}\!+\!\ketbra{\psi^-}{\psi^-}\right)_B  \nonumber \\
    &= \id_{1,B},
\end{align}
which implies that $\mathcal{K}^{1,B\to 2 }$ (and analogously $\mathcal{L}^{1,B\to 2}$) are valid channels for all choices of $n\neq m$ and $k, l$. It remains to show that the channels are contained in $\DIO$. Starting with $\mathcal{N}, \mathcal{M}$, let $k\neq l$ and let $n,m$ be arbitrary. Then, on one hand, we know that
    \begin{align}
        \Delta_1 J_\mathcal{N} &= \Delta_1 D +\frac{1}{2}\left(\ketbra{kn}{kn}+\ketbra{lm}{lm}\right)_{1,2} \otimes \left(\ketbra{\psi^+}{\psi^+}+\ketbra{\psi^-}{\psi^-}\right)_B \nonumber \\
        &=D+\frac{1}{2}\left(\ketbra{kn}{kn}+\ketbra{lm}{lm}\right)_{1,2} \otimes \Delta \left(\ketbra{\psi^+}{\psi^+}+\ketbra{\psi^-}{\psi^-}\right)_B,
    \end{align}
    where we used that $D$ is diagonal and Eq.~\eqref{eq:Psi_pm_diagonal}. 
On the other hand, we have
\begin{align}
    \Delta_{2,B} J_\mathcal{N} & =
    D+\frac{1}{2}\left(\ketbra{kn}{kn}+\ketbra{lm}{lm}\right)_{1,2} \otimes \Delta \left(\ketbra{\psi^+}{\psi^+}+\ketbra{\psi^-}{\psi^-}\right)_B.
\end{align}
Hence, $\Delta_1 J_\mathcal{N}= \Delta_{2,B} J_\mathcal{N} $, which is equivalent to  $\mathcal{N}^{1\to 2,B} \in \DIO$.  Analogously, by using Eq.~\eqref{eq:Psi_pm3_diagonal}, which implies
\begin{align}
      \Delta_B \left( \ketbra{\psi^+}{\psi^{-}} - \ketbra{\psi^-}{\psi^{+}} \right)=  \Delta_B (\idChannel-\Delta)_B\left( \ketbra{\psi^+}{\psi^{-}} - \ketbra{\psi^-}{\psi^{+}}\right)=0, 
\end{align}
we find that $\Delta_1 J_\mathcal{M}= \Delta_{2,B} J_\mathcal{M} $, and thus $\mathcal{M}^{1\to 2,B} \in \DIO.$  Similarly, exploiting Eq.~\eqref{eq:Psi_pm2_diagonal}, and Eq.~\eqref{eq:Psi_pm3_diagonal}, for $\mathcal{K}$ and $\mathcal{L}$, we find on one hand that
\begin{align}
    \Delta_{1,B} J_\mathcal{K} =\Delta_{1,B} J_\mathcal{L} &= \Delta_{1,B} Q + \frac{1}{2} \left(\ketbra{kn}{kn}\!+\!\ketbra{lm}{lm}\right)_{1,2} \otimes \Delta_B\left(\ketbra{\psi^+}{\psi^+}\!+\!\ketbra{\psi^-}{\psi^-}\right)_B,
\end{align}
and on the other hand that for every $n\neq m$ and every $k,l$
\begin{align}
    \Delta_{2} J_\mathcal{K} =\Delta_{2} J_\mathcal{L} &= Q + \frac{1}{2} \left(\ketbra{kn}{kn}\!+\!\ketbra{lm}{lm}\right)_{1,2} \otimes \left(\ketbra{\psi^+}{\psi^+}\!+\!\ketbra{\psi^-}{\psi^-}\right)_B \nonumber \\
    &= \Delta_{1,B} Q + \frac{1}{2} \left(\ketbra{kn}{kn}\!+\!\ketbra{lm}{lm}\right)_{1,2} \otimes \Delta_B\left(\ketbra{\psi^+}{\psi^+}\!+\!\ketbra{\psi^-}{\psi^-}\right)_B,
\end{align}
where we used that $Q$ is diagonal and Eq.~\eqref{eq:Psi_pm_diagonal}. This implies that $\mathcal{K}^{1,B\to 2},\mathcal{L}^{1,B\to 2} \in \DIO$ for all $n\neq m$ and all $k,l$. Lastly, recall that $\DIO \subset \MIO,\DI$, which completes the proof.
\end{proof}

We are now ready to show that for the resource theories of dynamical coherence, there exists a largest set of compatible supermaps in the sense of the discussion following Def.~\ref{def:cO}.

\CompatibleSupermaps*
\begin{proof}
We treat the three cases of $\O\in\{\MIO,\DI,\DIO\}$ separately.

a) In Ref.~\cite[Thm.~8 and Lem.~5 of the Supplemental Material]{Ahnefeld2025} it has been shown that the conditions in Eq.~\eqref{eq:MIOCompFromMaxSet} are necessary and sufficient for a supermap to be contained in $\cMIO$.

b) Consider DI next. Let $\mS_N\in \mathfrak{C} $, where $\mathfrak{C}$ denotes a set of supermaps compatible with $\DI$ in the sense of Def.~\ref{def:cO}. As a direct consequence of the set $\DI$ satisfying the assumptions~\ref{assump:identity} to ~\ref{assumpt:complfree}, the set of free networks must be compatible and $\nDI \subseteq \mathfrak{C}$. Since $\mathfrak{C}$ is compatible, it must hold that composing the fixed supermap $\mS_N$ with a network of free channels results in a compatible supermap. To exploit this observation, we will now construct free networks that will allow us to deduce that the dephasing conditions in Eq.~\eqref{eq:DICompFromMaxSet} are necessary. 

Let $B_k$ denote a copy of system $k$, which is fixed by the supermap $\mathcal{S}_N$ via the labeling convention in Fig.~\ref{fig:NetworkvsComb}. We treat the case of $\mS_N$ being a superchannel, i.e., $N=1$, separately; thus, we assume $N \geq 2$ for now and will return to $N=1$ later. Now take an arbitrary but fixed integer $1\leq j\leq N$ and let $\mathcal{K}^{2j-1, B_{2j},B_{1},\ldots, B_{2j-3} \to 2j}, \mathcal{L}^{2j-1, B_{2j},B_{1},\ldots, B_{2j-3} \to 2j} \in \DI$ as in  Lem.~\ref{lem:DIOTester} (where the systems 1 and 2 have to be renamed to $2j-1$ and $2j$, respectively). Let further $\mathcal{E}_{k}^{2k-1,B_{2k} \to B_{2k-1},2k } $ denote the channel that acts as the identity from system $2k-1$ to $B_{2k-1}$ and $B_{2k}$ to $2k$ (i.e., that acts as a SWAP channel). Suppressing all system labels, we consider two quantum networks comprised of the sequence of channels
\begin{subequations}\label{eq:freeNetworks_Testers}
 \begin{align}
    &\mathcal{E}_1,\ldots, \mathcal{E}_{j-1}, \mathcal{K}, \mathcal{E}_{j+1},\ldots, \mathcal{E}_N, \\
    &\mathcal{E}_1,\ldots, \mathcal{E}_{j-1}, \mathcal{L}, \mathcal{E}_{j+1},\ldots, \mathcal{E}_N,
\end{align}
\end{subequations}
which we insert into $\mS_N$ as depicted in Fig.~\ref{fig:SupermapMIOCompatible}. Recall that according to Lem.~\ref{lem:DIOTester}, the channels $\mK, \mL \in \DI$ for all choices of integers $k,l$ and $n\neq m$ (which were arbitrary but fixed in the definition of $\mathcal{K}$ and $\mathcal{L}$). Obviously, $\mathcal{E}_k \in \DI$, which implies that the two networks are free networks.

\begin{figure}[ht]
    \centering
    \scalebox{1}{\includegraphics[width=1\linewidth]{ 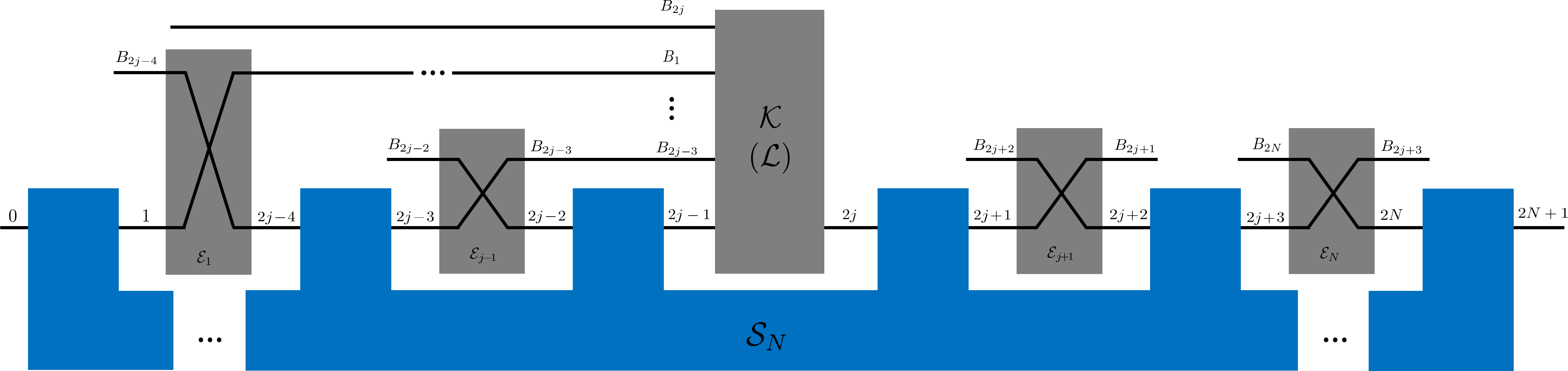}}
    \caption{Depiction of the free network (in gray) that is inserted into a supermap $\mS_N$ in order to derive necessary conditions for $\mS_N$ being an element of a set of supermaps that is compatible with $\DI$.}
    \label{fig:SupermapMIOCompatible}
\end{figure}

Let $\mathcal{T}$ and $\mathcal{W}$ denote the supermaps resulting from the composition of $\mathcal{S}_N$ with these networks, respectively. As discussed above, this implies that the supermaps $\mathcal{T}$ and $\mathcal{W}$ must be DI channels from all input to all output systems (if $n\neq m$). This is equivalent to 
\begin{subequations}\label{eq:Composition_DI_Tester}
        \begin{align}
        &\Delta_{2N+1,B_{2N-1},\ldots, B_{2j+1}} J_\mathcal{T}=  \Delta_{2N+1,B_{2N-1},\ldots, B_{2j+1}}  \Delta_{B_{2N},B_{2N-2},\ldots, B_{2j+2}, B_{2j},B_{2j-2},\ldots B_2,0} J_\mathcal{T}, \\
        &\Delta_{2N+1,B_{2N-1},\ldots, B_{2j+1}} J_\mathcal{W}=  \Delta_{2N+1,B_{2N-1},\ldots, B_{2j+1}}  \Delta_{B_{2N},B_{2N-2},\ldots, B_{2j+2}, B_{2j},B_{2j-2},\ldots B_2,0} J_\mathcal{W}.
    \end{align}
\end{subequations}
Importantly, notice that the right-hand side is being dephased on system $B_{2j}$, and thus, applying $\left(\idChannel_{B_{2j}}-\Delta_{B_{2j}}\right)$ to this expression yields
\begin{subequations}
        \begin{align}
        &\left(\idChannel_{B_{2j}}-\Delta_{B_{2j}}\right)\Delta_{2N+1,B_{2N-1},\ldots, B_{2j+1}} J_\mathcal{T}= 0, \\
         &\left(\idChannel_{B_{2j}}-\Delta_{B_{2j}}\right)\Delta_{2N+1,B_{2N-1},\ldots, B_{2j+1}} J_\mathcal{W}= 0.
    \end{align}
\end{subequations}
     By construction, composition with the channels $\mathcal{E}_j$ simply relabels system $j$ as $B_j$. Using this fact in the definition of the link product in Eq.~\eqref{eq:defLinkProduct} to express $J_\mathcal{T}$ and $J_\mathcal{W}$, respectively, we obtain
    \begin{subequations}\label{eq:Offdiag_expanded}
    \begin{align}
        &0= \left(\idChannel_{B_{2j}}-\Delta_{B_{2j}}\right) \Delta_{2N+1,B_{2N-1},\ldots, B_{2j+1}} \partTr{B_{1}\ldots, B_{2j-3},2j-1,2j}{\left( \id  \otimes \left( J_{\mathcal{K}}\right)^{T_{B_{1}\ldots, B_{2j-3},2j-1,2j}} \right)\left( \id \otimes J_{\mathcal{S}_N}\right)}, \label{eq:TOffDiagsNetwork} \\
        &0=\left(\idChannel_{B_{2j}}-\Delta_{B_{2j}}\right)\Delta_{2N+1,B_{2N-1},\ldots, B_{2j+1}} \partTr{B_{1}\ldots, B_{2j-3},2j-1,2j}{\left( \id  \otimes \left( J_{\mathcal{L}}\right)^{T_{B_{1}\ldots, B_{2j-3},2j-1,2j}} \right)\left( \id \otimes J_{\mathcal{S}_N}\right)},\label{eq:WOffDiagsNetwork}
    \end{align}  
    \end{subequations}
    where we suppressed the (renamed) system labels of $J_{\mathcal{S}_N}$ and the identities. Next, we invoke the structure of the channels $\mathcal{K}$ and $\mathcal{L}$ for the first time.  As defined in Lem.~\ref{lem:DIOTester}, let
        \begin{align}
        \ket{\psi^\pm}&= \frac{1}{\sqrt{2}} \left( \ket{0}_{B_{2j}} \otimes \ket{ k_{B_{1}}\ldots k_{B_{2j-3}}}_{B_{1},\ldots, B_{2j-3}} \pm \ket{1}_{B_{2j}} \otimes \ket{l_{B_{1}}\ldots l_{B_{2j-3}}}_{B_{1},\ldots, B_{2j-3}} \right) \nonumber \\
        &= \frac{1}{\sqrt{2}} \left(  \ket{0}_{B_{2j}} \otimes \ket{ \Vec{k}}_{B_{1},\ldots, B_{2j-3}} \pm \ket{1}_{B_{2j}} \otimes \ket{\Vec{l}}_{B_{1},\ldots, B_{2j-3}} \right),
    \end{align}
 for some arbitrary but fixed indices abbreviated as $\Vec{k}$ and $\Vec{l}$. Note that the only contribution in Eq.~\eqref{eq:Offdiag_expanded} on system $B_{2j}$ originates from the Choi state of $J_\mathcal{K}$ and $J_\mathcal{L}$, respectively. Crucially, according to Eq.~\eqref{eq:Psi_pm_identites}, which states that
\begin{align}
    &\ketbra{\psi^+}{\psi^{+}}_B - \ketbra{\psi^-}{\psi^{-}}_B=  \ketbra{0}{1}_{B_{2j}}  \otimes \ketbra{\Vec{k}}{\Vec{l}}_{B_1,\ldots, B_{2j-3}} + \ketbra{1}{0}_{B_{2j}}  \otimes \ketbra{\Vec{l}}{\Vec{k}}_{B_1,\ldots, B_{2j-3}},  \\
    &\ketbra{\psi^+}{\psi^{-}}_B- \ketbra{\psi^-}{\psi^{+}}_B=  \ketbra{1}{0}_{B_{2j}} \otimes \ketbra{\Vec{l}}{\Vec{k}}_{B_1,\ldots, B_{2j-3}}-\ketbra{0}{1}_{B_{2j}}  \otimes \ketbra{\Vec{k}}{\Vec{l}}_{B_1,\ldots, B_{2j-3}},
\end{align}
the only off diagonal terms on system $B_{2j}$ are given by
\begin{align}
    &\left(\idChannel_{B_{2j}}-\Delta_{B_{2j}}\right) \!J_\mathcal{K}^{T_{B_{1}\ldots, B_{2j-3},2j-1,2j}} \!=\!\frac{1}{2}\! \left( \ketbra{kn}{lm}\!+\!\ketbra{lm}{kn} \right)_{2j-1,2j}^{T} \!\otimes \!\left( \ketbra{\psi^+}{\psi^+}\!-\!\ketbra{\psi^-}{\psi^-}\right)_{B_{1}\ldots, B_{2j-3},B_{2j}}^{T_{B_{1}\ldots, B_{2j-3}}}, \\
    &\left(\idChannel_{B_{2j}}-\Delta_{B_{2j}}\right) \!J_\mathcal{L}^{T_{B_{1}\ldots, B_{2j-3},2j-1,2j}}  \!=\!\frac{1}{2}\! \left( \ketbra{kn}{lm}\!-\!\ketbra{lm}{kn} \right)_{2j-1,2j}^{T} \!\otimes \!\left( \ketbra{\psi^+}{\psi^-}\!-\!\ketbra{\psi^-}{\psi^+}\right)_{B_{1}\ldots, B_{2j-3},B_{2j}}^{T_{B_{1}\ldots, B_{2j-3}}}.
\end{align}
Inserting this into Eqs.~\eqref{eq:TOffDiagsNetwork} and~\eqref{eq:WOffDiagsNetwork}, respectively, yields
\begin{align}
    0\!&=\! \Delta_{2N+1,B_{2N-1},\cldots, B_{2j+1}}  \bra{{\Vec{k}}}_{B_1,\ldots, B_{2j-3}} \left( \bra{l m}_{2j-1,2j} J_{\mathcal{S}_N} \ket{kn}_{2j-1,2j} \! +\! \bra{k n}_{2j-1,2j}J_{\mathcal{S}_N} \ket{lm}_{2j-1,2j} \right)\ket{{\Vec{l}}}_{B_1,\ldots, B_{2j-3}}, \\
    0\!&=\!\Delta_{2N+1,B_{2N-1},\cldots, B_{2j+1}}  \bra{{\Vec{k}}}_{B_1,\ldots, B_{2j-3}} \left( \bra{l m}_{2j-1,2j} J_{\mathcal{S}_N} \ket{kn}_{2j-1,2j}  \!-\! \bra{k n}_{2j-1,2j}J_{\mathcal{S}_N} \ket{lm}_{2j-1,2j} \right)\ket{{\Vec{l}}}_{B_1,\ldots, B_{2j-3}}. 
\end{align}
Finally, adding the latter two expressions yields
\begin{align}
    0=\Delta_{2N+1,B_{2N-1},\ldots, B_{2j+1}}  \bra{{\Vec{k}}}_{B_1,\ldots, B_{2j-3}} \bra{l m}_{2j-1,2j} J_{\mathcal{S}_N} \ket{kn}_{2j-1,2j}\ket{{\Vec{l}}}_{B_1,\ldots, B_{2j-3}},
\end{align}
which is valid for all choices of $\Vec{k}$ and $\Vec{l}$, and all choices of $k,l$ and $n\neq m.$ After renaming the systems $B_j$ into $j$ again, this is equivalent to
    \begin{align}\label{eq:nextInput_dephased}
        \left(\idChannel_{{2j}}-\Delta_{{2j}}\right) \Delta_{2N+1,2N-1,\ldots, 2j+1} J_{\mathcal{S}_N}=0.
    \end{align}
Since this construction works for arbitrary $1\leq j\leq N$, Eq.~\eqref{eq:nextInput_dephased} holds for all $1\leq j\leq N$, and thus $\Delta_{2N+1,\ldots,2j+1} J_{\mathcal{S}_N}=\Delta_{2N+1,\ldots,2j+1} \Delta_{2N,\ldots, 2j} J_{\mathcal{S}_N} \quad \forall j: 1\leq j\leq N$. Notice the absence of $\Delta_{2N+1,\ldots,1} J_{\mathcal{S}_N}=\Delta_{2N+1,\ldots,1} \Delta_{2N,\ldots, 0} J_{\mathcal{S}_N}$ here, which, however, follows directly from the fact that $\mathcal{S}_N$ must be a $\DI$ channel from all its input to all its output systems. Since $\mathcal{S}_N$ is compatible by assumption, this yields $\Delta_{2N+1,\ldots 1}J_{\mathcal{S}_N} =\Delta_{2N+1,\ldots 1} \Delta_{2N,\ldots 0} J_{\mathcal{S}_N}$.

As promised earlier, we now address the $N=1$ case, i.e., it remains to show that the dephasing conditions in Eq.~\eqref{eq:DICompFromMaxSet} hold for any superchannel $\mS_1\in \mathfrak{C}$. The proof works analogously, where instead of composing $\mS_1$ with a network of free channels, we simply apply $\mS_1$ to the channels $\mK$ and $\mL$ from Lem.~\ref{lem:DIOTester} with the important caveat that we choose the systems $B_1,\ldots, B_{2j-3}$ to be trivial (cf. footnote~\ref{footnote:j_2_system}). As such $\mK^{1,B_2 \to 2}, \mL^{1,B_2 \to 2} \in \DI$ for the appropriate choice of indices as detailed in Lem.~\ref{lem:DIOTester}. In complete analogy, to the considerations from Eq.~\eqref{eq:Composition_DI_Tester} to Eq.~\eqref{eq:nextInput_dephased} we then obtain
\begin{align}\label{eq:DI_condition_superchannel}
    \Delta_3 J_{\mS_1}^{0123}=\Delta_{32} J_{\mS_1}^{0123} \quad \text{and} \quad  \Delta_{31} J_{\mS_1}^{0123}=\Delta_{3210} J_{\mS_1}^{0123},
\end{align}
where the second condition follows again from inserting the channel $\mE^{1,B_{2} \to 2, B_1}$ into $\mS_1$, i.e., considering $\mS_1$ as a channel from all input to all output systems. This finishes the proof that the conditions stated in the Theorem are necessary.

Moreover, that the dephasing conditions in Eq.~\eqref{eq:nextInput_dephased} already imply closedness under \textit{arbitrary} compositions, is essentially clear by inspection. More technically, if we take two arbitrary $\mS_N,\mT_M$ satisfying the dephasing conditions in Eq.~\eqref{eq:DICompFromMaxSet}, then we have to show that their composition $\mS_N*\mT_M$ must satisfy the conditions in Eq.~\eqref{eq:DICompFromMaxSet}. This follows similarly to Ref.~\cite[Lem.~5]{Ahnefeld2025}, and can be directly verified graphically by pulling dephasing maps through the composition as defined by the link product: For instance, using Fig.~\ref{fig:Composition}, a dephasing on the very last accessible output of $\mS_N*\mT_M$ can recursively be ``pulled through" all the composed systems until it can be ``pulled out" of the next accessible input system. The same can be recursively repeated for all other dephasing constraints to show that $\mS_N*\mT_M$ satisfies Eq.~\eqref{eq:DICompFromMaxSet} if  $\mS_N$ and $\mT_M$ do so to begin with. 

c) We finish the proof by noting that the free networks we used above are also contained in DIO. In complete analogy, following Eq.~\eqref{eq:Composition_DI_Tester} to Eq.~\eqref{eq:nextInput_dephased}, we obtain the necessary conditions of $\Delta_{2N+1,\ldots,2j+1} J_{\mathcal{S}_N}=\Delta_{2N+1,\ldots,2j+1} \Delta_{2N,\ldots, 2j} J_{\mathcal{S}_N} \quad \forall j: 0\leq j\leq N$. Moreover, the testers used in Ref.~\cite[Thm.~8]{Ahnefeld2025} are, in fact, also contained in $\DIO$ (cf. Lem.~\ref{lem:DIOTester}), and thus we obtain the remaining conditions of $\Delta_{0,2,\ldots, 2j} J_{\mathcal{S}_N}=\Delta_{0,2,\ldots, 2j} \Delta_{1,3,\ldots, 2j+1} J_{\mathcal{S}_N} \quad \forall j: 0\leq j\leq N,$ which shows that the conditions are necessary. That these conditions are also sufficient follows analogously from pulling the dephasing maps through the composition.
\end{proof}

In Eq.~\eqref{eq:completely_free_Superchannels} in the main text, we defined the set of completely free superchannels. The defining property for such superchannels is that they map free channels into free channels in a complete sense, and in particular, when the superchannel is applied to $\mK^{1,B_2\to 2}$. We emphasize that according to the latter proof, and particularly Eq.~\eqref{eq:DI_condition_superchannel} and the surrounding discussion, it holds that the completely free superchannels are precisely characterized by Eq.~\eqref{eq:DI_condition_superchannel}, and as such 
\begin{align}
\cO_1=\cfO_1 \quad \forall \O \in \{\MIO,\DI,\DIO\},
\end{align}
see also Ref.~\cite[Cor.~10]{Ahnefeld2025} for an explicit treatment of $\O=\MIO$. It remains an open question whether this property holds in general for any resource theory satisfying assumptions~\ref{assump:identity} to \ref{assumpt:complfree}.

In Sec.~\ref{sec:FreeSupermaps}, we discussed two other classes of free supermaps, namely free networks and maximally free supermaps. As claimed in the main text, these sets are, in general, not the same. We now set out to show this at the example of MIO and DI, respectively, as the set of free operations, and thereby provide the proof of Prop.~\ref{prop:nonfreenetworks}. To this end, we require the following Lemma.
\begin{lemma}\label{lem:BlockMatrix_supp}
    Let $A\in \mathbb{C}^{n\times n}$, $B\in \mathbb{C}^{m\times m}$, and $C\in \mathbb{C}^{n\times m}$, such that
    \begin{align}
    M=
        \begin{pmatrix}
            A & C \\
            C^\dagger & B
        \end{pmatrix}
        \geq 0
    \end{align}
    and let $\Pi_A$ and $\Pi_B$ denote projectors onto the support of $A$ and $B$, respectively. Then 
    \begin{align}
        C=\Pi_A C \Pi_B.
    \end{align}
\end{lemma}
\begin{proof}
 First notice that since $M\geq 0$, it follows that if $\braket{u|M|u}=0$, then $M\ket{u} =0$. This holds since $0=\braket{u|M|u}=\braket{\sqrt{M}u| \sqrt{M}u}= || \sqrt{M}\ket{u} ||^2$, and thus $\sqrt{M} \ket{u}=0$. Left-multiplying with $\sqrt{M}$ gives $M \ket{u}=0$. 

 Now let $\ket{x}\in \ker(A) \subseteq \mathbb{C}^n$, and define $\ket{u}=\ket{x} \oplus \ket{0}_m \in \mathbb{C}^{n+m}$, denoting that we attach the zero vector to $\ket{x}$. Then $\braket{u|M|u}= \braket{x|A|x}=0$. Therefore, 
 \begin{align}
     0= \begin{pmatrix}
            A & C \\
            C^\dagger & B
        \end{pmatrix}
        \begin{pmatrix}
            \ket{x} \\
            0
        \end{pmatrix}
        =\begin{pmatrix}
            A \ket{x}\\ C^\dagger \ket{x}
        \end{pmatrix},
 \end{align}
and thus, $C^\dagger \ket{x}=0$ for every $\ket{x}\in \ker(A)$. As such, for every $\ket{z} \in \mathbb{C}^m$ and every $\ket{x}\in \ker(A)$ we have $\braket{x| C| z}=\braket{C^\dagger x|z}=0$. Hence, $C\ket{z}\in \ker(A)^\perp =\supp(A)$. Projecting onto the support of $A$ leaves $C\ket{z}$ invariant, and since $\ket{z}$ is arbitrary, we find $\Pi_A C= C$. Secondly, let $\ket{y}\in \ker(B) \subseteq  \mathbb{C}^{m}$, and set $\ket{v}=\ket{0}_n \oplus \ket{y} \in \mathbb{C}^{n+m}$. Analogously, we obtain $C\ket{y}=0$ for every $\ket{y}\in \ker(B)$. Take any $\ket{z}  \in \mathbb{C}^{m}$. Then, decomposing $\ket{z}=\Pi_B \ket{z} +(\id-\Pi_B) \ket{z}$ yields
 \begin{align}
     C \ket{z}= C\Pi_B \ket{z} +C(\id-\Pi_B) \ket{z} =  C\Pi_B \ket{z},
 \end{align}
 where we used that $(\id-\Pi_B) \ket{z} \in \ker(B)$. Since $\ket{z}$ was arbitrary, we obtain $C=C\Pi_B$, and lastly, $C=\Pi_A C\Pi_B$.
\end{proof}

\nonfreenetworks*
\begin{proof}
    Consider the superchannel $\mS$ depicted in Fig.~\ref{fig:noFreeNetwork}, where systems $0,1,2,3,R^\prime$ are qubits. The post-processing channel $\mE_1^{2R^{\prime} \to 3}$, containing the controlled Pauli $X$ gate, acts on any $A,B$  (expressed in the incoherent basis by elements $B_{k,l}$) as
    \begin{align}\label{eq:E_2action}
        \mE_1^{2R^{\prime}\to 3}(A_2 \otimes B_{R^{\prime}}) &= \sum_{k,l=0}^1 \Tr{X^k A X^l}  B_{k,l} \ketbra{k}{l}_3= \sum_{k} \Tr{A} B_{k,k} \ketbra{k}{k}_3 + \sum_{k\neq l} \Tr{AX} B_{k,l} \ketbra{k}{l}_3 \nonumber \\
        &= \Tr{A} \Delta(B_3) +\Tr{AX}\left( B_3-\Delta(B_3)\right).
    \end{align}
 Let $\ket{\phi_0} :=\ket{0}$ and $\ket{\phi_1}:=H\ket{0}=\ket{+}$. The action of the superchannel $\mS$ on an arbitrary basis element $0\leq i,j,k,l\leq 1$ is
    \begin{align} \label{eq:E_2_action_basis}
        \mS(\ketbra{i}{j}_0 \otimes \ketbra{k}{l}_2)&= \left(\idChan^1 \otimes \, \mE_1^{2R^{\prime}\to 3} \right) \left(\mE_0^{0\to 1R^{\prime}}\left(\ketbra{i}{j}_0 \right)\otimes \ketbra{k}{l}_2\right) = \left(\idChan^1 \otimes \, \mE_1^{2R^{\prime}\to 3} \right) \left(\ketbra{i}{j}_1 \otimes \ketbra{\phi_i}{\phi_j}_{R^{\prime}}\otimes \ketbra{k}{l}_2\right) \nonumber \\
        &= \ketbra{i}{j}_1 \otimes  \left( \Delta\left(\ketbra{\phi_i}{\phi_j}_3\right) \delta_{k,l} + \left(\ketbra{\phi_i}{\phi_j}_3- \Delta(\ketbra{\phi_i}{\phi_j}_3) \right) \delta_{k+1,l} \right),
    \end{align}
    where the addition in the subscript of $\delta_{k+1,l}$ is modulo 2, or in other words, the term contributes iff $k\neq l$. In particular, Eq.~\eqref{eq:E_2_action_basis} implies that $\mS(\ketbra{i}{i}_0 \otimes \ketbra{k}{k}_2)= \ketbra{i}{i}_1 \otimes \Delta(\ketbra{\phi_i}{\phi_i}_3)$, which gives $\Delta_{02} J_{\mS}= \Delta_{0123} J_{\mS}$. Since $\Delta_{0} J_{\mS}= \Delta_{01} J_{\mS}$ also clearly holds, the superchannel $\mS\in \cMIO$ (see Thm.~\ref{thm:CompatibleSupermaps}). Moreover, for all elements $0\leq i,j\leq 1$,
    \begin{align}\label{eq:Action_S}
        \mS(\ketbra{i}{j}_0 \otimes \ketbra{+}{+}_2) &=\frac{1}{2} \sum_{k,l=0}^1\mS(\ketbra{i}{j}_0 \!\otimes\! \ketbra{k}{l}_2) = \ketbra{i}{j}_1\! \otimes\! \left( \Delta(\ketbra{\phi_i}{\phi_j}_3)  + \left(\ketbra{\phi_i}{\phi_j}_3- \Delta(\ketbra{\phi_i}{\phi_j}_3) \right) \right) \nonumber \\
        &=\ketbra{i}{j}_1 \otimes \ketbra{\phi_i}{\phi_j}_3.
    \end{align}
    We will now use this property to show by contradiction that there exists no decomposition of $\mS$ into a network of MIO operations. Let $R$ denote an auxiliary system of arbitrary finite dimension, and let $\mM_0^{0\to 1R},\mM_1^{2R\to 3} \in \MIO$ implement $\tilde{\mS}$ as shown in Fig.~\ref{fig:noFreeNetwork}. Assume that $\mS$ is equal to $\tilde{\mS}$. For convenience, define the states
    \begin{align}\label{eq:def_mu_i}
        \mu_i^{1R} := \mM_0^{0\to 1R}(\ketbra{i}{i}_0) \quad \forall i\in\{0,1\}.
    \end{align}

\begin{figure}[ht]
    \centering
    \includegraphics[width=0.8\linewidth]{ 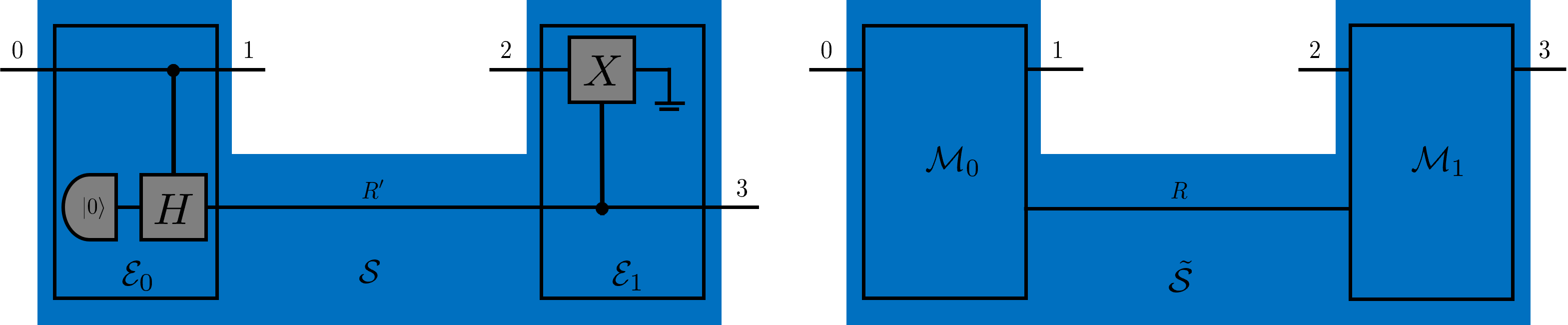}
    \caption{A superchannel $\mS\in \cMIO$, comprised of the pre- and post-processing channels $\mE_0,\mE_1$ . Here $H$ denotes the Hadamard and $X$ the Pauli $X$ gate. No superchannel $\tilde{\mS}$ comprised of arbitrary $\mM_0,\mM_1 \in \MIO$ can implement $\mS$.}
    \label{fig:noFreeNetwork}
\end{figure}

    If $\tilde{\mS}$ is equal to the superchannel $\mS$, it must in particular hold that $\partTr{3}{\mS(\ketbra{i}{i}_0 \otimes \ketbra{+}{+}_2)} = \partTr{3}{\tilde{\mS}(\ketbra{i}{i}_0 \otimes \ketbra{+}{+}_2)}$. Thus, as most easily seen from the figure,
    \begin{align}\label{eq:CondPartTrMu}
      \ketbra{i}{i}_1 \overset{\eqref{eq:Action_S}}{=} \partTr{3}{\mS(\ketbra{i}{i}_0 \otimes \ketbra{+}{+}_2)} &\overset{!}{=}  \partTr{3}{\tilde{\mS}(\ketbra{i}{i}_0 \otimes \ketbra{+}{+}_2)} = \partTr{3}{\left(\idChan^1 \otimes \, \mM_1^{2R\to 3} \right) \left(\mM_0^{0\to 1R}\left(\ketbra{i}{i}_0 \right)\otimes \ketbra{+}{+}_2\right)}\nonumber  \\
      &= \partTr{3}{\left(\idChan^1 \otimes \, \mM_1^{2R\to 3} \right) \left( \mu_i^{1R} \otimes \ketbra{+}{+}_2\right)} =  \partTr{R}{\mu_i^{1R}},
    \end{align}
    where in the last line we used that $\mM_1^{2R\to 3}$ is trace-preserving. Therefore,
    it must hold that 
    \begin{align}
        \partTr{R}{\mu_i^{1R}} \overset{!}{=}\ketbra{i}{i}_1  \quad \forall i\in\{0,1\},
    \end{align}
    and since the right-hand side is a pure state, there exist states $\tau_i^R $ such that
    \begin{align}\label{eq:mu_i_product}
        \mu_i^{1R} = \ketbra{i}{i}_1 \otimes \tau_i^R  \quad \forall i\in\{0,1\}.
    \end{align}
    However, the assumption $\mM_0 \in \MIO$ implies $\mu_i^{1R}=\Delta_{1R} \left(\mu_i^{1R} \right)$, and thus, there exists a conditional probability distribution $p_{j|i}$ such that
    \begin{align}\label{eq:mu_i_prop}
        \mu_i^{1R} =\ketbra{i}{i}_1 \otimes \tau_i^R = \ketbra{i}{i}_1 \otimes \sum_j p_{j|i}\ketbra{j}{j}_R  \quad \forall i\in\{0,1\},
    \end{align}
    since $\tau_i$ must be diagonal.  Defining the channel $\mN^{R\to 3}(\cdot):=\mM_1^{2R\to 3} \left(  \ketbra{+}{+}_2 \otimes (\cdot)\right)$ implies (using the same arguments as in Eq.~\eqref{eq:CondPartTrMu}) that
    \begin{align}
        \tilde{\mS}(\ketbra{i}{i}_0 \otimes \ketbra{+}{+}_2) &= \left( \idChan^1\otimes \, \mM_1^{2R\to 3}\right)(\mM_0^{0\to 1R}\left(\ketbra{i}{i}_0\right)\otimes \ketbra{+}{+}_2) \overset{\eqref{eq:def_mu_i}}{=}\left( \idChan^1\otimes \, \mM_1^{2R\to 3}\right)( \mu_i^{1R} \otimes \ketbra{+}{+}_2) \nonumber \\
        & \overset{\eqref{eq:mu_i_product}}{=}\ketbra{i}{i}_1\otimes \mM_1^{2R\to 3} \left(  \ketbra{+}{+}_2 \otimes \tau_i^R\right)= \ketbra{i}{i}_1\otimes \mN^{R\to 3}\left( \tau_i^R\right) \nonumber\\
        &\overset{\eqref{eq:Action_S}}{=} \ketbra{i}{i}_1 \otimes \ketbra{\phi_i}{\phi_i}_3 \quad \forall i\in \{0,1\},
    \end{align}
where in the last line we used that $\tilde{\mS}$ implements $\mS$ by assumption. This implies that
\begin{align}
    \ketbra{\phi_i}{\phi_i}_3=\mN^{R\to 3}\left( \tau_i^R\right) \overset{\eqref{eq:mu_i_prop}}{=} \sum_j p_{j|i} \,\mN^{R\to 3}\left(\ketbra{j}{j}_R \right) \quad \forall i\in\{0,1\}.
\end{align}
Since the left-hand side is pure, this, in turn, implies that for all $i\in\{0,1\}$,
\begin{align}\label{eq:phi_i_Eq}
    \ketbra{\phi_i}{\phi_i}_3= \mN^{R\to 3}\left(\ketbra{j}{j}_R \right) \quad \forall j: p_{j|i} >0,
\end{align}
which corresponds to basis elements $j$ contained in the supports of $\tau_i$, respectively. Additionally, note that $\tau_0$ and $\tau_1$ have orthogonal support: If this were not true, i.e., if there existed a $j$ such that $p_{j|0},p_{j|1}>0$, then Eq.~\eqref{eq:phi_i_Eq} would imply $\ket{\phi_0}=\ket{\phi_1}$, which is false. Let us now define the orthogonal projectors $\Pi_{i}$ onto the supports of $\tau_i$ respectively, i.e.,
\begin{align}\label{eq:Supp_tau}
    \Pi_{i} := \sum_{j: p_{j|i}>0} \ketbra{j}{j}.
\end{align}
Next, consider the Choi state of $\mM_0$ as a block matrix in the incoherent basis of  system $0$,
\begin{align}
    0\leq J_{\mM_0}^{01R}=\sum_{i,j=0}^{1} \ketbra{i}{j}_0 \otimes \mM_0^{\tilde{0}\to 1R}(\ketbra{i}{j}_{\tilde{0}})=
    \begin{pmatrix}
        \mu_0 & C \\
        C^\dagger &\mu_1
    \end{pmatrix},
\end{align}
 where $C:=\mM_0^{0\to 1R}(\ketbra{0}{1})$. We can now exploit the positive-semidefinitness of the Choi state and invoke Lem.~\ref{lem:BlockMatrix_supp}: Let $\Pi_{\mu_i}$ denote the projector onto the support of $\mu_i$, and $\Pi_i$ the projector onto the support of $\tau_i$ as defined in Eq.~\eqref{eq:Supp_tau}. Then $C$ must satisfy 
\begin{align}      
    C\overset{\mathrm{Lem.}~\ref{lem:BlockMatrix_supp}}{=} \Pi_{\mu_0} C \Pi_{\mu_1}= \left(\ketbra{0}{0}\otimes \Pi_0\right) C \left(\ketbra{1}{1}\otimes \Pi_1\right) =\ketbra{0}{1} \otimes \tilde{C}
\end{align}
for some $\tilde{C}=\Pi_0 \tilde{C} \Pi_1$. Therefore, we obtain
\begin{align} \label{eq:M_1_action}
    \mM_0^{0\to 1R}(\ketbra{0}{1}_0)= \ketbra{0}{1}_1 \otimes \tilde{C}_R \quad \text{with} \quad  \tilde{C}=\Pi_0 \tilde{C} \Pi_1.
\end{align}
Next, we consider 
\begin{align}\label{eq:Tilde_mS_action}
    \tilde{\mS}(\ketbra{0}{1}_0 \otimes \ketbra{+}{+}_2) &=\left(\idChan^{1} \otimes \,\mM_1^{2R\to 3}\right) \left(\mM_0^{0\to 1R}\left( \ketbra{0}{1}_0\right) \otimes \ketbra{+}{+}_2 \right) \overset{\eqref{eq:M_1_action}}{=} \ketbra{0}{1}_1 \otimes \mM_1^{2R\to3}\left( \ketbra{+}{+}_2 \otimes \tilde{C}_R\right) \nonumber \\
     &= \ketbra{0}{1}_1 \otimes \mN^{R\to 3}(\tilde{C}_R) \overset{\eqref{eq:M_1_action}}{=}\ketbra{0}{1}_1 \otimes \mN^{R\to 3}(\Pi_0 \tilde{C} \Pi_1).
\end{align}
Now take a Kraus representation for $\mN^{R\to 3}(\cdot)=\sum_n K_n (\cdot) K_n^\dagger $. Let $\ket{k} \in \supp(\tau_0)$ and $\ket{l} \in \supp(\tau_1)$, and recall from the above discussion that this implies that $\braket{k|l}=0$. From Eq.~\eqref{eq:phi_i_Eq}, we know that
\begin{align}
    \ketbra{\phi_0}{\phi_0}= \mN(\ketbra{k}{k}) =\sum_n K_n\ketbra{k}{k} K_n^\dagger,
\end{align}
and since the left-hand side is pure, each term must be proportional to $\ket{\phi_0}$, i.e., $K_n\ket{k}= \alpha_{n,k} \ket{\phi_0}$ for some $\alpha_{n,k}$. Likewise, from $ \ketbra{\phi_1}{\phi_1}= \mN(\ketbra{l}{l})$, one obtains $K_n\ket{l}= \alpha_{n,l} \ket{\phi_1}$. We can now invoke that the channel $\mN$ is trace-preserving, which yields, for all $k,l$ in the respective support
\begin{align}
    0=\braket{k|l} = \braket{k| \sum_n K_n^\dagger K_n |l} = \braket{\phi_0| \phi_1} \sum_n \alpha_{n,k}^* \alpha_{n,l}  \quad \forall k,l:  \ket{k} \in \supp(\tau_0),\ket{l} \in \supp(\tau_1).
\end{align}
Since $\braket{\phi_0| \phi_1} \neq 0$, this implies that for all indices $k,l$ such that $\ket{k} \in \supp(\tau_0),\ket{l} \in \supp(\tau_1)$, we have $\sum_n \alpha_{n,k}^* \alpha_{n,l} =0$. Using (the complex conjugate of) the latter expression yields that
\begin{align}\label{eq:N_Channel_0}
    \mN(\ketbra{k}{l})= \left( \sum_n \alpha_{n,k}\alpha_{n,l}^* \right) \ketbra{\phi_0}{\phi_1} =0 \quad \forall k,l:  \ket{k} \in \supp(\tau_0),\ket{l} \in \supp(\tau_1).
\end{align}
Finally, combining everything, we find that
\begin{align}
    \tilde{\mS}(\ketbra{0}{1}_0 \otimes \ketbra{+}{+}_2)& \overset{\eqref{eq:Tilde_mS_action}}{=}\ketbra{0}{1}_1 \otimes \mN^{R\to 3}(\Pi_0 \tilde{C}\Pi_1) = \ketbra{0}{1}_1 \otimes \sum_{\substack{\ket{k}\in \supp(\tau_0) \\ \ket{l}\in \supp(\tau_1)}} \tilde{C}_{k,l} \,\mN^{R\to 3}(\ketbra{k}{l})  \overset{\eqref{eq:N_Channel_0}}{=} 0.
\end{align}
This leads to the promised contradiction to Eq.~\eqref{eq:Action_S} and thus $\nMIO \subsetneq \cMIO$. That $\cMIO \subsetneq \mMIO$ was shown in Ref.~\cite{Ahnefeld2025}, see Fig.~6 in its Supplemental Material for an explicit example and the corresponding discussion around it.

In the detection-incoherent setting, consider the superchannel $\mS$ depicted in Fig.~\ref{fig:noFreeNetwork_DI}, where again all systems are qubits. The action of $\mS$ on an arbitrary basis element with $0\leq i,j,k,l \leq 1$ can be expressed as
 \begin{align} \label{eq:SuperChan_DI_action}
     \mS(\ketbra{i}{j}_0 \otimes \ketbra{k}{l}_2)=\delta_{k,l} \ketbra{i}{j}_1\otimes H^k\ketbra{i}{j}_3H^k.
 \end{align}
This implies that $J_\mS^{0123}$ satisfies $\Delta_3 J_{\mS}^{0123}=\Delta_{32} J_{\mS}^{0123}$ and $\Delta_{321} J_{\mS}^{0123}=\Delta_{3210} J_{\mS}^{0123}$ and thus $\mS\in \cDI$ (see Thm.~\ref{thm:CompatibleSupermaps}).
In particular, Eq.~\eqref{eq:SuperChan_DI_action} also implies that
\begin{align}
    \Delta_3\circ \mS(\ketbra{0}{1}_0 \otimes \ketbra{1}{1}_2) &= \ketbra{0}{1}_1\otimes \Delta(\ketbra{+}{-}_3) \neq 0, \label{eq:DI_contradiction}
\end{align}
which we will use to show that this superchannel cannot be implemented by a DI network. 

\begin{figure}[ht]
    \centering
    \includegraphics[width=0.8\linewidth]{ 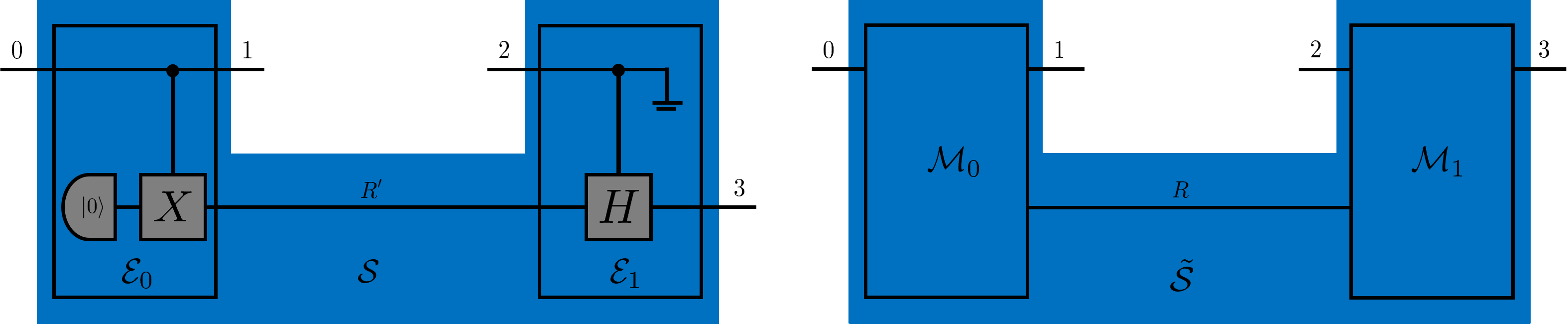}
    \caption{A superchannel $\mS\in\cDI$, comprised of the pre- and post-processing channels $\mE_0,\mE_1$ . Here $H$ denotes the Hadamard and $X$ the Pauli $X$ gate. No superchannel $\tilde{\mS}$ comprised of arbitrary $\mM_0,\mM_1 \in \DI$ can implement $\mS$.}
    \label{fig:noFreeNetwork_DI}
\end{figure}

Again, we proceed with a proof by contradiction using a quantum system $R$ of arbitrary but finite dimension. Let $\mM_0^{0\to 1R},\mM_1^{2R\to 3} \in \DI$, let $\tilde{\mS}$ be the superchannel implemented by these quantum channels, and define 
\begin{align}
    \sigma_i^{1R}:= \mM_0^{0\to 1R}(\ketbra{i}{i}_0) \quad \forall i\in \{0,1\}.
\end{align}
If $\tilde{\mS}$ is equal to $\mS$, it must hold that $\partTr{R}{ \sigma_i^{1R}}=\ketbra{i}{i}_1$, as most easily seen from Fig.~\ref{fig:noFreeNetwork_DI}, or more technically from
\begin{align}
     \partTr{R}{ \sigma_i^{1R}} =& \partTr{R}{\mM_0^{0\to 1R}(\ketbra{i}{i}_0)} = \partTr{3}{\left( \idChan^1\otimes \, \mM_1^{2R\to 3}\right)(\mM_0^{0\to 1R}\left(\ketbra{i}{i}_0\right)\otimes \ketbra{0}{0}_2) }\overset{!}{=}\partTr{3}{\mS(\ketbra{i}{i}_0 \otimes \ketbra{0}{0}_2)} \nonumber \\
     \overset{\eqref{eq:SuperChan_DI_action}}{=}&\ketbra{i}{i}_1,
\end{align}
where we used that $\mM_1$ is trace-preserving. Since the right-hand side is already pure, there must exist quantum states $\tau_i^R$  such that 
\begin{align}\label{eq:sigam_product}
    \sigma_i^{1R}=\ketbra{i}{i}_1\otimes \tau_i^R \quad \forall i\in\{0,1\}.
\end{align}
We next characterize the memory states $\tau_i^R$. In particular, we show that their supports lie in mutually orthogonal incoherent subspaces of $R$; this will provide the ingredient to invoke Lem.~\ref{lem:BlockMatrix_supp} and lead to a contradiction. To this end, define the channel $\mN^{R\to 3}(\cdot):=\mM_1^{2R\to 3}\left( \ketbra{0}{0}_2\otimes (\cdot)_R\right)$ and notice that
\begin{align}
    \ketbra{i}{i}_1 \otimes \ketbra{i}{i}_3&\overset{\eqref{eq:SuperChan_DI_action}}{=}\mS(\ketbra{i}{i}_0 \otimes \ketbra{0}{0}_2) \overset{!}{=} \left( \idChan^1\otimes \, \mM_1^{2R\to 3}\right)(\mM_0^{0\to 1R}\left(\ketbra{i}{i}_0\right)\otimes \ketbra{0}{0}_2) \nonumber \\
    &\overset{\eqref{eq:sigam_product}}{=}\ketbra{i}{i}_1\otimes  \mM_1^{2R\to 3}\left(\ketbra{0}{0}_2\otimes \tau_i^R  \right) = \ketbra{i}{i}_1\otimes  \mN^{R\to 3}(\tau_i^R)
\end{align}
enforces
\begin{align}\label{eq:DI_mN_map}
    \mN^{R\to 3}(\tau_i^R)=\ketbra{i}{i}_3 \quad \forall i\in\{0,1\}.
\end{align}
Moreover, $\mN\in \DI$, since by assumption, $\mM_1\in \DI$, and thus
\begin{align}
    \Delta_3 \mN(X_R)=\Delta_3\mM_1\left(\ketbra{0}{0}_2 \otimes X_R\right)=\Delta_3\mM_1\left(\ketbra{0}{0}_2 \otimes \Delta_R (X_R)\right)=\Delta_3 \mN(\Delta_R(X_R)).
\end{align} 
Choosing $M_i:= \left(\mN^{R\to 3}\right)^\dagger(\ketbra{i}{i}_3)$ defines a POVM $\{M_0,M_1\}$ since $\mN$ is a channel. Noticing that $\tau_i, M_i \geq 0 \, \forall i\in\{0,1\}$ together with
\begin{align}
    0 \overset{\eqref{eq:DI_mN_map}}{=} \Tr{\ketbra{i}{i}_3 \, \mN^{R\to 3}(\tau_{1-i}^R)} = \Tr{M_{i}^R \tau_{1-i}^R} \quad \forall i\in\{0,1\},
\end{align}
implies that $\supp(\tau_i) \subseteq \ker(M_{1-i})$. Now define the projectors $\Pi_{i}=\Pi_{\ker(M_{1-i})}$ therefore satisfying $\tau_i=\Pi_i \tau_i \Pi_i \quad \forall i\in\{0,1\}$. Moreover, $(\mN^{R\to 3})^\dagger \circ \Delta_3=\Delta_{R}\circ (\mN^{R\to 3})^\dagger \circ \Delta_3$ (or, in other words, $\mN$ cannot turn an incoherent measurement coherent) because $\mN\in \DI$. This implies $M_i=\Delta(M_i)$, and thus, there exists a conditional probability distribution $q_{i|r}$ such that 
\begin{align}
    M_i^R =\sum_{r=0}^{d_R-1} q_{i|r} \ketbra{r}{r}\quad \forall i\in\{0,1\} \quad \text{and}\quad \sum_{i=0}^{1}q_{i|r}=1 \quad \forall r: \{0,\ldots, d_R-1\}.
\end{align}
As such, we have $\Pi_{i}=\Delta(\Pi_{i})$, and from $M_0+M_1=\id$, it follows that $\Pi_0\Pi_1=0$. Next, define
\begin{align}
    C^{1R}:=\mM_0^{0\to 1R}(\ketbra{0}{1}_0).
\end{align}
Analogously to the MIO case, we want to exploit the positive-semidefinitness of the Choi state of $\mM_0$, i.e.,
\begin{align}
     0\leq J_{\mM_0}^{01R}=\sum_{i,j=0}^{1} \ketbra{i}{j}_0 \otimes \mM_0^{\tilde{0}\to 1R}(\ketbra{i}{j}_{\tilde{0}})=
    \begin{pmatrix}
        \sigma_0 & C \\
        C^\dagger &\sigma_1
    \end{pmatrix}.
\end{align}
To this end, let us define the projectors $P_i:=\Pi_{\supp(\tau_i)}$. Then, Lem.~\ref{lem:BlockMatrix_supp} together with Eq.~\eqref{eq:sigam_product} implies that
\begin{align}\label{eq:C_1R_projected}
    C^{1R}= \left(\ketbra{0}{0}_1 \otimes P_0^R \right)C^{1R}\left(\ketbra{1}{1}_1 \otimes P_1^R\right) = \ketbra{0}{1}_1 \otimes D^R
\end{align}
for some operator $D=P_0DP_1$. From $\tau_i=\Pi_i \tau_i \Pi_i$, it follows that the support of $\tau_i$ lies inside the subspace onto which $\Pi_i$ projects, and thus $\Pi_i P_i= P_i = P_i\Pi_i$. Hence, $D=P_0DP_1 = \Pi_0 P_0 D P_1  \Pi_1  = \Pi_0 D \Pi_1$. Moreover, it follows that
\begin{align}\label{eq:Delta_D_R}
\Delta_R(D_R)=\Delta_R(\Pi_0 D_R \Pi_1)=\Pi_0\,\Delta_R(D_R)\,\Pi_1=\Delta_R(D_R)\Pi_0\Pi_1=0,
\end{align}
where we used that $\Delta_R(D_R)$ commutes with the diagonal projectors $\Pi_0,\Pi_1$. Lastly, consider that
    \begin{align}\label{eq:DI_contradiction_impl}
        \Delta_3 \circ \tilde{\mS}(\ketbra{0}{1}_0 \otimes \ketbra{1}{1}_2) &= \Delta_3 \circ\left( \idChan^1\otimes \, \mM_1^{2R\to 3}\right)(\mM_0^{0\to 1R}\left(\ketbra{0}{1}_0\right)\otimes \ketbra{1}{1}_2) = \Delta_3 \circ\left( \idChan^1\otimes \, \mM_1^{2R\to 3}\right)(C^{1R}\otimes \ketbra{1}{1}_2) \nonumber \\
        &\overset{\eqref{eq:C_1R_projected}}{=} \ketbra{0}{1}_1 \otimes \Delta_3 \!\circ\!\mM_1^{2R\to 3}\left(\ketbra{1}{1}_2\otimes D^R\right) \!=\! \ketbra{0}{1}_1 \otimes \Delta_3 \circ\mM_1^{2R\to 3}\left(\ketbra{1}{1}_2\otimes \Delta_R(D^R)\right) \nonumber \\
        &\overset{\eqref{eq:Delta_D_R}}{=} 0,
    \end{align}
where in the second-to-last equality, we used that $\mM_1 \in \DI$ by assumption. This is a contradiction to Eq.~\eqref{eq:DI_contradiction}, and as such $\nDI_1 \subsetneq \cDI_1$. Lastly, that $\cDI\subsetneq \mDI$ follows directly from Prop.~\ref{cor:DIConversionDistance}: different conversion distances imply that there exists a superchannel contained in $\mDI_1$ but not in $\cDI_1$. 
\end{proof}

\subsection{Manipulation of dynamical coherence}\label{sec:Appen_Coherence_Manipulation}
In this section, we investigate the manipulation of dynamical coherence under different sets of free supermaps. As a key technical tool, we use a twirling channel, which we introduce and characterize in the following section.

\subsubsection{Weyl-Heisenberg twirling}~\label{sec:Twirling}
To investigate the manipulation of dynamical resources under compatible supermaps, and in particular, to show optimality of conversion distances, we rely on a twirling map (on the level of Choi states, or equivalently channels) based on a symmetry of the Fourier transform. In general, for any compact group $G$, let us define the twirling channel (acting on Choi states) as
\begin{align}\label{eq:GTwirl}
    \mT_G(J_\mN)=\int_G d\mu(g) (U_g\otimes V_g) J_\mN(U_g\otimes V_g)^\dagger,
\end{align}
where $d\mu(g)$ denotes the normalized Haar measure, and $U_g,V_g$ are unitary representations of the group $G$ on the input and output systems of the channel, respectively. Such a twirling can be interpreted as taking random concatenations (with respect to the Haar measure) with unitary pre- and post-processing channels $\mU_g,\mV_g$, i.e.,
\begin{align}
    \mT_G(J_\mN)= \int_G d\mu(g) J_{\mV_g\circ \mN\circ \,\mU_g^T} =: J_{\tilde{\mN}} \quad \text{where } \tilde{\mN}= \int_G d\mu(g) \mV_g\circ \mN\circ \mU_g^T
\end{align}
and $\mU_g^T(\cdot)=U_g^T(\cdot) U_g^*$ and  $\mV_g(\cdot)=V_g(\cdot) V_g^\dagger$ \footnote{Note that, more generally, one could also consider a twirling that involves random concatenations with unitaries and an auxiliary system. This could be useful for other applications, but is not necessary for our purposes.}. 
Using the standard identity 
\begin{align}\label{eq:MaxEntSwitch}
    (U\otimes \id) \ket{\phi}= (\id \otimes U^T) \ket{\phi},
\end{align}
where $\ket{\phi}=\tfrac{1}{\sqrt{d}}\sum_i\ket{ii}$, this follows from
\begin{align}
    (U_g \otimes V_g) J_\mN (U_g \otimes V_g)^\dagger &= \sum_{i,j} U_g\ketbra{i}{j}U_g^\dagger \otimes V_g\mN(\ketbra{i}{j})V_g^\dagger =\sum_{i,j} \ketbra{i}{j} \otimes V_g\mN(U_g^T\ketbra{i}{j}U_g^*)V_g^\dagger = J_{\mV_g\circ \mN\circ \,\mU_g^T}.
\end{align}
Specifying to dynamical coherence, recall that the Choi state of the Fourier transform is given by
\begin{align}
    J_{\mF_d}^{03}= d\, (\id^0 \otimes F_d^3) \ketbra{\phi}{\phi}^{03} (\id^0 \otimes F_d^3)^\dagger.
\end{align}
Using again Eq.~\eqref{eq:MaxEntSwitch}, it becomes clear that $J_{\mF_d}$ is invariant under the action of $U=(\id \otimes F_d) (V\otimes V^*) (\id \otimes F_d)^\dagger =V \otimes (F_d V^* F_d^\dagger)$ for \textit{any} unitary $V$. Our goal is to use this property to investigate resource manipulation. To this end, we start with a characterization of unitaries that commute with the complete dephasing operation $\Delta$, see also Ref.~\cite[Thm.~31]{Chitambar2016}.

\begin{lem}\label{lem:MonomialUnitary}
Let $\Delta$ denote the complete dephasing channel and let $\mU(\cdot) =U(\cdot)U^\dagger$ be a unitary channel. Then, 
\begin{align}
    \mU \circ \Delta=  \Delta\circ \mU,
\end{align}
iff $U$ is a monomial unitary, i.e., $U=P_\pi D$, where $P_\pi=\sum_k \ketbra{\pi(k)}{k}$ is a permutation matrix and $D=\sum_k \omega_k \ketbra{k}{k}$ is a diagonal unitary.
\end{lem}
\begin{proof}
If $U$ is a monomial unitary, then $U \ketbra{i}{i} U^\dagger=P_\pi D \ketbra{i}{i} D^\dagger P_\pi^\dagger=P_\pi \ketbra{i}{i} P_{\pi}^\dagger= \ketbra{\pi(i)}{\pi(i)}$. For any operator $X$, this implies that
\begin{align}
    \mU \Delta(X) &= \sum_i U\ketbra{i}{i} X\ketbra{i}{i} U^\dagger =  \sum_i \ketbra{\pi(i)}{\pi(i)} UXU^\dagger \ketbra{\pi(i)}{\pi(i)}=\sum_k \ketbra{k}{k} U X U^\dagger \ketbra{k}{k} \nonumber \\
    &= \Delta \mU (X),
\end{align}
where in the second-to-last step we relabeled the summation index $\pi(i)$ to $k$, which is possible since $\pi$ is a bijection. Conversely, let $U$ be any unitary satisfying $\mU \Delta=  \Delta \mU$. Let $\ket{\psi_i}:= U \ket{i} =\sum_k c_k^{(i)} \ket{k}$, then $\mU \Delta=  \Delta \mU$ implies that $\Delta \ketbra{\psi_i}{\psi_i}= \ketbra{\psi_i}{\psi_i}.$
Since $\ket{\psi_i}$ is rank-one, this implies that for each $i$, exactly one $c_k^{(i)}$ is non-vanishing. Moreover, since $U$ is unitary, we have $\braket{\psi_i|\psi_j}=\delta_{i,j}$, and thus, there exists a permutation $\pi$ such that $\ket{\psi_i}=\omega_i \ket{\pi(i)}$. Hence, the unitary $U$ is composed of a diagonal unitary $D=\sum_i \omega_i \ketbra{i}{i}$ and a permutation matrix $P_\pi=\sum_i \ketbra{\pi(i)}{i}$ with $U=P_\pi D$. \qedhere
\end{proof}

Next, consider the shift and clock matrices~\cite{Schwinger1960}
\begin{align}\label{eq:def_shiftandclock}
    S:=\sum_{k=0}^{d-1} \ketbra{k+1}{k} \quad \text{and} \quad \Omega:=\sum_{k=0}^{d-1} \omega_d^k\ketbra{k}{k},
\end{align}
where $\omega_d=e^{2\pi i/d}$ and the addition is to be understood modulo $d$, i.e., $\ket{k+1}= \ket{k+1 \,\text{mod} \,d}$. From the shift and clock matrices, one can construct monomial unitaries known as Weyl–Heisenberg matrices~\cite{Schwinger1960} 
\begin{align}
    W_{n,m}:=S^n \Omega^m,
\end{align}
of which there are a total of $d^2$ distinct matrices. Importantly, this is one of the few places where the superscripts are to be understood as matrix exponents.  Let us define the unitaries 
\begin{align}\label{eq:U_nm}
    U_{n,m} := W_{n,m} \otimes F W_{n,m}^* F^\dagger,
\end{align}
which are clearly of the form $U_{n,m}=(\id \otimes F) (W_{n,m}\otimes W_{n,m}^*) (\id \otimes F)^\dagger$. Moreover, we define a twirling channel $\mT$ via
\begin{align}\label{eq:TwirlingChannel}
    \mT(X) := \frac{1}{d^2} \sum_{n,m} U_{n,m} X U_{n,m}^\dagger =\frac{1}{d^2} \sum_{n,m} (W_{n,m} \otimes F W_{n,m}^* F^\dagger) X (W_{n,m} \otimes F W_{n,m}^* F^\dagger )^\dagger.
\end{align}
To characterize -- and later use -- this twirling channel, we require the following Lemma.
\begin{lem}\label{lem:UnitaryRewriting}
Let $ \ket{\Psi_{k,l}}:= \frac{1}{d} \sum_{a,b} \omega_d^{a(b-k)-lb} \ket{a,b}$ and $\ket{\phi_k}:=\frac{1}{\sqrt{d}} \sum_{b} \omega_d^{bk}\ket{b}$, and define
\begin{subequations}
    \begin{align}
        \Pi_{k,l}:=&\ketbra{\Psi_{k,l}}{\Psi_{k,l}}, \label{eq:Pi_kl}\\
         P_l :=& \sum_a \ketbra{a}{a}_0 \otimes \ketbra{\phi_{a-l}}{\phi_{a-l}}_3, \\
         Q_k := & \sum_b\ketbra{\phi_{b-k}}{\phi_{b-k}}_0 \otimes \ketbra{b}{b}_3,
    \end{align}
\end{subequations}
where the addition in the indices is modulo $d$. Then, the unitary $U_{n,m}$ defined in Eq.~\eqref{eq:U_nm} can be written as
\begin{align}\label{eq:U_asMonomialUnitaries}
    U_{n,m}= W_{n,m} \otimes F W_{n,m}^* F^\dagger=\omega_d^{nm} (W_{n,m} \otimes W_{m,n})= \sum_{k,l} \omega_d^{nk+lm} \,\Pi_{k,l}.
\end{align}
Moreover, $\{ \Pi_{k,l} \}_{0\leq k,l\leq d-1}$,  $\{P_l \}_{0\leq l \leq d-1} $, and $\{Q_k \}_{0\leq k \leq d-1}$ are sets of complete and orthogonal projectors satisfying  
\begin{subequations}
   \begin{align}
       &\Delta_0 \Pi_{k,l}^{03} = \frac{1}{d} P_l^{03} = \frac{1}{d}\sum_k \Pi_{k,l}^{03} , \label{eq:P_l}\\
        &\Delta_3 \Pi_{k,l}^{03} = \frac{1}{d} Q_k^{03} =\frac{1}{d}\sum_l \Pi_{k,l}^{03}. \label{eq:Q_k}
   \end{align}
\end{subequations}
\end{lem}
\begin{proof}
That $\{ \Pi_{k,l} \}_{0\leq k,l\leq d-1}$,  $\{P_l \}_{0\leq l \leq d-1} $, and $\{Q_k \}_{0\leq k \leq d-1}$ are sets of complete and orthogonal projectors is essentially clear by inspection. More technically, note that for $0\leq k,k^\prime, l,l^\prime \leq d-1$ we have
\begin{subequations}
\begin{align}
    \braket{\Psi_{k,l}| \Psi_{k^\prime,l^\prime}}& = \frac{1}{d^2} \sum_{a,b} \omega_d^{-a(b-k)+lb} \omega_d^{a(b-k^\prime)-l^\prime b} =\frac{1}{d^2} \sum_{a,b}\omega_d^{a(k-k^\prime)+b(l-l^\prime)} = \delta_{k,k^\prime} \delta_{l,l^\prime}, \label{eq:ProjectorsOrthonormal} \\
    \braket{\phi_k|\phi_{k^\prime}} &=  \frac{1}{d} \sum_{b} \omega_d^{bk^\prime-b^\prime k} =\delta_{k,k^\prime},
\end{align}
\end{subequations}
and thus, $\Pi_{k,l} \Pi_{k^\prime,l^\prime}=\delta_{k,k^\prime} \delta_{l,l^\prime} \Pi_{k,l}$, $P_l P_{l^\prime}=\delta_{l,l^\prime}P_l$, and $Q_k Q_{k^\prime}=\delta_{k,k^\prime}Q_k$. Moreover, 
\begin{align}
    \sum_{k,l} \Pi_{k,l}&= \sum_{k,l}\frac{1}{d^2} \sum_{a,b} \sum_{a^\prime,b^\prime} \omega_d^{a(b-k)-lb}\omega_d^{-a^\prime(b^\prime-k)+lb^\prime} \ketbra{a,b}{a^\prime,b^\prime}=\sum_{a,b} \sum_{a^\prime,b^\prime} \omega_d^{ab-a^\prime b^\prime} \left( \frac{1}{d^2} \sum_{k,l} \omega_d^{k(a^\prime-a)+l(b^\prime-b)} \right) \ketbra{a,b}{a^\prime,b^\prime} \nonumber \\
    &= \id\otimes\id,
\end{align}
and analogously for $\{P_l \}_{0\leq l\leq d-1}$ and $\{Q_k \}_{0\leq k\leq d-1}$. Additionally, a straightforward calculation reveals that 
\begin{align}
    \sum_k \Pi_{k,l}^{03}\mkern-4mu=\mkern-4mu\frac{1}{d^2}\sum_k \sum_{a,a^\prime} \sum_{b,b^\prime} \omega_d^{a(b\mkern-1mu-\mkern-1muk)\mkern-1mu-\mkern-1mu lb}\omega_d^{-a^\prime(b^\prime-k)+lb^\prime} \ketbra{a}{a^\prime}_0\! \otimes\! \ketbra{b}{b^\prime}_3 \!=\!\frac{1}{d}\sum_a \ketbra{a}{a}_0 \!\otimes\! \sum_{b,b^\prime}\omega_d^{(a-l)(b-b^\prime)} \ketbra{b}{b^\prime}_3=P_l,
\end{align}
and analogously for $Q_k$ (summing the $\Pi_{k,l}$ over $l$ instead). Thus, the sets are complete since $\{\Pi_{k,l} \}_{0\leq k,l\leq d-1}$ is a complete set of projectors. Furthermore, it holds that
\begin{align}
    \Delta_0 \Pi_{k,l}^{03}&=\frac{1}{d^2} \sum_{a} \sum_{b,b^\prime} \omega_d^{a(b-k)-lb}\omega_d^{-a(b^\prime-k)+lb^\prime} \ketbra{a}{a}_0 \otimes \ketbra{b}{b^\prime}_3  = \frac{1}{d^2} \sum_{a} \sum_{b,b^\prime} \omega_d^{(a-l)(b-b^\prime)}\ketbra{a}{a}_0 \otimes \ketbra{b}{b^\prime}_3 \nonumber \\
    &= \frac{1}{d} \sum_a \ketbra{a}{a}_0 \otimes \ketbra{\phi_{a-l}}{\phi_{a-l}}_3 = \frac{1}{d} P_l^{03},
\end{align}
and analogously for $\Delta_3\Pi_{k,l}^{03}=\frac{1}{d} \sum_b\ketbra{\phi_{b-k}}{\phi_{b-k}}_0 \otimes \ketbra{b}{b}_3= \frac{1}{d} Q_k^{03}$. Next, we show that $U_{n,m}$ can be written as a product of two Heisenberg-Weyl matrices (up to a phase) as in Eq.~\eqref{eq:U_asMonomialUnitaries}. The straightforward calculation
\begin{subequations}
\begin{align}
    &F S F^\dagger \ket{n}\!=\! F \frac{1}{\sqrt{d}} \sum_k \omega_d^{-kn} \ket{k+1}= \frac{1}{d}\sum_{k,l}  \omega_d^{-kn} \omega_d^{l(k+1)} \ket{l} \!=\!\sum_l \omega_d^l\left( \frac{1}{d} \sum_k \omega_d^{k(l-n)} \right) \ket{l}\!= \!\omega^n \ket{n} \!=\!\Omega \ket{n}, \\
    &F \Omega^{-1} F^\dagger \ket{n} \!=\!F \frac{1}{\sqrt{d}} \sum_k \omega_d^{-kn} \omega_d^{-k}\ket{k}\!=\!\frac{1}{d}\sum_{k,l}  \omega_d^{-kn} \omega_d^{-k} \omega_d^{kl} \ket{l} \!=\!\sum_l \left( \frac{1}{d} \sum_k \omega_d^{k(l-(n+1))} \right) \ket{l}\!=\!\ket{n+1}\!=\! S \ket{n},
\end{align}
\end{subequations}
reveals that $F S F^\dagger = \Omega$ and $F \Omega^{-1} F^\dagger= S$. Using this to evaluate $F W_{n,m}^* F^\dagger$ yields
\begin{align}
    F W_{n,m}^* F^\dagger= F S^n \Omega^{-m} F^\dagger= (F S^n F^\dagger)(F\Omega^{-m} F^\dagger)= \Omega^n S^m.
\end{align}
Moreover, note that $\Omega S=\omega_d S \Omega$ since
\begin{align}
    \Omega S = \sum_k \omega_d^{k+1} \ketbra{k+1}{k}= \omega_d \sum_k \omega_d^{k} \ketbra{k+1}{k}= \omega_d S \Omega.
\end{align}
By repeated application of the last identity, this implies $\Omega^n S^m= \omega_d^{nm} S^m \Omega^n$, and thus
\begin{align}
    U_{n,m}=W_{n,m} \otimes F W_{n,m}^* F^\dagger=  \omega_d^{nm} (W_{n,m} \otimes W_{m,n}).
\end{align}
Next, we show that $U_{n,m}$ can be decomposed into the projectors $\Pi_{k,l}$. Recall that $\Omega^a S^b=\omega_d ^{ab}S^b\Omega^a$. Thus, 
\begin{align}\label{eq:WeylProduct}
    W_{n,m} W_{k,l} & = S^n \Omega^m S^k \Omega^l =\omega_d^{mk} S^{n+k} \Omega^{m+l}=\omega_d^{mk} W_{n+k,m+l}.
\end{align}
This implies that all $U_{n,m}$ commute since 
\begin{subequations}
\begin{align}
    &(W_{n,m} \otimes W_{m,n})(W_{k,l} \otimes W_{l,k})= (W_{n,m}W_{k,l} \otimes W_{m,n}W_{l,k})= \omega_d^{mk+nl} (W_{n+k,m+l} \otimes W_{m+l,n+k}),  \\
    &(W_{k,l} \otimes W_{l,k})(W_{n,m} \otimes W_{m,n})= (W_{k,l}W_{n,m} \otimes W_{l,k}W_{m,n})= \omega_d^{ln+mk} (W_{n+k,m+l} \otimes W_{m+l,n+k})
\end{align}
\end{subequations}
implies $[U_{n,m},U_{k,l}]=0 \, \forall n,m,k,l$. Next, we show that the states $\ket{\Psi_{k,l}}$ are eigenstates of $U_{n,m}$. To this end, note that the two operators
\begin{align}
    G_1:= S\otimes \Omega \quad \text{and} \quad G_2:= \Omega \otimes S
\end{align}
generate all possible $U_{n,m}$ since 
\begin{align}
    G_1^n G_2^m=S^n\Omega^m\otimes \Omega^n S^m =\omega_d^{nm} W_{n,m} \otimes W_{m,n} =U_{n,m}.
\end{align}
A straightforward calculation reveals that 
\begin{align}\label{eq:G_1_action}
    G_1 \ket{\Psi_{k,l}}&\!=\!\frac{1}{d} \sum_{a,b}\! \omega_d^{a(b-k)-lb} \omega_d^{b}\ket{a+1,b}\!=\!\frac{1}{d} \sum_{a^\prime,b}\! \omega_d^{(a^\prime-1)(b-k)-lb+b} \ket{a^\prime,b}\!=\!\frac{\omega_d^k}{d} \sum_{a^\prime,b}\! \omega_d^{a^\prime(b-k)-lb} \ket{a^\prime,b}  =\omega_d^k  \ket{\Psi_{k,l}},
\end{align}
and
\begin{align}\label{eq:G_2_action}
    G_2  \!\ket{\Psi_{k,l}} &\!=\!\frac{1}{d} \!\sum_{a,b}\!\omega_d^{a(b-k)-lb} \omega_d^{a}\ket{a,b\!+\!1}\!=\!\frac{1}{d} \!\sum_{a,b^\prime}\! \omega_d^{a(b^\prime-1-k)-l(b^\prime-1)+a} \ket{a,b^\prime}\!=\!\frac{\omega_d^l}{d} \!\sum_{a,b^\prime}\! \omega_d^{a(b^\prime-k)-lb^\prime} \ket{a,b^\prime} \!=\!\omega_d^l \ket{\Psi_{k,l}}.
\end{align}
Combining the two expressions, we find that $U_{n,m} \ket{\Psi_{k,l}}=G_1^nG_2^m \ket{\Psi_{k,l}}= \omega_d^{nk+lm}\ket{\Psi_{k,l}}$. 
Lastly, expanding $U_{n,m}$ in terms of the complete set of projectors $\Pi_{k,l}$ yields
\begin{align}
    U_{n,m}=\sum_{k,l} U_{n,m} \Pi_{k,l}= \sum_{k,l} \omega_d^{nk+lm} \Pi_{k,l}.
\end{align}
\end{proof}

We are now ready to provide a characterization of (Choi states of) superchannels and supermaps that are invariant under the action of the twirling channel $\mT$ defined in Eq.~\eqref{eq:TwirlingChannel} when applied to parts of its input and output systems. 

\begin{lem}\label{lem:TwirlingInvariantChannel}
    Let $d_{0}=d_{3}=d$. Then $J_{\mS_1}^{0123} \in \Comb_1$ and 
    \begin{align}
        \idChan^{12} \otimes \mT^{03} \left(J_{\mS_1}^{0123}\right) =J_{\mS_1}^{0123}
    \end{align}
    iff 
    \begin{align}
       J_{\mS_1}^{0123}= \sum_{k,l} Y_{k,l}^{12} \otimes \Pi_{k,l}^{03},
    \end{align}
    with $Y_{k,l}^{12} \geq 0 \, \forall k,l\in \{0,\ldots, d-1\}$ and $\sum_{k,l} Y_{k,l}^{12} = M^1 \otimes \id^2$ for some $M^1\geq 0$ with $ \Tr{M^1}=d$.
\end{lem}

\begin{proof}
  The Weyl-Heisenberg matrices form a basis of $\mathcal{L}(\mathbb{C}^d)$~\cite{Schwinger1960}, and thus, we can expand any $J_{\mS_1}^{0123}$ as
\begin{align}
    J_{\mS_1}^{0123}= \sum_{a,b,c,d} J_{abcd}^{12} \otimes W_{a,b}^{0} \otimes W_{c,d}^{3}.
\end{align}
 Moreover, note that Eq.~\eqref{eq:WeylProduct} implies 
\begin{align}\label{eq:WeylSandwiched}
    W_{n,m} W_{a,b} W_{n,m}^\dagger& = \omega_d^{ma} S^{n+a} \Omega^{m+b} \Omega^{-m} S^{-n}= \omega_d^{ma} S^{n+a} \Omega^{b}  S^{-n} = \omega_d^{ma-bn} S^a \Omega^b = \omega_d^{ma-bn} W_{a,b}.
\end{align}
Thus, the action of the twirling channel $\idChan^{12} \otimes \mT^{03}$ can be written as (omitting system indices for better readability)
\begin{align}
    \idChan^{12} \otimes \mT^{03} \left(J_{\mS_1}^{0123}\right) \!&=\! \frac{1}{d^2} \sum_{n,m}\sum_{a,b,c,d}\! J_{abcd} \otimes W_{n,m}W_{a,b}W_{n,m}^\dagger \otimes W_{m,n}W_{c,d}W_{m,n}^\dagger \nonumber \\
    &=\sum_{a,b,c,d}  \left(\!\frac{1}{d^2}\sum_{n,m}\!\omega_d^{ma-bn} \omega_d^{nc-dm}\! \right) J_{abcd}\otimes W_{a,b} \otimes W_{c,d} =\sum_{a,b}  J_{abba}\otimes W_{a,b} \otimes W_{b,a}.
\end{align}
This implies that any operator $J_{\mS_1}^{0123}$ that is invariant under the action of the twirling can be written as
\begin{align}
    J_{\mS_1}^{0123}=\sum_{a,b} X_{ab}^{12} \otimes W_{a,b}^{0} \otimes W_{b,a}^{3}
\end{align}
for some $X_{ab}^{12}$. Using Lem.~\ref{lem:UnitaryRewriting}, and in particular, Eq.~\eqref{eq:U_asMonomialUnitaries}, we obtain 
\begin{align}\label{eq:J_Twirled}
    J_{\mS_1}^{0123}=\sum_{k,l} \left(\sum_{a,b}  \omega_d^{-ab}\omega_d^{ak+lb}X_{ab}^{12}\right) \otimes  \Pi_{k,l}^{03} =: \sum_{k,l} Y_{k,l}^{12} \otimes \Pi_{k,l}^{03}.
\end{align}
Now, we derive the conditions that $J_{\mS_1}^{0123}$ must satisfy in order to correspond to a superchannel $\mS_1$. Since the $\Pi_{k,l}^{03}$ are orthogonal projectors, we have that $J_{\mS_1}^{0123} \geq 0$ iff $Y_{k,l}^{12}\geq 0 \, \forall k,l$. Furthermore, using that 
\begin{align}\label{eq:partTr_Pi_kl}
    \partTr{3}{\Pi_{k,l}^{03}}&= \frac{1}{d^2} \sum_{a,a^\prime} \sum_{b} \omega^{a(b-k)-lb}\omega^{-a^\prime(b-k)+lb}  \ketbra{a}{a^\prime}_0 = \frac{1}{d^2} \sum_{a,a^\prime} \left(\sum_{b}\omega^{b(a-a^\prime)} \right) \omega^{k(a^\prime -a)}\ketbra{a}{a^\prime}_0 =\frac{\id_0}{d},
\end{align}
we see that the normalization constraint $\partTr{3}{J_{\mS_1}^{0123}}=L^{01} \otimes \id^2$, for some $L^{01}\geq 0$ with $\partTr{1}{L^{01}}=\id^0$, is satisfied iff
\begin{align}\label{eq:NormalizationConstrained_symmetry}
    \sum_{k,l} Y_{k,l}^{12} \otimes \id^0 = d\, L^{01} \otimes \id^2.
\end{align}
Next, we show that 
\begin{align}\label{eq:NormalizationConstrained_reduced}
    \sum_{k,l} Y_{k,l}^{12} = M^1\otimes \id^2, \quad \text{and}\quad L^{01}= \frac{1}{d} \id^0 \otimes M^1.
\end{align}
for some $M \geq 0$ with $\Tr{M}=d$. To this end, assume that Eq.~\eqref{eq:NormalizationConstrained_symmetry} holds for a feasible $L^{01}$. Taking partial traces over systems $0$ and $2$ of Eq.~\eqref{eq:NormalizationConstrained_symmetry} respectively yields
\begin{align}\label{eq:NormalizationConstraint_partialTRaces}
    \sum_{k,l} Y_{k,l}^{12}= \partTr{0}{L^{01}} \otimes \id^2, \quad \text{and} \quad \partTr{2}{\sum_{k,l}Y_{k,l}^{12}} \otimes \id^0= d d_2 L^{01},
\end{align}
where $d_X$ denotes the dimension of system $X$ (and $d_0=d_3=d$). Inserting the former into the latter gives
\begin{align}
    dd_2 \,L^{01}= \partTr{2}{\partTr{0}{L^{01}} \otimes \id^2} \otimes \id^0=  d_2 \,\partTr{0}{L^{01}}\otimes \id^0,
\end{align}
and thus $L^{01}$ must be of the form $L^{01}=\frac{1}{d} M^1 \otimes \id^0$ for some $M\geq 0$ with $\Tr{M}=d.$ Inserting this back into Eq.~\eqref{eq:NormalizationConstraint_partialTRaces} yields
\begin{align}
    \sum_{k,l} Y_{k,l}^{12} = M^1 \otimes \id^2.
\end{align}
Conversely, if Eq.~\eqref{eq:NormalizationConstrained_reduced} holds, then Eq.~\eqref{eq:NormalizationConstrained_symmetry} is obviously satisfied.
\end{proof}

\subsubsection{Dilution of dynamical coherence}\label{sec:coherece_dilution}
In this section, we provide the optimal conversion distance between multiple copies of a Fourier transform $\mF_d$ and an arbitrary target channel $\mE$ in the sequential and parallel case. To this end, we need the following two Lemmas.

\begin{lem}\label{cor:R_Fd}
For any channel $\mE^{A\to B}$, it holds that 
\begin{align}\label{eq:R_max_bounds}
    R_{\max,\MIO}(\mE^{A\to B}) \leq d_B \qquad\text{and} \qquad  R_{\max,\DI}(\mE^{A\to B}) \leq d_A,
\end{align}
where $R_{\max,\cdot}$ is defined in Eq.~\eqref{eq:def_R_max}. Furthermore, for $\O \in \{\MIO,\DI\}$ and $\mF_d$ the $d$-dimensional Fourier transform, 
    \begin{align}\label{eq:R_min_R_max_d}
         R_{\max,\O}(\mF_d)=d, 
    \end{align}
    and
    \begin{align}
        R_{\max,\MIO}(\mF_d)=R_{\min,\aff(\MIO)}(\mF_d),
    \end{align}
    where $R_{\min,\aff(\MIO)}$ is defined in Eq.~\eqref{eq:min_entropy_aff}.
\end{lem}
\begin{proof}
The generalized robustness can be rewritten as~\cite{Takagi2019}
\begin{align}\label{eq:R_max_rewritten}
    R_{\max,\O}(\mN) \!= \!\inf\left\{ 1+\lambda :\mathcal{N}+\lambda\mathcal{L} \!= \!(1+\lambda)\mathcal{M}, \mathcal{L}\in\CPTP,\mathcal{M}\in\O,\lambda\geq 0\right\}\! = \!\inf\left\{ s\geq 1: J_\mN \leq sJ_\mM,\mM\in \O \right\}.
\end{align}
For completeness, the last equality follows from the following argument: Consider a feasible point such that $\mathcal{N}+\lambda\mathcal{L} = (1+\lambda)\mathcal{M}$. With $s:=1+\lambda$, the Choi states satisfy $s J_\mM =J_\mN+\lambda J_\mL\geq J_\mN $. Conversely, assume $J_\mN\leq s J_\mM$ for some $s\geq 1$ and $\mM\in\O$. Without loss of generality, assume that $s>1$, otherwise there is nothing to show (since $0\leq J_\mN\leq  J_\mM$ together with $\Tr{J_\mN- J_\mM}=0$ implies that $J_\mN=  J_\mM$). Then, define the (by assumption) positive semidefinite operator $J_\mL=(s-1)^{-1}(sJ_\mM-J_\mN) \geq 0$, which corresponds to the Choi state of a channel $\mL$. Substituting $\lambda=s-1$, we obtain $J_{\mN}+\lambda J_{\mL}=sJ_{\mM}=(1+\lambda)J_{\mM}$. 

Next, we proceed to show the bound in Eq.~\eqref{eq:R_max_bounds}. To this end, notice that from the pinching inequality~\cite{Bhatia2000, Hayashi_2002, Winter2025}, it follows that 
\begin{align}\label{eq:Pinching}
    \left(\sum_{i=0}^{d_A-1} \ketbra{i}{i}_A\otimes \id_B\right) J_\mE^{AB}  \left(\sum_{i=0}^{d_A-1} \ketbra{i}{i}_A\otimes \id_B\right) \leq d_A \sum_{i=0}^{d_A-1} (\ketbra{i}{i}_A\otimes \id_B) J_\mE^{AB}(\ketbra{i}{i}_A\otimes \id_B) = d_A \Delta_A J_\mE^{AB},
\end{align}
and analogously for a projection on system $B$. Thus, $J_\mE^{AB} \leq d_X \Delta_X J_{\mE}^{AB}$ for $X \in \{A,B\}$. Clearly, the channel defined by $J_{\mE\circ \Delta}^{AB}=\Delta_A J_{\mE}^{AB}$ is detection-incoherent, and $J_{\Delta \circ \mE}^{AB}=\Delta_B J_{\mE}^{AB}$ is creation-incoherent. Inserting this as a feasible point into Eq.~\eqref{eq:R_max_rewritten} yields the desired upper bounds for $R_{\max, \MIO}(\mE)$ and $R_{\max, \DI}(\mE)$, respectively. 

 As a direct consequence, we have that $R_{\max, \O}(\mF_d)\leq d$ for $\O\in \{\MIO,\DI\}$. To get the reverse inequality, let $s\geq 1$ and $\mM\in \O$ denote a feasible pair $(s,\mM)$ such that $J_{\mF_d} \leq sJ_{\mM}$. This implies that $\mF_d(\rho)\leq s\mM(\rho)$ for every state $\rho$. In particular, for $\O=\MIO$, choose $\rho=\ketbra{0}{0}$ and let $\ket{\phi}=F_d \ket{0}$. Thus 
\begin{align}
    1=\braket{\phi|\mF_d(\rho)|\phi}\leq s\braket{\phi|\mM(\ketbra{0}{0})|\phi}= s \braket{\phi|\Delta\mM(\ketbra{0}{0})|\phi} = \frac{s}{d},
\end{align}  
where we used that $\mM(\ketbra{0}{0})$ is diagonal and has trace one together with $\Delta(\phi)=\tfrac{\id}{d}.$ Since $s$ was feasible but arbitrary, we have $R_{\max,\MIO}(\mF_d)\geq d$. For $\O=\DI$, let $\rho_i=F_d^\dagger\ketbra{i}{i} F_d$, which satisfies $\Delta(\rho_i)=\tfrac{\id}{d}$. We obtain
\begin{align}
    1=\braket{i|\mF_d(\rho_i)|i}\leq s\braket{i|\mM(\rho_i)|i} =s\braket{i|\mM(\Delta\rho_i)|i}= \frac{s}{d}\braket{i|\mM(\id)|i},
\end{align}
which holds for arbitrary $i$. Thus, we can average over all of them, yielding $d\leq s$. This establishes $R_{\max,\O}(\mF_d)=d$ for $\O\in \{\MIO,\DI\}$.

Next, we show that $R_{\max,\MIO}(\mF_d)=R_{\min,\aff(\MIO)}(\mF_d)$. To this end, let us choose the state $\psi=\ketbra{0}{0}\otimes \ketbra{0}{0}$ in Eq.~\eqref{eq:min_entropy_aff}. Then, we have $\Tr{(\idChan \otimes \mF_d)(\psi)(\idChan \otimes \mM)(\psi)}=\braket{\phi|\mM(\ketbra{0}{0})| \phi}=\braket{\phi|\Delta \mM(\ketbra{0}{0})| \phi}=\tfrac{1}{d} \, \forall \mM\in \MIO$, and as such
\begin{align}
    R_{\min,\aff(\MIO)}(\mF_d)^{-1}\leq \frac{1}{d} \Leftrightarrow  R_{\min,\aff(\MIO)}(\mF_d) \geq d.
\end{align}
To conclude, note that the reverse inequality $R_{\min,\aff(\MIO)}(\mF_d)\leq R_{\max,\MIO}(\mF_d)=d$ follows directly from Eqs.~\eqref{eq:min_entropy_combined}.
\end{proof}

In particular, in the creation-incoherent setting, the fact that $R_{\max,\MIO}(\mF_d) = R_{\min,\aff(\MIO)}(\mF_d)$ will become relevant in the distillation of dynamical resources under maximally free superchannels because it allows us to directly apply the results from Sec.~\ref{sec:Appen_GeneralConversion}. However, in the detection-incoherent setting, we cannot do the same since $R_{\max,\DI}(\mF_d) \neq R_{\min,\aff(\DI)}(\mF_d)$ in general, as we will demonstrate in a second. Nevertheless, we will still be able to find a separate argument to characterize the optimal conversion distances under maximally free superchannels in the detection-incoherent setting. As promised, we now show that in general, $R_{\max,\DI}(\mF_d)\neq R_{\min,\aff(\DI)}(\mF_d)$, and in particular that $R_{\max,\DI}(\mF_2)=2 \neq 0=R_{\min,\aff(\DI)}(\mF_2)$. To see that this is true, consider
Eq.~\eqref{eq:min_entropy_aff} and assume that $(\lambda,\psi_{RA})$ denotes a feasible point, which must satisfy
\begin{align}\label{eq:DI_affine_cond}
    \Tr{(\idChan\otimes \mF_2)(\psi_{RA})(\idChan\otimes \mM)(\psi_{RA})}=\lambda \quad \forall \mM\in \DI.
\end{align}
We now show that this leads to a contradiction, and thus there exists no such feasible point. In particular, Eq.~\eqref{eq:DI_affine_cond} would need to hold for the replacement channel $\mR_\tau(X)=\Tr{X} \tau$, where $\tau$ is an arbitrary but fixed state since $\mR_\tau\in \DI$ for every $\tau$. With $\rho_A:=\partTr{R}{\psi_{RA}}$ and $\rho_R:=\partTr{A}{\psi_{RA}}$, this would imply that 
\begin{align}
    \lambda&=  \Tr{(\idChan\otimes \mF_2)(\psi_{RA})(\idChan\otimes \mR_\tau)(\psi_{RA})}= \Tr{(\idChan\otimes \mF_2)(\psi_{RA})( \rho_R \otimes \tau_A)}=\Tr{\psi_{RA}(\rho_R\otimes H^\dagger\tau_A H)}  \nonumber \\
    &= \Tr{\partTr{R}{(\rho_R\otimes \id_A)\psi_{RA})} H^\dagger \tau H} = \Tr{H\rho_A^2 H^\dagger \tau_A},
\end{align}
where we used that  $\rho_A^2=\partTr{R}{(\rho_R \otimes \id) \psi_{AR}}$ (following, for instance, from the Schmidt decomposition of $\psi$). Since $\tau$ was arbitrary, choosing $\tau$ to be any eigenstate of $H \rho_A^2 H^\dagger$ implies that all eigenvalues of $H \rho_A^2 H^\dagger$ are equal to $\lambda$. As such, $\rho_A^2=\lambda \id$. Thus, positive semidefiniteness and the normalization of $\rho_A$ yields $\rho_A=\id/2$ and $\lambda=1/4$. However, in Eq.~\eqref{eq:DI_affine_cond}, we could also choose the identity channel, which results in 
\begin{align}
    \lambda=\Tr{(\idChan\otimes\mF_2)(\psi_{RA}) \,\psi_{RA}}=\left|\bra{\psi_{RA}}(\id\otimes H)\ket{\psi_{RA}}\right|^2=\left|\Tr{\rho_A H}\right|^2=\frac{1}{4}\left|\Tr{H}\right|^2=0,
\end{align}
and is a contradiction to $\lambda=1/4$. Thus, the feasible set is empty, and by our convention, $R_{\min,\aff(\DI)}(\mF_2)=0$.

\begin{lemma}\label{lem:DI_xiBound}
    Let $\mF_d$ denote the $d$-dimensional quantum Fourier transform, and let $\mU=U \cdot U^\dagger$ denote a $D$-dimensional unitary channel $\mU$. Suppose $D>d$, then every $\mS \in \nDI$ satisfies
    \begin{align}
        \frac{1}{2}\norm{\mS[\mF_d]-\mU}_\diamond \geq \xi(U):= \min_{0\leq i,j\leq D-1} |\braket{i|U|j}|^2.  
    \end{align}
\end{lemma}
\begin{proof}
Let $R$ denote an arbitrary finite-dimensional auxiliary system and set
    \begin{align}\label{eq:N_def}
         \mN:=\mS[\mF_d]=\mM_1 \circ \left(\mF_d\otimes \idChan^R\right)\circ \mM_0 \quad \mM_0,\mM_1 \in \DI.
    \end{align}
 Define $\phi_i:=U^\dagger\ketbra{i}{i}U \, \forall i\in \{0,\ldots D-1\}$. Then, $\mU(\phi_i)=\ketbra{i}{i}$ and $\sum_{i=0}^{D-1} \phi_i=\id$. Choose the state $\psi=\tfrac{1}{D} \sum_{i=0}^{D-1} \ketbra{i}{i}_{R^\prime} \otimes \phi_i$, then
    \begin{align}\label{eq:Diamond_lowerBound_N}
         \frac{1}{2}\norm{\mN-\mU}_\diamond &\geq \frac{1}{2}\norm{\frac{1}{D}\sum_{i=0}^{D-1} \ketbra{i}{i}_{R^\prime} \otimes \left(\mN(\phi_i)-\ketbra{i}{i}\right)}_1 = \frac{1}{2D} \sum_{i=0}^{D-1}\norm{\left(\mN(\phi_i)-\ketbra{i}{i}\right)}_1  \nonumber \\
         &\overset{\eqref{eq:boundTraceDist}}{\geq }  1-  \frac{1}{D} \sum_{i=0}^{D-1} \Tr{\ketbra{i}{i}\mN(\phi_i)} =:1- B(\mN,\mU).
    \end{align}
We now proceed to further bound $B(\mN,\mU)$. To this end, notice that $\mM_1\in\DI$ implies that $\mM_1^\dagger(\ketbra{i}{i})=\Delta(\mM_1^\dagger(\ketbra{i}{i}))$ and $\sum_i \mM_1^\dagger(\ketbra{i}{i})=\id$. Thus, there exists a conditional probability distribution $p_{i|j,r}\geq0$ with $\sum_i p_{i|j,r}=1$ such that we can express
\begin{align}
    \mM_1^\dagger(\ketbra{i}{i}) =\sum_{j=0}^{d-1} \sum_{r=0}^{d_R-1} p_{i|j,r} \ketbra{j,r}{j,r}_{2R}.
\end{align}
Also define the operators
\begin{align}
    X_{j,r}:= \mM_0^\dagger\left(F_d^\dagger \ketbra{j}{j}_2 F_d \otimes \ketbra{r}{r}_R\right), \quad \text{and} \quad Y_r:= \sum_{j=0}^{d-1}X_{j,r},
\end{align}
which satisfy $X_{j,r}\geq 0$, $\sum_{j,r} X_{j,r}=\id$, $Y_r\geq 0$. Additionally, because $ Y_r=\sum_{j=0}^{d-1} \mM_0^\dagger\left(F_d^\dagger \ketbra{j}{j}_2 F_d \otimes \ketbra{r}{r}_R\right) =  \mM_0^\dagger\left(\id^2 \otimes \ketbra{r}{r}_R\right)= \Delta \mM_0^\dagger\left(\id^2 \otimes \ketbra{r}{r}_R\right)$ as $\mM_0\in \DI$, it follows that
\begin{align}\label{eq:Y_normalization}
    Y_r=\Delta(Y_r) \quad \forall r\in\{0,\ldots, d_R-1\} \quad \text{and} \quad \sum_{r=0}^{d_R-1} \Tr{Y_r}=D.
\end{align}
Moreover, note that for each fixed pair $(j,r)$ and each $i$ we can bound 
\begin{align}
    \Tr{X_{j,r} \phi_i} \leq \Tr{X_{j,r} \phi_{i_{j,r}}},
\end{align}
where $i_{j,r}:= \min \argmax_{0\leq i\leq D-1} \Tr{X_{j,r} \phi_i}$. Picking up from Eq.~\eqref{eq:Diamond_lowerBound_N}, we can now insert Eq.~\eqref{eq:N_def} to rewrite and bound $B(\mN,\mU)$ (using that $p_{i|j,r}$ is a conditional probability distribution) as
\begin{align}\label{eq:B_bound_intermediate}
    B(\mN,\mU)&=\frac{1}{D} \sum_{i=0}^{D-1} \Tr{\ketbra{i}{i}\mN(\phi_i)}= \frac{1}{D} \sum_{i=0}^{D-1} \Tr{\left(\mM_0^\dagger\circ \left(\idChan^R\otimes \mF_d^\dagger\right) \mM_1^\dagger(\ketbra{i}{i}) \right)\phi_i} = \frac{1}{D} \sum_{i=0}^{D-1} \sum_{j=0}^{d-1}\sum_{r=0}^{d_R-1} p_{i|j,r} \Tr{X_{j,r} \phi_i} \nonumber \\\
    &\leq \frac{1}{D}\sum_{j=0}^{d-1}\sum_{r=0}^{d_R-1}\Tr{X_{j,r} \phi_{i_{j,r}}}.
\end{align}
For a fixed $r\in \{0,\ldots, d_R-1\}$ let $S_r:= \{i_{j,r}: j\in \{0\ldots,d-1\}\}$. Hence, $|S_r|\leq d < D$, and thus there exists an integer $k_r \in \{0,\ldots, D-1\}\setminus S_r$. As such, for every fixed $r$, we have $k_r \neq i_{j,r} \, \forall j\in \{0,\ldots,d-1\}$. Recall that $\{\phi_i\}$ are mutually orthogonal rank-one projectors, and thus $\phi_{i_{j,r}} \leq \id - \phi_{k_r}$ for each $r$. Moreover, since $Y_r$ is diagonal according to Eq.~\eqref{eq:Y_normalization} we have $\Tr{Y_r \phi_{k_r}}=\Tr{Y_r \Delta(\phi_{k_r})} \geq \xi(U) \Tr{Y_r} $ because $ \Delta(\phi_{k_r})=\sum_{n=0}^{D-1} |\braket{k_r|U|n}|^2 \ketbra{n}{n} \geq \xi(U) \id$. Inserting this into Eq.~\eqref{eq:B_bound_intermediate} yields
\begin{align}
    B(\mN,\mU)&\overset{\eqref{eq:B_bound_intermediate}}{\leq}\frac{1}{D}\sum_{r=0}^{d_R-1} \sum_{j=0}^{d-1}\Tr{X_{j,r} \phi_{i_{j,r}}} \leq \frac{1}{D}\sum_{r=0}^{d_R-1} \sum_{j=0}^{d-1}  \left(\Tr{X_{j,r}}-\Tr{X_{j,r}\phi_{k_r}} \right) =1-\frac{1}{D}\sum_{r=0}^{d_R-1}\Tr{Y_r\phi_{k_r}}  \nonumber \\
    &\leq 1-\frac{1}{D}\sum_{r=0}^{d_R-1}\Tr{Y_r} \xi(U) \overset{\eqref{eq:Y_normalization}}{=} 1 - \xi(U),
\end{align}
where we used $Y_r \geq 0$. Together with Eq.~\eqref{eq:Diamond_lowerBound_N}, this finishes the proof.
\end{proof}

\DynamicalCoherenceCost*
\begin{proof}
We start with the parallel case. Let $\O \in\{\MIO,\DI\}$, and note that from $\nO_1 \subseteq \cO_1 \subseteq \mO_1$ we obtain
\begin{align}\label{eq:cost_inequality_chain}
      \inf_{\mS_1 \in \nO_1} \tfrac{1}{2}\norm{\mS_1[\mF_d^{\,\otimes N}]-\mE}_\diamond\geq \min_{\mS_1 \in \cO_1} \tfrac{1}{2}\norm{\mS_1[\mF_d^{\,\otimes N}]-\mE}_\diamond \geq  \min_{\mS_1 \in \mO_1} \tfrac{1}{2}\norm{\mS_1[\mF_d^{\,\otimes N}]-\mE}_\diamond ,
\end{align}
and analogously, the chain of inequalities holds for the worst-case fidelity. We now relax the optimization problem over maximally free superchannels further. To this end, take a feasible superchannel $\mS_1\in \mO_1$,
where $\mS_1$ is a superchannel from systems $1,2$ to $0,3$, with $d_1=d_2=d^N$ (since we are considering the parallel setting) and $d_0, d_3$ fixed by the target channel $\mE^{0\to 3}$. Then, using that $R_{\max , \O}$ is a sub-multiplicative monotone (which follows directly from its definition in Eq.~\eqref{eq:def_R_max}), Lem.~\ref{cor:R_Fd} yields 
\begin{align}
    R_{\max , \O}(\mS_1[\mF_d^{\otimes N}]) \leq R_{\max , \O}(\mF_d^{\otimes N}) \leq R_{\max , \O}(\mF_d)^N \overset{\ref{cor:R_Fd}}{=}d^N,
\end{align}
and thus $ \mS_1[\mF_d^{\otimes N}]\in  \Xi_{\O}^{(d,N)}:= \{ \mM \in \CPTP, R_{\max , \O}(\mM)\leq d^N\}$. Therefore, for any $\mS_1 \in \mO_1$ it holds that $\tfrac{1}{2}\norm{\mS_1[\mF_d^{\otimes N}]-\mE}_\diamond \geq \min_{\mM \in \Xi_{\O}^{(d,N)}} \tfrac{1}{2}\norm{\mM-\mE}_\diamond $. Since $\mS_1 \in \mO_1$ was arbitrary, take the minimum over all such feasible $\mS_1$ to obtain
\begin{align}\label{eq:XI_bound_cost_a}
    \min \left\{\tfrac{1}{2}\norm{\mM-\mE}_\diamond: R_{\max,\O}(\mM)\leq d^N, \mM \in \CPTP \right\} &\leq \min_{\mS_1\in \mO_1}\tfrac{1}{2}\norm{\mS_1[\mF_d^{\otimes N}]-\mE}_\diamond.
\end{align}
Analogously, for the worst-case fidelity it follows that
\begin{align}\label{eq:XI_bound_cost_b}
    \max \left\{F_{\min}(\mM,\mE): R_{\max,\O}(\mM)\leq d^N, \mM \in \CPTP \right\} &\geq \max_{\mS_1\in \mO_1}  F_{\min}(\mS_1[\mF_d^{\otimes N}],\mE)
\end{align}
To show that this is also achievable, let $\mM \in  \Xi_{\O}^{(d,N)}$ be arbitrary but fixed. We will now construct a compatible superchannel satisfying $\mS[\mF_d^{\otimes N}]=\mM$. To this end, note that by the definition of $R_{\max,\O}(\mM)$, there exists $\mK \in \O$ such that $J_\mM\leq d^N J_\mK =: D J_\mK$ (see the proof of Lem.~\ref{cor:R_Fd}). Define the quantum channel $\mQ$ in the Choi representation
\begin{align}\label{eq:Q_def}
    J_\mQ =\frac{D J_\mK-J_\mM}{D-1},
\end{align}
where $\mK\in \O$ for each $\O\in \{\MIO,\DI\}$ respectively. For $\O=\MIO$, let $\omega_0=\ketbra{0}{0}^{\otimes N}$ and $\psi= (F_d\ketbra{0}{0}F_d^\dagger)^{\otimes N}$. Define the channels $\mN_0,\mN_1$ acting as
\begin{align}
    \mN_0^{0\to 1R}(\rho^0) \! &:= \! \omega_0^1\otimes\rho^R,\\
    \mN_1^{2R\to 3}(X^{2R})\! &:= \!\mM^{R\to 3}\!\left(
        \partTr{2}{(\psi^2\!\otimes\!\id_R)\, X^{2R}\, (\psi^2\!\otimes\!\id_R)}\right)\!+\!\mQ^{R\to3}\!\left(\partTr{2}{((\id_2\!-\!\psi^2)\!\otimes\!\id_R)\, X^{2R}\,((\id_2\!-\!\psi^2)\!\otimes\!\id_R)}\right),
\end{align}
where $R\cong 0$ denotes the memory system. The action of $\mN_1$ on an arbitrary incoherent basis element is 
\begin{align}
    \mN_1^{2R\to 3}(\ketbra{i}{i}_2 \otimes \ketbra{j}{j}_R)= \frac{1}{D} \mM^{R\to 3}(\ketbra{j}{j}_R) +\frac{D-1}{D} \mQ^{R\to 3}(\ketbra{j}{j}_R) \overset{\eqref{eq:Q_def}}{=}\mK^{R\to 3}(\ketbra{j}{j}_R) = \Delta_3\mK^{R\to 3}(\ketbra{j}{j}_R),
\end{align}
where the last line follows since $\mK\in \MIO$. Thus, $\mN_1 \in \MIO$ and since clearly $\mN_0 \in \MIO$, we have that the superchannel $\mS_1$ acting as
\begin{align}
    \mS_1[\mL^{1\to 2}]:=\mN_1^{2R\to 3}\circ(\mL^{1\to 2}\otimes\idChan^R)\circ\mN_0^{0\to 1R},
\end{align}
is implemented by a MIO network. Moreover, for any state $\rho$ we have
\begin{align}
    \mS[\mF_d^{\otimes N}](\rho^0)=\mN_1^{2R\to 3}\!\left[ (\mF_d^{\otimes N}\otimes\idChan^R)(\omega_0^1\otimes\rho^R) \right]=\mN_1^{2R\to 3}(\psi^2\otimes\rho^R)=   \mM^{R\to 3}(\rho^R)=\mM^{0\to 3}(\rho^0).
\end{align}

Next, for $\O=\DI$ consider 
\begin{align}
    P^{12}:=\frac{1}{D} J_{\mF_d^{\otimes N}}^{12} \qquad \text{and} \qquad Q^{12}:= D \Delta_2 P^{12}.
\end{align}
Now, $P$ is a rank-one projector with $\partTr{1}{P}=\id^2/D$. Moreover, $\Delta_2 P$ is block-diagonal with each block positive semidefinite, of rank one, and of trace $1/D$. Its non-zero eigenvalues are therefore all equal to $1/D$, and as such, $Q=D\Delta_2 P$ is a projector too. Using the pinching inequality~\cite{Bhatia2000, Hayashi_2002, Winter2025} gives
\begin{align}
    P \leq D\sum_j \left(\id^1\otimes \ketbra{j}{j}_2\right) P\left(\id^1\otimes \ketbra{j}{j}_2\right) = D\Delta_2 P=Q.
\end{align}
Therefore, we have $0\leq P\leq D\Delta_2(P)=Q\leq\id$. Since $P$ and $Q$ are projectors, $P\leq Q$ implies that the range of $P$ is contained in the range of $Q$, and thus
\begin{align}\label{eq:P_Q}
    QP=P, \qquad \text{and} \qquad P(Q-P)=P(\id-Q)=0.
\end{align}
Additionally, from its definition it follows that $\Delta_{12}P^{12}=\id^{12}/D^2$ and idempotence and commutativity
of the local dephasing maps yields
\begin{align}
    \Delta_2Q^{12}=Q^{12},
    \qquad
    \Delta_1Q^{12}=D\Delta_{12}P^{12}=\tfrac{\id^{12}}{D}.
\end{align}
Then, define the superchannel $\tilde{\mS}_1$ in the Choi representation as
\begin{align}
    J_{\tilde{\mS}_1}^{0123}:= \frac{1}{D} \left((P^{12})^T \otimes J_{\mM}^{03} +(Q^{12}-P^{12})^T \otimes \Delta_3 J_{\mM}^{03} +(\id^{12}-Q^{12})^T \otimes \Delta_3 J_{\mQ}^{03}\right).
\end{align}
Since $0\leq P\leq Q\leq \id$, it follows that $J_{\tilde{\mS}_1}^{0123} \geq 0$. Trace-preservation of $\mM,\mQ$ yields $\partTr{3}{J_{\tilde{\mS}_1}^{0123}}=\tfrac{\id^{01}}{D} \otimes \id^2$, so $\tilde{\mS}_1$ is a valid superchannel. Moreover, using $\Delta_2 Q^{12}=Q^{12}$ yields
\begin{align}
    \Delta_3 J_{\tilde{\mS}_1}^{0123}=\frac{1}{D} \left((Q^{12})^T \otimes \Delta_3 J_{\mM}^{03} +(\id^{12}-Q^{12})^T \otimes \Delta_3 J_{\mQ}^{03}\right) = \Delta_{32} J_{\tilde{\mS}_1}^{0123},
\end{align}
and additionally we have that
\begin{align}
    \Delta_{321} J_{\tilde{\mS}_1}^{0123}= \frac{\id^{12}}{D^2}\otimes \Delta_3\left( J_{\mM}^{03} +(D-1) J_{\mQ}^{03}\right) =\frac{\id^{12}}{D}\otimes \Delta_{3}J_{\mK}^{03}= \frac{\id^{12}}{D}\otimes \Delta_{30}J_{\mK}^{03},
\end{align}
where we used that $\mK\in \DI$. This implies that $\tilde{\mS}_1\in \cDI$ (cf. Thm.~\ref{thm:CompatibleSupermaps}). Finally, the action of $\tilde{\mS}_1$ on $J_{\mF_d^{\otimes N}}= D P$ is given by
\begin{align}
    J_{\tilde{\mS}_1[\mF_d^{\otimes N}]}=\Tr{PP} J_\mM +\Tr{P(Q-P)} \Delta_3 J_{\mM} +\Tr{P(\id-Q)} \Delta_3 J_\mQ \overset{\eqref{eq:P_Q}}{=} J_\mM,
\end{align}
where we used Eq.~\eqref{eq:P_Q} and $\Tr{P P}=1$ since $P$ is a normalized rank-one projector. We have therefore constructed a feasible superchannel $\mS_1$ and $\tilde{\mS}_1$ such that for $\O\in\{\MIO,\DI\}$ we have 
\begin{subequations}
\begin{align}
     \min_{\mS_1\in \mO_1}\tfrac{1}{2}\norm{\mS_1[\mF_d^{\otimes N}]-\mE}_\diamond, &\leq \tfrac{1}{2} \norm{\mM-\mE}_\diamond,  \qquad \text{and} \qquad
    \max_{\mS_1\in \mO_1}  F_{\min}(\mS_1[\mF_d^{\otimes N}],\mE) \geq F_{\min}(\mM,\mE).
\end{align}
\end{subequations}
Since $\mM \in \Xi_{\O}^{(d,N)}$ was arbitrary, we can take the minimum and maximum over all such channels respectively, which yields
\begin{subequations}\label{eq:Cost_MIO_Bound}
\begin{align}
    \min \left\{\tfrac{1}{2}\norm{\mM-\mE}_\diamond: R_{\max,\O}(\mM)\leq d^N, \mM \in \CPTP \right\} &\geq \min_{\mS_1\in \mO_1}\tfrac{1}{2}\norm{\mS_1[\mF_d^{\otimes N}]-\mE}_\diamond, \\
    \max \left\{F_{\min}(\mM,\mE): R_{\max,\O}(\mM)\leq d^N, \mM \in \CPTP \right\} &\leq \max_{\mS_1\in \mO_1}  F_{\min}(\mS_1[\mF_d^{\otimes N}],\mE),
\end{align}
\end{subequations}
which, together with Eq.~\eqref{eq:XI_bound_cost_a} and Eq.~\eqref{eq:XI_bound_cost_b}, implies the equality in the latter expressions. Since the constructed superchannels belong to $\nMIO_1$ and $\cDI_1$, respectively, Eq.~\eqref{eq:cost_inequality_chain} implies that conversion under $\nMIO_1$ and $\cDI_1$ yields the same optimal conversion error.

In the parallel case, it remains to show that
    \begin{align}
       \inf_{\mS_1 \in \nDI_1} \tfrac{1}{2}\norm{\mS_1[\mF_2]-\mU}_\diamond > \min_{\mS_1 \in \cDI} \tfrac{1}{2}\norm{\mS_1[\mF_2]-\mU}_\diamond.  
    \end{align}
To this end, let $d=2$, $D=3$, and let $\ket{\mu}=\tfrac{1}{\sqrt{3}}(\ket{0}+\ket{1}+\ket{2})$. Consider the unitary
\begin{align}
    U:= (\id-\ketbra{\mu}{\mu}) + e^{i\arctan(3/4)} \ketbra{\mu}{\mu}.
\end{align}
Since $\tilde{P}:=\ketbra{\mu}{\mu}$ and $\id-\tilde{P}$ are orthogonal projectors, $U$ is unitary, and a straightforward calculation reveals that $U= \alpha \id +\beta \sum_{i\neq j} \ketbra{i}{j}$, where
$\alpha=\frac{14+3i}{15}$ and $\beta=\frac{-1+3i}{15}$. Invoking Lem.~\ref{lem:DI_xiBound}, yields $\inf_{\mS_1\in \nDI_1} \tfrac{1}{2}\norm{\mS_1[\mF_2]-\mU}_\diamond  \geq \xi(U)= \min\{\tfrac{41}{45},\tfrac{2}{45}\}= \tfrac{2}{45}$. On the contrary, we first show that $R_{\max,\DI}(\mU)\leq 2$, from which we conclude that $ \min \left\{\tfrac{1}{2}\norm{\mM-\mU}_\diamond: R_{\max,\O}(\mM)\leq d^N, \mM \in \CPTP \right\}=0$. To this end, set
\begin{subequations}
\begin{align}
    s&:= \sum_{j=0}^{2} |U_{ij}|\!=\! \sum_{j=0}^{2} |U_{ji}|\!=\! |\alpha| + 2|\beta|= \frac{\sqrt{205}+2\sqrt{10}}{15} \\
    t&:= s^2= \frac{49+4\sqrt{82}}{45} \approx 1.89, \\
    H_i&:= \frac{1}{s}\sum_{j=0}^{2}|U_{ij}|\,\ketbra{j}{j}
\end{align}
\end{subequations}
Then $H_i \geq 0$ and $\sum_i H_i=\id$ by construction. Using $\phi_i= U^\dagger \ketbra{i}{i} U$, the Cauchy-Schwarz inequality implies $\bra{x}\phi_i\ket{x}= \left|\sum_{j=0}^{2} U_{ij}x_j  \right|^2 \leq s\sum_{j=0}^{2} |U_{ij}|\,|x_j|^2 = t \bra{x}H_i\ket{x}$, and thus $K_i:= t H_i-\phi_i \geq 0$. Notice that $t>1$ and define the channels
\begin{align}
    \mM(\rho):=\frac{1}{t-1}\sum_{i=0}^{2} \Tr{(tH_i-\phi_i)\rho}\ketbra{i}{i}, \quad \text{and} \quad  \mK:=\frac{1}{t}\left(\mU+(t-1)\mM\right).
\end{align}
Now $\mM$ is a valid channel since $tH_i-\phi_i\geq 0$, and $\tfrac{1}{t-1}\sum_{i=0}^{2}(tH_i-\phi_i)=\id$. Moreover, $\mK$ is a channel since it is a convex combination of channels. Additionally, from $\braket{i|\mU(\rho)|i}=\Tr{\ketbra{i}{i}U\rho U^\dagger}=\Tr{\phi_i \rho }$, we obtain
\begin{align*}
    \Delta\mK(\rho)=\frac{1}{t}\left(\Delta \mU(\rho)+\sum_{i=0}^{2}\Tr{(tH_i-\phi_i)\rho}\ketbra{i}{i} \right)=\frac{1}{t}\sum_{i=0}^{2}\Tr{(\phi_i+tH_i-\phi_i)\rho}\ketbra{i}{i} =\sum_{i=0}^{2}\Tr{H_i\rho}\ketbra{i}{i} = \Delta\mK(\Delta\rho),
\end{align*}
where the last equality follows since all $H_i$ are diagonal. Now let $\lambda:=t-1$, and thus $\mU+\lambda\mM=(1+\lambda)\mK$. As such, we have constructed a feasible solution to $R_{\max,\DI}(\mU)= \inf_{\lambda\geq 0} \big\{ 1\!+\!\lambda :  \mU+\lambda \mM = (1\!+\!\lambda) \mK, \, \mM\in \CPTP, \mK \in \DI\big\}$, and thus
\begin{align}
    R_{\max,\DI}(\mU) \leq 1+\lambda= t =\frac{49+4\sqrt{82}}{45}\approx 1.89<2.
\end{align}
This implies that $\mU$ is a feasible point in
\begin{align}
    0= \tfrac{1}{2}\norm{\mU-\mU}_\diamond \geq \min_{\mL \in \CPTP} \left\{\tfrac{1}{2}\norm{\mL-\mU}_\diamond: R_{\max,\DI}(\mL)\leq 2 \right\}  \geq 0,
\end{align}
which concludes the proof with $\inf_{\mS_1 \in \nDI} \tfrac{1}{2}\norm{\mS_1[\mF_2]-\mU}_\diamond \geq \tfrac{2}{45}>0=\min_{\mS_1 \in \cDI} \tfrac{1}{2}\norm{\mS_1[\mF_2]-\mU}_\diamond$. 

Moving on to the sequential case, we now show that the one-shot cost is identical to the one-shot cost in the parallel case. Let $\mS_N\in \cO_N$ be arbitrary but fixed. For fixed channels in all $N$ slots of the supermap, the supermap is multilinear, i.e., for any $1\leq j\leq N$ we have
\begin{align}
    \mS_N\bigl[\mL_1,\ldots,a\mK_j+b\mT_j,\ldots,\mL_N\bigr]  =a\,\mS_N\bigl[\mL_1,\ldots,\mK_j,\ldots,\mL_N \bigr] +b\,\mS_N\bigl[ \mL_1,\ldots,\mT_j,\ldots,\mL_N\bigr],
\end{align}
where $\mK_j$ and $\mT_j$ denotes that the channel is inserted into the $j$-th slot. This multilinearity follows directly from the Choi representation in Eq.~\eqref{eq:defLinkProduct}, which is linear in each Choi operator of $\mL_j$. Let $\mQ \in \CPTP$, and let $\mK\in \O$ such that $\mF_d+(d-1)\mQ=d\mK$. Then,  multilinearity gives
\begin{align}\label{eq:Action_S_N_Q}
    d^N\mS_N[\mK,\cldots,\mK] &=\mS_N\left[\mF_d+(d-1)\mQ,\cldots,\mF_d+(d-1)\mQ\right] \nonumber \\
    &= \mS_N\left[\mF_d,\cldots,\mF_d\right]\! +\!(d\!-\!1) \!\sum_{1\leq j\leq N}\!\mS_N\left[\mF_d,\cldots,\mQ_j,\cldots, \mF_d\right] \!+\!(d\!-\!1)^2 \sum_{1\leq j<k \leq N}\! \mS_N\left[\mF_d,\cldots,\mQ_j,\cldots, \mQ_k,\cldots, \mF_d\right] \nonumber \\
    & + \cldots + (d-1)^N \mS_N[\mQ,\cldots, \mQ],
\end{align}
where $\mQ_j$ again denotes that $\mQ$ is inserted into slot j. Now, define 
   \begin{align}
       \mM:=\mS_N[\mF_d,\ldots, \mF_d],
   \end{align} 
    \begin{align*}
        \tilde{\mQ}\!:=\! \tfrac{1}{d^N\!-\!1}\Bigr(  (d\!-\!1) \!\sum_{1\leq j\leq N}\mS_N\left[\mF_d,\cldots,\mQ_j,\cldots, \mF_d\right] \!+\!(d\!-\!1)^2\! \sum_{1\leq j<k \leq N}\! \mS_N\left[\mF_d,\cldots,\mQ_j,\cldots, \mQ_k,\cldots, \mF_d\right]\!+\! \cldots \!+\! (d\!-\!1)^N \mS_N[\mQ,\ldots, \mQ]\!\Bigr),
    \end{align*}
    where each term proportional to $(d-1)^k$ is the sum over all $\binom{N}{k}$ possible placements of $k$ copies of $\mQ$ and $N-k$ copies of $\mF_d$ into $\mS$. Since
    \begin{align}
        \sum_{k=1}^{N} \binom Nk (d-1)^k=(1+(d-1))^N-1=d^N-1,
    \end{align}
    we have that $\tilde{\mQ}$ is a convex combination of channels, and thus $\tilde{\mQ}$ is a channel. Rewriting Eq.~\eqref{eq:Action_S_N_Q} as $d^N\mS_N[\mK,\ldots,\mK]=\mM+(d^N-1)\tilde{\mQ}$, notice that $\tilde{\mK}:=\mS_N[\mK,\ldots,\mK] \in \O$ as $\mK\in \O$ and $\mS_N \in \cO$. Hence, $\mM+(d^N-1)\tilde{\mQ}=d^N \tilde{\mK}$ implies that $R_{\max, \O}(\mM)\leq d^N$, where we recall that $\mM=\mS_N[\mF_d,\ldots, \mF_d]$. Therefore, for any choice of $\mS_N\in \cO_N$ we find
    \begin{align}
        \min_{\mS_N \in \cO} \tfrac{1}{2}\norm{\mS_N[\mF_d,\ldots, \mF_d]-\mE}_\diamond \geq  \min \{ \tfrac{1}{2}\norm{\mM-\mE}_\diamond: R_{\max, \O}(\mM)\leq d^N, \mM\in \CPTP\}.
    \end{align}

    Lastly, the reverse inequality follows by noting that on the left hand-side of the latter expression, we can always choose a protocol that implements a parallel protocol (which is always possible via a sequence of SWAP and identity channels). Therefore, we obtain
        \begin{align}
        \min_{\mS_N \in \cO} \tfrac{1}{2}\norm{\mS_N[\mF_d,\ldots, \mF_d]-\mE}_\diamond \leq \min_{\mS_1 \in \cO} \tfrac{1}{2}\norm{\mS_1[\mF_d^{\otimes N}]-\mE} _\diamond =\min \{ \tfrac{1}{2}\norm{\mM-\mE}_\diamond: R_{\max, \O}(\mM)\leq d^N, \mM\in \CPTP\} \nonumber,
    \end{align}
    where the last equality has been shown in the parallel case. The same argument applies if we use the worst-case fidelity instead of the diamond distance.\qedhere
\end{proof}

\subsubsection{Optimal conversion distances to the Fourier transform under compatible supermaps}

In this section, we prove that the optimal conversion distance between an arbitrary quantum channel $\mE^{A\to B}$ and the quantum Fourier transform, with respect to the diamond distance, worst-case fidelity, and the (normalized) Choi state fidelity, all lead to the same conversion error in the one-shot regime. Recall that the monotones governing the optimal conversion error are given by
\begin{subequations}
\begin{align}
    M_{\cMIO,d}(\mE^{A\to B}) &:=\max\left\{\Tr{X^B \mE^{A\to B}(\ketbra{j}{j}_A)} : 0\leq X\leq d\id, \quad \Delta X=\id, 0\le j \le d_A-1\right\}, \label{eq:C_D_appen} \\
    M_{\cDI,d}(\mE^{A\to B}) &:= \max \left\{ \sum_{i=0}^{d_B-1}\braket{i|\mE^{A\to B}(\rho_i^A)| i}_B \!: \rho_i,\sigma \in \D, \rho_i \leq d \sigma, \Delta \rho_i =\Delta \sigma \, \forall i: 0\leq i \leq d_B-1\right\}. \label{eq:D_d_appen}
\end{align}
\end{subequations}
To this end, we rely on the twirling channel as defined in Sec.~\ref{sec:Twirling}.

\CoherenceOptimalConversion*
\begin{proof}
    We start by proving the optimality part of the statement, i.e., we first show the last inequality of 
    \begin{align}\label{eq:InequalityChain_conversion_distance}
        \min_{\mS \in \cO} \frac{1}{2}\norm{\mS[\mE]-\mF_d}_\diamond \overset{\eqref{eq:FuchsVanDeGraaf_channels_sharp}}{\geq}  1-\max_{\mS \in \cO}F_{\min}(\mS[\mE],\mF_d) & \overset{\eqref{eq:InequalityChain}}{\geq} 1-\max_{\mS \in \cO}F_{\ChoiF}(\mS[\mE],\mF_d) \geq 1-\frac{1}{d} M_{\cO,d}(\mE).
    \end{align}
To this end, we will employ the twirling channel introduced in Sec.~\ref{sec:Twirling}, and specifically Eq.~\eqref{eq:TwirlingChannel}. Let us define $\tilde{\mT}^{0123}=\idChan^{12} \otimes \mT^{03}$. Since $\tilde{\mT}^{0123}((\cdot)^{12}\otimes J_{\mF_d}^{03})=(\cdot)^{12}\otimes J_{\mF_d}^{03}$, the objective function in the Choi fidelity optimization is left invariant as
\begin{align}\label{eq:lowBound_objFunc}
    \max_{\mS \in \cO_1}F_{\ChoiF}(\mS[\mE],\mF_d)\!= \frac{1}{d^2} \max_{\mS \in \cO_1}  \Tr{J_{\mS}^{0123} \left(\left(J_{\mE}^{12}\right)^T \otimes J_{\mF_d}^{03}\right)} =\frac{1}{d^2}\max_{\mS \in \cO_1}  \Tr{\tilde{\mathcal{T}}^{0123}\left(J_{\mS}^{0123}\right) \left(\left(J_{\mE}^{12}\right)^T \otimes J_{\mF_d}^{03}\right)},
\end{align}
where we used that $\tilde{\mT}$ is self-dual. The twirling therefore leaves the objective function invariant. It remains to show that the twirling preserves the feasible set. The comb constraints in Eq.~\eqref{eq:CombSet} are automatically satisfied since conjugation by local unitaries preserves positive semidefiniteness and the partial trace constraints Eq.~\eqref{eq:TPcondition}.  Moreover, the coherence constraints are also preserved: Combining Lem.~\ref{lem:MonomialUnitary} with Lem.~\ref{lem:UnitaryRewriting}, and in particular, that according to Eq.~\eqref{eq:U_asMonomialUnitaries} the unitary $U_{n,m}$ (defining the twirling channel $\mT$) can be written as a tensor product of monomial unitaries, implies that
\begin{align}
    \Delta_X \tilde{\mT}^{0123}(J_{\mS}^{0123})= \tilde{\mT}(\Delta_X J_{\mS}^{0123}),
\end{align}
for any subsystem X. Hence, if $\mS \in \cO_1$, then the channel defined by the Choi state $\tilde{\mT}^{0123}(J_{\mS}^{0123})$ is contained in $\cO_1$. Since $\tilde{\mT}$ is self-adjoint, we can thus restrict the optimization over all superchannels to superchannels that are invariant under the action of the extended twirling channel $\tilde{\mT}$. This allows us to invoke Lem.~\ref{lem:TwirlingInvariantChannel}, which forces the superchannel to be of the form
\begin{align}
       J_{\mS}^{0123}= \sum_{k,l=0}^{d-1} Y_{k,l}^{12} \otimes \Pi_{k,l}^{03},
\end{align}
where $\Pi_{k,l}$ are orthonormal projectors as in Eq.~\eqref{eq:Pi_kl}, and $Y_{k,l}$ must satisfy $Y_{k,l}^{12} \geq 0 \, \forall k,l \in \{0,\ldots, d-1\}$, and $\sum_{k,l} Y_{k,l}^{12} = M^1 \otimes \id^2$ for some $M\geq 0,\Tr{M}=d$. We still have to incorporate the coherence constraints. To this end, we use Lem.~\ref{lem:UnitaryRewriting}, and in particular Eq.~\eqref{eq:P_l} and Eq.~\eqref{eq:Q_k}, yielding 
\begin{align}\label{eq:CohCOnst_0}
   &\Delta_0 J_{\mS}= \Delta_{01} J_{\mS} \quad \Leftrightarrow \quad  \frac{1}{d} \sum_{l} \left( \sum_k Y_{k,l}^{12} \right) \otimes P_l^{03} = \frac{1}{d} \sum_{l} \left( \Delta_1 \sum_k Y_{k,l}^{12} \right) \otimes P_l^{03}, \\
   &\Delta_3 J_{\mS}= \Delta_{32} J_{\mS} \quad \Leftrightarrow \quad \frac{1}{d} \sum_{k} \left( \sum_l Y_{k,l}^{12} \right) \otimes Q_k^{03} = \frac{1}{d} \sum_{k} \left( \Delta_2 \sum_l Y_{k,l}^{12} \right) \otimes Q_k^{03}.
\end{align}
Since the projectors $\{P_l^{03} \}_{0\leq l\leq d-1}$ and $\{Q_k^{03} \}_{0\leq k\leq d-1}$ are mutually orthogonal within the respective set (cf. Lem.~\ref{lem:UnitaryRewriting}), we find that 
\begin{align}
   &\Delta_0 J_{\mS}= \Delta_{01} J_{\mS} \quad \Leftrightarrow \quad  \sum_k Y_{k,l}^{12}= \Delta_1 \sum_k   Y_{k,l}^{12} \quad \forall l: 0\leq l \leq d-1 \label{eq:Twirled_Delta01},\\
   &\Delta_3 J_{\mS}= \Delta_{32} J_{\mS} \quad \Leftrightarrow \quad\sum_l Y_{k,l}^{12}= \Delta_2 \sum_l   Y_{k,l}^{12} \quad \forall k: 0\leq k \leq d-1. \label{eq:Twirled_Delta32}
\end{align}
Moreover, the second coherence constraint to ensure that $\mS \in \cO_1$ translates to
\begin{align}
    &\Delta_{012} J_{\mS}= \Delta_{0123}J_{\mS} \quad \Leftrightarrow \quad      \frac{1}{d} \sum_{l} \left(\sum_k \Delta_{12} Y_{k,l}^{12} \right) \otimes P_l^{03} = \frac{1}{d^2} \sum_{k,l}  \Delta_{12}  Y_{k,l}^{12}  \otimes \id^{03}, \\
    &\Delta_{321} J_{\mS}= \Delta_{3210}J_{\mS} \quad \Leftrightarrow \quad  \frac{1}{d} \sum_{k} \left(\sum_l \Delta_{12} Y_{k,l}^{12} \right) \otimes Q_k^{03} = \frac{1}{d^2} \sum_{k,l}  \Delta_{12}  Y_{k,l}^{12}  \otimes \id^{03},
\end{align}
which can be simplified to (again using the orthogonality of the projectors) 
\begin{align}
    &\Delta_{012} J_{\mS}= \Delta_{0123}J_{\mS} \quad \Leftrightarrow \quad   \sum_{k} \Delta_{12} Y_{k,l}^{12} = \frac{1}{d} \sum_l \left(\sum_{k} \Delta_{12} Y_{k,l}^{12} \right) \quad \forall l: 0\leq l \leq d-1,    \\
    &\Delta_{321} J_{\mS}= \Delta_{3210}J_{\mS} \quad \Leftrightarrow \quad  \sum_{l} \Delta_{12} Y_{k,l}^{12} = \frac{1}{d} \sum_k \left(\sum_{l} \Delta_{12} Y_{k,l}^{12} \right) \quad \forall k: 0\leq k \leq d-1.
\end{align}
For each $S_l:=\sum_{k} Y_{k,l}^{12}$, and $S_k:=\sum_{l} Y_{k,l}^{12}$, this implies that $\Delta_{12} S_l$ (and respectively  $\Delta_{12} S_k$) must be equal to its arithmetic average over all $l$ (respectively $k$), and thus,
\begin{align}
    &\Delta_{012} J_{\mS}= \Delta_{0123}J_{\mS} \quad \Leftrightarrow \quad \sum_{k} \Delta_{12} Y_{k,l}^{12} = \sum_{k} \Delta_{12} Y_{k,l^\prime}^{12} \quad \forall l,l^\prime: 0\leq l,l^\prime \leq d-1, \\
    &\Delta_{321} J_{\mS}= \Delta_{3210}J_{\mS} \quad \Leftrightarrow \quad \sum_{l} \Delta_{12} Y_{k,l}^{12} = \sum_{l} \Delta_{12} Y_{k^\prime,l}^{12} \quad \forall k,k^\prime: 0\leq k,k^\prime \leq d-1, \label{eq:CohCOnst_321}
\end{align}
or in other words, all $\Delta_{12} S_l=\Delta_{12}\sum_{k} Y_{k,l}^{12}$  and  $\Delta_{12} S_k= \Delta_{12}\sum_{l} Y_{k,l}^{12}$ must be equal. Lastly, note that the projector $\Pi_{0,0}^{03}$ is proportional to $J_{\mF}^{03}$ since
\begin{align}
    \Pi_{0,0}^{03}= \frac{1}{d^2} \left( \sum_{a,b} \omega_d^{ab} \ket{ab}\right)\left(\sum_{a^\prime,b^\prime} \omega_d^{-a^\prime b^\prime} \bra{a^\prime b^\prime}\right)= \frac{1}{d} J_{\mF_d}^{03},
\end{align}
and as such, the objective function of Eq.~\eqref{eq:lowBound_objFunc} takes the simple form  
\begin{align}\label{eq:ObjectiveFunction_Twirled}
    \Tr{\tilde{\mT}(J_{\mS}^{0123}) \left(\left(J_{\mE}^{12}\right)^T \otimes J_{\mF_d}^{03}\right)} &= \sum_{k,l} \Tr{Y_{k,l}^{12} \left(J_{\mE}^{12}\right)^T} \Tr{\Pi_{k,l}^{03} J_{\mF_d}^{03}}=  \sum_{k,l} \Tr{Y_{k,l}^{12} \left(J_{\mE}^{12}\right)^T} \, d\Tr{\Pi_{k,l}^{03} \Pi_{0,0}^{03} } \nonumber \\
    &= d\, \Tr{Y_{0,0}^{12} \left(J_{\mE}^{12}\right)^T},
\end{align}
where we used that the projectors $\Pi_{k,l}^{03}$ are orthonormal according to Eq.~\eqref{eq:ProjectorsOrthonormal}. Thus, the lower bound in Eq.~\eqref{eq:InequalityChain} can be written as follows
\begin{align}\label{eq:LowerBound_Twirled_b}
    &d^2\! \max_{\mS \in \cO_1}F_{\ChoiF}(\mS[\mE],\mF_d)= \max_{\mS \in \cO_1}  \Tr{J_{\mS}^{0123} \left(\left(J_{\mE}^{12}\right)^T \otimes J_{\mF_d}^{03}\right)} \nonumber \\
    &= \max \Bigr\{ d \Tr{Y_{0,0}^{12} \left(J_{\mE}^{12}\right)^T}: \! Y_{k,l}^{12}, M^1 \!\geq \!0, \, \forall\, k,l, \Tr{M^1}\!=\!d, \sum_{k,l} \!Y_{k,l}^{12}\!=\! M^1 \otimes \id^2, \sum_{k,l} Y_{k,l}^{12} \otimes \Pi_{k,l}^{03} \!=\!J_{\mS}^{0123} \!\in\! \cO_1 \Bigr\}
\end{align}

We now proceed case by case with the different sets of free superchannels to show that the quantities Eq.~\eqref{eq:LowerBound_Twirled_b} reduce to the monotones as defined in Eq.~\eqref{eq:C_D} and Eq.~\eqref{eq:D_d}.

\vspace{0.25cm}
\begin{minipage}{0.48\textwidth}
\begin{center}
   \textbf{cMIO-superchannels}
\end{center}
\begin{subequations}\label{eq:LowerBound_Twirled}
\begin{alignat}{2} 
    &\text{maximize:}\quad && d\,\Tr{Y_{0,0}^{12} \left(J_{\mE}^{12}\right)^T}\\
    &\text{subject to:}\quad && Y_{k,l}^{12} \geq 0, \quad \forall\, k,l  \\
    & && M^1\geq 0\\
    & && \Tr{M^1}=d \\
    & && \sum_{k,l} Y_{k,l}^{12}= M^1 \otimes \id^2\\
    & && \sum_k Y_{k,l}^{12}= \Delta_1\sum_k   Y_{k,l}^{12} \quad \forall\, l  \\
    & && \Delta_{12}\sum_{k}  Y_{k,l}^{12} = \Delta_{12} \sum_{k}  Y_{k,l^\prime}^{12} \, \forall\, l,l^\prime 
\end{alignat}
\end{subequations}
\end{minipage}
\hfill
\begin{minipage}{0.48\textwidth}
\begin{center}
    \textbf{cDI-superchannels}
\end{center}
\begin{subequations}\label{eq:LowerBound_Twirled_detection}
\begin{alignat}{2} 
    &\text{maximize:}\quad && d\,\Tr{Y_{0,0}^{12} \left(J_{\mE}^{12}\right)^T}\\
    &\text{subject to:}\quad && Y_{k,l}^{12} \geq 0 \forall\, k,l  \\
    & && M^1\geq 0\\
    & && \Tr{M^1}=d \\
    & && \sum_{k,l} Y_{k,l}^{12}= M^1 \otimes \id^2\\
    & &&    \sum_l Y_{k,l}^{12}= \Delta_2 \sum_l  Y_{k,l}^{12} \quad \forall\, k\\
    & && \Delta_{12}\sum_{l}  Y_{k,l}^{12} = \Delta_{12} \sum_{l} Y_{k^\prime,l}^{12} \, \forall\, k,k^\prime
\end{alignat}
\end{subequations}
\end{minipage}
\vspace{0.25cm}

Note that the two SDPs only differ in their second-to-last constraint. Let $\{Y_{k,l}\}_{0\leq k,l\leq d-1}$ denote a feasible point of Eq.~\eqref{eq:LowerBound_Twirled} (or Eq.~\eqref{eq:LowerBound_Twirled_detection}, respectively).  For the cMIO case, we can introduce a new set of variables  $\{Y,T,\{S_l\}_{1\leq l \leq d-1}\}$ by setting
\begin{align}
    Y=Y_{0,0},\quad
    T= \sum_{k\geq 1} Y_{k,0}, \quad \text{and}\quad S_l=\sum_k Y_{k,l} \, \forall l\geq 1.
\end{align}
And analogously for cDI, set variables  $\{Y,T,\{S_k\}_{1\leq k \leq d-1}\}$ defined via $Y=Y_{0,0}$, $T= \sum_{l\geq 1} Y_{0,l}$, and $ S_k=\sum_l Y_{k,l} \, \forall k\geq 1.$ Then, Eq.~\eqref{eq:LowerBound_Twirled} (and Eq.~\eqref{eq:LowerBound_Twirled_detection}, respectively) can equivalently be written as (by simply grouping variables $Y_{k,l}$ into sums over $k$ and $l$, respectively)

\vspace{0.25cm}
\begin{minipage}{0.48\textwidth}
\begin{center}
   \textbf{cMIO-superchannels}
\end{center}
\begin{subequations}\label{eq:twirledSDP_S_l}
\begin{alignat}{2} 
    &\text{maximize:}\quad && d\,\Tr{Y^{12} \left(J_{\mE}^{12}\right)^T}\\
    &\text{subject to:}\quad && Y \geq 0, T \geq 0, S_l \geq 0 \quad \forall \, l \\
    & && M^1\geq 0\\
    & && \Tr{M^1}=d \\
    & && Y^{12}+T^{12}+\sum_{l\geq 1} S_l^{12}= M^1 \otimes \id^2\\
    & && Y^{12}+T^{12}=\Delta_1\left(Y^{12}+T^{12}\right)\\
    & && S_l^{12}=\Delta_1 S_l^{12}\,\ \forall l\geq 1\\
    & &&\Delta_{12}\left(Y^{12}\!+\!T^{12}\right)=\Delta_{12} S_l^{12} \,\forall \,l\geq 1.
\end{alignat}
\end{subequations}
\end{minipage}
\hfill
\begin{minipage}{0.48\textwidth}
\begin{center}
    \textbf{cDI-superchannels}
\end{center}
\begin{subequations}\label{eq:twirledSDP_S_k}
\begin{alignat}{2} 
    &\text{maximize:}\quad && d\,\Tr{Y^{12} \left(J_{\mE}^{12}\right)^T}\\
    &\text{subject to:}\quad && Y \geq 0, T \geq 0, S_k \geq 0 \quad \forall \, k \\
    & && M^1\geq 0\\
    & && \Tr{M^1}=d \\
    & && Y^{12}+T^{12}+\sum_{k\geq 1} S_k^{12}= M^1 \otimes \id^2\\
    & && Y^{12}+T^{12}=\Delta_2\left(Y^{12}+T^{12}\right)\\
    & && S_k^{12}=\Delta_2 S_k^{12}\, \forall k \geq 1\\
    & &&\Delta_{12}\left(Y^{12}\!+\!T^{12}\right)=\Delta_{12} S_k^{12} \,\forall k\geq 1.
\end{alignat}
\end{subequations}
\end{minipage}
\vspace{0.25cm}

Firstly, note in both SDPs we can always choose $T=0$ by introducing a new variable $\tilde{Y}=Y+T$, for which, the constraints remain unchanged, and the objective function is bounded by 
\begin{align}
    \Tr{Y J_{\mE}^T}=\Tr{\tilde{Y} J_{\mE}^T}-\Tr{TJ_{\mE}^T}\leq \Tr{\tilde{Y} J_{\mE}^T},
\end{align}
since $T\geq 0$ (and $J_{\mE}\geq 0$). Moreover, we can rewrite the SDPs again, first by relaxing it as follows: Let $\{Y,T,\{S_l \}_{1\leq l\leq d-1}\}$ and $\{Y,T,\{S_k \}_{1\leq k\leq d-1}\}$ be a feasible point in Eq.~\eqref{eq:twirledSDP_S_l} and Eq.~\eqref{eq:twirledSDP_S_k}, respectively. Then, defining $S:=\sum_{l\geq 1} S_l$  (and $S:=\sum_{k\geq 1} S_k$, respectively) we find that $\{Y,S\}$ is a feasible point of

\vspace{0.25cm}
\begin{minipage}{0.48\textwidth}
\begin{center}
   \textbf{cMIO-superchannels}
\end{center}
\begin{subequations}\label{eq:twirledSDP_S_final*}
\begin{alignat}{2} 
    &\text{maximize:}\quad && d\,\Tr{Y^{12} \left(J_{\mE}^{12}\right)^T} \label{eq:MIO_SDP_obj}\\
    &\text{subject to:}\quad && Y^{12}\geq 0, S^{12}\geq 0, M^1 \geq 0 \\
    & && \Tr{M^1}=d \\
    & && Y^{12}+S^{12}= M^1 \otimes \id^2 \label{eq:Y+S=M} \\
    & && Y^{12}=\Delta_1Y^{12}\label{eq:Delta1_Y}\\
    & && S^{12}=\Delta_1 S^{12}\label{eq:Delta1_S} \\
    & &&(d-1)\Delta_{12}Y^{12}=\Delta_{12} S^{12}.
\end{alignat}
\end{subequations}
\end{minipage}
\hfill
\begin{minipage}{0.48\textwidth}
\begin{center}
    \textbf{cDI-superchannels}
\end{center}
\begin{subequations}\label{eq:twirledSDP_S_final_detection}
\begin{alignat}{2} 
    &\text{maximize:}\quad && d\,\Tr{Y^{12} \left(J_{\mE}^{12}\right)^T}\\
    &\text{subject to:}\quad && Y^{12}\geq 0, S^{12}\geq 0, M^1\geq 0 \\
    & && \Tr{M^1}=d \\
    & && Y^{12}+S^{12}= M^1 \otimes \id^2 \\
    & && Y^{12}=\Delta_2Y^{12} \\
    & && S^{12}=\Delta_2 S^{12} \\
    & &&(d-1)\Delta_{12}Y^{12}=\Delta_{12} S^{12}.
\end{alignat}
\end{subequations}
\end{minipage}
\vspace{0.25cm}

Conversely, in Eq.~\eqref{eq:twirledSDP_S_l} we can always choose all $S_l$ identical, thus, Eq.~\eqref{eq:twirledSDP_S_final*} and Eq.~\eqref{eq:twirledSDP_S_l} are equal. The same applies to the detection-incoherent setting.

Next, consider the MIO setting. Starting from Eq.~\eqref{eq:twirledSDP_S_final*}, and in particular Eq.~\eqref{eq:Delta1_Y} and Eq.~\eqref{eq:Delta1_S}, we can, without loss of generality, take 
\begin{align}\label{eq:Block_diag_SDP_variable}
     Y^{12}=\sum_{i=0}^{d_1-1} \ketbra{i}{i}_1 \otimes X_i^2, \quad  S^{12}=\sum_{i=0}^{d_1-1} \ketbra{i}{i}_1 \otimes Z_i^2, \quad \text{and} \quad  M^1=\sum_{i=0}^{d_1-1} m_i \ketbra{i}{i}_1,
\end{align}
for some $X_i$, $Z_i$, and for some $m_i\geq 0$ with $ \sum_i m_i=d$. Here, $M$ being diagonal follows from Eq.~\eqref{eq:Y+S=M}. The remaining constraints of Eq.~\eqref{eq:twirledSDP_S_final*} imply 
\begin{align}\label{eq:Delta_X_I}
     X_i^2+Z_i^2= m_i \id^2, \quad \text{and} \quad  (d-1) \Delta X_i^2= \Delta Z_i^2 \qquad \forall i \in\{0,\ldots, d_1-1 \},
\end{align}
where we recall that in our notation the superscripts label systems (and not matrix exponents). Since $Z_i^2= m_i \id^2-X_i^2$, positive semidefiniteness of $Z_i$ (and $X_i$) implies
\begin{align}\label{eq:X_i_inequ}
    0\leq X_i\leq m_i \id.
\end{align}
Moreover, substituting $Z_i^2= m_i \id^2-X_i^2$ into the second equation in Eq.~\eqref{eq:Delta_X_I} yields
\begin{align}
    (d-1) \Delta X_i^2= \Delta Z_i^2 = m_i \id^2 -\Delta X_i^2,
\end{align}
and hence
\begin{align}
    \Delta X_i^2=\frac{m_i}{d} \id^2.
\end{align}
Conversely, let $X_i^2$ satisfy the latter condition and Eq.~\eqref{eq:X_i_inequ}. Then, define $Z_i^2=m_i\id^2-X_i$, and the corresponding variables in Eq.~\eqref{eq:Block_diag_SDP_variable}, which clearly satisfy the constraints in Eq.~\eqref{eq:twirledSDP_S_final*}. Then, the constraints can equivalently be expressed as $0\leq X_i\leq m_i \id$ and $\Delta X_i= \frac{m_i}{d}\id$ for $m_i\geq 0$ with $\sum_i m_i=d$.

Moreover, note that the objective function in Eq.~\eqref{eq:MIO_SDP_obj} can be written as 
\begin{align}
    d\,\Tr{Y^{12} \left(J_{\mE}^{12}\right)^T}= d\sum_i\Tr{(\ketbra{i}{i}_1\otimes X_i^2)^T J_{\mE}^{12}}=d\sum_i\Tr{ X_i^T \mE(\ketbra{i}{i})}.
\end{align}
Thus, after substituting the variable $X_i$ by its transpose (which preserves all constraints), and after eliminating $Z_i^2$, we can rewrite Eq.~\eqref{eq:twirledSDP_S_final*} as
\begin{align}\label{eq:L_modified_again}
     d^2\! \max_{\mS \in \cMIO}F_{\ChoiF}(\mS[\mE],\mF_d)\!&=\! \max \Big\{d \! \sum_{i=0}^{d_1-1} \Tr{X_i^2\mE(\ketbra{i}{i}_1)} :0\!\leq\! X_i \!\leq\!  m_i \id, \Delta(X_i)\!=\!\tfrac{m_i}{d}\id, m_i\geq \!0 \, \forall i, \sum_{i=0}^{d_1-1} \!m_i\!=\!d \Big\}, \nonumber \\
     &=\max\left\{ \sum_{i=0}^{d_1-1}\! \Tr{m_i \tilde{X}_i^2\mE(\ketbra{i}{i}_1)} : 0\!\leq\! \tilde{X}_i \!\leq \!d \id, \Delta \tilde{X}_i \!=\!\id, m_i\!\geq\! 0 \, \forall i, \sum_{i=0}^{d_1-1} m_i\!=\!d \right\},
\end{align}
where the second line follows from first multiplying $0\leq X_i\leq m_i \id$ with $d/m_i$, the constraint is equivalent to $0\leq (\idChan-\Delta)\left(X_i\frac{d}{m_i}\right)+\id \leq d \id$ for all $i: m_i>0$ (all $X_i$ with $m_i=0$ are irrelevant since this enforces $X_i=0$). Defining $\tilde{X}_i=X_i\frac{d}{m_i}$, we clearly have that $\Delta \tilde{X}_i= \id$. Next, we show that the latter expression reduces to $d M_{\cMIO,d}(\mE)$. To this end, notice that for any feasible point in the latter expression, we have
\begin{align}
      \sum_{i=0}^{d_1-1} m_i \Tr{\tilde{X}_i \mE(\ketbra{i}{i})} \leq \left(\max_{j} \Tr{\tilde{X}_j \mE(\ketbra{j}{j})}\right) \sum_{i=0}^{d_1-1}m_i = d\max_{j} \Tr{\tilde{X}_j \mE(\ketbra{j}{j})}.
\end{align}
Dropping the dependence of $\tilde{X}_i$ on $i$, which we can do from here on, we can relax 
\begin{align}\label{eq:MIO_upperBound}
       d^2 \max_{\mS \in \cMIO}F_{\ChoiF}(\mS[\mE],\mF_d)&=\max\left\{ \sum_{i=0}^{d_1-1} \Tr{m_i \tilde{X}_i^2\mE(\ketbra{i}{i}_1)} : 0\leq \tilde{X}_i \leq d \id, \Delta(\tilde{X}_i)=\id, m_i\geq 0 \, \forall i, \sum_{i=0}^{d_1-1} m_i=d \right\} \nonumber \\
      &\leq d  \max\left\{ \Tr{\tilde{X} \mE(\ketbra{j}{j})}: 0\!\leq \!\tilde{X}\!\leq\! d\id, \Delta \tilde{X}=\id,0\leq j\leq d_1-1\right\}\\ \nonumber 
      &=d M_{\cMIO,d}(\mE).
\end{align}
Conversely, let $(j^\star,\tilde{X}^\star)$ be an optimal feasible point of the SDP for $ M_{\cMIO,d}(\mE)$. In Eq.~\eqref{eq:L_modified_again} choose the variables
\begin{align}
    m_i=\delta_{i,j^\star} d \quad\text{and}\quad     X_i= \delta_{i,j^\star} \tilde{X}^\star.
\end{align}
Then, clearly $m_i\geq 0, \sum_i m_i=d$ and $0\leq X_i =\delta_{i,j^\star} \tilde{X} \leq \delta_{i,j^\star} d\id =m_i \id$, and $\Delta X_i =\delta_{i,j^\star} \id= \tfrac{m_i}{d} \id$. Thus, this is a feasible point of Eq.~\eqref{eq:L_modified_again}, and this choice of variables establishes the equality with $d M_{\cMIO,d}(\mE)$.

Next, consider the DI setting. Starting from Eq.~\eqref{eq:twirledSDP_S_final_detection}, eliminate the variable $S^{12}$ by using $S^{12}= M^1 \otimes \id^2-Y^{12}$ and thus $\Delta_{12}S^{12}= \Delta M^1 \otimes \id^2-\Delta_{12} Y^{12}$, which allows us to reduce Eq.~\eqref{eq:twirledSDP_S_final_detection} to
\begin{align}\label{eq:twirledSDP_DI}
     &d \max_{\mS \in \cDI}F_{\ChoiF}(\mS[\mE],\mF_d)\nonumber\\
     &=\max \left\{ \Tr{Y^{12} \left(J_{\mE}^{12}\right)^T}:\Tr{M^1}=d, Y^{12}=\Delta_2Y^{12}, M^1\!\otimes\!\id^2 \geq Y^{12} \geq 0, d \Delta_{12} Y^{12}=\Delta M^1\!\otimes\!\id^2 \right\}.
\end{align}
Without loss of generality, one can take $Y^{12}=\sum_{i=0}^{d_2-1} X_i^1 \otimes \ketbra{i}{i}_2$. The constraints of Eq.~\eqref{eq:twirledSDP_DI} are equivalent to $0\leq X_i^1 \leq M^1 $, $d \Delta X_i^1=\Delta M$ for some $M^1\geq 0$ with $\Tr{M^1}=d$. Thus, Eq.~\eqref{eq:twirledSDP_DI} can equivalently be written as
\begin{align}\label{eq:twirledSDP_DI*}
   d \max_{\mS \in \cDI}F_{\ChoiF}(\mS[\mE],\mF_d)=\max \left\{ \sum_{i=0}^{d_2-1} \braket{i|\mE(X_i^1)|i} : 0\leq X_i^1 \leq M^1, \, \forall i, \Tr{M^1}=d, d\Delta_{1}X_i^1=\Delta M^1 \right\}.
\end{align}
From the last two constraints, it follows that $\Tr{X_i^1}=\Tr{\Delta X_i^1}=\frac{1}{d}\Tr{M^1}=1\quad \forall i$, and $X_i^1=:\rho_i$ are therefore quantum states. Introducing the new variable $\sigma^1:=M^1/d$, which is also a state, Eq.~\eqref{eq:twirledSDP_DI*} finally takes the form
\begin{align}
     d \max_{\mS \in \cDI}F_{\ChoiF}(\mS[\mE],\mF_d)=\max \left\{ \sum_{i=0}^{d_2-1} \braket{i|\mE(\rho_i^1)| i}_2 : \rho_i,\sigma \in \mD, \rho_i \leq d \sigma,\Delta \rho_i =\Delta \sigma \, \forall i\right\}.
\end{align}

This finishes the optimality part of the proof, and we now show that this optimal conversion error is achievable. To this end, consider the left-hand side of Eq.~\eqref{eq:InequalityChain_conversion_distance}, and choose the variable $Z^{03}$ for the optimization problem for the diamond distance in Eq.~\eqref{eq:SPD_DiamondDistance} as
\begin{align}
    Z^{03}&:=\Big[\partTr{12}{J_{\mS}^{0123}\left(\left( \id^{03}\otimes J_{\mE}^{12} \right)^T \right)}-J_{\mF_d}^{03} \Big]_+.
\end{align}
This is clearly a feasible choice as $Z^{03} \geq 0$, and $Z^{03} \geq \partTr{12}{J_{\mS}^{0123}\left(\left( \id^{03}\otimes J_{\mE}^{12} \right)^T \right)}-J_{\mF_d}^{03}$. In the case of $\O=\MIO$, let $(j^\star,X^\star)$ be an optimal feasible point of $M_{\cMIO,d}(\mE)=\max\big\{ \Tr{X \mE(\ketbra{j}{j})}: 0\leq X\leq d\id,\Delta(X)=\id ,0\leq j\leq d_1-1\big\}$. Define a POVM $\mathbb{M}:=\{P,\id-P\}$ with $P:=\tfrac{1}{d}(X^\star)^T$ that satisfies $\Delta P=\frac{1}{d}\id$. Now, consider the channel and superchannel 
\begin{subequations}
  \begin{align}
    J_{\mN}^{03}&:=  \frac{d}{d-1} \sum_{l=1}^{d-1} \Pi_{0,l}^{03}=\frac{d}{d-1}\left( (F_d\otimes \id) \tilde{\Pi}^{03}(F_d\otimes \id)^\dagger - \frac{1}{d}J_{\mF}^{03}\right), \label{eq:J_M_rewriting} \\
    J_{\mS}^{0123}&=\ketbra{j^\star}{j^\star}_1\otimes \Big(P^2 \otimes J_{\mF_d}^{03} +\left(\id^2-P^2 \right) \otimes J_{\mN}^{03} \Big), \label{eq:J_optimal_measureAct}
\end{align}  
\end{subequations}
where $\tilde{\Pi}=\sum_n \ketbra{nn}{nn}$ denotes the unnormalized maximally correlated state. That $\mN$ is a quantum channel follows since $\partTr{3}{J_{\mN}^{03}}=\tfrac{d}{d-1}(\id^0-\frac{1}{d} \id^0)= \id^0$ and $J_{\mN}^{03}\geq 0$ (recall $\Pi_{k,l}$ are rank-one projectors)\footnote{Alternatively, the latter also follows from  $\tilde\Pi \geq \ketbra{\phi}{\phi}$ (and the Fourier transform being symmetric), which allows us to write $J_{\mF_d}=(F_d\otimes\id)\ketbra{\phi}{\phi} (F_d\otimes \id)^\dagger $. Let $\ket{\psi}=\sum_{k,l}c_{k,l} \ket{kl}$ be arbitrary and let $C=\sum_n c_{n,n}$. Then, 
\begin{align}
    \bra{\psi}(\tilde\Pi- \ketbra{\phi}{\phi}) \ket{\psi}&=\sum_j |c_{j,j}|^2-\frac{1}{d} \left(\sum_n c_{n,n}\right)^*\left(\sum_m c_{m,m}\right)= \sum_j |c_{j,j}|^2-\frac{1}{d} |C|^2 \geq 0, \nonumber
\end{align}
where the last inequality follows from the Cauchy-Schwarz inequality.}.
For this superchannel $\mS$, we obtain
\begin{align}\label{eq:evaluate_Z}
    Z^{03}&=\Big[\partTr{12}{J_{\mS}^{0123}\left(\left( \id^{03}\otimes J_{\mE}^{12} \right)^T \right)}-J_{\mF_d}^{03} \Big]_+= \Big[ \Tr{P^T\mE(\ketbra{j^\star}{j^\star})}-1\Big]_+ J_{\mF_d}^{03}+\Tr{(\id-P)^T\mE(\ketbra{j^\star}{j^\star})} J_{\mN}^{03} \nonumber \\
    &= \left(1-\Tr{P^T\mE(\ketbra{j^\star}{j^\star})}\right) J_{\mN}^{03}=\left(1- \frac{1}{d}\Tr{X^\star\mE(\ketbra{j^\star}{j^\star})}\right) J_{\mN}^{03}= \left(1-\frac{1}{d}M_{\cMIO,d}(\mE)\right) J_{\mN}^{03},
\end{align}
and thus $\norm{\partTr{3}{Z^{03}}}_\infty=1-\frac{1}{d}M_{\cMIO,d}(\mE) $. It remains to show that $\mS\in\cMIO$, which we show by proving $\mS \in \nMIO \subset \cMIO$. To this end, let $\mE_1$ and $\mE_2$ be the quantum channels defined in Fig.~\ref{fig:freeNetwork_combined}. Then, the superchannel in Eq.~\eqref{eq:J_optimal_measureAct} is implemented by the network in Fig.~\ref{fig:freeNetwork_combined}. Moreover, clearly $\mE_1^{0\to 1,R} \in \MIO$ and for $\mE_2^{2,R \to 3}$ we find that
\begin{align}
    \mE_2^{2,R \to 3} (\ketbra{i}{i}_2 \otimes \ketbra{j}{j}_R) &= \Tr{P\ketbra{i}{i}} \mF_d(\ketbra{j}{j}) + \Tr{(\id-P)\ketbra{i}{i}} \mN(\ketbra{j}{j}) \nonumber \\
    &=\frac{1}{d} \mF_d(\ketbra{j}{j}) + \frac{d-1}{d} \mN(\ketbra{j}{j}) =\frac{1}{d} \mF_d(\ketbra{j}{j}) + \left(\frac{1}{d}\id^3 - \frac{1}{d} \mF_d(\ketbra{j}{j}) \right) = \frac{1}{d} \id^3,
\end{align}
where  we used that 
\begin{align}
    \mN(\ketbra{j}{j})&= \frac{d}{d-1} \partTr{R}{\left( (F_d\otimes \id) \tilde{\Pi}^{R3}(F_d\otimes \id)^\dagger- \frac{1}{d} J_{\mF_d^{R3}} \right) \left(\ketbra{j}{j}_R\otimes \id^3\right)} =\frac{d}{d-1} \left(\frac{1}{d} \id -\frac{1}{d} \mF_d(\ketbra{j}{j}\right).
\end{align}
Thus, $\mE_2^{2,R\to 3} \in \MIO$. As such $\mS \in \nMIO$, and thus the optimal conversion distance (under cMIO superchannels and MIO networks) is upper bounded by 
\begin{align}
    \min_{\mS \in \cMIO} \frac{1}{2}\norm{\mS[\mE]-\mF_d}_\diamond \leq \min_{\mS \in \nMIO} \frac{1}{2}\norm{\mS[\mE]-\mF_d}_\diamond \leq \norm{\partTr{3}{Z^{03}}}_\infty =1-\frac{1}{d}M_{\cMIO,d}(\mE),
\end{align}
which together with Eq.~\eqref{eq:MIO_upperBound}, proves the claim that the optimal conversion distance is determined by $\tfrac{1}{d}M_{\cMIO,d}(\mE)$, and it is also achievable by a $\MIO$ network. 
\begin{figure}[ht]
    \centering
    \scalebox{0.4}{\includegraphics[width=0.99\linewidth]{ 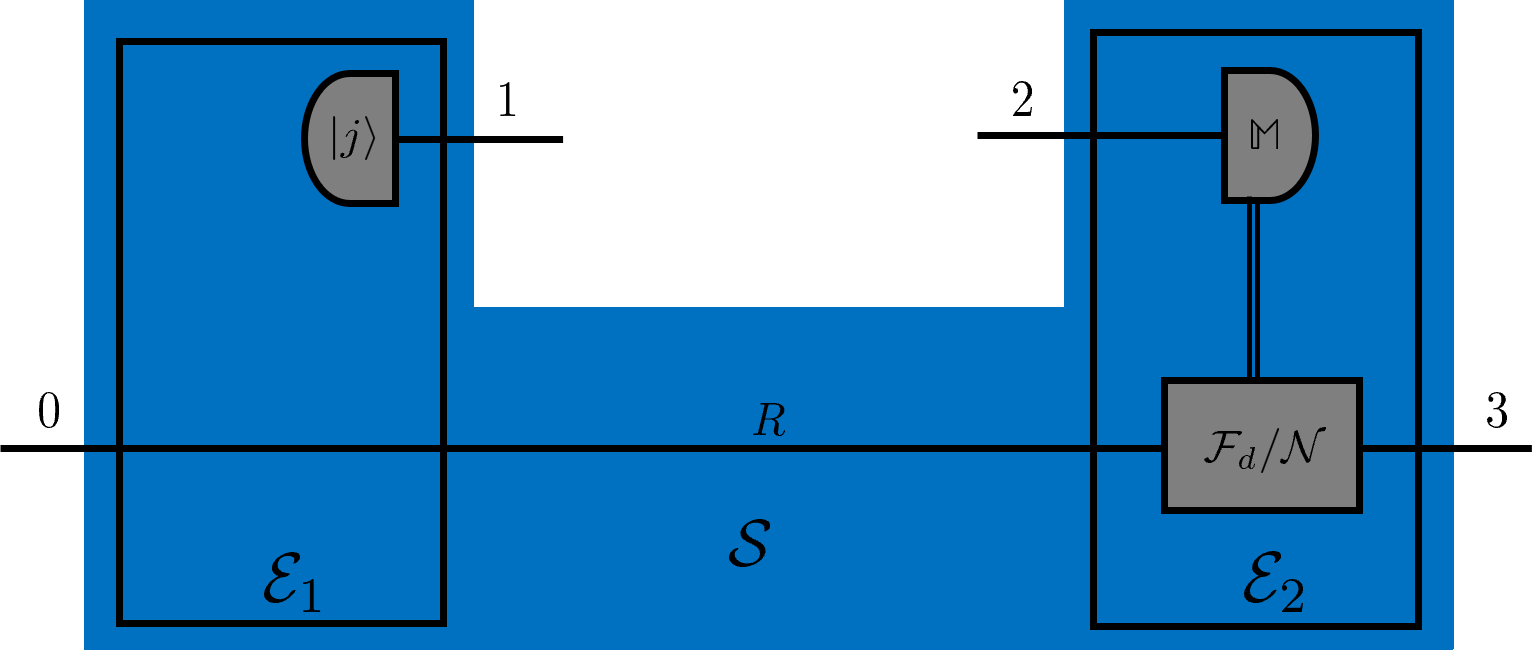}}
    \caption{Implementation of the superchannels $\mS$ via ``measure-and-act" networks, where upon measuring $\mathbb{M}=\{P^T,\id-P^T\}$ we apply $\mF_d$ or $\mN$ as defined in Eq.~\eqref{eq:J_M_rewriting}.}
    \label{fig:freeNetwork_combined}
\end{figure}

In the detection-incoherent setting, let $\{\rho_i^\star,\sigma^\star \}_i$ denote an optimal feasible point of $M_{\cDI,d}(\mE) \!=\! \max \Big\{ \sum_{i=0}^{d_B-1}\braket{i|\mE(\rho_i)| i}_2 \!:\! \rho_i,\sigma \in \D, \rho_i\! \leq\! d \sigma, \Delta \rho_i \!=\!\Delta \sigma \, \forall i: 0\leq i \leq d_B-1\Big\}$. Define
\begin{subequations}
    \begin{align}
            J_{\mG}^{03}&:= \frac{d}{d-1} \sum_{k=1}^{d-1} \Pi_{k,0}^{03} = \frac{d}{d-1} \left( (\id\otimes F_d) \tilde{\Pi}^{03}(\id\otimes F_d)^\dagger -\frac{1}{d} J_{\mF_d}^{03}\right), \label{eq:mG} \\
            J_{\tilde{\mS}}^{0123}&:= \frac{1}{d} X^{12} \otimes J_{\mF_d}^{03} + \left( (\sigma^\star)^1 \otimes \id^2 -\frac{1}{d}X^{12}\right) \otimes J_{\mG}^{03}, \label{eq:J_optimal_measureAct_detection}
    \end{align}
\end{subequations}
where $X^{12}:=\sum_{i=0}^{d_2-1} (\rho_i^\star)^1 \otimes \ketbra{i}{i}_2$ with $\Delta_1 X^{12}=\Delta \sigma^\star \otimes \id$ since $\Delta \rho_i^\star=\Delta \sigma^\star$. This is a valid choice of a superchannel $\tilde{\mS}\in \cDI$, since $\Delta_{3}J_{\tilde{\mS}}^{0123}= \Delta_{32}J_{\tilde{\mS}}^{0123}$ by construction, and 
\begin{align}
    \Delta_{31}J_{\tilde{\mS}}^{0123} &= \id^0 \otimes \Delta(\sigma^\star)^1 \otimes \id^2 \otimes  \frac{1}{d} \id^3,
\end{align}
which is already diagonal, and thus $\Delta_{31}J_{\tilde{\mS}}^{0123}=\Delta_{3210}J_{\tilde{\mS}}^{0123}$. Analogously to the consideration in the creation-incoherent setting in Eq.~\eqref{eq:evaluate_Z}, this yields 
\begin{align}
    \min_{\mS \in \cDI} \frac{1}{2}\norm{\mS[\mE]-\mF_d}_\diamond  \leq \norm{\partTr{3}{Z^{03}}}_\infty =1-\frac{1}{d}M_{\cDI,d}(\mE),
\end{align}
Analogously to the creation-incoherent case, this provides an upper bound to the left-hand side of Eq.~\eqref{eq:InequalityChain_conversion_distance}. \qedhere
\end{proof}

\begin{cor}\label{cor:Monotones_properties}
Let $\mE^{A\to B}$ be a quantum channel, let $\O \in \{\MIO,\DI\}$, and let $d\geq 2$. Then, the functionals $M_{\cO,d}$ are convex monotones under $\cO$. Moreover, the monotones are bounded by $1\leq M_{\cMIO,d}(\mE) \leq \min\{d,d_B\}$ and $1\leq M_{\cDI,d}(\mE) \leq \min\{d,d_A,d_B\}$, and are faithful (for $d>1$) in the sense that $M_{\cO,d}\geq 1$ with $M_{\cO,d}=1$ iff $\mE\in\O$. Additionally, $M_{\cO,d}(\mE)/d$ is non-increasing in $d$ for every channel $\mE$.
\end{cor}
\begin{proof}
    That the functionals are monotones follows immediately from Thm.~\ref{thm:CoherenceOptimalConversion_main}. Convexity follows immediately since (for fixed $d$) the feasible sets are independent of the channel, while the corresponding objective functions are linear in $\mE$. Hence, $M_{\cO,d}$ is a pointwise maximum of linear functionals of $\mE$, and therefore convex. Before we show faithfulness, we derive the maximal values of the monotones. Note that in $M_{\cMIO,d}(\mE^{A\to B})=\max\big\{\Tr{X^B \mE^{A\to B}(\ketbra{j}{j}_A)} : 0\leq X^B\leq d\id^B, \quad \Delta(X^B)=\id^B, 0\le j \le d_A-1\big\}$, the first constraint can equivalently be written as $0\leq X^B\leq \min\{d,d_B\}\id^B$. Thus, $\min\{d,d_B\} \geq M_{\cMIO,d}(\mE^{A\to B})\geq 1 $, where the last inequality follows since $X^B=\id^B$ is feasible. 
    For  $M_{\cDI,d}(\mE^{A\to B}) = \max \big\{ \sum_{i=0}^{d_B-1}\braket{i|\mE^{A\to B}(\rho_i^A)| i}_B \!: \rho_i,\sigma \in \D, \rho_i \leq d \sigma, \Delta \rho_i =\Delta \sigma \, \forall i: 0\leq i \leq d_B-1\big\}$, let $M_i^A:= \mE^\dagger(\ketbra{i}{i})$ which satisfies $M_i^A\geq 0$ and $\sum_i M_i^A =\id$. As such, the objective function can be written as $\sum_{i=0}^{d_B-1}\braket{i|\mE(\rho_i)| i}_B=\sum_{i=0}^{d_B-1} \Tr{\rho_i^A M_i^A}$. From $\rho_i\leq d\sigma$, $\rho_i^A\leq \id^A$, and $M_i^A\leq \id^A$, we obtain $M_{\cDI,d}(\mE)\leq \min \{d,d_A,d_B \}$. Conversely, choosing $\rho_i=\sigma\, \forall i$ yields $M_{\cDI,d}(\mE)\geq 1$. 
    
    Next, note that $M_{\cO,d}(\mE)= 1$ for every $\mE\in \O$ is clear from inspection. Conversely, let $\mE\notin \MIO$, then there exists some $j\in \{0,\ldots, d_A-1\}$ such that $\tau_j^B:=\mE^{A\to B}(\ketbra{j}{j}_A) \notin \I$. Now let $H:= \tau_j -\Delta(\tau_j)$, and choose $X:= \id+\epsilon H$ where $\epsilon:= \tfrac{1}{2\norm{H}_\infty} $. Then, $0\leq X\leq d\id$, and 
    \begin{align}
        M_{\cMIO,d}(\mE^{A\to B})\geq 1+ \epsilon \Tr{\tau_j H} = 1+ \epsilon \Tr{ H^2} >1,
    \end{align}
    where we used that $H$ is off-diagonal, and as such, only the off-diagonal part of $\tau_j$ contributes. If $\mE\notin \DI$, then there exists $i^\prime \in \{0,\ldots, d_B-1\}$ such that $M_{i^\prime}\neq\Delta(M_{i^\prime})$. Analogously to before, define $\tilde{H}:=M_{i^\prime}-\Delta(M_{i^\prime})$. Now choose $\sigma=\id^A/d_A$ and $\rho_i= \sigma \, \forall i\neq i^\prime$, and $\rho_{i^\prime}= \sigma +\tilde{\epsilon} \tilde{H}$ with $\tilde{\epsilon}=\tfrac{1}{2d_A\norm{\tilde{H}}_\infty}$. Then, $\{\sigma,\rho_i\}$ is in the feasible set, and we again obtain
    \begin{align}
        M_{\cDI,d}(\mE^{A\to B})\geq 1+ \tilde{\epsilon} \Tr{\tilde{H}^2} >1.
    \end{align}
    Lastly, let $1< d_1\leq d_2$. Then, $0\leq \xi:=\tfrac{d_1-1}{d_2-1} \leq 1$. Now let $X$ denote a feasible point of $M_{\cMIO,d_2}(\mE)$ and let $\{\sigma, \rho_i\}$ denote a feasible point of $M_{\cDI,d_2}(\mE)$. Now define
    \begin{subequations}
    \begin{align}
        X^\prime:= (1-\xi)\id+\xi X, \quad \text{and} \quad \rho_i^\prime := (1-\xi) \sigma +\xi \rho_i \, \forall i\in \{0,\ldots, d_B-1\}.
    \end{align}
    \end{subequations}
    Then, clearly $\Delta(X^\prime)=\id$, and $0\leq X^\prime \leq (1+\xi(d_2-1))\id=d_1\id$. Analogously, $\rho_i^\prime \in \D$, and $\rho_i^\prime\leq (1+\xi(d_2-1))\sigma=d_1\sigma$. As such, these are feasible points in $M_{\cO,d_1}(\mE)$ respectively, yielding
    \begin{align}
        M_{\cO,d_1}(\mE) \geq 1+\xi\left(M_{\cO,d_2}(\mE)-1 \right) = 1+\frac{d_1-1}{d_2-1}\left(M_{\cO,d_2}(\mE)-1 \right) \geq \frac{d_1}{d_2} M_{\cO,d_2}(\mE),
    \end{align}
    where the last inequality follows from $d_1\leq d_2$ and $M_{\cO,d_2}(\mE)\leq d_2$.
\end{proof}

As mentioned in the main text, this result can be directly extended to the case of creation-detection-incoherent operations. Here, we require the following monotones
\begin{subequations}\label{eq:DIO_d}
\begin{alignat}{2} 
    M_{\cDIO,d}(\mE):=&\text{ maximize}\quad && \Tr{Y^{12} \left(J_{\mE}^{12}\right)^T}\\
    &\text{ subject to}\quad && Y^{12},S^{12},T^{12} \geq 0, \sigma \in \I\\
    & &&  Y^{12}+T^{12}+S^{12} \leq d\sigma^1 \otimes \id^2\\
    & && \Delta_1 (Y^{12}+T^{12})=Y^{12}+T^{12} \\
    & && \Delta_2 (Y^{12}+S^{12})=Y^{12}+S^{12} \\
    & && \Delta_{12} (Y^{12}+T^{12})=\Delta_{12} (Y^{12}+S^{12})= \sigma^1\otimes \id^2,
\end{alignat}
\end{subequations}

to establish the analogous Corollary.
\begin{proposition}\label{cor:DIO_Conv}
Let $\mE$ be any quantum channel. Then, for $d\geq 2$
    \begin{align}
        \min_{\mS \in \cDIO_1} \tfrac{1}{2}\norm{\mS[\mE]-\mF_d}_\diamond &=  1-\max_{\mS \in \cDIO_1}F_{\min}(\mS[\mE],\mF_d) = 1-\max_{\mS \in \cDIO_1}F_{\ChoiF}(\mS[\mE],\mF_d) =1-\tfrac{1}{d} M_{\cDIO,d}(\mE).
    \end{align}    
\end{proposition}
\begin{proof}
Let $d\geq2$.For creation-detection-incoherent operations, we can just repeat the argumentation from Eq.~\eqref{eq:InequalityChain_conversion_distance} to Eq.~\eqref{eq:LowerBound_Twirled_b}. Following Thm.~\ref{thm:CompatibleSupermaps}, combining the coherence constraints from Eq.~\eqref{eq:CohCOnst_0} to Eq.~\eqref{eq:CohCOnst_321} yields
\begin{subequations}\label{eq:DIO_SymmetrySDP}
\begin{alignat}{2}
    d \max_{\mS \in \cDIO}F_{\ChoiF}(\mS[\mE],\mF_d)=& \, \text{maximize}\quad && \Tr{Y_{0,0}^{12} \left(J_{\mE}^{12}\right)^T}\\
    &\,\text{subject to}\quad && M^1,Y_{k,l}^{12} \geq 0, \quad \forall\, k,l: 0\leq k,l\leq d-1  \\
    & && \Tr{M^1}=d \quad \text{and} \quad \sum_{k,l=0}^{d-1} Y_{k,l}^{12}= M^1 \otimes \id^2\\
    & && \sum_{k=0}^{d-1} Y_{k,l}^{12}= \Delta_1\sum_{k=0}^{d-1}  Y_{k,l}^{12} \quad \text{and} \quad \sum_{l=0}^{d-1} Y_{k,l}^{12}= \Delta_2\sum_{l=0}^{d-1}
    Y_{k,l}^{12} \quad \forall\, k,l \label{eq:Delta_1Constraint}\\
    & && \Delta_{12}\!\sum_{k=0}^{d-1} \!Y_{k,l}^{12} \!=\! \Delta_{12}\! \sum_{k=0}^{d-1}  \!Y_{k,l^\prime}^{12} \,\,\, \text{and} \,\,\, \Delta_{12}\!\sum_{l=0}^{d-1}\!  Y_{k,l}^{12} \!=\!\Delta_{12}\! \sum_{l=0}^{d-1}  Y_{k^\prime,l}^{12} \,\,\forall\, l,l^\prime, k,k^\prime.
\end{alignat}
\end{subequations}
Next, we show that this SDP equals
\begin{subequations}\label{eq:DIO_SymmetrySDP*}
\begin{alignat}{2}
    d \max_{\mS \in \cDIO}F_{\ChoiF}(\mS[\mE],\mF_d)=&\,\text{maximize}\quad && \Tr{Y^{12} \left(J_{\mE}^{12}\right)^T}\\
    &\,\text{subject to}\quad && Y^{12},M^1,S^{12},T^{12},C^{12} \geq 0, \\
    & && \Tr{M^1}=d \\
    & &&  Y^{12}+T^{12}+S^{12}+C^{12}= M^1 \otimes \id^2\\
    & && \Delta_1 (Y^{12}+T^{12})=Y^{12}+T^{12} \\
    & && \Delta_1 (S^{12}+C^{12})=S^{12}+C^{12} \label{eq:const_S+C-deph} \\
    & && \Delta_2 (Y^{12}+S^{12})=Y^{12}+S^{12} \\
    & && \Delta_2 (T^{12}+C^{12})=T^{12}+C^{12} \label{eq:const_T+C-deph} \\
    & && (d-1)\Delta_{12} (Y^{12}+T^{12})= \Delta_{12} (S^{12}+C^{12}) \label{eq:Delta_12Y+T}  \\
    & && (d-1)\Delta_{12} (Y^{12}+S^{12})= \Delta_{12} (T^{12}+C^{12}). \label{eq:Delta_12Y+S}
\end{alignat}
\end{subequations}
To this end, let $\{Y_{k,l} \}_{0\leq k,l\leq d-1}$ denote a set of feasible variables in Eq.~\eqref{eq:DIO_SymmetrySDP}, then define
\begin{align}
    Y:=Y_{0,0},\quad T:=\sum_{k=1}^{d-1} Y_{k,0}, \quad  S:=\sum_{l=1}^{d-1} Y_{0,l}, \quad
    A:= \sum_{l=1}^{d-1}\sum_{k=0}^{d-1} Y_{k,l}, \quad B:= \sum_{k=1}^{d-1}\sum_{l=0}^{d-1} Y_{k,l}, \quad \text{and}\quad C:= \sum_{l=1}^{d-1}\sum_{k=1}^{d-1} Y_{k,l}.
\end{align}
Clearly, $Y,T,S,A,B,C \geq 0$ and the normalization constraint of
\begin{align}
    M^1\otimes \id^2 = \sum_{k,l} Y_{k,l} = Y+T+S+C.
\end{align}
is satisfied. From Eq.~\eqref{eq:Delta_1Constraint} follows that $\Delta_1 A= A$ and $\Delta_1 (Y+T)= Y+T$. Analogously, $\Delta_2 B= B$ and $\Delta_2(Y+S)=Y+S$ follows. Moreover, since $A=S+C$ and $B=T+C$ holds by construction, this can equivalently be written as $\Delta_1 (C+S)= C+S$ and likewise $\Delta_2 (C+T)= C+T$. Lastly,
\begin{align}
    \Delta_{12} A^{12}= \Delta_{12}\left(\sum_{l=1}^{d-1}\sum_{k=0}^{d-1} Y_{k,l}\right) =  \sum_{l=1}^{d-1} \Delta_{12} (Y+T) =(d-1) \Delta_{12} (Y+T),
\end{align}
which is again equivalent to $\Delta_{12}(S+C)=(d-1) \Delta_{12} (Y+T)$, and analogously for $\Delta_{12}(T+C)=(d-1) \Delta_{12} (Y+S)$. Thus, $(M,Y,T,S,C)$ is a feasible point in Eq.~\eqref{eq:DIO_SymmetrySDP*}. Conversely, let $(M^\star, Y^\star, T^\star, S^\star, C^\star)$ denote an optimal feasible point of Eq.~\eqref{eq:DIO_SymmetrySDP*}. Then, choosing $M=M^\star$ and optimization variables in Eq.~\eqref{eq:DIO_SymmetrySDP} of
\begin{align}
    Y_{0,0}=Y^\star, \quad Y_{k,0}=\frac{T^\star}{d-1}, \quad Y_{0,l}=\frac{S^\star}{d-1}, \quad Y_{k,l}=\frac{C^\star}{(d-1)^2} \quad \forall k,l: 1\leq k,l\leq d-1,
\end{align}
provides a feasible solution to Eq.~\eqref{eq:DIO_SymmetrySDP}, which implies the equality of Eq.~\eqref{eq:DIO_SymmetrySDP} and Eq.~\eqref{eq:DIO_SymmetrySDP*}. Next, notice that $M^1=\Delta_1 M^1$ must hold since
\begin{align}
    M^1\otimes \id^2= Y^{12}+T^{12}+S^{12}+C^{12} =\Delta_1\left(Y^{12}+T^{12}+S^{12}+C^{12}\right).
\end{align}
From $\Delta_1 M^1=M^1$ together with $\Delta_1(Y^{12}+T^{12})=Y^{12}+T^{12}$ and $\Delta_2(Y^{12}+S^{12})=Y^{12}+S^{12}$, it follows that the constraints in Eq.~\eqref{eq:const_S+C-deph} and Eq.~\eqref{eq:const_T+C-deph} are redundant. Furthermore, the constraints in Eqs.~\eqref{eq:Delta_12Y+T} and \eqref{eq:Delta_12Y+S} can be expressed as $d \Delta_{12}(Y^{12}+S^{12})= M^1 \otimes \id^2$ and $d \Delta_{12}(Y^{12}+T^{12})= M^1 \otimes \id^2$. Lastly, we can eliminate the slack variable $C^{12}$ by noting that $0\leq C^{12}= \Delta_1(M^1) \otimes \id^2-(Y^{12}+T^{12}+S^{12})$. Thus, we can replace the normalization constraint by $Y^{12}+S^{12}+T^{12}\leq M^1\otimes \id^2$, and if we define $\sigma^1=\frac{1}{d} M^1$ we can write the optimization as
\begin{align}
     d \max_{\mS \in \cDIO}F_{\ChoiF}(\mS[\mE],\mF_d)=M_{\cDIO,d}(\mE).
\end{align}
It remains to show that this is achievable: Let $(Y,S,T,\sigma)$ denote an optimal feasible point of Eq.~\eqref{eq:DIO_d}, and set $R^{12}:= d\,\sigma^1 \otimes \id^2-Y^{12}-S^{12}-T^{12}$
\begin{align}
J_{\mS}^{0123}&=Y^{12}\otimes \Pi_{0,0}^{03}+ \frac{T^{12}}{d-1}\otimes\sum_{k\neq 0}\Pi_{k,0}^{03}+\frac{S^{12}}{d-1}\otimes\sum_{l\neq 0}\Pi_{0,l}^{03}+\frac{R^{12}}{(d-1)^2}\otimes\sum_{k,l \neq 0}\Pi_{k,l}^{03} \nonumber \\
&= \frac{Y^{12}}{d}\otimes J_{\mF}^{03}+ \frac{T^{12}}{d} \otimes J_{\mG}^{03} +\frac{S^{12}}{d} \otimes J_{\mN}^{03}+\frac{R^{12}}{d} \otimes J_{\mL}^{03},
\end{align}
where the channels $\mN$ and $\mG$ are defined as in Eq.~\eqref{eq:J_M_rewriting} and Eq.~\eqref{eq:mG}, and
\begin{align}
    J_{\mL}^{03} := \frac{d}{(d-1)^2} \sum_{k,l \neq 0}\Pi_{k,l}^{03} = \frac{d}{(d-1)^2} \left( \id^{03}-(\id\otimes F) \tilde{\Pi}(\id\otimes F)^\dagger -(F\otimes \id) \tilde{\Pi}(F\otimes \id)^\dagger+\frac{1}{d}J_{\mF}\right).
\end{align}
This defines a valid superchannel. Moreover, $\mS\in \cDIO$ since $\Delta_1(Y+T)=Y+T$ implies $\Delta_0 J_{\mS}^{0123}= \Delta_{01} J_{\mS}^{0123} $ (see Eq.~\eqref{eq:Twirled_Delta01} and the following discussion). Analogously,  $\Delta_3 J_{\mS}^{0123}= \Delta_{32} J_{\mS}^{0123} $, $\Delta_{02} J_{\mS}^{0123}= \Delta_{0123} J_{\mS}^{0123} $, and $\Delta_{31}J_{\mS}^{0123}= \Delta_{3210} J_{\mS}^{0123} $ follow from $\Delta_2(Y+S)=Y+S$ and $\Delta_{12}(Y+T)=\Delta_{12}(Y+S)=\sigma \otimes \id$
respectively. Then,
\begin{align}
    Z^{03}= \left[J_{\mS[\mE]}^{03}-J_{\mF}^{03}\right]_+ =\frac{1}{d}\Tr{T^{12}\left(J_{\mE}^{12}\right)^T}J_{\mG}^{03}+\frac{1}{d}\Tr{S^{12}\left(J_{\mE}^{12}\right)^T}J_{\mN}^{03}+\frac{1}{d}\Tr{R^{12}\left(J_{\mE}^{12}\right)^T}J_{\mL}^{03},
\end{align}
where we used that the Choi operators have mutually orthogonal support. Thus,
\begin{align}
    \min_{\mS \in \cDIO} \frac{1}{2}\norm{\mS[\mE]-\mF_d}_\diamond &\leq \frac{1}{d} \left(\Tr{\left(T^{12}+S^{12}+R^{12}\right) \left(J_{\mE}^{12}\right)^T} \right)=  \frac{1}{d} \left(\Tr{\left( d\,\sigma^1 \otimes \id^2 -Y^{12} \right) \left(J_{\mE}^{12}\right)^T} \right) \nonumber \\
    &= 1-\frac{1}{d} M_{\cDIO,d}(\mE).
\end{align}

\end{proof}

So far, we have only considered distillation in parallel protocols. We now address the question of whether adaptive protocols offer an advantage compared to parallel protocols. Already in the case of $\cMIO$ superchannels, such an advantage exists, as we show next.

\SequentialAdvantage*
\begin{proof}
    Consider two copies of the unitary qubit channel $\mU$ defined by
    \begin{align}
        U:= \begin{pmatrix}
             \cos(\pi/8) & -\sin(\pi/8) \\
             \sin(\pi/8) & \cos(\pi/8)
        \end{pmatrix},
    \end{align}
    for which we now show that sequential distillation (to the Hadamard gate) outperforms parallel distillation. Note that 
    \begin{align}
        U^2Z= \frac{1}{\sqrt{2}}\begin{pmatrix}
            1 & -1 \\ 
            1 & 1
        \end{pmatrix} Z = H, 
    \end{align}
    and thus a sequential protocol can achieve the conversion to a Hadamard gate exactly by simply applying a Z gate followed by the unitary $U$ twice. Next, we show that this is impossible in a parallel protocol. 
    Recall that according to Thm.~\ref{thm:CoherenceOptimalConversion_main}, the parallel conversion distance to a Hadamard gate is determined by Eq.~\eqref{eq:C_D}, which takes the form of
    \begin{align}
        M_{\cMIO,2}(\mU^{\otimes 2}) &= \max\Big\{ \Tr{A \,\psi_{i,j}}: 0\leq A\leq 2\id, \Delta(A)=\id, 0\le i,j \leq 1 \Big\},
    \end{align}
    where $\ket{\psi_{i,j}}= U^{\otimes 2} \ket{i,j}$. The four states are related to each other by basis permutations combined with sign flips. Since the feasible set is invariant under such operations, the optimization problem is simply given by 
    \begin{align}\label{eq:M_MIO_U_primal}
        M_{\cMIO,2}(\mU^{\otimes 2}) &= \max\Big\{ \Tr{A \,\psi}: 0\leq A\leq 2\id, \Delta(A)=\id \Big\},
    \end{align}
    where $\ket{\psi}:=\ket{\psi_{0,0}}= \cos^2(\pi/8)\ket{0,0}+\sin(\pi/8)\cos(\pi/8) \left(\ket{1,0}+\ket{0,1}\right)+\sin^2(\pi/8) \ket{1,1}$. Switching to the dual problem, introduce dual variables $Y, D \in \Herm$, such that the Lagrangian of Eq.~\eqref{eq:M_MIO_U_primal} is 
    \begin{align}
        L(A;Y,D)\!=\!\Tr{\ketbra{\psi}{\psi} A}\!+\!\Tr{Y\left(2\id\!-\!A\right)}\!+\!\Tr{D\left(\id\!-\!\Delta(A)\right)}\!=\!2 \Tr{Y}\!+\!\Tr{D}\!+\!\Tr{A \left(\ketbra{\psi}{\psi}\!-Y-\Delta(D)\right)}.
    \end{align}
    The dual, for which strong duality holds, is thus
    \begin{align}
        M_{\cMIO,2}(\mU^{\otimes 2}) =\min\left\{2 \Tr{Y}+\Tr{D}: Y\geq 0, Y+D \geq \ketbra{\psi}{\psi},\Delta D= D, \, Y,D\in \Herm \right\}.
    \end{align}
    Now let $c:=\cos^2(\pi/8)$, and define
    \begin{align}
        \ket{\xi} := \ket{0,0}, \quad \text{and} \quad \ket{\mu} := \frac{1}{\sqrt{1-c^2}} \left( \ket{\psi} -c \ket{0,0}\right)\quad \text{and} \quad \ket{\phi_{\pm}} :=\frac{1}{\sqrt{2}} (\ket{\xi} \pm \ket{\mu}),
    \end{align}
    from which follows that $\braket{\xi|\mu}=0$, $\Delta(\ketbra{\xi}{\mu})=0$, and $\ket{\psi}=c \ket{\xi} +\sqrt{1-c^2} \ket{\mu}$. We can now choose variables
    \begin{subequations}
    \begin{align}
        Y^\star&:= (1-c^2+ c\sqrt{1-c^2} )\ketbra{\phi_+}{\phi_+}, \\
        D^\star&:= (2c^2-1) \ketbra{\xi}{\xi}.
    \end{align}
    \end{subequations}
    Clearly, $Y^\star\geq 0$, and $\Delta D^\star= D^\star$. Next, note that
    \begin{align}
        D^\star-\ketbra{\psi}{\psi}&=(2c^2-1) \ketbra{\xi}{\xi} - c^2 \ketbra{\xi}{\xi}-(1-c^2) \ketbra{\mu}{\mu}- c\sqrt{1-c^2}\left(\ketbra{\mu}{\xi}+\ketbra{\xi}{\mu}\right) \nonumber \\
        &= -(1-c^2) \ketbra{\xi}{\xi}-(1-c^2) \ketbra{\mu}{\mu}- c\sqrt{1-c^2}\left(\ketbra{\mu}{\xi}+\ketbra{\xi}{\mu}\right) \nonumber \\
        &= -\left(1-c^2+c\sqrt{1-c^2)}\right) \ketbra{\phi_+}{\phi_+} - \left(1-c^2-c\sqrt{1-c^2)}\right) \ketbra{\phi_-}{\phi_-},
    \end{align}
    where the last line follows from a straightforward calculation. Thus, 
    \begin{align}
        Y^\star+D^\star-\ketbra{\psi}{\psi}= -\left(1-c^2-c\sqrt{1-c^2)}\right) \ketbra{\phi_-}{\phi_-} \geq 0,
    \end{align}
    where the last inequality holds since for $c=\cos^2(\pi/8)$,  $1-c^2-c\sqrt{1-c^2)}<0$. From this dually feasible point, we obtain
    \begin{align}
        1+2c\sqrt{1-c^2}=2\Tr{Y^\star} +\Tr{D^\star}\geq M_{\cMIO,2}(\mU^{\otimes 2}).
    \end{align}
    Lastly, reinserting $c=\cos^2(\pi/8)$ yields 
    \begin{align}
        M_{\cMIO,2}(\mU^{\otimes 2}) \leq 1+2c\sqrt{1-c^2}=\frac{1}{4} \left(4+\sqrt{7+4\sqrt{2}}\right) \approx 1,89 <2,
    \end{align}
    where we note that the inequality is, in fact, an equality, which can be shown from the primal problem. However, to conclude the proof by Thm.~\ref{thm:CoherenceOptimalConversion_main}, we do only need the inequality to see that
    \begin{align}
        \min_{\mS \in \cMIO} \frac{1}{2}\norm{\mS\left[\mU^{\otimes 2}\right]-\mF_2}_\diamond =1-\frac{1}{2} M_{\cMIO,2}\left(\mU^{\otimes 2}\right)>0,
    \end{align}
    i.e., there exists no parallel protocol achieving this conversion exactly.\qedhere
\end{proof}

\subsubsection{Optimal conversion distances to the Fourier transform under maximally free superchannels}\label{sec:ConvMaxFree}

In this section, we provide the optimal conversion distance under maximally free superchannels for dynamical coherence. To this end, we require the following Lemma.

\begin{lem}\label{lem:dualCones}
Let $A,B$ denote arbitrary finite-dimensional quantum systems. Let $\O \in \{ \MIO, \DI, \DIO\}$, and let $\aff\left(\O(A\to B)\right)^*=\left\{X^{AB}: \Tr{X^{AB}J_\mM^{AB}} = 0 \, \forall \mM^{A\to B} \in\O(A\to B)\right\}$ as in Eq.~\eqref{eq:dualCones_appen_aff}. Then, 
\begin{subequations}\label{eq:dualCones_Coherence}
    \begin{align}
      \aff\left(\MIO(A\to B)\right)^*&= \left\{X^{AB}: X^{AB}= P^A\otimes \id^B +\sum_{i=0}^{d_A-1} \ketbra{i}{i}_A\otimes Q_i^B, \Tr{P}=0, \Delta Q_i =0, P,Q_i \in \Herm \right\}, \\
      \aff\left(\DI(A\to B)\right)^*&= \left\{ X^{AB}: X^{AB}= P^A\otimes \id^B +\sum_{i=0}^{d_B-1} Q_i^A \otimes \ketbra{i}{i}_B, \Tr{P}=0, \Delta Q_i =0, P,Q_i \in \Herm \right\}, \\
      \aff\left(\DIO(A\to B)\right)^*&=\left\{X^{AB}: X^{AB}= P^A\otimes \id^B + (\Delta_A-\Delta_B) Q^{AB}, \,\Tr{P}=0, P,Q \in \Herm \right\}.
    \end{align}
\end{subequations}
\end{lem}
\begin{proof}
Define the linear and self-adjoint maps
\begin{align}\label{eq:freeSet_annahilatorMap}
    \mD_{\O}^{AB} =
    \begin{cases}
        \Delta_A\otimes \idChan^B-\Delta_{AB}, \quad \quad \text{if} \, \O=\MIO, \\
        \idChan^A\otimes \Delta_B-\Delta_{AB}, \quad \quad  \text{if} \,\O=\DI, \\
        \Delta_A\otimes \idChan^B-\idChan^A\otimes \Delta_{B},  \text{if}  \,\O=\DIO. \\
    \end{cases}
\end{align}
For every $\O\in \{\MIO,\DI,\DIO\}$ and every quantum channel $\mM^{A\to B}$ we have $ \mM^{A\to B} \in \O(A\to B)$ iff $\mD_{\O}^{AB}(J_\mM^{AB})=0$. We first characterize the affine hull of the Choi operators of free quantum channels, and we claim that
\begin{align}\label{eq:S_1,S_2}
    \{J_Y^{AB} \!\in\! \Herm: Y^{A\to B} \!\in\! \aff(\O(A\to B))\} \! = \!\{J_Y^{AB}\!\in\! \Herm: \partTr{B}{J_Y^{AB}}\!=\!\id^A, \mD_{\O}^{AB}(J_Y^{AB})\!=\!0 \}=:\Xi_{\O},
\end{align}
where $J_Y^{AB}$ denotes the Choi operator of the linear map $Y^{A\to B}$. To see that this holds, take $Y^{A\to B} \!\in\! \aff(\O(A\to B))$ with $Y=\sum_j \xi_j \mM_j$ with $\mM_j^{A\to B} \in \O(A\to B)$, $\xi_j\in \mathbb{R}$, and $\sum_j \xi_j=1$. Then, we have that $\partTr{B}{J_{Y}^{AB}}=\id^A$ and $\mD_{\O}^{AB}(J_{Y}^{AB})=0$, which proves the inclusion from left to right in Eq.~\eqref{eq:S_1,S_2}. Conversely, take any $J_Y^{AB}\in \Xi_{\O}$, and let
\begin{align}
    J_{\mathcal{P}}^{AB}=\id^{AB}/d_B
\end{align}
be the Choi state of the (total) depolarizing channel $\mP^{A\to B} \in \O(A\to B)$. Now $J_{\mathcal{P}}^{AB}>0$, $ \partTr{B}{J_{\mathcal{P}}^{AB}}=\id^A$, and $ \mD_{\O}^{AB}(J_{\mathcal{P}}^{AB})=0$. The operator $J^{AB}:=J_Y^{AB}-J_{\mathcal{P}}^{AB}$ satisfies $\partTr{B}{J^{AB}}=0$ and  $\mD_{\O}^{AB}(J^{AB})=0$. Now let $\lambda > d_B \norm{J^{AB}}_\infty$, then define the positive semidefinite operator
\begin{align}
    J_{\mK}^{AB} := J_{\mathcal{P}}^{AB}+\frac{1}{\lambda} J^{AB} =  \left(1-\frac{1}{\lambda}\right) \frac{\id^{AB}}{d_B}+\frac{1}{\lambda} J_{Y}^{AB}\geq 0,
\end{align}
which follows from $J^{AB}\geq -\norm{J^{AB}}_\infty \id^{AB}$ and the choice of $\lambda$. Additionally, from $\partTr{B}{J_{\mK}^{AB}}=\id^A$ it follows that $\mK^{A\to B}$ defines a channel and $\mD_{\O}^{AB}(J_{\mK}^{AB})=0$ is equivalent to $ \mK^{A\to B} \in \O(A\to B)$. Lastly, since we can write $J_Y^{AB}=\lambda J_{\mK}^{AB}+(1-\lambda) J_{\mP}^{AB}$, we have  $J_{Y}{AB} \in\{J_Y^{AB} \!\in\! \Herm: Y^{A\to B} \!\in\! \aff(\O(A\to B))\}$, and using the {\Choi} isomorphism we know that $Y^{A\to B} \in \aff(\O(A\to B))$. This proves the equality in Eq.~\eqref{eq:S_1,S_2}.

Next, because the Choi representation is linear, a linear functional vanishes on all free Choi operators iff it vanishes on their affine hull. Therefore, 
\begin{align}\label{eq:aff_O_intermediate}
    \aff(\O(A\to B))^* &=\left\{X^{AB} \in \Herm: \Tr{X^{AB}J_{\mM}^{AB}} = 0 \, \forall \mM^{A\to B} \in\O(A\to B)\right\} \nonumber \\
    &= \left\{X^{AB} \in \Herm: \Tr{X^{AB}J_{Y}^{AB}} = 0 \, \forall Y^{A\to B} \in \aff(\O(A\to B))\right\} \nonumber \\
    &\overset{\eqref{eq:S_1,S_2}}{=}\left\{X^{AB} \in \Herm: \Tr{X^{AB}J^{AB}} = 0 \, \forall J^{AB} \in \Xi_{\O}\right\}
\end{align}
More technically, $X^{AB}\in \left\{X^{AB} \in \Herm: \Tr{X^{AB}J_{\mM}^{AB}} = 0 \, \forall \mM^{A\to B} \in\O(A\to B)\right\}$, take $Y^{A\to B}=\sum_j \xi_j \mM_j^{A\to B} \in \aff(\O(A\to B))$, then $\Tr{X^{AB} J_{Y}^{AB}}=\sum_j \xi_j \Tr{ X^{AB} J_{\mM_j}^{AB} }=0$ since $\mM_j^{A\to B} \in \O(A\to B) \, \forall j$. Conversely, if $\Tr{X^{AB} J_{Y}^{AB}}=0$ for all  $Y^{A\to B}\in\aff(\O(A\to B))$, in particular it must vanish for all  $Y^{A\to B} \in \O(A\to B)$ itself. Next, define the subspaces
\begin{align}
    U:=\{J_{Y}^{AB} \in  \Herm: \partTr{B}{J_{Y}^{AB}}=0 \}, \quad \text{and} \quad W_{\O}:=\{J_{Y}^{AB} \in \Herm: \mD_{\O}^{AB}(J_{Y}^{AB})=0 \}.
\end{align}
Since $J_{\mathcal P}^{AB}\in \Xi_{\O}$, Eq.~\eqref{eq:S_1,S_2} yields
\begin{align}
     \Xi_{\O}=J_{\mathcal P}^{AB}+ \left(U\cap W_{\O}\right),
\end{align}
where $ J_{\mathcal P}^{AB}+\left(U\cap W_{\O}\right):=\left\{J_{\mathcal P}^{AB}+J_{Y}^{AB}:J_{Y}^{AB}\in U\cap W_{\O}\right\}$. Combining this with Eq.~\eqref{eq:aff_O_intermediate} yields
\begin{align}\label{eq:annihilator_conditions}
    X^{AB}\in\aff\left(\O(A\to B)\right)^*\quad\Longleftrightarrow\quad\Tr{X^{AB}\left(J_{\mathcal P}^{AB}+J_{Y}^{AB}\right)}=0\quad \forall\,J_{Y}^{AB}\in U\cap W_{\O},
\end{align}
We next rewrite the condition on the right. Taking $J_{Y}^{AB}=0 \in U\cap W_{\O}$  implies $\Tr{X^{AB}J_{\mathcal P}^{AB}}=0$. Eq.~\eqref{eq:annihilator_conditions} then implies 
\begin{align}
    0=\Tr{X^{AB}\left(J_{\mathcal P}^{AB}+J_{Y}^{AB}\right)}=\Tr{X^{AB}J_{\mathcal P}^{AB}}+\Tr{X^{AB}J_{Y}^{AB}}=\Tr{X^{AB}J_{Y}^{AB}} \quad J_{Y}^{AB}\in U\cap W_{\O},
\end{align}
which is equivalent to $X^{AB} \in(U\cap W_{\O})^\perp $, where the orthogonal complements are taken with respect to the Hilbert–Schmidt inner product. This, together with $\Tr{X^{AB}J_{\mathcal P}^{AB}}=0$ clearly imply Eq.~\eqref{eq:annihilator_conditions}, and thus we can equivalently write
\begin{align}
    X^{AB}\in\aff\left(\O(A\to B)\right)^*\quad\Longleftrightarrow \quad \Tr{X^{AB}J_{\mathcal P}^{AB}}=0 \quad \text{and} \quad X^{AB} \in(U\cap W_{\O})^\perp.
\end{align}
Since we are considering finite-dimensional spaces, we have $\left(U\cap W_{\O}\right)^\perp=U^\perp+W_{\O}^\perp= \{u+v: u \in U^\perp, v\in W_{\O}^\perp \}$. Using the adjoint of the partial trace, and the self-adjointness of $\mD_{\O}^{AB}$, we obtain
\begin{align}
    U^\perp=\left\{P^A\otimes\id^B: P^A\in\Herm \right\} \quad \text{and} \quad W_{\O}^\perp=\left\{\mD_{\O}^{AB}(Q^{AB}):Q^{AB}\in\Herm\right\}.
\end{align}
It follows that $X^{AB}\in\left(U\cap W_{\O}\right)^\perp\quad\Longleftrightarrow\quad X^{AB}=P^A\otimes\id^B+\mD_{\O}^{AB}(Q^{AB})$ for some $P^A,Q^{AB}\in\Herm$. It remains to impose the condition that $\Tr{X^{AB}J_{\mathcal P}^{AB}}=0$. Using $J_{\mathcal P}^{AB}=\id^{AB}/d_B$ and $\mD_{\O}^{AB}(J_{\mathcal P}^{AB})=0$, we find
\begin{align}
    0\!=\!\Tr{X^{AB}J_{\mathcal P}^{AB}}\!=\!\Tr{\left(P^A\otimes\id^B\right)J_{\mathcal P}^{AB}} \!+\!\Tr{\mD_{\O}^{AB}(Q^{AB})J_{\mathcal P}^{AB}}\!=\!\Tr{P^A}\!+\!\Tr{Q^{AB}\mD_{\O}^{AB}(J_{\mathcal P}^{AB})} \!=\! \Tr{P^A}.
\end{align}
Hence,
\begin{align}\label{eq:unified_dual_characterization}
    \aff\left(\O(A\to B)\right)^*=\left\{ X^{AB}\in \Herm:X^{AB} = P^A\otimes\id^B +  \mD_{\O}^{AB}(Q^{AB}),  \ \Tr{P^A}=0, \ P^A,Q^{AB}\in\Herm\right\}.
\end{align}
Lastly, the block-diagonal form in the MIO and DI case in Eq.~\eqref{eq:dualCones_Coherence} follows directly from Eq.~\eqref{eq:freeSet_annahilatorMap}.
\end{proof}

For the proof of the following Propositions, recall that the relevant monotones are given by 
 
\begin{subequations}
\begin{align}
    M_{\cMIO,d}(\mE) &=\max\big\{\Tr{X \mE(\ketbra{j}{j})} : 0\leq X\leq d\id, \Delta X=\id, 0\le j \le d_A\!-\!1\big\} \label{eq:C_D_appen2} , \\
    M_{\mMIO,d}(\mE) &=\max\Big\{\sum_{i=0}^{d_A-1} \Tr{X_i\mE(\ketbra{i}{i}_A)}: p,q\in \Prob_{d_A}, 0\leq X_i\leq d q_i \id, \Delta X_i=p_i \id \Big \}, \label{eq:M_mMIO,d_appen} \\
    M_{\cDI,d}(\mE) &= \max \Big\{ \sum_{i=0}^{d_B-1}\braket{i|\mE(\rho_i)| i}_2 \!:\! \rho_i,\sigma \in \D, \rho_i\! \leq\! d \sigma,\Delta \rho_i \!=\!\Delta \sigma \, \forall i: 0\leq i \leq d_B\!-\!1\Big\}, \label{eq:D_d_appen2} \\
    M_{\mDI,d}(\mE) &= \max \Big\{\!\sum_{i=0}^{d_B-1} \!\braket{i|\mE(\rho_i)| i}_2 \!:\! \rho_i,\sigma \!\in \!\D, \rho_i\! \leq\! d \sigma,\Delta \rho_i \!=\!\Delta \rho_j \, \forall i,j: 0\!\leq\! i,j \!\leq \!d_B\!-\!1\Big\}. \label{eq:M_mDI,d_appen}
\end{align}
\end{subequations}

\ConversionMaximalSet*
\begin{proof}
For MIO, Thm.~\ref{thm:DiamondBound_Bounds_G} and Lem.~\ref{cor:R_Fd} reduce the claim to a reformulation of $G_{\aff(\MIO)}$. For DI, we first derive a direct bound on the diamond distance and worst-case fidelity, and then construct a superchannel attaining it. For $\O= \MIO$, we have 
\begin{align}
    \min_{\mS \in \mMIO_1}\tfrac{1}{2}\norm{\mS[\mE]\!-\!\mF_d}_\diamond \!\overset{\ref{thm:DiamondBound_Bounds_G}}{=}\! 1\!-\!\max_{\mS \in \mMIO_1} F_{\min}(\mS[\mE],\mF_d)\!\overset{\ref{thm:DiamondBound_Bounds_G}}{=}\! 1\!-\!G_{\aff(\MIO)}(\mE,R_{\max,\MIO}(\mF_d)) \overset{\ref{cor:R_Fd}}{=} 1\!-\!G_{\aff(\MIO)}(\mE,d),
\end{align}
where $G_{\aff(\MIO)}$ is defined in Eq.~\eqref{eq:G_aff_O}. It remains to show that $G_{\aff(\MIO)}(\mE,d)=\tfrac{1}{d}M_{\mO_1,d}(\mE)$. Invoking Lem.~\ref{lem:dualCones} yields
    \begin{align}\label{eq:G_O_affMIO}
    G_{\aff(\MIO)}(\mE,d)\! &= \!\sup \big\{ \!\Tr{WJ_\mE}: 0\leq W \leq \rho\otimes \id, \rho \in \D, \tfrac{ \id}{d_A}-d W \in \aff(\MIO)^*\big\}  \\
    &= \! \sup \big\{\! \Tr{WJ_\mE}: 0\!\leq \!W \!\leq\! \rho\otimes \id,\rho \in \D, \tfrac{ \id}{d_A}\!-\!d W \!=\! P\otimes \id \!+\!\sum_{i=0}^{d_A-1} \!\ketbra{i}{i}\otimes Q_i, \Tr{P}=0, \Delta Q_i\!=0\big\} \nonumber,
\end{align}
where we suppressed system indices being ordered as $A$, $B$ throughout. Now let $W$ be any feasible point of Eq.~\eqref{eq:G_O_affMIO}. Then, define
\begin{align}
    \tilde{W}:=\Delta_A W =\frac{1}{d}\left(\left(\frac{\id}{d_A}-\Delta(P)\right)\otimes \id -\sum_{i=0}^{d_A-1} \ketbra{i}{i}\otimes Q_i \right),
\end{align}
which satisfies $W-\tilde{W}= -\frac{1}{d}\left(P-\Delta(P)\right) \otimes \id$. Hence, from
\begin{align}
    \Tr{(W-\tilde{W})J_\mE}=-\frac{1}{d}\Tr{\left(P-\Delta(P)\right) \otimes \id J_\mE }= -\frac{1}{d}\Tr{P-\Delta(P)}=0,
\end{align}
where we used that $\mE$ is a channel and $\partTr{B}{J_{\mE}^{AB}}=\id^A$, it follows that the objective function remains invariant. Moreover, since $0\leq \tilde{W}\leq \Delta(\rho) \otimes \id$ because $\Delta$ is order-preserving, we can relax the optimization in Eq.~\eqref{eq:G_O_affMIO} to
    \begin{align}
    G_{\aff(\MIO)}(\mE,d)\mkern-4mu\leq  \sup \Bigr\{\!\Tr{\tilde{W}J_\mE}\!: \tilde{W}\mkern-4mu=\!\tfrac{1}{d}\!\Big(\!\tfrac{\id}{d_A}\!-\!\Delta P\otimes \id \!-\mkern-4mu\sum_{i=0}^{d_A-1}\mkern-4mu\ketbra{i}{i}\otimes Q_i \Big), 0\!\leq\! \tilde{W}\!\leq \Delta\rho\otimes \! \id, \Tr{P}\!=\!0, \Delta Q_i\!=\!0\Bigr\}.
\end{align}
Without loss of generality, for some $\alpha_i\in\mathbb{R}$ such that $\sum_i \alpha_i=0$, and a probability distribution $q_i$ take 
\begin{align}
    P=\sum_i \alpha_i \ketbra{i}{i}, \quad \text{and} \quad \Delta(\rho)=\sum_i q_i \ketbra{i}{i}.
\end{align}
 Then, let $ p_i=\frac{1}{d_A}-\alpha_i$ allows us to rewrite any feasible $\tilde{W}$ as $\tilde{W} =\frac{1}{d}\sum_i \ketbra{i}{i} \otimes \left( \id \left( \frac{1}{d_A}-\alpha_i\right)-Q_i\right) = \frac{1}{d}\sum_i \ketbra{i}{i}\otimes X_i$ with $X_i:= p_i \id -Q_i$. Thus, for any feasible $\tilde{W}$, we find
 \begin{align}
     \Tr{\tilde{W}J_\mE}=\frac{1}{d} \sum_{i=0}^{d_A-1} \Tr{X_i \mE(\ketbra{i}{i})}.
 \end{align}
 We now show that the $\{X_i\}$ form a feasible point of Eq.~\eqref{eq:M_mMIO,d_appen}. The constraint $0\leq \tilde{W}\leq \Delta(\rho)\otimes \id$ is equivalent to $0\leq X_i \leq d q_i \id$. Because $\Delta Q_i =0$, it follows that $\Delta X_i=p_i \id$. Moreover, $X_i \geq 0$ implies $p_i\geq 0$ and $\Tr{P}=0$ implies $\sum_i p_i=\sum_i \left(\frac{1}{d_A}-\alpha_i\right)=1$. As such, both $p_i,q_i$ are probability distributions, and thus
 \begin{align}\label{eq:G_upperBound}
     G_{\aff(\MIO)}(\mE,d)\leq \frac{1}{d}M_{\mMIO,d}(\mE).
 \end{align}
Conversely, let $(\{p_i\},\{q_i\},\{X_i\})$ be feasible in Eq.~\eqref{eq:M_mMIO,d_appen}. Then, define
\begin{align}
    \rho=\sum_{i=0}^{d_A-1} q_i \ketbra{i}{i}, \quad P:=\sum_{i=0}^{d_A-1} \left(\tfrac{1}{d_A}-p_i\right)\ketbra{i}{i} ,\quad  Q_i=p_i\id-X_i, \quad \text{and} \quad W= \frac{1}{d} \sum_{i=0}^{d_A-1}  \ketbra{i}{i} \otimes X_i.
\end{align}
Then, $0\leq X_i \leq d q_i\id$ implies $0\leq W\leq \rho\otimes \id$. And from $\Tr{P}=0$ and $\Delta (Q_i)=0$ (together with Lem.~\ref{lem:dualCones}) we have $\tfrac{ \id}{d_A}-d W \in \aff(\MIO)^*$, and as such $\frac{1}{d}M_{\mMIO,d}(\mE) \leq G_{\aff(\MIO)}(\mE,d)$, which, together with Eq.~\eqref{eq:G_upperBound}, finishes the proof for $\O=\MIO$.

We now turn to the case of $\O=\DI$. Let $\mS_1^{(A\to B) \to (C\to D)}\in \mDI_1$ with $d_C=d_D=d$. Let $\rho_i=F_d^\dagger \ketbra{i}{i}F_d$ for $i\in\{0,\ldots,d-1\}$. Using the sharpened Fuchs-van de Graaf inequality in Eq.~\eqref{eq:FuchsVanDeGraaf_channels_sharp_appen} gives
\begin{align}\label{eq:DI_mDI_bound}
    \tfrac{1}{2}\norm{\mS[\mE]-\mF_d}_\diamond&\geq  1- F_{\min}(\mS[\mE],\mF_d) =1- \inf_\rho F(\idChan\otimes \mS[\mE](\rho),\idChan\otimes\mF_d(\rho)) \geq 1-\Tr{\mS[\mE](\rho_i) \mF_d(\rho_i)} \nonumber \\
    &= 1-\braket{i|\mS[\mE](\rho_i)|i},
\end{align}
where we used that $\mF_d$ is unitary and that $\mF_d(\rho_i)=\ketbra{i}{i}$. This holds for any $i\in\{0,\ldots,d-1\}$, and we can thus average over all of them to obtain $\tfrac{1}{2}\norm{\mS[\mE]-\mF_d}_\diamond\geq 1-\tfrac{1}{d}\sum_{i=0}^{d-1} \braket{i|\mS[\mE](\rho_i)|i}$. Define
\begin{align}
    T^{CD}=  \frac{1}{d}\sum_{i=0}^{d-1} (\rho_i^C)^T\otimes \ketbra{i}{i}_D,\qquad \text{and} \qquad W_{\mS}^{AB}:= \left(\partTr{CD}{J_{\mS}^{ABCD}(\id^{AB}\otimes T^{CD})}\right)^T.
\end{align}
A straightforward calculation reveals that
\begin{align}\label{eq:S_i_detection}
    \frac{1}{d}\sum_{i=0}^{d-1} \braket{i|\mS[\mE](\rho_i)|i} &= \frac{1}{d}\sum_{i=0}^{d-1} \Tr{\left((\rho_i^C)^T\otimes \ketbra{i}{i}_D\right) J_{\mS[\mE]}^{CD}} = \Tr{T^{CD} J_{\mS[\mE]}^{CD}} = \Tr{J_{\mS}^{ABCD}\left(T^{CD}\otimes (J_{\mE}^{AB})^T\right)} \nonumber \\
    &=\Tr{W_{\mS}^{AB} J_{\mE}^{AB}}.
\end{align}
As such, we can rewrite the bound in Eq.~\eqref{eq:DI_mDI_bound} as
\begin{align}
    \tfrac{1}{2}\norm{\mS[\mE]-\mF_d}_\diamond&\geq 1- \Tr{W_{\mS}^{AB} J_{\mE}^{AB}}.
\end{align}
Next, we show that this $W$ is feasible in
\begin{align}\label{eq:G-aff_DI_rest}
    G_{\aff(\DI)}(\mE^{A\to B}, d)\!=\!\sup\left\{ \Tr{J_\mE^{AB} W^{AB}}\!:\! 0\!\leq\! W^{AB} \!\leq\! \rho^A\otimes \id^B, \rho\!\in\! \D, \Tr{W^{AB} J_\mM^{AB}}\!=\!\tfrac{1}{d} \,\forall \mM^{A\to B} \in \O(A\!\to\! B)\right\}.
\end{align}
To this end, first note that $\partTr{CD}{J_{\mS}^{ABCD}(\id^{AB}\otimes T^{CD})}=\partTr{CD}{(\id^{AB}\otimes \sqrt{T^{CD}})J_{\mS}^{ABCD}(\id^{AB}\otimes \sqrt{T^{CD}})} \geq 0 $, and thus, $W\geq 0$. Moreover, we have that $T^{CD}\leq \tfrac{1}{d} \id^{CD}$, and thus
\begin{align}
    W^{AB}\leq \left(\partTr{CD}{J_{\mS}^{ABCD}(\id^{AB}\otimes \tfrac{1}{d} \id^{CD})}\right)^T = \frac{\partTr{CD}{J_{\mS}^{ABCD}}^T}{d}= \tfrac{1}{d}\partTr{C}{M^{CA}}^T \otimes \id^B=\sigma^A\otimes \id^B,
\end{align}
where we used that $\partTr{D}{J_{\mS}^{ABCD}}=M^{CA}\otimes \id^B$ for some $\partTr{A}{M^{CA}}=\id^C$, and we can thus introduce the state $\sigma^A:=\tfrac{1}{d}\partTr{C}{M^{CA}}^T$. To show that $\Tr{J_{\mM}^{AB} W_{\mS}^{AB}}=\frac{1}{d} \, \forall \mM\in \DI$, note that we can use Eq.~\eqref{eq:S_i_detection} to express (simply using an channel $\mM$ instead of $\mE$) this as
\begin{align}
    \Tr{J_{\mM}^{AB} W_{\mS}^{AB}}= \frac{1}{d}\sum_{i=0}^{d-1} \braket{i|\mS[\mM](\rho_i)|i}= \frac{1}{d}\sum_{i=0}^{d-1} \braket{i|\tilde{\mM}(\Delta \rho_i)|i}= \frac{1}{d}\sum_{i=0}^{d-1} \braket{i|\tilde{\mM}(\tfrac{1}{d}\id)|i}= \frac{1}{d}
\end{align}
where we used that $\tilde{\mM}:=\mS[\mM] \in \DI$ since $\mS \in \mDI_1$, and $\Delta(\rho_i)=\frac{1}{d}\id$. Thus, this $W_{\mS}^{AB}$ is feasible in Eq.~\eqref{eq:G-aff_DI_rest}, and $ G_{\aff(\DI)}(\mE, d)\geq \Tr{W_{\mS}^{AB} J_{\mE}^{AB}}$. Therefore, combining this with Eq.~\eqref{eq:DI_mDI_bound} and the fact that we can take the minimum over all free superchanels yields
\begin{align}
    \min_{\mS\in \mDI_1}\tfrac{1}{2}\norm{\mS[\mE]-\mF_d}_\diamond&\geq 1- G_{\aff(\DI)}(\mE,d).
\end{align}
Conversely, define 
\begin{align}
    J_{\mK}^{CD} := (\id^C\otimes F_d^D) \tilde{\Pi}^{CD}(\id^C\otimes F_d^D)^\dagger \qquad \text{and} \qquad J_{\mQ}^{CD}:=\frac{dJ_{\mK}^{CD}-J_{\mF_d}^{CD}}{d-1}
\end{align}
where $\tilde{\Pi}^{CD}$ denotes the unnormalized maximally correlated state as in Eq.\eqref{eq:mG}. These are Choi states since the corresponding channel is $\mK=\mF_d\circ\Delta$, and as shown around Eq.~\eqref{eq:mG} we have $J_{\mG}^{CD}\geq 0$ and $(d-1) J_{\mG}^{CD}=dJ_{\mK}-J_{\mF_d}$ . Additionally, $\partTr{D}{J_{\mQ}}=\id^C$. Now, take a feasible point $(W,\sigma)$ in Eq.~\eqref{eq:G-aff_DI_rest}. Now define the superchannel
\begin{align}
    J_{\mS}^{ABCD} = (W^{AB})^T \otimes J_{\mF_d}^{CD} +(\sigma^A\otimes \id^B-W^{AB})^T\otimes J_{\mQ}^{CD}.
\end{align}
For any channel $\mL$, the superchannel acts as
\begin{align}
    \mS[\mL] &=\Tr{W^{AB} J_{\mL}^{AB}} \mF_d\!+\!\left(\Tr{(\sigma^A\otimes \id^B)J_{\mL}^{AB}}\!-\!\Tr{W^{AB} J_{\mL}^{AB}}\right) \mQ \nonumber \\
    &=\Tr{W^{AB} J_{\mL}^{AB}} \mF_d+\left(1-\Tr{W^{AB} J_{\mL}^{AB}}\right) \mQ.
\end{align}
In particular, for any $\mM\in \DI$ we have $\Tr{W^{AB} J_{\mM}^{AB}} =\frac{1}{d}$, and thus, $\mS[\mM]=\tfrac{1}{d}(\mF_d+(d-1)\mQ) =\mK \in \DI$. Additionally, we have
\begin{align}
    \min_{\mS\in \mDI_1}\tfrac{1}{2}\norm{\mS[\mE]-\mF_d}_\diamond \leq \tfrac{1}{2} \left(1-\Tr{W^{AB} J_{\mE}^{AB}}\right) \norm{\mQ-\mF_d}_\diamond \leq  1-\Tr{W^{AB} J_{\mE}^{AB}},
\end{align}
and taking the infimum over all feasible points $(W,\sigma)$ yields $\min_{\mS\in \mDI_1}\tfrac{1}{2}\norm{\mS[\mE]-\mF_d}_\diamond\leq 1- G_{\aff(\DI)}(\mE,d)$.

Next, using Lem.~\ref{lem:dualCones} yields
\begin{align}
    G_{\aff(\DI)}(\mE,d) \!&=\! \sup\! \Big\{ \Tr{WJ_\mE}: 0\leq W \leq \sigma\otimes \id, \tfrac{ \id}{d_A}-d W \in \aff(\DI)^*,\sigma \in \D\Big\}  \nonumber \\
    &= \max \Big\{\!\Tr{WJ_\mE}:\!0\!\leq\! W \!\leq \sigma\otimes \id, \frac{\id}{d_A}\!-\!d W\!=\!P\!\otimes\!\id\!+\!\sum_{i=0}^{d_B-1}\!Q_i\!\otimes\!\ketbra{i}{i}, \Tr{P}\!=\!0, \Delta Q_i\!=\!0,\sigma \in \D\Big\},
\end{align}
where $P,Q \in \Herm$. This allows us to write 
\begin{align}
    W=\frac{1}{d}\sum_i \left(\frac{\id}{d_A}-P-Q_i\right) \otimes \ketbra{i}{i} =: \frac{1}{d}\sum_i \rho_i^T \otimes \ketbra{i}{i},
\end{align}
where the $\rho_i$ are states since $W\geq 0$ implies non-negativity and $\Tr{\rho_i}=1-\Tr{P}-\Tr{Q_i}=1$. Moreover,  $\Delta(\rho_i)=\frac{\id}{d_A}-\Delta(P)$, i.e., the diagonals of all $\rho_i$ are identical. Moreover, the constraint $0\leq W\leq \sigma \otimes \id$ becomes equivalent to $0\leq \rho_i \leq d\sigma^T \, \forall i$. Thus,
    \begin{align}
    G_{\aff(\DI)}(\mE,d)&= \frac{1}{d} \max \left\{ \sum_{i=0}^{d_B-1} \braket{i|\mE(\rho_i)|i}: 0\leq \rho_i \leq d\sigma^T, \Delta \rho_i=\Delta \rho_j , \rho_i,\sigma \in \mD \, \forall i,j\right\} = \frac{1}{d} M_{\mDI,d}(\mE),
\end{align}
where the last equality follows from substituting the variable $\sigma$ with its transpose.\qedhere
\end{proof}

Lastly, we show that the conversion distances under compatible superchannels and maximally free superchannels differs.

\DIConversionDistance*
\begin{proof}
In the $\O =\MIO$ case, let $d_A=2$ and $d_B=3$. Furthermore, let $\ket{\psi}=\tfrac{1}{\sqrt{3}}(\ket{0}+\ket{1}+\ket{2})$, and define the ``measure-and-prepare" channel
    \begin{align}
        \mE^{A\to B}(Y)=\braket{0|Y|0} \tau_0^B+\braket{1|Y|1} \tau_1^B, \quad \text{where} \quad  \tau_0^B=\frac{\id^B}{6}+\frac{\ketbra{\psi}{\psi}_B}{2}, \quad \text{and} \quad \tau_1^B=\frac{\id^B-\ketbra{\psi}{\psi}_B}{2}.
    \end{align}
We will now prove the claim for this specific channel. To this end, recall that for $M_{\cMIO,2}(\mE)$ in Eq.~\eqref{eq:C_D_appen2} we have $0\leq X^B\leq 2\id^B$ and $\Delta X^B=\id^B$. Therefore, $\Tr{X}=3$ and $0\leq\braket{\psi|X^B|\psi}\leq 2$, and thus
\begin{align}
    M_{\cMIO,2}(\mE) =\max \left\{ \max\left\{\frac{1+\braket{\psi|X^B|\psi}}{2}, \frac{3-\braket{\psi|X^B|\psi}}{2}\right\}: 0\leq X^B\leq 2\id^B, \Delta X^B=\id^B\right\} \leq \frac{3}{2}.
\end{align}
On the other hand, consider $M_{\mMIO,2}(\mE)$ and choose the variables $p=(\tfrac{1}{3},\tfrac{2}{3})$, $q=(\tfrac{1}{2},\tfrac{1}{2})$, $X_0^B=\ketbra{\psi}{\psi}_B$, and $X_1^B= \id^B-\ketbra{\psi}{\psi}_B$. Then, $0\leq X_i^B \leq 2q_i \id^B$, $\Delta X_0^B=\tfrac{1}{3} \id^B$, and $\Delta X_1^B=\tfrac{2}{3} \id^B$. Thus, using Eq.~\eqref{eq:M_mMIO,d_appen} yields
\begin{align}
    M_{\mMIO,2}(\mE) \geq \Tr{X_0^B\tau_0^B}+\Tr{X_1^B\tau_1^B}= \frac{2}{3}+1 =\frac{5}{3} > \frac{3}{2} \geq M_{\cMIO,2}(\mE).
\end{align}

For $\O=\DI$, we proceed in a similar manner. First, we construct a specific channel and then we bound the relevant monotones on it. To this end, let $d_A=4$, $d_B=3$, and 
\begin{align}
    \ket{\phi_0}=\tfrac{1}{\sqrt{2}} (\ket{0}+\ket{1}), \quad \ket{\phi_1}=\tfrac{1}{\sqrt{2}} (\ket{0}-\ket{1}),\quad \ket{\phi_2}=\tfrac{1}{\sqrt{2}} (\ket{2}+\ket{3}), \quad \ket{\phi_3}=\tfrac{1}{\sqrt{2}} (\ket{2}-\ket{3}).
\end{align}
Set $\Pi_i^A=\ketbra{\phi_i}{\phi_i}_A$ and define $K^A:= 2 (\ketbra{\phi_1}{\phi_3}_A+\ketbra{\phi_3}{\phi_1}_A)$ and $L^A:=\Pi_0^A-\Pi_1^A-\Pi_2^A+\Pi_3^A$. Define 
\begin{align}
    M_0^A:= \frac{\id^A}{3}+\frac{K^A+L^A}{8}, \quad M_1^A:=\frac{\id^A}{3}+\frac{L^A-K^A}{8}, \quad \text{and} \quad  M_2^A:= \frac{\id^A}{3}-\frac{L^A}{4}.
\end{align}
Since $(L^A)^2 = I^A $, $ (K^A)^2 = 4(\Pi_1+\Pi_3)$, and $KL+LK = 0$,
we have
\begin{align}
 (L^A\pm K^A)^2 &= (L^A)^2+(K^A)^2\pm(L^AK^A+K^AL^A)=I^A+4(\Pi_1^A+\Pi_3^A)\leq5I^A.
\end{align}
As $L^A\pm K^A$ are Hermitian, this implies $-\sqrt{5}\,I^A\leq L^A\pm K^A\leq\sqrt{5}\,I^A$, and therefore $ M_j\geq\left(\tfrac13-\tfrac{\sqrt{5}}8\right)I^A\geq0,\quad j\in\{0,1\}$ as well as $M_2^A=\frac{I^A}{3}-\frac{L^A}{4}\geq\frac{1}{12}I^A\geq0$, following from $(L^A)^2=I^A$, and thus $-I^A\leq L^A\leq I^A$. Together with $M_0^A+M_1^A+M_2^A=I^A$, this shows that $\{M_0,M_1,M_2\}$ is a POVM. Now define the channel
\begin{align}
    \mE^{A\to B}(X^A):= \sum_{i=0}^{2} \Tr{M_i^A X^A} \ketbra{i}{i}_B.
\end{align}
Next, we provide a lower bound on $M_{\mDI,2}(\mE)$. Consider the states
\begin{align}
 \sigma &=\frac9{44}(\Pi_0+\Pi_2)+\frac{13}{44}(\Pi_1+\Pi_3), \quad \text{and}\quad \rho_0 =\frac9{22}\Pi_0+\frac1{22}\bigl(2\ket{\phi_1}+3\ket{\phi_3}\bigr)\bigl(2\bra{\phi_1}+3\bra{\phi_3}\bigr), \\
 \rho_1 &=\frac9{22}\Pi_0+\frac1{22}\bigl(2\ket{\phi_1}-3\ket{\phi_3}\bigr)\bigl(2\bra{\phi_1}-3\bra{\phi_3}\bigr), \quad \text{and}\quad \rho_2 =\frac{13}{22}\Pi_1+\frac9{22}\Pi_2.
\end{align}
To verify the domination constraints $2\sigma\geq \rho_i$, note that
\begin{align}
 2\sigma\!-\!\rho_0&\!=\!\frac9{22}\Pi_2+\frac{13}{22}(\Pi_1\!+\!\Pi_3)-\frac1{22}\bigl(2\ket{\phi_1}+3\ket{\phi_3}\bigr)\bigl(2\bra{\phi_1}+3\bra{\phi_3}\bigr) \nonumber \\
 &=\frac9{22}\Pi_2+\frac1{22}\bigl(9\Pi_1+4\Pi_3-6\ket{\phi_1}\!\bra{\phi_3}-6\ket{\phi_3}\!\bra{\phi_1}\bigr)\!=\!\frac9{22}\Pi_2\!+\!\frac1{22}\bigl(3\ket{\phi_1}\!-\!2\ket{\phi_3}\bigr)\bigl(3\bra{\phi_1}\!-\!2\bra{\phi_3}\bigr)\!\geq\!0,
\end{align}
where we used $\Pi_j=\ket{\phi_j}\!\bra{\phi_j}$. Analogously, by reversing the signs of the off-diagonal terms we find $2\sigma-\rho_1\geq0$. Additionally,  $2\sigma-\rho_2=\frac9{22}\Pi_0+\frac{13}{22}\Pi_3\geq0$. Thus, $\rho_i\leq2\sigma$ for all $i\in\{0,1,2\}$, and clearly $\Delta \rho_0 =\Delta \rho_1 =\Delta \rho_2$. A straightforward calculation yields $\Tr{K\rho_0} =\frac{12}{11}$, $ \Tr{L\rho_0} =\frac{7}{11}$, $\Tr{K\rho_1} =-\frac{12}{11}$, $\Tr{L\rho_1} =\frac{7}{11}$, and $\Tr{L\rho_2} =-1$. As such, $\Tr{M_0\rho_0} =\frac{1}{3}+\frac{1}{8}\left(\frac{12}{11}+\frac{7}{11}\right)=\frac{145}{264}$, $\Tr{M_1\rho_1} =\frac{1}{3}+\frac{1}{8}\left(\frac{12}{11}+\frac{7}{11}\right)=\frac{145}{264}$, and $\Tr{M_2\rho_2} =\frac{1}{3}-\frac{1}{4}(-1)=\frac{7}{12}$. This yields
\begin{align}\label{eq:M_mDI_bound}
 M_{\mDI,2}(\mE)
 &\geq\sum_{i=0}^{2}\bra{i}\mE_(\rho_i)\ket{i}
 =\sum_{i=0}^{2}\Tr{M_i\rho_i}
 =\frac{145}{264}+\frac{145}{264}+\frac7{12}
 =\frac{37}{22}.
\end{align}

Next, we provide an upper bound to $M_{\cDI,2}(\mE)$. To this end, we switch Eq.~\eqref{eq:D_d_appen2} to its dual problem. Introduce the dual variables $Y_i\geq 0$ for $d\sigma-\rho_i\geq 0$, diagonal Hermitian $H_i=\Delta(H_i)$ for $\Delta(\rho_i)-\Delta(\sigma)=0$, and $\lambda_i,\mu\in\mathbb R$ for the trace constraints. Define $M_i=\mE^\dagger(\ketbra{i}{i})$, then the Lagrangian is
\begin{align}
L(\sigma,\{\rho_i\}; \{Y_i\}, \{\lambda_i \}, \mu)&=\sum_i \Tr{M_i\rho_i}
+\sum_i \operatorname{Tr}\!\left[Y_i(d\sigma-\rho_i)\right]+\sum_i \Tr{H_i\bigl(\Delta(\rho_i)-\Delta(\sigma)}\bigr)+\sum_i\lambda_i\left(1-\Tr{\rho_i} \right) \nonumber  \\
& \qquad +\mu\left(1-\Tr{\sigma}\right) \nonumber\\
&=\mu+\sum_i\lambda_i+\sum_i\Tr{\left(M_i-Y_i+H_i-\lambda_i \id\right)\rho_i}+\Tr{\left(d\sum_iY_i-\sum_iH_i-\mu \id\right)\sigma},
\end{align}
where we used that each $H_i$ is diagonal. Taking the supremum over $\rho_i\geq 0$ and $\sigma\geq 0$ is finite iff
\begin{align}
M_i-Y_i+H_i-\lambda_i \id\leq0, \quad\text{and}\quad d\sum_iY_i-\sum_iH_i-\mu \id \leq0.
\end{align}
Thus, the dual problem is (all sums over $i$ run over $0\leq i\leq d_2-1$)
\begin{align}
    \min\left\{\!\mu\!+\!\sum_{i}\!\lambda_i: \lambda_i \id\!+\!Y_i\!-H_i\! \geq M_i, \mu \id\geq d\sum_{i}\!Y_i-\sum_i H_i, Y_i \geq 0, H_i=\Delta(H_i), H_i \in \Herm ,\lambda_i,\mu \in \mathbb{R}\right\},
\end{align}
where strong duality holds by Slater's conditions since $\rho_i=\sigma=\frac{\id}{d_A} \,\forall i$ is strictly feasible in the primal problem. In the dual problem, now choose $\lambda_i=0$ for every $i$, $\mu=13/8$, together with
\begin{align}
 Y_0=\frac1{16}\left[\Pi_0+\bigl(\ket{\phi_1}+2\ket{\phi_3}\bigr)\bigl(\bra{\phi_1}+2\bra{\phi_3}\bigr)\right],\quad Y_1=\frac1{16}\left[\Pi_0+\bigl(\ket{\phi_1}-2\ket{\phi_3}\bigr)\bigl(\bra{\phi_1}-2\bra{\phi_3}\bigr)\right], \quad Y_2=\frac12\Pi_2,
\end{align}
and
\begin{align}
 H_0=H_1&=-\frac{19}{48}(\Pi_0+\Pi_1)-\frac{13}{48}(\Pi_2+\Pi_3),\quad \text{and} \quad H_2=-\frac7{12}(\Pi_0+\Pi_1)-\frac1{12}(\Pi_2+\Pi_3).
\end{align}
The operators $Y_i$ are positive by construction. Moreover, the $H_i$ are diagonal in the computational basis because
$\Pi_0+\Pi_1=\ket{0}\!\bra{0}+\ket{1}\!\bra{1}$ and
$\Pi_2+\Pi_3=\ket{2}\!\bra{2}+\ket{3}\!\bra{3}$. A direct calculation yields
\begin{subequations}
\begin{align}
 Y_0-H_0-M_0&=\frac1{16}\Pi_2+\frac1{16}\bigl(2\ket{\phi_1}-\ket{\phi_3}\bigr)\bigl(2\bra{\phi_1}-\bra{\phi_3}\bigr)\geq0,\\
 Y_1-H_1-M_1&=\frac1{16}\Pi_2+\frac1{16}\bigl(2\ket{\phi_1}+\ket{\phi_3}\bigr)\bigl(2\bra{\phi_1}+\bra{\phi_3}\bigr)\geq0,\\
 Y_2-H_2-M_2&=\frac12\Pi_0\geq0.
\end{align}
\end{subequations}
Furthermore, 
\begin{align}
 2\sum_{i=0}^{2}Y_i-\sum_{i=0}^{2}H_i=\frac{13}{8}\id=\mu\id.
\end{align}
Hence this is a dual-feasible point with objective value $\mu+\sum_{i=0}^{2}\lambda_i=\frac{13}{8}.$ As such, this choice yields
\begin{align}
 M_{\cDI,2}(\mE)\leq\frac{13}{8}=1.625 < 1.68 \approx \frac{37}{22} \overset{\eqref{eq:M_mDI_bound}}{\leq} M_{\mDI,2}(\mE),
\end{align}
which proves the separation between the two monotones.
\end{proof}

\end{document}